\RequirePackage[l2tabu, orthodox]{nag}

\documentclass{amsart}

\usepackage{mathtools}
\usepackage{etoolbox}
\usepackage{amssymb}
\usepackage[full]{textcomp}
\usepackage{bm}
\usepackage{slashed}
\usepackage{tikz-cd}
\usepackage{amscd}

\usepackage{mleftright}
\usepackage[T1]{fontenc}
\usepackage[sb,osf]{libertine} 
\usepackage{textcomp}
\usepackage[varqu,varl]{zi4}
\usepackage{libertinust1math}
\usepackage[scr=boondoxo,bb=boondox]{mathalfa} 

\usepackage{amsrefs}

\usepackage[unicode]{hyperref}

\newtheorem{theorem}{Theorem}[section]
\newtheorem{proposition}[theorem]{Proposition}
\newtheorem{lemma}[theorem]{Lemma}
\newtheorem{corollary}[theorem]{Corollary}

\theoremstyle{definition}

\newtheorem{definition}[theorem]{Definition}
\newtheorem{theoremdefinition}[theorem]{Theorem-Definition}

\newtheorem{propositiondefinition}[theorem]{Proposition-Definition}
\newtheorem{example}[theorem]{Example}

\newtheorem{remark}[theorem]{Remark}

\numberwithin{equation}{section}

\newcommand{\bA}{\mathbf{A}}

\newcommand{\bC}{\mathbf{C}}

\newcommand{\bN}{\mathbf{N}}
\newcommand{\bR}{\mathbf{R}}

\newcommand{\bT}{\mathbf{T}}

\newcommand{\bZ}{\mathbf{Z}}

\newcommand{\bF}{\mathbf{F}}

\newcommand{\cA}{\mathcal{A}}
\newcommand{\cC}{\mathcal{C}}

\newcommand{\cE}{\mathcal{E}}
\newcommand{\cF}{\mathcal{F}}
\newcommand{\cG}{\mathcal{G}}
\newcommand{\cH}{\mathcal{H}}

\newcommand{\cO}{\mathcal{O}{}}

\newcommand{\cS}{\mathcal{S}}
\newcommand{\cU}{\mathcal{U}}

\newcommand{\fg}{\mathfrak{g}}

\def\eu{\ensuremath{\mathrm{e}}}
\def\iu{\ensuremath{\mathrm{i}}}
\def\du{\ensuremath{\mathrm{d}}}

\DeclareMathOperator{\ad}{ad}
\DeclareMathOperator{\Ad}{Ad}
\DeclareMathOperator{\Aut}{Aut}
\DeclareMathOperator{\Der}{Der}

\DeclareMathOperator{\End}{End}
\DeclareMathOperator{\GL}{GL}
\DeclareMathOperator{\Hom}{Hom}
\DeclareMathOperator{\id}{id}
\DeclareMathOperator{\Ob}{Ob}
\DeclareMathOperator{\Proj}{Proj}
\DeclareMathOperator{\PSL}{PSL}

\DeclareMathOperator{\SL}{SL}

\DeclareMathOperator{\Span}{Span}

\DeclareMathOperator{\vol}{vol}
\DeclareMathOperator{\Unit}{U}
\DeclareMathOperator{\Int}{int}

\newcommand{\inj}{\hookrightarrow}
\newcommand{\iso}{\xrightarrow{\sim}}
\newcommand{\surj}{\twoheadrightarrow}

\newcommand{\inv}[1]{{{#1}^{-1}}}
\providecommand\given{}
\newcommand\SetSymbol[1][]{\nonscript\:#1\vert\nonscript\:\allowbreak}
\DeclarePairedDelimiterX\set[1]\{\}{%
\renewcommand\given{\SetSymbol[\delimsize]}
#1
}
\newcommand{\rest}[2]{\left.{#1}\right\rvert_{#2}}

\DeclarePairedDelimiterX\norm[1]\lVert\rVert{
\ifblank{#1}{\:\cdot\:}{#1}
}
\DeclarePairedDelimiterX\hp[2](){
\ifblank{#1}{\ifblank{#2}{\:\cdot\:,\:\cdot\:}{\:\cdot\:,#2}}{\ifblank{#2}{#1,\:\cdot\:}{#1,#2}}
}
\DeclarePairedDelimiterX\ip[2]\langle\rangle{
\ifblank{#1}{\ifblank{#2}{\:\cdot\:,\:\cdot\:}{\:\cdot\:,#2}}{\ifblank{#2}{#1,\:\cdot\:}{#1,#2}}
}
\DeclarePairedDelimiterX\abs[1]\lvert\rvert{
	\ifblank{#1}{\:\cdot\:}{#1}
}

\newcommand{\sa}{\mathrm{sa}}

\newcommand{\hotimes}{\mathbin{\widehat{\otimes}}}

\DeclareMathOperator{\can}{can}

\allowdisplaybreaks

\newcommand{\e}{\epsilon}

\newcommand{\fr}[1]{\mathfrak{#1}}

\newcommand{\ver}{\mathrm{ver}}
\newcommand{\hor}{\mathrm{hor}}
\DeclareMathOperator{\Inn}{Inn}
\DeclareMathOperator{\Out}{Out}
\DeclareMathOperator{\ran}{ran}

\newcommand{\cm}[2]{{#1}_{(#2)}}
\newcommand{\ca}[2]{{#1}_{[#2]}}
\newcommand{\p}[1]{{#1}^\prime}
\newcommand{\pp}[1]{{#1}^{\prime\prime}}
\newcommand{\act}{\triangleright}
\newcommand{\dt}[1]{\operatorname{d}_{{#1}}}
\newcommand{\dv}[1]{\operatorname{d}_{{#1},\ver}}

\newcommand{\coinv}[2]{\prescript{\operatorname{co}{#1}}{}{{#2}}}
\newcommand{\ract}{\triangleleft}

\newcommand{\dva}[1]{{#1}^{\mathrm{pr}}}

\newcommand{\pr}[1]{{#1}^{\mathrm{pr}}}
\newcommand{\tot}[1]{{#1}^{\mathrm{tot}}}

\DeclareMathOperator{\CS}{CS}
\DeclareMathOperator{\ZS}{ZS}
\DeclareMathOperator{\BS}{BS}
\DeclareMathOperator{\HS}{HS}
\DeclareMathOperator{\CH}{CH}
\DeclareMathOperator{\ZH}{ZH}
\DeclareMathOperator{\BH}{BH}
\DeclareMathOperator{\HH}{HH}

\DeclareMathOperator{\Cent}{C}
\DeclareMathOperator{\Zent}{Z}
\DeclareMathOperator{\MC}{MC}

\DeclareMathOperator{\Op}{Op}

\newcommand{\fodc}{\textsc{fodc}}
\newcommand{\sodc}{\textsc{sodc}}

\newcommand{\ddt}{\tfrac{\partial}{\partial t}}

\let\Im\relax
\DeclareMathOperator{\Im}{Im}

\newcommand{\dvatimes}{\widehat{\otimes}^{\leq 2}}

\newcommand{\rel}{\mathrm{rel}}

\begin{document}

\title{Classical gauge theory on quantum principal bundles}


\author{Branimir \'Ca\'ci\'c}
\address{Department of Mathematics and Statistics, University of New Brunswick, PO Box 4400, Fredericton, NB\, E3B 5A3, Canada}
\email{bcacic@unb.ca}



\date{}


\begin{abstract}
We propose a conceptually economical and computationally tractable completion of the foundations of gauge theory on quantum principal bundles \`{a} la Brzezi\'{n}ski--Majid to the case of general differential calculi and strong bimodule connections. In particular, we use explicit groupoid equivalences to reframe the basic concepts of gauge theory---gauge transformation, gauge potential, and field strength---in terms of reconstruction of calculi on the total space (to second order) from given calculi on the structure quantum group and base, respectively. We therefore construct gauge-equivariant moduli spaces of all suitable first- and second-order total differential calculi, respectively, compatible with these choices. As a first illustration, we relate the gauge theory of a crossed product algebra \emph{qua} trivial quantum principal bundle to lazy Sweedler and Hochschild cohomology with coefficients. As a second illustration, we show that a noncommutative \(2\)-torus with real multiplication is the base space of a non-trivial \(\Unit_{q^2}(1)\)-gauge theory admitting Connes's constant curvature connection as a \(q\)-monopole connection, where---in the spirit of Manin's Alterstraum---one must take \(q\) to be the norm-positive fundamental unit of the corresponding real quadratic field.
\end{abstract}

\maketitle

\tableofcontents

\section{Introduction and summary of results}

Noncommutative geometry permits the semiclassical modelling of quantum physics by classical physics on `quantum' (noncommutative) spaces, whose underlying geometry serves as both generator of and recepticle for quantum corrections; indeed, in much of theoretical physics, this might as well be the operational definition of noncommutative geometry itself. From this perspective, it may be surprising to note that a faithful and more-or-less complete formulation of classical Yang--Mills gauge theory on `quantum' principal bundles does not yet exist. Nonetheless, partial foundations have already been laid by Brzezi\'{n}ski--Majid~\cite{BrM}, who make sense of principal (Ehresmann) connections on quantum principal bundles in terms of suitable first-order differential calculi (\textsc{fodc}) on principal comodule algebras in a manner compatible with the theory of quantum groups and quantum homogeneous spaces.

There are gaps, however, in these partial foundations. Principal connections correspond to gauge potentials---their curvature, which would then encode field strength, is typically only accessed in a piecemeal and somewhat indirect fashion through the curvature of induced module connections on quantum associated vector bundles. More significantly, gauge transformations are generally only defined in the theoretically convenient but geometrically and physically unnatural special case of universal \fodc{}. Furthermore, to the author's best knowledge, almost all calculations of affine spaces of principal connections on quantum principal bundles are restricted to this special case; Brzezi\'{n}ski--Gaunt--Schenkel's computations on the \(\theta\)-deformed complex Hopf fibration~\cite{BGS} provide a valuable and instructive exception.

We aim to fill these gaps in a conceptually economical and computationally tractable fashion that readily interfaces with noncommutative Riemannian geometry, whether that mean Connes's functional-analytic theory of spectral triples~\citelist{\cite{Connes95}\cite{Connes96}} or, e.g., noncommutative K\"{a}hler geometry as championed by \'{O} Buachalla~\cite{OB}. We take our cue from \'{C}a\'{c}i\'{c}--Mesland~\cite{CaMe}, who construct a comprehensive theory of principal connections and global gauge transformations on noncommutative Riemannian principal bundles with compact connected Lie structure group in terms of spectral triples. For our purposes, their main formal innovation is to encode principal connections in terms of horizontal covariant derivatives (cf. \DJ{}ur\dj{}evi\'{c}~\cite{Dj97}), which permits a workable notion of gauge transformation acting affinely on principal connections in a manner that faithfully generalises the classical case. There is an apparent drawback to this change of perspective: neither variation of the principal connection nor application of a gauge transformation generally preserves the total \fodc{}. However, we shall have reason to embrace this as a feature, not a bug.


We begin in Section~\ref{sec2} by effecting the aforementioned  change of perspective at the level of \fodc{}. Let \(H\) be a Hopf \(\ast\)-algebra and let \(P\) be a principal left \(H\)-comodule \(\ast\)-subalgebra. Fix a horizontal calculus on \(P\), which therefore consists of a \fodc{} \((\Omega^1_B,\dt{B})\) on \(B \coloneqq \coinv{H}{P}\) and a suitable \emph{projectable horizontal lift} \(\Omega^1_{P,\hor} = P \cdot \Omega^1_B \cdot P = P \cdot \Omega^1_B\) of \(\Omega^1_B = \coinv{H}{\Omega^1_{P,\hor}}\). For us, a \emph{gauge transformation} is an \(H\)-covariant \(\ast\)-automorphism \(f\) of \(P\) satisfying \(\rest{f}{B} = \id_B\) (i.e., a \emph{vertical} \(\ast\)-automorphism) that induces via \(\id_{\Omega^1_B}\) an \(H\)-comodule automorphism \(f_{\ast,\hor}\) of \(\Omega^1_{P,\hor}\), while a \emph{gauge potential} is a lift of \(\dt{B} : B \to \Omega^1_B\) to an \(H\)-covariant \(\ast\)-derivation \(P \to \Omega^1_{P,\hor}\). The group \(\fr{G}\) of gauge transformations now acts affine-linearly on the real affine space \(\fr{At}\) of gauge potentials, whose space of translations \(\fr{at}\) consists of \emph{relative gauge potentials}. In turn, we can define the subgroup \(\Inn(\fr{G})\) of inner gauge transformations and the subspace \(\Inn(\fr{at})\) of inner relative gauge potentials and check that \(\Out(\fr{G}) \coloneqq \fr{G}/\Inn(\fr{G})\) still acts affine-linearly on \(\Out(\fr{At}) \coloneqq \fr{At}/\Inn(\fr{at})\).

We now relate these constructions to the standard theory. Fix a bicovariant \fodc{} \((\Omega^1_H,\dt{H})\) on \(H\) with respect to which \(P\) is \emph{locally free}. We simultaneously consider all \(H\)-covariant \fodc{} \((\Omega^1_P,\dt{P})\) on \(P\) that make \((P;\Omega^1_P,\dt{P})\) into a quantum principal \((H;\Omega^1_H,\dt{H})\)-bundle inducing the horizontal calculus \((\Omega^1_B,\dt{B};\Omega^1_{P,\hor})\) and admitting a \(\ast\)-preserving bimodule connection---for convenience, call such an \fodc{} \emph{admissible}. Indeed, let \(\cG[\Omega^1_H]\) be the groupoid of \emph{abstract gauge transformations}, whose objects are admissible \fodc{} on \(P\) and whose arrows are vertical \(\ast\)-automorphisms of \(P\) that are differentiable with respect to the source \fodc{} on the domain and the target \fodc{} on the codomain; let \(\mu[\Omega^1_H] : \cG[\Omega^1_H] \to \Aut(P)\) be the forgetful homomorphism. Likewise, let \(\cA[\Omega^1_H]\) be the set of all triples \((\Omega^1_P,\dt{P};\Pi)\), where \((\Omega^1_P,\dt{P})\) is admissible and where \(\Pi\) is a (necessarily strong) \(\ast\)-preserving bimodule connection on \((P;\Omega^1_P,\dt{P})\). Thus, the classical affine action of gauge transformations on principal Ehresmann connections generalises to the action of \(\cG[\Omega^1_H]\) on \(\cA[\Omega^1_H]\) by conjugation of bimodule connections. We can now summarise the non-trivial results of Section~\ref{sec2} as follows:

\begin{theorem}
	Reconstruction of \fodc{} induces an explicit equivalence of groupoids
	\[
		\Sigma[\Omega^1_H] : \fr{G} \ltimes \fr{At} \to \cG[\Omega^1_H] \ltimes \cA[\Omega^1_H]
	\]
	with explicit homotopy inverse manifesting the action groupoid \(\fr{G} \ltimes \fr{At}\) as a deformation retraction of the action groupoid \(\cG[\Omega^1_H] \ltimes \cA[\Omega^1_H]\). Furthermore:
	\begin{enumerate}
	\item the forgetful homomorphism \(\mu[\Omega^1_H] : \cG[\Omega^1_H] \to \Aut(P)\) maps surjectively onto \(\fr{G}\);
	\item the equivalence of groupoids \(\Sigma[\Omega^1_H]\) descends to a groupoid isomorphism \[\fr{G} \ltimes \fr{At}/\fr{at}[\Omega^1_H] \iso \cG[\Omega^1_H]/\ker\mu[\Omega^1_H],\] where \(\fr{at}[\Omega^1_H]\) is the \(\fr{G}\)-invariant subspace of \emph{\((\Omega^1_H,\dt{H})\)-adapted} relative gauge potentials.
\end{enumerate}
Hence, in particular, the quotient affine space \(\fr{At}/\fr{at}[\Omega^1_H]\) defines a \(\fr{G}\)-equivariant moduli space of  admissible \fodc{}.
\end{theorem}

\noindent This justifies our definition of gauge transformation and proves that we can work directly with the group \(\fr{G}\) and real affine space \(\fr{At}\) without any loss of geometrical or physical information, at least to first order.

Next, in Section~\ref{sec3}, we turn to making sense of curvature---field strength---in a conceptually minimalistic fashion. This requires the careful refinement of the constructions and results from Section~\ref{sec2} to the context of second-order differential calculi (\sodc{}), i.e., \(\ast\)-differential calculi truncated at degree \(2\). Suppose that we have prolonged our horizontal calculus to a \emph{second-order horizontal calculus} \((\Omega_B,\dt{B};\Omega_{P,\hor})\) on \(P\). A gauge transformation \(f\) is \emph{prolongable} if it extends via \(f_\ast : \Omega^1_{P,\hor} \to \Omega^1_{P,\hor}\) to an \(H\)-covariant \(\ast\)-automorphism of the graded left \(H\)-comodule \(\ast\)-algebra \(\Omega_{P,\hor}\); a gauge potential \(\nabla\) is \emph{prolongable} if it extends via \(\dt{B} : \Omega^1_B \to \Omega^2_B\) and \(0 : \Omega^2_{P,\hor} \to 0\) to an \(H\)-covariant degree \(1\) \(\ast\)-derivation of \(\Omega_{P,\hor}\), in which case its \emph{field strength} is \(\bF[\nabla] \coloneqq \rest{\nabla^2}{P}\). The subgroup \(\pr{\fr{G}}\) of prolongable gauge transformations continues to act affine-linearly on the affine subspace \(\pr{\fr{At}}\) of prolongable gauge transformations, whose space of translations \(\pr{\fr{at}}\) consists of \emph{prolongable} relative gauge potentials.

Now, suppose that we have prolonged the bicovariant \fodc{} \((\Omega^1_H,\dt{H})\) to a bicovariant \(\ast\)-differential calculus \((\Omega_H,\dt{H})\). We distill a proposal of Beggs--Majid~\cite{BeM}*{\S 5.5} into a notion of \emph{strong second-order quantum principal \((H;\Omega_H,\dt{H})\)-bundle} and propose a compatible notion of \emph{prolongable} \(\ast\)-preserving bimodule connection, which turns out to be related to \DJ{}ur\dj{}evi\'{c}'s notion of multiplicative connection~\cite{Dj97}. We can still define a groupoid \(\cG[\Omega^{\leq 2}_H]\) of \emph{prolongable} abstract gauge transformations on \(P\), a set \(\cA[\Omega^{\leq 2}_H]\) of prolongable \(\ast\)-preserving bimodule connections, and an action of \(\cG[\Omega^{\leq 2}_H]\) on \(\cA[\Omega^{\leq 2}_H]\) by conjugation. However, constructing an analogue \(\Sigma[\Omega^{\leq 2}_H]\) of the groupoid equivalence \(\Sigma[\Omega^1_H]\) turns out to be a subtle matter. 

As it turns out, we can still reconstruct from \(\Omega_{P,\hor}\) and \((\Omega_H)^{\mathrm{co}H}\) a canonical \(H\)-covariant graded \(\ast\)-algebra \(\Omega_{P,\oplus}\) of total differential forms on \(P\) through degree \(2\) (cf.\ \DJ{}ur\dj{}evi\'{c}~\cite{Dj10}). However, a prolongable gauge potential \(\nabla\) will yield an object \((\Omega_{P,\oplus},\dt{P,\nabla})\) of \(\cG[\Omega^{\leq 2}_H]\) if and only if its field strength factors as \(\bF[\nabla] = F[\nabla] \circ \dv{P}\), where \((\Omega^1_{P,\ver},\dv{P})\) is the first-order \emph{vertical calculus} of \(P\) induced by \((\Omega^1_H,\dt{H})\) and \(F[\nabla] : \Omega^1_{P,\ver} \to \Omega^2_{P,\hor}\) is a (necessarily unique) \(H\)-covariant morphism of \(P\)-\(\ast\)-bimodules. We say that such a gauge potential \(\nabla\) is \emph{\((\Omega^1_H,\dt{H})\)-adapted}, in which case, we call \(F[\nabla]\) its \emph{curvature \(2\)-form}; we denote the \(\pr{\fr{G}}\)-invariant quadric subset of all \((\Omega^1_H,\dt{H})\)-adapted prolongable gauge potentials by \(\pr{\fr{At}}[\Omega^1_H]\). Reconstruction of \sodc{} yields an explicit equivalence of groupoids
\[
	\Sigma[\Omega^{\leq 2}_H] : \pr{\fr{G}} \ltimes \pr{\fr{At}}[\Omega^1_H] \to \cG[\Omega^{\leq 2}_H] \ltimes \cA[\Omega^{\leq 2}_H]
\]
with explicit homotopy inverse manifesting the action groupoid \(\pr{\fr{G}} \ltimes \pr{\fr{At}}[\Omega^1_H]\) as a deformation retraction of the action groupoid \(\cG[\Omega^{\leq 2}_H] \ltimes \cA[\Omega^{\leq 2}_H]\), so that one can work directly with \(\pr{\fr{G}}\) and \(\pr{\fr{At}}[\Omega^1_H]\) without any loss of geometric or physical information, at least to second order. Once more, it follows that the obvious forgetful homomorphism \(\mu[\Omega^{\leq 2}_H] : \cG[\Omega^{\leq 2}_H] \to \Aut(P)\) maps surjectively onto \(\pr{\fr{G}}\), justifying our notion of prolongable gauge potential. However, the construction of a \(\pr{\fr{G}}\)-equivariant moduli space of admissible \sodc{} on \(P\) becomes considerably more involved:

\begin{theorem}
	Suppose that \((\Omega_H,\dt{H})\) is given through degree \(2\) by the canonical prolongation \`{a} la Woronowicz of \((\Omega^1_H,\dt{H})\). Then \(\Sigma[\Omega^{\leq 2}_H]\) descends to a groupoid isomorphism
	\[
		\pr{\fr{G}} \ltimes \pr{\fr{At}}[\Omega^1_H]/\pr{\fr{at}}_{\can}[\Omega^{1}_H] \iso \cG[\Omega^{\leq 2}_H]/\ker\mu[\Omega^{\leq 2}_H],
	\]
	where \(\pr{\fr{at}}_{\can}[\Omega^1_H]\) is the \(\pr{\fr{G}}\)-invariant subspace of \emph{canonically \((\Omega^1_H,\dt{H})\)-adapted} prolongable relative gauge potentials. Thus, the \(\pr{\fr{G}}\)-invariant quadric subset \(\pr{\fr{At}}[\Omega^1_H]/\pr{\fr{at}}_{\can}[\Omega^{1}_H]\) of the affine space \(\pr{\fr{At}}/\pr{\fr{at}}_{\can}[\Omega^{1}_H]\) defines a \(\pr{\fr{G}}\)-equivariant moduli space of admissible \sodc{}.
\end{theorem}

Next, in Section~\ref{sec4}, we illustrate our definitions and results in the case of crossed product algebras, which one views as trivial quantum principal bundles. As \'{C}a\'{c}i\'{c}--Mesland showed in the context of spectral triples~\cite{CaMe}, the group \(\fr{G}\) of gauge transformations and the affine space \(\fr{At}\) of gauge potentials of a crossed product by \(\bZ^n\) \emph{qua} trivial quantum principal \(\bT^n\)-bundle can be computed in terms of the degree \(1\) group cohomology of \(\bZ^n\) with coefficients in a certain group of unitaries and a certain \(\bR[\bZ^m]\)-module of noncommutative \(1\)-forms, respectively. Generalising their calculations from the commutative and cocommutative Hopf \(\ast\)-algebra \(\bC[\bZ^m] \cong \cO(\bT^m)\) to arbitary Hopf \(\ast\)-algebras requires the construction of novel \emph{ad hoc} degree \(1\) cohomology groups, which we term `lazy' on account of their formal resemblance to a construction of Bichon--Carnovale~\cite{BC}. These generalisations reduce appropriately to conventional group cohomology in the case of a group algebra and to Lie algebra cohomology in the case of the universal enveloping algebra of a real Lie algebra.

Let \(H\) be a Hopf \(\ast\)-algebra, let \(B\) be a right \(H\)-module \(\ast\)-algebra, and let \(M\) be an \(H\)-equivariant \(B\)-\(\ast\)-bimodule. On the one hand, we can define the degree \(1\) \emph{lazy \((B,M)\)-valued Sweedler cohomology} 
\(
	\HS^1_\ell(H;B,M) \coloneqq \ZS^1_\ell(H;B,M)/\BS^1_\ell(H;B,M)
\)
of \(H\), where \(\ZS^1_\ell(H;B,M)\) is the group of \emph{lazy Sweedler \(1\)-cocycles} and \(\BS^1_\ell(H;B,M)\) is the central subgroup of \emph{lazy Sweedler \(1\)-coboundaries}; this generalises Sweedler cohomology~\cite{Sweedler} through degree \(1\) to the case of not-necessarily-cocommutative Hopf \(\ast\)-algebras and non-trivial coefficients. On the other hand, we can refine the degree \(1\) Hoschchild cohomology of \(H\) with coefficients in \(M\) with the given right \(H\)-action and the trivial left \(H\)-action to degree \(1\) \emph{lazy \(M\)-valued Hoschchild cohomology}
\(
	\HH^1_\ell(H;M) \coloneqq \ZH^1_\ell(H;M)/\BH^1_\ell(H;M),
\)
where \(\ZH^1_\ell(H;M)\) is the real vector space of \emph{lazy Hochschild \(1\)-cocycles} and \(\BH^1_\ell(H;M)\) is the subspace of \emph{lazy Hochschild \(1\)-coboundaries}. Conjugation with respect to convolution on \(H\) now yields a representation of \(\ZS^1_\ell(H;B,M)\) on \(\ZH^1_\ell(H;M)\) that descends to a representation of \(\HS^1_\ell(H;B,M)\) on \(\HH^1_\ell(H;M)\). Furthermore, any \(H\)-equivariant \(\ast\)-derivation \(\partial : B \to M\) induces a canonical group \(1\)-cocycle \(\MC[\partial] : \ZS^1_\ell(H;B,M) \to \ZH^1_\ell(H;M)\) that descends, in turn, to a group \(1\)-cocycle \(\widetilde{MC}[\partial] : \HS^1_\ell(H;B,M) \to \HH^1_\ell(H;M)\).

We can now exemplify the use of lazy Sweedler and Hochschild cohomology to compute the group \(\fr{G}\), the real affine space \(\fr{At}\), and the affine-linear action of \(\fr{G}\) on \(\fr{At}\) for a trivial quantum principal \(H\)-bundle \(B \rtimes H\); note that none of this structure would be visible if we did not consider all possible admissible \fodc{} simultaneously.

\begin{proposition}
	Let \(P \coloneqq B \rtimes H\), so that \(B = \coinv{H}{P}\). Suppose that \((\Omega^1_B,\dt{B})\) is a an \(H\)-equivariant \fodc{} on \(B\), and endow \(P\) with the horizontal calculus \((\Omega^1_B,\dt{B};\Omega^1_B \rtimes H)\). We have a group isomorphism  \(\Op : \ZS^1_\ell(H;B,\Omega^1_B) \to \fr{G}\) and an isomorphism \(\Op : \ZH^1_\ell(H;\Omega^1_B) \iso \fr{At}\) of real affine spaces given, respectively, by
	\begin{gather*}
		\forall \sigma \in \ZS^1_\ell(H;B,\Omega^1_B), \, \forall h \in H, \, \forall b \in B, \quad \Op(\sigma)(hb) \coloneqq \cm{h}{1}\sigma(\cm{h}{2})b,\\
		\forall \mu \in \ZH^1_\ell(H;\Omega^1_B), \, \forall h \in H, \, \forall b \in B, \quad \Op(\mu)(hb) \coloneqq h \cdot \dt{B}(b) + \cm{h}{1} \cdot \mu(\cm{h}{2}) \cdot b;
	\end{gather*}
	these descend, respectively, to a group isomorphism \(\widetilde{\Op} : \HS^1_\ell(H;B,\Omega^1_B) \iso \Out(\fr{G})\) and an affine-linear isomorphism \(\widetilde{\Op} : \HH^1_\ell(H;\Omega^1_B) \iso \Out(\fr{At})\). Moreover for every lazy Sweedler \(1\)-cocycle \(\sigma \in \ZS^1_\ell(H;B,\Omega^1_B)\) and lazy Hochschild \(1\)-cocycle \(\mu \in \ZH^1_\ell(H;\Omega^1_B)\),
	\begin{gather*}
		\Op(\sigma) \act \Op(\mu) = \Op(\sigma \act \mu + \MC[\dt{B}](\sigma)),\\
		\widetilde{\Op}([\sigma]) \act \widetilde{\Op}([\mu]) = \widetilde{\Op}\mleft([\sigma]\act[\mu]+\widetilde{\MC}[\dt{B}]([\sigma])\mright).
	\end{gather*}
\end{proposition}

\noindent We can similarly compute the second-order gauge theory of \(B \rtimes H\) in terms of suitable further refinements of lazy Sweedler and lazy Hochschild cohomology. As we shall show in forthcoming work~\cite{CaTi}, one can similarly analyse general quantum principal \(\Unit(1)\)-bundles with commutative base space \(C^\infty(M)\)  in terms of the degree \(1\) group cohomology of \(\bZ\) with coefficients in \(C^\infty(M,\Unit(1))\) and \(\Omega^1_{\mathrm{dR}}(M,\iu{}\bR)\); although \(\fr{G} \ltimes \fr{At}\) will only depend on \(M\), all other groupoids of interest will turn out to be sensitive to the underlying dynamics.

Finally, in Section~\ref{sec5}, we apply our formalism to the example of the non-trivial quantum principal \(\cO(\Unit(1))\)-bundle implicit to Manin's `Alterstraum'~\cite{Manin}. Let \(\theta \in \bR\) be a quadratic irrationality, and let \(\cA_\theta\) be the corresponding smooth noncommutative \(2\)-torus endowed with the canonical \sodc{} \((\Omega_{\cA_\theta},\dt{\cA_\theta})\). Then, by combining results of Schwarz~\cite{Schwarz98}, Dieng--Schwarz~\cite{DiengSchwarz}, Polishchuk--Schwarz~\cite{PolishchukSchwarz}, Po\-li\-shchuk \cite{Polishchuk}, and Vlasenko~\cite{Vlasenko}, we canonically assemble the self-Morita equivalence bimodules among the basic Heisenberg modules over \(\cA_\theta\) into a non-cleft principal \(\cO(\Unit(1))\)-comodule \(\ast\)-algebra \(P\) with base \(\coinv{\cO(\Unit(1))}{P} = \cA_\theta\). Moreover, using results of Poli\-shchuk--Schwarz~\cite{PolishchukSchwarz}, we canonically assemble Connes's constant curvature connections~\cite{Connes80} on the isotypical components into a prolongable gauge potential \(\nabla_0\) on \(P\) with respect to a certain second-order horizontal calculus constructed from \((\Omega_{\cA_\theta},\dt{\cA_\theta})\). From there, we can show that \(\fr{G} = \pr{\fr{G}} \cong \Unit(1)\), where \(\Unit(1)\) acts as the structure group on \(P\), that \(\Inn(\fr{G}) = \set{1}\), that every relative gauge potential is inner prolongable and given by supercommutation by a \(1\)-form in \(\bR^2 \subset \cA_{\theta}^{\oplus 2} = \Omega^1_{\cA_\theta}\), and that \(\fr{G} = \pr{\fr{G}}\) acts trivially on \(\fr{At} \cong \bR^2\). In particular, we can prove the following:

\begin{theorem}
	Let \(\e  = c_1\theta + d_1\) be the norm-positive fundamental unit of \(\mathbf{Q}[\theta]\), where \(c_1 \in \bN\) and \(d_1 \in \bZ\) are uniquely determined. Given \(q \in \bR^\times\), let \((\Omega^1_q,\dt{q})\) denote the corresponding \(q\)-deformed bicovariant \fodc{} on \(\cO(\Unit(1))\), so that \((\Omega^1_1,\dt{1})\) is the de Rham calculus. Then
	\[
		\forall q \in \bR^\times, \quad \pr{\fr{At}}[\Omega^1_q] = \begin{cases} \fr{At} &\text{if \(q = \e^2\),}\\ \emptyset &\text{else,}\end{cases}
	\]
	and for every \(\nabla \in \fr{At} = \pr{\fr{At}}[\Omega^1_{\e^2}]\), the curvature \(2\)-form \(F[\nabla]\) is non-zero and given by
		\[
			F[\nabla](\dt{\e^2}t) = - \iu{}\e c_1\vol_{\cA_\theta},
		\]
		where \(\dt{\e^2}t \coloneqq \tfrac{1}{2\pi\iu{}}\dt{\e^2}(z) \cdot z^{-1}\) and where \(\vol_{\cA_{\theta}} \coloneqq 1_{\cA_\theta} \in \cA_{\theta} = \Omega^2_{\cA_\theta}\). Furthermore, 
	\[
		\forall q \in \bR^\times, \quad \fr{at}[\Omega^1_q] = \pr{\fr{at}}_{\can}[\Omega^{1}_q] = \begin{cases} \fr{at} &\text{if \(q = \e\),}\\ 0 &\text{else,}\end{cases}
	\]
	so that \(\fr{At}/\fr{at}[\Omega^1_{\e^2}] = \pr{\fr{At}}[\Omega^1_{\e^2}]/\pr{\fr{at}}_{\can}[\Omega^{1}_{\e^2}] = \fr{At} \cong \bR^2\).
\end{theorem}

\noindent Thus, when \(q = \e\), we obtain a \(q\)-monopole over \(\cA_\theta\) exactly analogous to the \(q\)-monopole constructed by Brzezi\'{n}ski--Majid~\cite{BrM} on the \(q\)-deformed complex Hopf fibration; furthermore, distinct gauge potentials are gauge-inequivalent and yield non-isomorphic admissible \fodc{}. Note that the elementary algebraic number theory of the real quadratic irrational \(\theta\) appears again and again in all constructions and calculations.

\addtocontents{toc}{\protect\setcounter{tocdepth}{0}}
\subsection*{Notation and conventions} We shall mostly follow the notation, terminology, and conventions of Beggs--Majid~\cite{BeM} with certain exceptions. In this work, unless otherwise stated, all algebras are unital \(\bC\)-algebras and all modules are unital modules.

We shall use Sweedler notation as follows. If \(h\) is an element of a bialgebra \(H\), its coproduct will be denoted by \(\Delta(h) \eqqcolon \cm{h}{1} \otimes \cm{h}{2} \in H \otimes H\). If \(p\) is an element of a left \(H\)-comodule \(P\), we denote its left \(H\)-coaction by \(\delta(p) \eqqcolon \ca{p}{-1} \otimes \ca{p}{0} \in H \otimes P\), and we denote the subspace of left \(H\)-coinvariants of \(P\) by \(\coinv{H}{P}\); similarly, if \(q\) is an element of a right \(H\)-comodule \(Q\), we denote its right \(H\)-coaction by \(\delta(q) \eqqcolon \ca{q}{0} \otimes \ca{q}{1} \in P \otimes H\), and we denote the subspace of right \(H\)-coinvariants of \(Q\) by \(Q^{\mathrm{co}H}\). We shall also use the corresponding higher-order Sweedler notation, e.g.,
\begin{align*}
	\forall h \in H, \quad (\id_H \otimes \Delta) \circ \Delta(h) &= (\Delta \otimes \id) \circ \Delta(h) \eqqcolon \cm{h}{1} \otimes \cm{h}{2} \otimes \cm{h}{3},\\
	\forall p \in P, \quad (\id \otimes \delta) \circ \delta(p) &= (\Delta \otimes \id) \circ \delta(p) \eqqcolon \ca{p}{-2} \otimes \ca{p}{-1} \otimes \ca{p}{0},\\
	\forall q \in Q, \quad (\delta \otimes \id) \circ \delta(q) &= (\id \otimes \Delta) \circ \delta(q) \eqqcolon \ca{q}{0} \otimes \ca{q}{1} \otimes \ca{q}{2}.
\end{align*}

We use the following conventions related to bimodules. Given a \(\ast\)-algebra \(B\), a \emph{\(B\)-\(\ast\)-bimodule} is a \(B\)-bimodule \(M\)  with a \(\bC\)-antilinear involution \(\ast : M \to M\) satisfying
\[
	\forall b_1,b_2 \in B,\,\forall m \in M, \quad (b_1 \cdot m \cdot b_2)^\ast = b_2^\ast \cdot m^\ast \cdot b_1^\ast.
\]
In this case, we set \(M_{\sa} \coloneqq \set{m \in M \given m^\ast = m}\) and define the \emph{centre} of \(M\) with respect to \(B\) by
\(
	\Zent_B(M) \coloneqq \set{m \in M \given \forall b \in B,\, b \cdot m = m \cdot b};
\)
moreover, given a subset \(S \subset M\), we define the \emph{centraliser} of \(S\) in \(B\) by
\(
	\Cent_B(S) \coloneqq \set{b \in B \given \forall m \in S,\, b \cdot m = m \cdot b}
\). Given a \(\ast\)-algebra \(B\) and a \(B\)-\(\ast\)-module \(M\), we say that a derivation \(\partial : B \to M\) is a \emph{\(\ast\)-derivation} whenever
\[
	\forall b \in B, \quad \partial(b^\ast) = -\partial(b)^\ast,
\]
and we denote by \(\Der_B(M)\) the \(\bR\)-vector space of all \(\ast\)-derivations \(B \to M\); given \(m \in M\), the resulting \emph{inner} \(\ast\)-derivation \(\ad_m \in \Der_B(M)\) is defined by setting
\[
	\forall b \in B, \quad \ad_m(b) \coloneqq [m,b].
\]
Note that Beggs--Majid use the opposite convention, where a \(\ast\)-derivation is \(\ast\)-preserving (i.e., intertwines \(\bC\)-antilinear involutions). A \(\ast\)-algebra \(B\) admits an obvious category of \(B\)-\(\ast\)-bimodules, where a morphism is a left and right \(B\)-linear map that is \(\ast\)-preserving.

Finally, let us set the following terminology and conventions related to differential calculi. A \emph{graded \(\ast\)-algebra} is a \(\bZ\)-graded algebra \(\Omega = \bigoplus_{k \in \bZ} \Omega^k\) with \(\Omega^k = 0\) for \(k < 0\) together with a \(\bC\)-linear involution \(\ast : \Omega \to \Omega\) such that \(\ast(\Omega^k) = \Omega^k\) for all \(k \in \bZ\) and
\[
	\forall m,n \in \bZ_{\geq 0}, \, \forall \alpha \in \Omega^m, \, \forall \beta \in \Omega^n, \quad (\alpha \wedge \beta)^\ast = (-1)^{mn} \beta^\ast \wedge \alpha^\ast.
\]
Given a graded \(\ast\)-algebra \(\Omega\) and \(k \in \bZ\), a \emph{degree \(k\) \(\ast\)-derivation} is a \(\bC\)-linear map \(\partial : \Omega \to \Omega\) satisfying \(\partial(\Omega^m) \subseteq \Omega^{m+k}\) for \(m \in \bZ_{\geq 0}\) and
\begin{gather*}
	\forall m,n \in \bZ_{\geq 0}, \, \forall \alpha \in \Omega^m, \, \forall \beta \in \Omega^n, \quad \partial(\alpha \wedge \beta) = \partial(\alpha) \wedge \beta + (-1)^{km} \alpha \wedge \partial(\beta),\\
	\forall \alpha \in \Omega, \quad \partial(\alpha^\ast) = -\partial(\alpha)^\ast.
\end{gather*}
Hence, given a \(\ast\)-algebra \(B\), a \emph{\(\ast\)-differential calculus} on \(B\) is a pair \((\Omega_B,\dt{B})\), where \(\Omega_B\) is a graded \(\ast\)-algebra with \(\Omega^0_B = B\) and \(\dt{B} : \Omega_B \to \Omega_B\) is a degree \(1\) \(\ast\)-derivation with \(\dt{B}^2 = 0\), such that \(\Omega_B\) is generated over \(B\) by \(\dt{B}(B) \subset \Omega^1_B\); in this case, we say that \((\Omega_B,\dt{B})\) is a \emph{prolongation} of the first-order differential calculus (\fodc{}) \((\Omega^1_B,\dt{B})\) on \(B\). In particular, a \emph{second-order differential calculus} (\emph{\sodc{}}) on a \(\ast\)-algebra \(B\) is a \(\ast\)-differential calculus \((\Omega_B,\dt{B})\) on \(B\), such that \(\Omega_B\) is truncated at degree \(2\), i.e., \(\Omega^k_B = 0\) for \(k > 2\).

\subsection*{Acknowledgements} The author is grateful to Edwin Beggs, Viqar Husain, Will Jagy, Andrey Krutov, R\'{e}amonn \'{O} Buachalla, Matilde Marcolli, Bram Mesland, Pavle Pand\v{z}i\'{c}, Karen Strung, and V.\ Karthik Timmavajjula for useful comments and conversations. This research was supported by NSERC Discovery Grant RGPIN-2017-04249 and by a Harrison McCain Foundation Young Scholar Award.

\addtocontents{toc}{\protect\setcounter{tocdepth}{2}}
\section{Gauge theory to first order}\label{sec2}

In this section, we use explicit groupoid equivalences to [re]formulate the notions of gauge transformation and principal connection on quantum principal bundles in terms of horizontal covariant derivatives, thereby yielding a computationally tractable generalisation of the gauge action on principal connections to quantum principal bundles with general first-order differential calculi (\textsc{fodc}); in the process, we obtain a gauge-equivariant moduli space of all relevant \fodc{} inducing the same vertical and horizontal calculi.

\subsection{Deconstruction of quantum principal bundles to first order}\label{firstorderdecon}

Let \(H\) be a Hopf \(\ast\)-algebra over \(\bC\), and let \((\Omega^1_H,\dt{H})\) be a bicovariant \fodc{} on \(H\). We give a novel review of the theory of quantum principal \(H\)-bundles \`{a} la Brzezi\'{n}ski--Majid~\cite{BrM} incorporating recent insights of Beggs--Majid~\cite{BeM} while recovering earlier insights of \DJ{}ur\dj{}evi\'{c}~\cite{Dj97}. In particular, we shall see how a strong bimodule connection decomposes  the total \textsc{fodc} of a quantum principal \(H\)-bundle into the direct sum of independent vertical and horizontal calculi, where the latter can be viewed as a horizontal lift of the induced \textsc{fodc} on the base.

Let us first recall the relevant notion of topological quantum principal \(H\)-bundle.

\begin{definition}[Brzezi\'{n}ski--Hajac~\cite{BH}]
	A left \(H\)-comodule \(\ast\)-algebra \(P\) is \emph{principal} if and only if both of the following conditions hold:
\begin{enumerate}
	\item the canonical map \(P \otimes P \to H \otimes P\) defined by \(p \otimes \p{p} \mapsto \ca{p}{-1} \otimes \ca{p}{0}\p{p}\) descends to a bijection \(P \otimes_B P \to H \otimes P\);
	\item there exists a unital bicovariant map \(\omega : H \to P \otimes P\), such that \(m_P \circ \omega = \epsilon(\cdot) 1_P\), where \(m_P : P \otimes P \to P\) is multiplication in  \(P\) and \(\epsilon\) is the counit of \(H\).
\end{enumerate}
\end{definition}

\begin{example}[Schneider~\cite{Schneider}]
	If the Hopf \(\ast\)-algebra \(H\) is cosemisimple, e.g., if \(H = \cO(G)\) for \(G\) a compact Lie group or \(H = \bC[\Gamma]\) for \(\Gamma\) a discrete group, then \(P\) is principal if and only if its canonical map is surjective. 
\end{example}

\begin{remark}[Baum--De Commer--Hajac~\cite{BDH}]
	One can interpret principality of a left \(H\)-co\-module \(\ast\)-algebra \(P\) as topological freeness of the \(H\)-coaction on \(P\) together with vestigial local triviality of the topological quantum principal \(H\)-bundle \(P\) to the extent that quantum associated vector bundles are finitely generated and projective as \(\coinv{H}{P}\)-modules \cite{DGH}*{Cor.\ 2.6}.
\end{remark}

We now consider differentiable quantum principal \(H\)-bundles compatible with the given bicovariant \textsc{fodc} \((\Omega^1_H,\dt{H})\) on \(H\). To this end, we will find it convenient to encode \((\Omega^1_H,\dt{H})\) in terms of the \(\ast\)-closed left \(H\)-subcomodule
\begin{equation}
	\Lambda^1_H \coloneqq (\Omega^1_H)^{\operatorname{co}H}
\end{equation}
of right \(H\)-covariant \(1\)-forms and \emph{quantum Maurer--Cartan form} \(\varpi_H : H \to \Lambda^1_H\) given by
\begin{equation}
	\forall h \in H, \quad \varpi_H(h) \coloneqq \dt{H}(\cm{h}{1})\cdot S(\cm{h}{2});
\end{equation}
the relevant properties of \(\Lambda^1_H\) and \(\varpi_H\) are given by the following definition.

\begin{definition}
	A \emph{left crossed \(H\)-\(\ast\)-module} is a left \(H\)-module and comodule \(V\) over \(\bC\) together with a conjugate-linear involution \(\ast : V \to V\), such that
\begin{gather*}
	\forall h \in H, \, \forall v \in V, \quad \delta(h \act v) = \cm{h}{1} \ca{v}{-1} S(\cm{h}{3}) \otimes \cm{h}{2} \act \ca{v}{0},\\
	\forall h \in H, \, \forall v \in V, \quad (h \act v)^\ast = S(h)^\ast \act v^\ast,\\
	\forall v \in V, \quad  \delta(v^\ast) = (\ca{v}{-1})^\ast \otimes (\ca{v}{0})^\ast;
\end{gather*}
in this case, a \(V\)-valued \emph{\(1\)-cocycle} is a \(\bC\)-linear map \(\varpi : H \to V\) satisfying
\begin{gather*}
	\forall h,k \in H, \quad \varpi(hk) = h \act \varpi(k) + \varpi(h)\epsilon(k),\\
	\forall h \in H, \quad \varpi(h)^\ast = \varpi(S(h)^\ast),
\end{gather*}
which we call \emph{\(\Ad\)-covariant} whenever it also satisfies
\[
	\forall h \in H, \quad \delta(\varpi(h)) = \cm{h}{1}S(\cm{h}{3}) \otimes \varpi(\cm{h}{2}),
\]
\end{definition}

One can now show that \(\Lambda^1_H\) defines a left crossed \(H\)-\(\ast\)-module with respect to the left \emph{adjoint} action of \(H\) given by
\begin{equation}
	\forall h \in H, \, \forall \omega \in \Lambda^1_H, \quad h \act \omega \coloneqq \cm{h}{1} \cdot \omega \cdot S(\cm{h}{2})
\end{equation}
and that \(\varpi_H\) defines a surjective \(\Lambda^1_H\)-valued \(\Ad\)-covariant \(1\)-cocycle. Furthermore, one can show that the bicovariant \textsc{fodc} \((\Omega^1_H,\dt{H})\) can be recovered from the data \((\Lambda^1_H,\varpi_H)\) up to isomorphism~\cite{BeM}*{Proof of Thm.\ 2.26}. The proof of this fact, \emph{mutatis mutandis}, permits the following noncommutative generalisation---essentially due to \DJ{}ur\dj{}evi\'{c}---of a locally free action of a connected Lie group and its orbitwise differential calculus.

\begin{definition}[cf.\ \DJ{}ur\dj{}evi\'{c}~\cite{Dj97}*{Lemma 3.1}]
	The \emph{vertical calculus} of a left \(H\)-comodule \(\ast\)-algebra \(P\) with respect to the bicovariant \textsc{fodc} \((\Omega^1_H,\dt{H})\) is the pair \((\Omega^1_{P,\ver},\dv{P})\), where:
	\begin{enumerate}
		\item \(\Omega^1_{P,\ver} \coloneqq \Lambda^1_H \otimes P\) is a left \(H\)-comodule \(P\)-\(\ast\)-bimodule with respect to the left \(H\)-coaction, \(P\)-bimodule structure, and \(\ast\)-structure given respectively by
		\begin{gather}
		\forall p \in P, \, \forall \omega \in \Lambda^1_H, \quad \delta(\omega \otimes p) \coloneqq \ca{\omega}{-1}\ca{p}{-1} \otimes \ca{\omega}{0} \otimes \ca{p}{0},\\
		\forall p,q,\p{q} \in P, \, \forall \omega \in \Lambda^1_H, \quad q \cdot (\omega \otimes p) \cdot \p{q} \coloneqq \ca{q}{-1} \act \omega  \otimes \ca{q}{0}p\p{q},\\
		\forall \omega \in \Lambda^1_H, \, \forall p \in P, \quad (\omega \otimes p)^\ast \coloneqq \ca{p}{-1}^\ast \act \omega^\ast  \otimes \ca{p}{0}^\ast;
	\end{gather}
		\item \(\dv{P} : P \to \Omega^1_{P,\ver}\) is the left \(H\)-covariant \(\ast\)-derivation defined by
		\begin{equation}
		\forall p \in P, \quad \dv{P}(p) \coloneqq (\varpi_H \otimes \id) \circ \delta(p) = \varpi_H(\ca{p}{-1}) \otimes \ca{p}{0}.
	\end{equation}
	\end{enumerate}
	We say that the bicovariant \textsc{fodc} \((\Omega^1_H,\dt{H})\) is \emph{locally freeing} for \(P\) whenever 
	\begin{equation}
		\Omega^1_{P,\ver} = P \cdot \dv{P}(P),
	\end{equation}
	so that \((\Omega^1_{P,\ver},\dv{P})\) defines a left \(H\)-covariant \textsc{fodc} on \(P\).
\end{definition}

\begin{example}
	If \(P\) is a principal \(H\)-comodule \(\ast\)-algebra, then the universal \fodc{} on \(H\) is locally freeing for \(P\).
\end{example}

\begin{remark}
	If \((\Omega^1_H,\dt{H})\) is locally freeing for \(P\), we can interpret \((\Omega^1_{P,\ver},\dv{P})\) as the orbitwise differential calculus of the locally free \(H\)-coaction on \(P\); in particular, we can interpret \(\dv{P} : P \to \Omega^1_{P,\ver}\) as the dualised ``infinitesimal coaction'' of \(H\) on \(P\) with respect to \((\Omega^1_H,\dt{H})\).
\end{remark}

We can now give a suitable formulation of the standard notion of differentiable quantum principal \(H\)-bundle. From now on, let \(P\) be a principal left \(H\)-comodule \(\ast\)-algebra with \(\ast\)-subalgebra of coinvariants \(B \coloneqq \coinv{H}{P}\).

\begin{definition}[Brzezi\'{n}ski--Majid~\cite{BrM}*{Def.\ 4.9}, cf.\ Beggs--Majid~\cite{BeM}*{\S 5.4}]
	Let \((\Omega^1_P,\dt{P})\) be a left \(H\)-covariant \textsc{fodc} on \(P\). Then \((P;\Omega^1_P,\dt{P})\) defines a \emph{quantum principal \((H;\Omega^1_H,\dt{H})\)-bundle}  if and only if the \emph{vertical map} \(\ver[\dt{P}] : \Omega^1_P \to \Omega^1_{P,\ver}\) given by
	\begin{equation}\label{verdef}
		\forall p, \p{p} \in P, \quad \ver[\dt{P}]\mleft(p \cdot \dt{P}(\p{p})\mright) \coloneqq p \cdot \dv{P}(\p{p})
	\end{equation}
	is well-defined and surjective with kernel \(P \cdot \dt{P}(B) \cdot P\), where \(B \coloneqq \coinv{H}{P}\); in this case, we say that \((\Omega^1_P,\dt{P})\) is \emph{\((H;\Omega^1_H,\dt{H})\)-principal}.
\end{definition}

\begin{example}[Brzezi\'{n}ski--Majid~\cite{BrM}*{\S 4}]
	Let \((\Omega^1_{H,u},\dt{H,u})\) be the universal \fodc{} on \(H\). Then the universal \fodc{} on \(P\) is \((H;\Omega^1_{H,u},\dt{H,u})\)-principal.
\end{example}

\begin{remark}[Brzezi\'{n}ski--Majid~\cite{BrM}*{\S 4.1}]
	Suppose that \((\Omega^1_P,\dt{P})\) on \(P\) is an \((H;\Omega^1_H,\dt{H})\)-principal \fodc{} on \(P\). We view its vertical map \(\ver[\dt{P}] : \Omega^1_{P} \to \Omega^1_{P,\ver}\) as encoding contraction of \(1\)-forms with fundamental vector fields, so that \(\ker\ver[\dt{P}] = P \cdot \dt{P}(B) \cdot P\) is the \(P\)-\(\ast\)-bimodule of horizontal \(1\)-forms; note that \(\ker[\dt{P}]\) is correctly generated as a \(P\)-bimodule by the basic \(1\)-forms \(B \cdot \dt{P}(B)\).
\end{remark}

\begin{remark}\label{locallyfreeremark}
	If \(P\) admits an \((H;\Omega^1_H,\dt{H})\)-principal \fodc{} \((\Omega^1_P,\dt{P})\), then \((\Omega^1_H,\dt{H})\) is locally freeing for \(P\). Indeed, since \(\dv{P} = \ver[\dt{P}] \circ \dt{P}\), it follows that
	\[
		\Omega^1_{P,\ver} = \ver[\dt{P}](\Omega^1_P) = \ver[\dt{P}](P \cdot \dt{P}) = P \cdot \dv{P}(P).
	\]
	Thus, following \DJ{}ur\dj{}evi\'{c}~\cite{Dj97}*{\S 3}, we can also view \(\ver[\dt{P}]\) as encoding restriction of \(1\)-forms to orbitwise \(1\)-forms.
\end{remark}

The following now gives the most commonly used notion of principal connection on a differentiable quantum principal \(H\)-bundle.

\begin{definition}[Brzezi\'{n}ski--Majid~\cite{BrM}*{\S 4.2}, Hajac~\cite{Hajac}*{Def.\ 2.1}, Beggs--Majid~\cite{BeM}*{\S 5.4}]
	Suppose that \((\Omega^1_P,\dt{P})\) is an \((H;\Omega^1_H,\dt{H})\)-principal \textsc{fodc} on the principal left \(H\)-comodule \(\ast\)-algebra \(P\). A \emph{connection} on the quantum principal \((H;\Omega^1_H,\dt{H})\)-bundle \((P;\Omega^1_P,\dt{P})\) is a left \(H\)-covariant left \(P\)-linear map \(\Pi : \Omega^1_P \to \Omega^1_P\) satisfying \[\Pi^2 = \Pi, \quad \ker\Pi = \ker\ver[\dt{P}];\] in this case, \(\Pi\) is a \emph{bimodule connection} if and only if it is right \(P\)-linear and \(\ast\)-preserving, and it is \emph{strong} if and only if 
	\begin{equation}\label{strong}
		(\id-\Pi) \circ \dt{P}(P) \subseteq P \cdot \dt{P}(B).
	\end{equation}
\end{definition}

\begin{remark}[cf.\ Atiyah~\cite{Atiyah57}]\label{sesremark}
	Suppose that \((\Omega^1_P,\dt{P})\) is an \((H;\Omega^1_H,\dt{H})\)-principal \fodc{} on \(P\), so that we have a short exact sequence
	\begin{equation}\label{ses}
		0 \to P \cdot \dt{P}(B) \cdot P \to \Omega^1_P \xrightarrow{\ver[\dt{P}]} \Omega^1_{P,\ver}\to 0
	\end{equation}
	of left \(H\)-comodule \(P\)-\(\ast\)-bimodules generalising the Atiyah sequence of a smooth principal bundle~\cite{Atiyah57}. Then a connection \(\Pi\) corresponds to a splitting of~\eqref{ses} in the category of left \(H\)-comodule left \(P\)-modules, which is a splitting in the category of left \(H\)-comodule \(P\)-\(\ast\)-bimodules if and only if \(\Pi\) is a bimodule connection. In particular, a connection \(\Pi\) induces the right splitting \((\ver[\dt{P}] \circ \Pi)^{-1}\) and the left splitting \(\id-\Pi\).
\end{remark}

\begin{remark}
	Suppose that \((\Omega^1_P,\dt{P})\) is an \((H;\Omega^1_H,\dt{H})\)-principal \fodc{} on \(P\). If the quantum principal \((H;\Omega^1_H,\dt{H})\)-bundle \((P;\Omega^1_P,\dt{P})\) admits a strong connection, then every connection on \((P;\Omega^1_P,\dt{P})\) is strong.
\end{remark}

\begin{remark}
	Suppose that \((\Omega^1_P,\dt{P})\) is an \((H;\Omega^1_H,\dt{H})\)-principal \fodc{} on \(P\). If \(\Pi\) is a connection on \((P;\Omega^1_P,\dt{P})\), then the restriction of \[(\rest{\ver[\dt{P}] \circ \Pi}{\ran\Pi})^{-1} : \Omega^1_{P,\ver} \to \Omega^1_P\] to \(\Lambda^1_H\) is its (absolute) connection \(1\)-form in the sense of Brzezi\'{n}ski--Majid~\cite{BrM}*{Prop.\ 4.10} and \DJ{}ur\dj{}evi\'{c}~\cite{Dj97}*{Def.\ 4.1}. Thus, by a result of Beggs--Majid~\cite{BeM}*{Prop. 5.54}, bimodule connections correspond to regular connections in the sense of \DJ{}ur\dj{}evi\'{c}~\cite{Dj97}*{Def.\ 4.3}.
\end{remark}

The algebraic significance of the strong connection condition \eqref{strong} was already noted by Hajac~\cite{Hajac}, while its functional-analytic significance has recently been demonstrated by \'{C}a\'{c}i\'{c}--Mesland~\cite{CaMe}*{Appx.\ A}. However, a conceptual understanding of this condition has only recently been provided by Beggs--Majid~\cite{BeM}*{\S 5.4.2}: under standard hypotheses, \eqref{strong} is equivalent to requiring that \(\coinv{H}{\ker(\ver[\dt{P}])} = B \cdot \dt{P}(B)\), i.e., that a \(1\)-form is basic if and only if it is horizontal and \(H\)-coinvariant. We shall repeatedly use the following abstract restatement of this reformulation.

\begin{definition}
	Recall that \(B \coloneqq \coinv{H}{P}\) for \(P\) a principal left \(H\)-comodule \(\ast\)-subalgebra. Let \(E\) be a \(B\)-\(\ast\)-bimodule. A \emph{horizontal lift} of \(E\) is a pair \((\tilde{E},\iota)\), where \(\tilde{E}\) is a left \(H\)-covariant \(P\)-\(\ast\)-bimodule and \(\iota : E \inj \coinv{H}{\tilde{E}}\) is an injective morphism of \(B\)-\(\ast\)-bimodules, such that
	\(
		\tilde{E} = P \cdot \iota(E) \cdot P
	\);
	we say that \((\tilde{E},\iota)\) is \emph{projectable} whenever \(\coinv{H}{\tilde{E}} = \iota(E)\).
\end{definition}

\begin{proposition}[Beggs--Majid~\cite{BeM}*{Cor.\ 5.53}]\label{strongprop}
	Let \(E\) be a \(B\)-\(\ast\)-bimodule and let \((\tilde{E},\iota)\) be a horizontal lift of \(E\). Then  \((\tilde{E},\iota)\) is projectable if and only if \(\tilde{E} = P \cdot \iota(E)\).
\end{proposition}

We now expand upon Beggs--Majid's characterisation of strong connections to reinterpret a strong bimodule connection on a regular quantum principal \((H;\Omega^1_H,\dt{H})\)-bundle as a splitting of the total differential calculus into the direct sum of the vertical calculus and a projectible horizontal lift of the basic calculus.

\begin{proposition}[cf. \DJ{}ur\dj{}evi\'{c}~\cite{Dj97}*{Prop.\ 4.6, Thm.\ 4.12}, Beggs--Majid~\cite{BeM}*{Prop.\ 5.54}]\label{analysisthm}
	Suppose that \((\Omega^1_P,\dt{P})\) is an \((H;\Omega^1_H,\dt{H})\)-principal \textsc{fodc} on the principal left \(H\)-comodule \(\ast\)-algebra \(P\), such that quantum principal \((H;\Omega^1_H,\dt{H})\)-bundle \((P;\Omega^1_P,\dt{P})\) admits a bimodule connection. Define the restriction of \((\Omega^1_P,\dt{P})\) to an \textsc{fodc} on \(B \coloneqq \coinv{H}{P}\) by 
	\begin{equation}
		(\Omega^1_B,\dt{B}) \coloneqq (B \cdot \dt{P}(B),\rest{\dt{P}}{B}),
	\end{equation}
	Then the left \(H\)-covariant \(P\)-\(\ast\)-submodule
	\begin{equation}
		\Omega^1_{P,\hor} \coloneqq \ker \ver[\dt{P}] = P \cdot \dt{P}(B) \cdot P
	\end{equation}
	of \(\Omega^1_P\) together with the inclusion \(\Omega^1_B \inj \coinv{H}{\Omega^1_{P,\hor}}\) defines a horizontal lift of \(\Omega^1_B\), which is projectable if and only if every bimodule connection on \((P;\Omega^1_P,\dt{P})\) is strong. In this case, for every strong bimodule connection \(\Pi\) on \((P;\Omega^1_P,\dt{P})\):
	\begin{enumerate}
		\item \(\nabla_\Pi \coloneqq (\id-\Pi) \circ \dt{P} : P \to \Omega^1_{P,\hor}\) 
		is a left \(H\)-covariant \(\ast\)-derivation satisfying
		\[
			\Omega^1_{P,\hor} = P \cdot \nabla_\Pi(P), \quad \rest{\nabla_\Pi}{B} = \dt{B};
		\]
		\item the map \(\psi_\Pi : \Omega^1_P \to \Omega^1_{P,\ver} \oplus \Omega^1_{P,\hor}\) given by
		\begin{equation*}
			\forall \omega \in \Omega^1_P, \quad \psi_\Pi(\omega) \coloneqq \left(\ver[\dt{P}](\omega),(\id-\Pi)(\omega)\right)
		\end{equation*}
		defines a left \(H\)-covariant isomorphism of \(P\)-\(\ast\)-bimodules, such that
		\begin{equation*}
			\forall p \in P, \quad \psi_\Pi \circ \dt{P}(p) = \left(\dv{P}(p),\nabla_\Pi(p)\right).
		\end{equation*}
	\end{enumerate}
\end{proposition}

\begin{proof}
	By Proposition~\ref{strongprop}, \(\Omega^1_{P,\hor}\) is a projectable horizontal lift of \(\Omega^1_B\) if and only if \(\Omega^1_{P,\hor} = P \cdot \Omega^1_B\). On the one hand, if \((P;\Omega^1_P,\dt{P})\) admits a strong connection \(\Pi\), then
	\[
		\Omega^1_{P,\hor} = (\id-\Pi)(\Omega^1_P) = (\id-\Pi)(P \cdot \dt{P}(P)) = P \cdot (\id-\Pi) \circ \dt{P}(P) = P \cdot P \cdot \dt{P}(B) = P \cdot \Omega^1_B;
	\]
	on the other hand, if \(\Omega^1_{P,\hor} = P \cdot \Omega^1_B\), then, for every connection \(\Pi\) on \((P;\Omega^1_P,\dt{P})\),
	\[
		(\id-\Pi) \circ \dt{P}(B) \subset \Omega^1_{P,\hor} = P \cdot \Omega^1_B = P \cdot \dt{P}(B).
	\]
	
	Now, suppose that \(\Pi\) is a strong bimodule connection on \((P;\Omega^1_P,\dt{P})\). First, since \(\Pi\) is a left \(H\)-covariant morphism of \(P\)-\(\ast\)-bimodules, \(\nabla_\Pi \coloneqq (\id-\Pi) \circ \dt{P}\) is a left \(H\)-covariant \(\ast\)-derivation; since \(\ker \Pi = \Omega^1_{P,\hor}\) and \(\Pi^2 = \Pi\), it follows that \(\rest{\nabla_\Pi}{B} = \dt{B}\) and
	\[
		\Omega^1_{P,\hor} = (\id-\Pi)(P \cdot \dt{P}(P)) = P \cdot (\id-\Pi)\circ\dt{P}(P) = P \cdot \nabla_\Pi(P).
	\]
	Next, in light of Remark~\ref{sesremark}, the map \(\psi_\Pi : \Omega^1_P \to \Omega^1_{P,\ver} \oplus \Omega^1_{P,\hor}\) is simply the isomorphism of left \(H\)-comodule \(P\)-\(\ast\)-bimodules induced by splitting of the short exact sequence~\eqref{ses} of left \(H\)-comodule \(P\)-\(\ast\)-bimodules defined by \(\Pi\) as a bimodule connection. Finally, for all \(p \in P\),
	\[
		\psi_\Pi \circ \dt{P}(p) = (\ver[\dt{P}] \circ \dt{P}(p),(\id-\Pi) \circ \dt{P}(p)) = (\dv{P}(p),\nabla_\Pi(p)). \qedhere
	\]
\end{proof}

Thus, if a quantum principal \((H;\Omega^1_H,\dt{H})\)-bundle \((P;\Omega^1_P,\dt{P})\) admits strong bimodule connection, then a choice of strong bimodule connection \(\Pi\) is equivalent to a decomposition (up to isomorphism) of \((\Omega^1_P,\dt{P})\) as a direct sum of left \(H\)-covariant \textsc{fodc}
\begin{equation}\label{decomposition}
	(\Omega^1_P,\dt{P}) \cong (\Omega^1_{P,\ver},\dv{P}) \oplus (\Omega^1_{P,\hor},\nabla_\Pi).
\end{equation}
Here, the left \(H\)-covariant \textsc{fodc} \((\Omega^1_{P,\ver},\dv{P})\) on \(P\) is completely determined by the left \(H\)-coaction on \(P\) and by the bicovariant \textsc{fodc} \((\Omega^1_H,\dt{H})\) on \(H\), while the left \(H\)-covariant \textsc{fodc} \((\Omega^1_{P,\hor},\nabla_\Pi)\)  on \(P\) is a projectable horizontal lift of \((\Omega^1_B,\dt{B})\)  in the sense that
\[
	\left(\coinv{H}{\Omega^1_{P,\hor}},\rest{\nabla_\Pi}{B}\right) = (\Omega^1_B,\dt{B}).
\]
This strongly suggests that variation of the (principal) connection \(\Pi\) can be viewed as tantamount to variation of the lift \(\nabla_\Pi : P \to \Omega^1_{P,\hor}\) of \(\dt{B} : B \to \Omega^1_B\), a change of perspective that we shall explore in the next section and justify in \S\ref{reconsec}.

\subsection{Gauge transformations and gauge potentials}

Let \(H\) be a Hopf \(\ast\)-algebra, and let \(P\) be a principal left \(H\)-comodule \(\ast\)-algebra with algebra of coinvariants \(B \coloneqq \coinv{H}{P}\). We shall now will study the problem of lifting a \textsc{fodc} \((\Omega^1_B,\dt{B})\) on \(B\) to a left \(H\)-covariant \textsc{fodc} \((\Omega^1_{P,\hor},\nabla)\) on \(P\), such that \((\coinv{H}{\Omega^1_{P,\hor}},\rest{\nabla}{B})\) recovers \((\Omega^1_B,\dt{B})\) up to isomorphism. As \DJ{}ur\dj{}evi\'{c} observed~\cite{Dj98}, this suggests encoding principal connections through their induced horizontal covariant derivatives; we shall see that this also yields a tentative definition of gauge transformation for quantum principal bundles with general \fodc{} in the spirit of \'{C}a\'{c}i\'{c}--Mesland~\cite{CaMe}*{\S 3}. Our definitions will be rigorously justified in \S\ref{reconsec}.

We begin with the following definition, which encodes a choice of \fodc{} on the base \(B\) together with compatible left \(H\)-covariant \(P\)-\(\ast\)-bimodule of horizontal \(1\)-forms.

\begin{definition}[cf.\ \DJ{}ur\dj{}evi\'{c}~\cite{Dj10}*{\S\ 3.1}]
	A \emph{(first-order) horizontal calculus} on the principal left \(H\)-co\-module \(\ast\)-algebra \(P\) is a quadruple \((\Omega^1_B,\dt{B};\Omega^1_{P,\hor},\iota)\), where:
	\begin{enumerate}
		\item \((\Omega^1_B,\dt{B})\) is a (trivially left \(H\)-covariant) \fodc{} on \(B \coloneqq \coinv{H}{P}\);
		\item \((\Omega^1_{P,\hor},\iota)\) is a projectable horizontal lift of the \(B\)-\(\ast\)-bimodule \(\Omega^1_B\).
	\end{enumerate}
\end{definition}

\begin{example}
	Let \((\Omega^1_H,\dt{H})\) be a bicovariant \fodc{} for \(H\), and suppose that \((\Omega^1_P,\dt{P})\) is an \((H;\Omega^1_H,\dt{H})\)-principal \fodc{} on \(P\), such that the quantum principal \((H;\Omega^1_H,\dt{H})\)-bundle \((P;\Omega^1_P,\dt{P})\) admits a strong bimodule connection. Then, by Proposition~\ref{analysisthm}, the triple
	\begin{equation}
		(\Omega^1_B,\dt{B};\Omega^1_{P,\hor},\iota) \coloneqq \left(B \cdot \dt{P}(B),\rest{\dt{P}}{B};\ker\ver[\dt{P}],\rest{\id_{\Omega^1_P}}{B \cdot \dt{P}(B)}\right)
	\end{equation}
	defines a first-order horizontal calculus on \(P\), which we view as the \emph{canonical} first-order horizontal calculus induced by the \((H;\Omega^1_H,\dt{H})\)-principal \fodc{} \((\Omega^1_P,\dt{P})\).
\end{example}

Assume, therefore, that \(P\) admits a first-order horizontal calculus \((\Omega^1_B,\dt{B};\Omega^1_{P,\hor},\iota)\), which we now fix. Let us also assume that \(\dt{B} : B \to \Omega^1_B\) admits a lift to a left \(H\)-covariant \(\ast\)-derivation \(P \to \Omega^1_{P,\hor}\). To simplify notation, we suppress the inclusion map \(\iota\) and identify \(\Omega^1_B\) with its image in \(\Omega^1_{P,\hor}\); hence, where convenient, we denote the first-order horizontal calculus \((\Omega^1_B,\dt{B};\Omega^1_{P,\hor},\iota)\) by the triple \((\Omega^1_B,\dt{B};\Omega^1_{P,\hor})\).

We begin by considering automorphisms of horizontal lifts; this will permit us to define gauge transformations as automorphisms of \(\Omega^1_{P,\hor}\) \emph{qua} horizontal lift of \(\Omega^1_B\).

\begin{definition}
	Let \(E\) be a \(B\)-\(\ast\)-bimodule, where \(B \coloneqq \coinv{H}{P}\), and let \((\tilde{E},\iota)\) be a horizontal lift of \(E\). An \emph{automorphism} of \((\tilde{E},\iota)\) is a pair \((\phi,\phi_\ast)\), where \(\phi : P \to P\) is a left \(H\)-covariant \(\ast\)-automorphism satisfying \(\rest{\phi}{B} = \id_B\) and \(\phi_\ast : \tilde{E} \to \tilde{E}\) is a left \(H\)-covariant \(\ast\)-preserving \(\bC\)-linear bijection satisfying \(\phi_\ast \circ \iota = \iota\) and
	\begin{equation}
		\forall p, \p{p} \in P, \, \forall \eta \in \tilde{E}, \quad \phi_\ast(p \cdot \eta \cdot \p{p}) = \phi(p) \cdot \phi_\ast(\eta) \cdot \phi(\p{p}).
	\end{equation}
	We denote the group of all automorphisms of \((\tilde{E},\iota)\) by \(\Aut(\tilde{E},\iota)\).
\end{definition}

\begin{remark}\label{asthor}
Suppose that \(E\) is a \(B\)-\(\ast\)-bimodule with horizontal lift \((\tilde{E},\iota)\). Since \(\iota(E)\) generates \(E\) as a \(P\)-bimodule, an automorphism \((\phi,\phi_\ast) \in \Aut(\tilde{E},\iota)\) is uniquely determined by \(\phi\). By mild abuse of notation, we can therefore identify \(\Aut(\tilde{E},\iota)\) with the subgroup of \(\Aut(P)\) consisting of left \(H\)-covariant \(\phi \in \Aut(P)\), such that \(\rest{P}{B} = \id_B\) and the map
\[
	\phi_\ast : \tilde{E} \to \tilde{E}, \quad p \cdot \iota(e) \cdot \p{p} \mapsto \phi(p) \cdot \iota(e) \cdot \phi(\p{p})
\]
is well-defined and bijective; with this convention, it follows that
\[
	(\phi \mapsto \phi_\ast) : \Aut(\tilde{E},\iota) \to \GL(\tilde{E})
\]
is a group homomorphism.
\end{remark}

If \(B \coloneqq \coinv{H}{P}\) is not central in \(P\) (or even if \(B\) is central in \(P\) but does not centralise the lifted \(B\)-\(\ast\)-bimodule), then the inner automorphisms will form a non-trivial central subgroup of the group of automorphisms of a horizontal lift. More broadly, the significance of inner automorphisms in noncommutative gauge theory has been pointed out by Connes~\cite{Connes96}*{p.\ 165} in the context of the spectral action on spectral triples.

\begin{definition}
	Let \(E\) be a \(B\)-\(\ast\)-bimodule, and let \((\tilde{E},\iota)\) be a horizontal lift of \(E\). We say that an automorphism \(\phi \in \Aut(\tilde{E},\iota)\) is \emph{inner} if and only if
	\[
		\phi = \Ad_v \coloneqq (p \mapsto v p v^\ast), \quad \phi_\ast = \Ad_v \coloneqq (\eta \mapsto v \cdot \eta \cdot v^\ast)
	\]
	for some \(v \in \Unit(B)\). We denote the subset of all inner automorphisms of \((\tilde{E},\iota)\) by \(\Inn(\tilde{E},\iota)\).
\end{definition}

\begin{proposition}\label{innergaugeprop1}
	Let \(E\) be a \(B\)-\(\ast\)-bimodule, and let \((\tilde{E},\iota)\) be a horizontal lift of \(E\). For every \(v \in \Unit(B)\), the automorphism \(\Ad_v\) of \(P\) defines an inner automorphism of \((\tilde{E},\iota)\) if and only if \(v \in \Cent_B(B \oplus E)\). Hence, \(\Inn(\tilde{E},\iota)\) defines a central subgroup of \(\Aut(\tilde{E},\iota)\).
\end{proposition}

\begin{proof}
	First, \(\Ad_{v}\) restricts to the identity on \(B\) if and only if \(v \in \Zent(B)\). Now, given \(v \in \Unit(\Zent(B))\), for every \(p,\,\p{p}\in P\) and \(e \in E\),
	\begin{align*}
		\Ad_v(p \cdot e \cdot \p{p}) - \Ad_v(p) \cdot e \cdot \Ad_v(\p{p})	&= v \cdot (p \cdot e \cdot \p{p}) \cdot v^\ast - vpv^\ast \cdot e \cdot v\p{p}v^\ast\\
		&= vp \cdot (e - v^\ast \cdot e \cdot v) \cdot \p{p} v^\ast,
	\end{align*}
	so that \(\Ad_v \in \Inn(\tilde{E},\iota)\) if and only if \(v \in \Cent_B(\Omega^1_B)\). Thus, in particular, the group homomorphism
	\(
		(v \mapsto \Ad_v) : \Unit(P) \to \Aut(P)
	\) 
	restricts to a surjection \(\Unit(\Cent_B(B\oplus E)) \surj \Inn(\tilde{E},\iota)\), so that \(\Inn(\tilde{E},\iota) \leq \Aut(\tilde{E},\iota)\) since \(\Unit(\Cent_B(B\oplus E)) \leq \Unit(P)\).
	
	Let us now show that \(\Inn(\tilde{E},\iota)\) is central. Let \(\phi \in \Aut(\tilde{E},\iota)\) and let \(v \in \Unit(\Cent_B(B\oplus E))\). Then, for all \(p \in P\), since \(\rest{\phi}{B} = \id_B\),
	\[
		\phi \circ \Ad_v \circ \inv{\phi}(p) = \phi(v \inv{\phi}(p) v^\ast) = v \phi(\inv{\phi}(p)) v^\ast = \Ad_v(p). \qedhere
	\]
\end{proof}

At last, we can tentatively define a gauge transformation to be an automorphism of \(\Omega^1_{P,\hor}\) as a horizontal lift of \(\Omega^1_B\) and check its differentiability with respect to any vertical calculus on \(P\). This definition, which is a non-trivial refinement of Brzezi\'{n}ski's notion of vertical automorphism~\cite{Br96}*{\S 5}, will be justified by Proposition~\ref{equiv1} and Corollary~\ref{equiv1cor}.

\begin{definition}
A \emph{gauge transformation} of the principal left \(H\)-comodule \(\ast\)-algebra \(P\) with respect to the horizontal calculus \((\Omega^1_B,\dt{B};\Omega^1_{P,\hor},\iota)\) is an automorphism \(\phi\) of the projectable horizontal lift \((\Omega^1_{P,\hor},\iota)\) of \(\Omega^1_B\); hence, the \emph{gauge group}, \emph{inner gauge group}, and \emph{outer gauge group} of \(P\) with respect to \((\Omega^1_B,\dt{B};\Omega^1_{P,\hor},\iota)\) are given by
\[
	\fr{G} \coloneqq \Aut(\Omega^1_{P,\hor},\iota), \quad \Inn(\fr{G}) \coloneqq \Inn(\Omega^1_{P,\hor},\iota), \quad \Out(\fr{G}) \coloneqq \fr{G}/\Inn(\fr{G}),
\]
respectively. When clarity requires, for \(\phi \in \fr{G}\), we will denote \(\phi_\ast\) by \(\phi_{\ast,\hor}\).
\end{definition}

\begin{remark}
Thus, by Proposition~\ref{innergaugeprop1}, the map \((v \mapsto \Ad_v) : \Unit(\Cent_B(B\oplus\Omega^1_B)) \surj \Inn(\fr{G})\) yields the short exact sequence of Abelian groups
\[
	1 \to \Unit(\Cent_B(\Omega^1_B) \cap \Zent(P)) \to \Unit(\Cent_B(B \oplus \Omega^1_B)) \to \Inn(\fr{G}) \to 1,
\]
where \(\Inn(\fr{G})\) is central in \(\fr{G}\).
\end{remark}

We will find it useful to record the following straightforward observation, which guarantees that gauge transformations are differentiable with respect the vertical calculus induced by any locally freeing bicovariant \fodc{} on \(H\).

\begin{proposition}\label{astver}
	Suppose that \((\Omega^1_H,\dt{H})\) is a bicovariant \textsc{fodc} on \(H\) that is locally freeing for\(P\); let \((\Omega^1_{P,\ver},\dv{P})\) be the resulting vertical calculus. For every \(f \in \fr{G}\), the map 
	\begin{equation}
		f_{\ast,\ver} \coloneqq \id_{\Lambda^1_H} \otimes f
	\end{equation} 
	defines a left \(H\)-covariant \(\ast\)-preserving \(\bC\)-linear endomorphism of \(\Omega^1_{P,\ver}\) satisfying
	\begin{gather}
		\forall p, \p{p} \in P, \, \forall \omega \in \Omega^1_{P,\ver}, \quad f_{\ast,\ver}(p \cdot \omega \cdot \p{p}) = f(p) \cdot f_{\ast,\ver}(\omega) \cdot f(\p{p}),\\
		\dv{P} \circ f = f_{\ast,\ver} \circ \dv{P}.
	\end{gather}
	Furthermore, the map \((f \mapsto f_{\ast,\ver}) : \fr{G} \to \GL(\Omega^1_{P,\ver})\) is a group homomorphism.
\end{proposition}

Having tentatively defined gauge transformations to be symmetries of the projectable horizontal lift \(\Omega^1_{P,\hor}\) of \(\Omega^1_B\), we now tentatively define a gauge potential to be a lift of \(\dt{B} : B \to \Omega^1_B\) to a left \(H\)-covariant \(\ast\)-derivation \(P \to \Omega^1_{P,\hor}\), which, like \DJ{}ur\dj{}evi\'{c}~\cite{Dj98}, we view as the horizontal covariant derivative of a principal connection. The justification for this definition, which builds crucially on the role of (strong bimodule) connections in splitting the noncommutative Atiyah sequence \eqref{ses}, will be provided by Proposition~\ref{equiv1}.

\begin{definition}[cf.\ \DJ{}ur\dj{e}vi\'{c}~\cite{Dj98}*{Def.\ 1}]
	A \emph{gauge potential} on the principal left \(H\)-comodule \(\ast\)-algebra \(P\) with respect to the first-order horizontal calculus \((\Omega^1_B,\dt{B};\Omega^1_{P,\hor})\) is a left \(H\)-covariant \(\ast\)-derivation \(\nabla : P \to \Omega^1_{P,\hor}\), such that
	\(
		\rest{\nabla}{B} = \dt{B}
	\);	
	we define the \emph{Atiyah space} of \(P\) with respect to \((\Omega^1_B,\dt{B};\Omega^1_{P,\hor})\) to be the set  \(\fr{At}\) of all gauge potentials.
\end{definition}

\begin{example}
	Suppose that the \fodc{} \((\Omega^1_B,\dt{B})\) is inner, so that \(\dt{B} = \ad_\alpha\) for some \(1\)-form \(\alpha \in (\Omega^1_B)_{\sa}\). Then \(\nabla \coloneqq \ad_{\iota(\alpha)}\) is a gauge potential on \(P\) with respect to \((\Omega^1_B,\dt{B};\Omega^1_{P,\hor})\).
\end{example}

Our assumption on the horizontal calculus \((\Omega^1_B,\dt{B};\Omega^1_{P,\hor})\) is that \(\fr{At} \neq \emptyset\). Since \(\Omega^1_{P,\hor}\) is projectable as a horizontal lift of \(\Omega^1_B\), it follows that \(\Omega^1_{P,\hor} = P \cdot \Omega^1_B = P \cdot \dt{B}(B)\). Thus for every \(\nabla \in \fr{At}\), the pair \((\Omega^1_{P,\hor},\nabla)\) defines a left \(H\)-covariant \textsc{fodc} on \(P\) satisfying
\[
	(\Omega^1_B,\dt{B}) = (\coinv{H}{\Omega^1_{P,\hor}},\rest{\nabla}{\coinv{H}{\Omega^1_{P,\hor}}});
\]
in other words, it defines a projectable horizontal lift of \((\Omega^1_B,\dt{B})\).

It is now straightforward to check that the Atiyah space \(\fr{At}\) is an \(\fr{G}\)-invariant \(\bR\)-affine subspace of the \(\bR\)-vector space \(\Der_P(\Omega^1_{P,\hor})\) of  \(\ast\)-derivations on the \(P\)-\(\ast\)-bimodule \(\Omega^1_{P,\hor}\).

\begin{definition}[cf. \DJ{}ur\dj{}evi\'{c}~\cite{Dj98}*{p.\ 98}]
	A \emph{relative gauge potential} on \(P\) with respect to the horizontal calculus \((\Omega^1_B,\dt{B};\Omega^1_{P,\hor})\) is a left \(H\)-covariant \(\ast\)-derivation \(\bA: P \to \Omega^1_{P,\hor}\), such that \(\rest{\bA}{B} = 0\); we denote by \(\fr{at}\) the \(\bR\)-vector space of all relative gauge potentials on \(P\) with respect to \((\Omega^1_B,\dt{B};\Omega^1_{P,\hor})\).
\end{definition}

\begin{proposition}
	The Atiyah space \(\fr{At}\) of \(P\) with respect to \((\Omega^1_B,\dt{B};\Omega^1_{P,\hor})\) defines a \(\bR\)-affine space modelled on the \(\bR\)-vector space \(\fr{at}\) with respect to the subtraction map \(\fr{At} \times \fr{At} \to \fr{at}\) defined by \((\nabla,\nabla^\prime) \mapsto \nabla - \nabla^\prime\). Moreover, \(\fr{At}\) admits an affine action of the gauge group \(\fr{G}\) defined by
	\begin{equation}
		\forall \phi \in \fr{G}, \, \forall\, \nabla \in \fr{At}, \quad \phi \triangleright \nabla \coloneqq \phi_{\ast} \circ \nabla \circ \phi^{-1},
	\end{equation}
	whose linear part is the linear action of \(\fr{G}\) on \(\fr{at}\) defined by
	\begin{equation}
		\forall \phi \in \fr{G}, \, \forall \bA \in \fr{at}, \quad \phi \triangleright \bA \coloneqq \phi_{\ast} \circ \bA \circ \phi^{-1}.
	\end{equation}
\end{proposition}

If \(B \coloneqq \coinv{H}{P}\) is not central in \(P\) (or if \(B\) is central in \(P\) but does not centralise \(\Omega^1_{P,\hor}\)), one obtains a non-trivial \(\fr{G}\)-invariant \(\bR\)-subspace of inner relative gauge potentials. More broadly, the significance of inner derivations in noncommutative gauge theory has been pointed out by Connes~\cite{Connes96}*{p.\ 166} in the context of the spectral action on spectral triples.

\begin{definition}
	A relative gauge potential \(\bA \in \fr{at}\) on \(P\) with respect to the horizontal calculus \((\Omega^1_B,\dt{B};\Omega^1_{P,\hor})\) is \emph{inner} whenever \(\bA = \ad_\alpha\) for some \(\alpha \in (\Omega^1_B)_{\sa}\); we denote the subspace of all inner relative gauge potentials by \(\Inn(\fr{at})\). Thus, we define the \emph{outer Atiyah space} of \(P\) with respect to \((\Omega^1_B,\dt{B};\Omega^1_{P,\hor})\) to be the affine space \(\Out(\fr{At}) \coloneqq \fr{At}/\Inn(\fr{at})\) with space of translations \(\Out(\fr{at}) \coloneqq \fr{at}/\Inn(\fr{at})\).
\end{definition}

\begin{proposition}\label{innerprop}
	A map \(\bA : P \to \Omega^1_{P,\hor}\) is an inner relative gauge potential if and only if \(\bA = \ad_\alpha\) for some \(\alpha \in \Zent_B(\Omega^1_B)_{\sa}\). Thus, the map \((\alpha \mapsto \ad_\alpha) : Z_B(\Omega^1_B)_{\sa} \to \Inn(\fr{at})\) yields the short exact sequence
\begin{equation}\label{innerpropeq}
	0 \to \Zent_B(\Omega^1_B)_{\sa} \cap \Zent_P(\Omega^1_{P,\hor}) \to \Zent_B(\Omega^1_B)_{\sa} \to \Inn(\fr{at}) \to 0.
\end{equation}
Moreover, the subspace \(\Inn(\fr{at})\) of inner relative gauge potentials consists of \(\fr{G}\)-invariant vectors, so that the affine action of \(\fr{G}\) on \(\fr{At}\) descends to an affine action of \(\fr{G}\) on \(\Out(\fr{At})\).
\end{proposition}

\begin{proof}
	First, given \(\alpha \in (\Omega^1_B)_{\sa}\), the left \(H\)-covariant \(\ast\)-derivation \(\bA \coloneqq \ad_{\alpha}\) defines a relative gauge potential if and only if \(\ker \bA \supseteq B\), if and only if \(\alpha \in \Zent_B(\Omega^1_B)\); this immediately yields the short exact sequence \eqref{innerpropeq}.
\end{proof}


In fact, it turns out that the affine action of the gauge group \(\fr{G}\) on the Atiyah space \(\fr{At}\) descends further to an affine action of the outer gauge group \(\Out(\fr{G})\) on the outer Atiyah space \(\Out(\fr{At})\), which, as we shall see, will yield a non-trivial invariant of the principal left \(H\)-comodule \(\ast\)-algebra \(P\) endowed with the horiziontal calculus \((\Omega^1_{B},\dt{B};\Omega^1_{P,\hor})\).

\begin{proposition}\label{inducedprop}
	The inner gauge group \(\Inn(\fr{G})\) acts trivially on \(\Out(\fr{at})\), so that the induced action of \(\fr{G}\) on \(\Out(\fr{At})\) descends further to an affine action of \(\Out(\fr{G})\) on \(\Out(\fr{At})\).
\end{proposition}

\begin{proof}
	Let \(\phi \in \Inn(\fr{G})\), so that by Proposition~\ref{innergaugeprop1}, \(\phi = \Ad_v\) and \(\phi_\ast = \Ad_v\) for some unitary \(v \in \Unit(Z(B) \cap \Cent_B(\Omega^1_B))\); let \(\nabla \in \fr{At}\). Then, for all \(p \in P\),
	\begin{align*}
		(\phi \act \nabla - \nabla)(p) &= v \cdot \nabla(v^\ast p v) \cdot v^\ast - \nabla(p)\\
		&= v \cdot \left(\dt{B}(v^\ast) \cdot pv + v^\ast \cdot \nabla(p) \cdot v + v^\ast p \cdot \dt{B}(v)\right) \cdot v^\ast - \nabla(p)\\
		&= [-v^\ast \cdot \dt{B}(v),p]
	\end{align*}
	so that \(\phi \act \nabla - \nabla = \ad_{-v^\ast \dt{B}(v)} \in \Inn(\fr{at})\). Hence, \(\Inn(\fr{G})\) acts trivially on \(\Out(\fr{At})\).
\end{proof}

\subsection{Reconstruction of quantum principal bundles to first order}\label{reconsec}

Let \(H\) be a Hopf \(\ast\)-algebra and let \(P\) be a principal left \(H\)-comodule \(\ast\)-algebra with \(\ast\)-subalgebra of coinvariants \(B \coloneqq \coinv{H}{P}\). Given a horizontal calculus \((\Omega^1_B,\dt{B};\Omega^1_{P,\hor})\) for \(P\) and a bicovariant \fodc{} \((\Omega^1_H,\dt{H})\) on \(H\) that is locally freeing for \(P\), we consider \emph{all} \((H;\Omega^1_H,\dt{H})\)-principal \fodc{} on \(P\) inducing the first-order horizontal calculus \((\Omega^1_B,\dt{B};\Omega^1_{P,\hor})\). This will allow us to justify our notions of gauge transformation, gauge potential, and gauge action relative to the theory of Brzezi\'{n}ski--Majid~\cite{BrM} by means of a conceptually transparent equivalence of relevant groupoids. Furthermore, this equivalence will yield a gauge-equivariant affine moduli space of \((H;\Omega^1_H,\dt{H})\)-principal \fodc{} on \(P\) inducing \((\Omega^1_B,\dt{B};\Omega^1_{P,\hor})\). For relevant definitions from the basic theory of groupoids, see Appendix~\ref{Groupoids}.

From now on, fix a horizontal calculus \((\Omega^1_B,\dt{B};\Omega^1_{P,\hor},\iota)\) on the principal left \(H\)-co\-mod\-ule \(\ast\)-algebra \(P\). Given a bicovariant \fodc{} \((\Omega^1_H,\dt{H})\) on \(H\) that is locally freeing for \(P\), we construct a groupoid whose objects are \((H;\Omega^1_H,\dt{H})\)-principal \fodc{} on \(P\) inducing the horizontal calculus \((\Omega^1_B,\dt{B};\Omega^1_{P,\hor},\iota)\) on \(P\) and whose arrows will give an abstract notion of gauge transformation adapted to general \fodc{}; note that we do not require an abstract gauge transformation to be differentiable with respect to the same \fodc{} on \(P\) as both domain and codomain. In what follows, let \((\Omega^1_{P,\ver},\dv{P})\) denote the vertical calculus on \(P\) induced by a locally freeing bicovariant \fodc{} \((\Omega^1_H,\dt{H})\).

\begin{definition}\label{qpbdef}
	Let \((\Omega^1_H,\dt{H})\) be a bicovariant \fodc{} on \(H\) that is locally freeing for \(P\). We define the \emph{abstract gauge groupoid} with respect to \((\Omega^1_H,\dt{H})\) to be the groupoid \(\cG[\Omega^1_H]\) defined as follows:
	\begin{enumerate}
		\item an object is an \((H;\Omega^1_H,\dt{H})\)-principal \fodc{} \((\Omega^1_P,\dt{P})\) on \(P\), such that the quantum principal \((H;\Omega^1_H,\dt{H})\)-bundle \((P;\Omega^1_P,\dt{P})\) admits bimodule connections, and 
		\[
			(\ker\ver[\dt{P}], \dt{B}(b) \mapsto \dt{P}(b))
		\]
		defines a horizontal lift of \(\Omega^1_B\) admitting a (necessarily unique) left \(H\)-covariant isomorphism
		\(
			C[\Omega^1_P] : P \cdot \dt{B}(P) \to \Omega^1_{P,\hor}
		\)
		of \(P\)-\(\ast\)-bimodules satisfying \begin{equation}C[\Omega^1_P] \circ \rest{\dt{P}}{B} = \iota \circ \dt{B};\end{equation}
		\item given objects \((\Omega_1,\dt{1})\) and \((\Omega_2,\dt{2})\), an arrow \(f : (\Omega_1,\dt{1}) \to (\Omega_2,\dt{2})\) consists of a left \(H\)-covariant \(\ast\)-automorphism \(f : P \to P\), such that \(\rest{f}{B} = \id\), and such that
		\begin{equation}
			f_\ast : \Omega_1 \to \Omega_2, \quad p \cdot \dt{1}(\p{p}) \cdot \pp{p} \mapsto f(p) \cdot \dt{2}(f(\p{p})) \cdot f(\pp{p})
		\end{equation}
		is a well-defined bijection;
		\item composition of arrows is induced by composition of \(\ast\)-automorphisms of \(P\), and the identity of an object \((\Omega,\dt{})\) is given by \(\id_{(\Omega,\dt{})} \coloneqq \left(\id_P : (\Omega,\dt{}) \to (\Omega,\dt{})\right)\).
	\end{enumerate}
	Moreover, we define the star-injective homomorphism \(\mu[\Omega^1_H] : \cG[\Omega^1_H] \to \Aut(P)\)  by
	\begin{equation}
		\forall \left(f : (\Omega_1,\dt{1}) \to (\Omega_2,\dt{2})\right) \in \cG[\Omega^1_H], \quad \mu[\Omega^1_H]\mleft(f : (\Omega_1,\dt{1}) \to (\Omega_2,\dt{2})\mright) \coloneqq f.
	\end{equation}
\end{definition}

\begin{remark}
	Thus, the canonical horizontal calculus of an object \((\Omega_P,\dt{P})\) of \(\cG[\Omega^{1}_H]\) induces the given horizontal calculus \((\Omega^1_B,\dt{B};\Omega^1_{P,\hor})\) up to the canonical isomorphism \(C[\Omega_P]\) of left \(H\)-covariant \(P\)-\(\ast\)-bimodules.
\end{remark}

\begin{remark}
	Suppose that \(f : (\Omega_1,\dt{1}) \to (\Omega_2,\dt{2})\) is an arrow in \(\cG[\Omega^1_H]\). Then Proposition~\ref{astver} applies unchanged to the vertical automorphism \(f\), i.e., there exists a unique left \(H\)-covariant \(\ast\)-preserving bijection \(f_{\ast,\ver} : \Omega^1_{P,\ver} \to \Omega^1_{P,\ver}\), such that 
	\[
		\forall p, q \in P, \, \forall \omega \in \Omega^1_{P,\ver}, \quad f_{\ast,\ver}(p \cdot \omega \cdot q) = f(p) \cdot f_{\ast,\ver}(\omega) \cdot f(q)
	\]
	and \(f_{\ast,\ver} \circ \dv{P} = \dv{P} \circ f_{\ast,\ver}\). A straightforward calculation now shows that
	\[
		\ver[\dt{2}] \circ f_\ast = f_{\ast,\ver} \circ \ver[\dt{1}].
	\]	
\end{remark}

In keeping with our stated philosophy, given a bicovariant \fodc{} \((\Omega^1_H,\dt{H})\) on \(H\), we consider bimodule connections on all quantum principal \((H;\Omega^1_H,\dt{H})\)-bundles induced from \(P\) by \fodc{} in \(\Ob(\cG[\Omega^1_H])\). It is straightforward to check that the abstract gauge groupoid \(\cG[\Omega^1_H]\) admits a canonical action on this set of bimodule connections.

\begin{propositiondefinition}
	Let \((\Omega^1_H,\dt{H})\) be a bicovariant \fodc{} on \(H\) that is locally freeing for \(P\). Let \(\cA[\Omega^1_H]\) to be the set of all triples \((\Omega^1_P,\dt{P};\Pi)\), where \((\Omega^1_P,\dt{P}) \in \Ob(\cG[\Omega^1_H])\) and \(\Pi\) is a bimodule connection on the quantum principal \((\Omega^1_H,\dt{H})\)-bundle \((P;\Omega^1_P,\dt{P})\); hence, let \(p[\Omega^1_H]:\cA[\Omega^1_H] \to \Ob(\cG[\Omega^1_H])\) be the canonical surjection given by
	\[
		\forall (\Omega^1_P,\dt{P};\Pi) \in \cA[\Omega^1_H], \quad p[\Omega^1_H](\Omega^1_P,\dt{P};\Pi)  \coloneqq (\Omega^1_P,\dt{P}).
	\]
	Then the \emph{abstract gauge action} is the action of \(\cG[\Omega^1_H]\) on \(\cA[\Omega^1_H]\) via \(p[\Omega^1_H]\) defined by
	\begin{multline}
		\forall \left(f : (\Omega_1,\dt{1}) \to (\Omega_2,\dt{2})\right) \in \cG[\Omega^1_H], \, \forall (\Omega_1,\dt{1};\Pi) \in p[\Omega^1_H]^{-1}(\Omega_1,\dt{1}),\\ \left(f : (\Omega_1,\dt{1}) \to (\Omega_2,\dt{2})\right) \act (\Omega_1,\dt{1};\Pi) \coloneqq (\Omega_2,\dt{2};f_\ast \circ \Pi \circ f_\ast^{-1}).
	\end{multline}
	Hence, the canonical covering \(\pi[\Omega^1_H] : \cG[\Omega^1_H] \ltimes \cA[\Omega^1_H] \to \cG[\Omega^1_H]\) is given by
	\begin{multline}
		\forall \left(\left(f : (\Omega_1,\dt{1}) \to (\Omega_2,\dt{2})\right),(\Omega_1,\dt{1};\Pi)\right) \in \cG[\Omega^1_H] \ltimes \cA[\Omega^1_H], \\ \pi[\Omega^1_H]\mleft(\left(f : (\Omega_1,\dt{1}) \to (\Omega_2,\dt{2})\right),(\Omega_1,\dt{1};\Pi)\mright) \coloneqq \left(f : (\Omega_1,\dt{1}) \to (\Omega_2,\dt{2})\right).
	\end{multline}
\end{propositiondefinition}

As a convenient abuse of notation, we will denote an arrow \[\left(\left(f : (\Omega_1,\dt{1}) \to (\Omega_2,\dt{2})\right),(\Omega_1,\dt{1};\Pi)\right)\] of the action groupoid \(\cG[\Omega^1_H] \ltimes \cA[\Omega^1_H]\) by
\[
	f: (\Omega_1,\dt{1};\Pi_1) \to (\Omega_2,\dt{2};\Pi_2),
\]
where \(\Pi_2 \coloneqq f_\ast \circ \Pi_1 \circ f_\ast^{-1}\), so that, in particular
\[
	\pi[\Omega^1_H]\mleft(f: (\Omega_1,\dt{1};\Pi_1) \to (\Omega_2,\dt{2};\Pi_2)\mright) \coloneqq \left(f : (\Omega_1,\dt{1}) \to (\Omega_2,\dt{2})\right).
\]

\begin{example}
	Let \(\Omega^1_{P,\oplus} \coloneqq \Omega^1_{P,\ver} \oplus \Omega^1_{P,\hor}\), and let \(\Pi_\oplus : \Omega^1_{P,\oplus} \to \Omega^1_{P,\oplus}\) denote the projection onto	 \(\Omega^1_{P,\ver}\) along \(\Omega^1_{P,\hor}\). For every gauge potential \(\nabla \in \fr{At}\) on \(P\) with respect to \((\Omega^1_B,\dt{B};\Omega^1_{P,\hor})\), the left \(H\)-covariant \(\ast\)-derivation \(\dt{P,\nabla} : P \to \Omega^1_{P,\oplus}\) given by
	\begin{equation}
		\forall p \in P, \quad \dt{P,\nabla}(p) \coloneqq \left(\dv{P}(p),\nabla(p)\right)
	\end{equation}
	makes \((\Omega^1_{P,\oplus},\dt{P,\nabla};\Pi_{\oplus})\) into an element of \(\cA[\Omega^1_H]\), such that \(\nabla_{\Pi_\oplus} = \nabla\); in particular, the resulting vertical map \(\ver[\dt{P,\nabla}]\) is simply the projection onto \(\Omega^1_{P,\ver}\) along \(\Omega^1_{P,\hor}\).
\end{example}

Let \(\fr{G}\) and \(\fr{At}\) respectively denote the gauge group and Atiyah space of the principal left \(H\)-comodule \(\ast\)-algebra \(P\) with respect to the horizontal calculus \((\Omega^1_B,\dt{B};\Omega^1_{P,\hor})\), where, by the usual abuse of notation, we suppress the inclusion \(\iota : \Omega^1_B \inj \Omega^1_{P,\hor}\). We now promote the decomposition \eqref{decomposition} given by Proposition~\ref{analysisthm} to an explicit equivalence of categories that realises the action groupoid \(\fr{G} \ltimes \fr{At}\), which is independent of the choice of bicovariant \fodc{} \((\Omega^1_H,\dt{H})\) on \(H\), as a deformation retraction of the action groupoid \(\cG[\Omega^1_H] \ltimes \cA[\Omega^1_H]\) of the abstract gauge action. This rigorously justifies identifying the action of the gauge group \(\fr{G}\) on the Atiyah space \(\fr{At}\) as the affine action of global gauge transformations on principal connections for the quantum principal \(H\)-bundle \(P\) with respect to the bicovariant \fodc{} \((\Omega^1_H,\dt{H})\) on \(H\) and the horizontal calculus \((\Omega^1_B,\dt{B};\Omega^1_{P,\hor})\).

\begin{proposition}\label{equiv1}
	Let \((\Omega^1_H,\dt{H})\) be a bicovariant \fodc{} on \(H\) that is locally freeing for \(P\). The groupoid homomorphism \(\Sigma[\Omega^1_H] : \fr{G} \ltimes \fr{At} \to \cG[\Omega^1_H] \ltimes \cA[\Omega^1_H]\)
	given by
	\begin{equation}
		\forall (\phi,\nabla) \in \fr{G} \ltimes \fr{At}, \quad \Sigma[\Omega^1_H](\phi, \nabla) \coloneqq \left(\phi : (\Omega^1_{P,\oplus},\dt{P,\nabla};\Pi_{\oplus}) \to (\Omega^1_{P,\oplus},\dt{P,\phi\act\nabla};\Pi_{\oplus})\right)
	\end{equation}
	is an equivalence of groupoids with left inverse and homotopy inverse \(A[\Omega^1_H]\) given by
	\begin{multline}
		\forall \left(f: (\Omega_1,\dt{1};\Pi_1) \to (\Omega_2,\dt{2};\Pi_2)\right) \in \cG[\Omega^1_H] \ltimes \cA[\Omega^1_H], \\
			A[\Omega^1_H]\mleft(f: (\Omega_1,\dt{1};\Pi_1) \to (\Omega_2,\dt{2};\Pi_2)\mright) \coloneqq \left(f, C[\Omega_1] \circ (\id-\Pi_1) \circ \dt{1}\right).
	\end{multline}
	In particular, there exists a  homotopy \(\eta[\Omega^1_H] : \id_{\cG[\Omega^1_H] \ltimes \cA[\Omega^1_H]} \Rightarrow \Sigma[\Omega^1_H] \circ A[\Omega^1_H]\), which is necessarily unique, such that
	\begin{equation}\label{natiso}
		\forall (\Omega^1_P,\dt{P};\Pi) \in \cA[\Omega^1_H], \quad \mu[\Omega^1_H] \circ \pi[\Omega^1_H]\mleft(\eta[\Omega^1_H]_{(\Omega^1_P,\dt{P};\Pi)}\mright) = \id_P.
	\end{equation}
\end{proposition}

\begin{proof}
	First, let us check that \(\Sigma[\Omega^1_H]\) is well-defined. Let \((\phi,\nabla) \in \fr{G} \ltimes \fr{At}\). Then, by Remark~\ref{asthor} and Proposition~\ref{astver}, for all \(p, \p{p}, \pp{p} \in P\), 
	\begin{align*}
		\phi(p) \cdot \dt{P,\phi \act \nabla}(\phi(\p{p})) \cdot \phi(\pp{p}) &= \phi(p) \cdot \left(\dv{P}(\phi(\p{p})),(\phi \act \nabla)(\phi(\p{p}))\right) \cdot \phi(\pp{p})\\
		&=\phi(p) \cdot \left(\phi_{\ast,\ver}(\dv{P}(\p{p})),\phi_{\ast,\hor}(\nabla(\p{p}))\right) \cdot \phi(\pp{p})\\
		&= \left(\phi_{\ast,\ver}(p \cdot \dv{P}(\p{p}) \cdot \pp{p}),\phi_{\ast,\hor}(p \cdot \nabla(\p{p}) \cdot \pp{p})\right)\\
		&= (\phi_{\ast,\ver}\oplus \phi_{\ast,\hor})\mleft(p \cdot \dt{P,\nabla}(\p{p}) \cdot \pp{p}\mright)
	\end{align*}
	so that \(\phi_\ast = \phi_{\ast,\ver} \oplus \phi_{\ast,\hor}\) is a well-defined bijection that satisfies
	\begin{gather*}
		\phi_{\ast,\ver} \circ \ver[\dt{P,\nabla}] = \phi_{\ast,\ver} \circ \Proj_1 = \Proj_1 \circ (\phi_{\ast,\ver}\oplus\phi_{\ast,\hor}) = \ver[\dt{P,\phi \act \nabla}] \circ \phi_\ast, \\ \phi_\ast \circ \Pi_\oplus = (\phi_{\ast,\ver}\oplus\phi_{\ast,\hor}) \circ \Pi_\oplus = \Pi_\oplus \circ (\phi_{\ast,\ver}\oplus\phi_{\ast,\hor}) = \Pi_\oplus \circ \phi_\ast.
	\end{gather*}
	Hence, \(\phi : (\Omega^1_{P,\oplus},\dt{P,\nabla};\Pi_{\oplus}) \to (\Omega^1_{P,\oplus},\dt{P,\phi\act\nabla};\Pi_{\oplus})\) is an arrow in \(\cG[\Omega^1_H] \ltimes \cA[\Omega^1_H]\). Thus, \(\Sigma[\Omega^1_H]\) is well-defined as a function between sets of arrows. That \(\Sigma[\Omega^1_H]\) is a groupoid homomorphism now follows from the fact that both of
	\begin{gather*}
		(\phi \mapsto \phi_{\ast,\ver}) : \fr{G} \to \GL(\Omega^1_{P,\ver}), \quad (\phi \mapsto \phi_{\ast,\hor}) : \fr{G} \to \GL(\Omega^1_{P,\hor})
	\end{gather*}
	are group homomorphisms.

	Next, let us check that \(A[\Omega^1_H]\) is well-defined. Let \(f: (\Omega_1,\dt{1};\Pi_1) \to (\Omega_2,\dt{2};\Pi_2)\) be an arrow in \(\cG[\Omega^1_H] \ltimes \cA[\Omega^1_H]\), so that \(f_\ast : \Omega_1 \to \Omega_2\) is a well-defined bijection satisfying 
	\[
		\ver[\dt{2}] \circ f_\ast = f_{\ast,\ver} \circ \ver[\dt{1}], \quad \Pi_2 \circ f_\ast = f_\ast \circ \Pi_2.
	\]
	For \(i=1,2\), let \(\nabla_i \coloneqq C[\Omega_i] \circ (\id-\Pi_i) \circ \dt{P} \in \fr{At}\); we need to show that \(f \in \fr{G}\) and that \(\nabla_2 = f \act \nabla_1\). Now, for every \(p,\p{p} \in P\) and \(b \in B\),
	\begin{align*}
		f(p) \cdot \dt{B}(b) \cdot f(\p{p}) &= C[\Omega_2]\mleft(f(p) \cdot \dt{2}(b) \cdot f(\p{p})\mright)\\
		&= C[\Omega_2] \circ f_\ast \mleft(p \cdot \dt{1}(b) \cdot \p{p}\mright)\\
		&= C[\Omega_2] \circ f_\ast \circ C[\Omega_1]^{-1}\mleft(p \cdot \dt{B}(b) \cdot \p{p}\mright),
	\end{align*}
	so that \(f_{\ast,\hor} = C[\Omega_2] \circ f_\ast \circ C[\Omega_1]^{-1}\) is well-defined and bijective, hence \(f \in \fr{G}\), with
	\begin{multline*}
		f \act \nabla_1 = f_{\ast,\hor} \circ \nabla_1 \circ f^{-1} = C[\Omega_2] \circ f_\ast \circ C[\Omega_1]^{-1} \circ C[\Omega_1] \circ (\id-\Pi_1) \circ \dt{1} \circ f^{-1}\\
		= C[\Omega_2] \circ f_\ast \circ (\id-\Pi_1) \circ \dt{1} \circ f^{-1} = C[\Omega_2] \circ (\id-\Pi_2) \circ \dt{2} = \nabla_2.
	\end{multline*}
	As a result, \((f,\nabla_1) : \nabla_1 \to \nabla_2\) is a well-defined arrow in \(\fr{G} \ltimes \fr{At}\). Thus, \(A[\Omega^1_H]\) is well-defined as a function between sets of arrows. That \(A[\Omega^1_H]\) is a groupoid homomorphism now follows since \((\phi \mapsto \phi_{\ast,\hor}) : \fr{G} \to \GL(\Omega^1_{P,\hor})\) is a group homomorphism.
	
	On the one hand, \(C[\Omega^1_{P,\oplus}] = \id_{\Omega^1_{P,\hor}}\) by uniqueness of \(C[\Omega^1_{P,\oplus}]\); hence, it follows that \(A[\Omega^1_H] \circ \Sigma[\Omega^1_H] = \id_{\fr{G}\ltimes\fr{At}}\). On the other hand, by the proof of Proposition~\ref{analysisthm}, \emph{mutatis mutandis}, we can define a homotopy \(\eta[\Omega^1_H] : \id_{\cG[\Omega^1_H] \ltimes \cA[\Omega^1_H]} \Rightarrow \Sigma[\Omega^1_H] \circ A[\Omega^1_H]\) by
	\begin{multline*}
		\forall (\Omega^1_P,\dt{P};\Pi) \in \cA[\Omega^1_H], \\  \eta[\Omega^1_H]_{(\Omega^1_P,\dt{P};\Pi)} \coloneqq \left(\id_P : (\Omega^1_P,\dt{P};\Pi) \to A[\Omega^1_H] \circ \Sigma[\Omega^1_H](\Omega^1_P,\dt{P};\Pi)\right),
	\end{multline*}
	which, by construction, satisfies~\eqref{natiso}.
\end{proof}

In particular, this justifies the identification of \(\fr{G}\) as the group of global gauge transformations on the quantum principal \(H\)-bundle \(P\) with respect to the bicovariant \fodc{} \((\Omega^1_H,\dt{H})\) on \(H\) and horizontal calculus \((\Omega^1_B,\dt{B};\Omega^1_{P,\hor})\).

\begin{corollary}\label{equiv1cor}
	Let \((\Omega^1_H,\dt{H})\) be a bicovariant \fodc{} on \(H\) that is locally freeing for \(P\). The range of the star-injective groupoid homomorphism \(\mu[\Omega^1_H] : \cG[\Omega^1_H] \to \Aut(P)\) is \(\fr{G}\), so that, after restriction of codomain,
	\[
		\mu[\Omega^1_H] : \cG[\Omega^1_H] \to \fr{G}, \quad \mu[\Omega^1_H] \circ \pi[\Omega^1_H] : \cG[\Omega^1_H] \ltimes \cA[\Omega^1_H] \to \fr{G}
	\]
	both define coverings of groupoids.
\end{corollary}

Given a bicovariant \fodc{} \((\Omega^1_H,\dt{H})\) on \(H\), we now use the groupoid equivalence \(\Sigma[\Omega^1_H]\) of Proposition~\ref{equiv1} to construct a \(\fr{G}\)-equivariant moduli space of \((H;\Omega^1_H,\dt{H})\)-principal \fodc{} on \(P\) inducing the horizontal calculus \((\Omega^1_B,\dt{B};\Omega^1_{P,\hor})\). Indeed, this moduli space will turn out to be a quotient of the Atiyah space \(\fr{At}\) by the the space of all relative gauge potentials of the following form.

\begin{definition}
	Let \((\Omega^1_H,\dt{H})\) be a bicovariant \textsc{fodc} on \(H\) that is locally freeing for \(P\). We say that a relative gauge potential \(\bA \in \fr{at}\) for \(P\) with respect to \((\Omega^1_B,\dt{B};\Omega^1_{P,\hor})\) is \emph{\((\Omega^1_H,\dt{H})\)-adapted} whenever
	\begin{equation}
		\bA = \omega[\bA] \circ \dv{P}
	\end{equation}
	for some (necessarily unique) left \(H\)-covariant morphism \(\omega[\bA] : \Omega^1_{P,\ver} \to \Omega^1_{P,\hor}\) of \(P\)-\(\ast\)-bimodules; in this case, we call \(\omega[\bA]\) the \emph{relative connection \(1\)-form} of \(\bA\). We denote by \(\fr{at}[\Omega^1_H]\) the subspace of all \((\Omega^1_H,\dt{H})\)-adapted relative gauge potentials on \(P\) with respect to the horizontal calculus \((\Omega^1_B,\dt{B};\Omega^1_{P,\hor})\).
\end{definition}

\begin{remark}[cf.\ Brzezi\'{n}ski--Majid~\cite{BrM}*{\S 4.2}, \DJ{}ur\dj{}evi\'{c}~\cite{Dj97}*{\S 4}, Beggs--Majid~\cite{BeM}*{Prop.\ 5.54}]
	Using uniqueness, one can now check that the map
	\[
		\omega : \fr{at}[\Omega^1_H] \to \Hom_{P}(\Omega^1_{P,\ver},\Omega^1_{P,\hor}), \quad \bA \mapsto \omega[\bA]
	\]	
	is \(\bR\)-linear; moreover, given \(\bA \in \fr{at}[\Omega^1_H]\), the relative connection \(1\)-form \(\omega[\bA]\) is completely determined by its restriction to a map \(\Lambda^1_H \cong \Lambda^1_H \otimes 1_P \to \Omega^1_{P,\hor}\), which can be viewed as a noncommutative Lie-valued \(1\)-form.
\end{remark}

Now, let \((\Omega_1,\dt{1}),(\Omega_2,\dt{2})\in\Ob(\cG[\Omega^1_H])\), where \((\Omega^1_H,\dt{H})\) is a locally freeing bicovariant \fodc{} on \(H\). Observe that \((\Omega_1,\dt{1})\) and \((\Omega_2,\dt{2})\) admit an isomorphism of left \(H\)-covariant \textsc{fodc} for \(P\) if and only if \(\id_P : (\Omega_1,\dt{1}) \to (\Omega_2,\dt{2})\) is an arrow in \(\cG[\Omega^1_H]\). Since the subgroupoid of all such arrows is precisely \(\ker\mu[\Omega^1_H]\), it follows that \((\Omega_1,\dt{1})\) and \((\Omega_2,\dt{2})\) are isomorphic if and only if they define the same object in the quotient groupoid \(\cG[\Omega^1_H]/\ker\mu[\Omega^1_H]\), which will turn out to be well-defined and canonically isomorphic to \(\fr{G} \ltimes \fr{At}/\fr{at}[\Omega^1_H]\). Thus, the quotient affine space \(\fr{At}/\fr{at}[\Omega^1_H]\) yields the desired \(\fr{G}\)-equivariant affine moduli space of \((\Omega^1_H,\dt{H})\)-principal \fodc{} on \(P\) inducing \((\Omega^1_B,\dt{B};\Omega^1_{P,\hor})\).

\begin{theorem}\label{firstorderclassification}
	Let \((\Omega^1_H,\dt{H})\) be a bicovariant \fodc{} on \(H\). Suppose that the Atiyah space \(\fr{At}\) is non-empty. The subspace \(\fr{at}[\Omega^1_H]\) is \(\fr{G}\)-invariant; the subgroupoid \(\ker\mu[\Omega^1_H]\) of \(\cG[\Omega^1_H]\) is wide and has trivial isotropy groups, so that the quotient groupoid \(\cG[\Omega^1_H]/\ker\mu[\Omega^1_H]\) is well-defined; and there exists a unique groupoid isomorphism \[\tilde{\Sigma}[\Omega^1_H] : \fr{G} \ltimes \left(\fr{At}/\fr{at}[\Omega^1_H]\right) \iso \cG[\Omega^1_H]/\ker\mu[\Omega^1_H],\] such that
	\begin{equation}
		\forall (\phi,\nabla) \in \fr{G} \ltimes \fr{At}, \quad \tilde{\Sigma}[\Omega^1_H](\phi,\nabla+\fr{at}[\Omega^1_H]) = \left[\pi[\Omega^1_H] \circ \Sigma[\Omega^1_H](\phi,\nabla)\right]_{\ker\mu[\Omega^1_H]}.
	\end{equation}
	
\end{theorem}

\begin{proof}
	Let us first show that \(\fr{at}[\Omega^1_H]\) is \(\fr{G}\)-invariant. Let  \(\bA \in \fr{at}[\Omega^1_H]\), so that \(\bA = \omega[\bA] \circ \dt{P,\ver}\). Then, for any \(\phi \in \fr{G}\),
	\[
		\phi \act \bA = \phi_{\ast,\hor} \circ \omega[\bA] \circ \dt{P,\ver} \circ \phi^{-1} = \left(\phi_{\ast} \circ \omega[\bA] \circ \phi^{-1}_{\ast,\ver}\right) \circ \dt{P,\ver},
	\]
	where, by properties of \(\phi_{\ast}\) and \(\phi^{-1}\), the map \(\phi_{\ast,\hor} \circ \omega[\bA] \circ \phi_{\ast,\ver}^{-1} : \Omega^1_{P,\ver} \to \Omega^1_{P,\hor}\) remains a left \(H\)-covariant morphism of \(P\)-\(\ast\)-bimodules, so that \(\phi \act \bA \in \fr{at}[\Omega^1_H]\), with \[\omega[\phi\act\bA] = \phi_{\ast,\hor} \circ \omega[\bA] \circ \phi_{\ast,\ver}^{-1}.\]
	
	Now, by Proposition~\ref{equiv1}, the surjective covering \(\pi[\Omega^1_H] : \cG[\Omega^1_H] \ltimes \cA[\Omega^1_H] \to \cG[\Omega^1_H]\), the star-injective groupoid homomorphism \(\mu[\Omega^1_H] : \cG[\Omega^1_H] \to \Aut(P)\), the injective group\-oid homomorphism \(\Sigma[\Omega^1_H] : \fr{G} \ltimes \fr{At} \to \cG[\Omega^1_H] \ltimes \cA[\Omega^1_H]\), and the left inverse \(A[\Omega^1_H]\) of \(\Sigma[\Omega^1_H]\) combine to satisfy the hypotheses of Lemma~\ref{groupoidlemma}. Thus, the subgroupoid \(\ker\mu[\Omega^1_H]\) is wide and has trivial isotropy groups, the quotient groupoid \(\cG[\Omega^1_H]/\ker\mu[\Omega^1_H]\) is well-defined, the equivalence kernel \(\sim\) of the map
	\[
		\left(\nabla \mapsto \left[\pi[\Omega^1_H] \circ \Sigma[\Omega^1_H](\id_P,\nabla)\right]_{\ker\mu[\Omega^1_H]}\right) : \fr{At} \to \Ob(\cG[\Omega^1_H]/\ker\mu[\Omega^1_H])
	\]
	is a \(\fr{G}\)-invariant equivalence relation on \(\fr{At}\), and there exists a unique groupoid isomorphism \(\tilde{\Sigma}[\Omega^1_H] : \fr{G} \ltimes \fr{At}/\sim \iso \cG[\Omega^1_H]/\ker\mu[\Omega^1_H]\), such that
	\[
		\forall (\phi,\nabla) \in \fr{G} \ltimes \fr{At}, \quad \tilde{\Sigma}[\Omega^1_H](\phi,[\nabla]_\sim) = \left[\Sigma[\Omega^1_H](\phi,\nabla)\right]_{\ker\mu[\Omega^1_H]}.
	\]
	Thus, it remains to show that \(\sim\) is the orbit equivalence with respect to the translation action of \(\fr{at}[\Omega^1_H] \leq \fr{at}\) on \(\fr{At}\).
	
	On the one hand, suppose that \(\nabla_1,\nabla_2 \in \fr{At}\) satisfy \(\nabla_1 \sim \nabla_2\), so that \(\id_P\) defines an arrow
	\[
		\left(\id_P : (\Omega^1_{P,\oplus},\dt{P,\nabla_1}) \to (\Omega^1_{P,\oplus},\dt{P,\nabla_2})\right) \in \ker\mu[\Omega^1_H] \subset \cG[\Omega^1_H].
	\]
	Thus, \((\id_P)_\ast : \Omega^1_{P,\oplus} \to \Omega^1_{P,\oplus}\) is a left \(H\)-covariant automorphism of the \(P\)-\(\ast\)-bimodule \(\Omega^1_{P,\oplus}\), such that \(\dt{P,\nabla_2} = (\id_P)_{\ast} \circ \dt{P,\nabla_1}\) and 
	\begin{gather*}
		\Proj_1 \circ (\id_P)_\ast = \ver[\dt{P,\nabla_2}] \circ (\id_P)_{\ast} = (\id_P)_{\ast,\ver} \circ \ver[\dt{P,\nabla_1}] = \Proj_1.
	\end{gather*}
	Let \(\mathcal{N} \coloneqq (\id_P)_\ast - \id\). Now, for all \(p \in P\),
	\[
		0 = \dv{P}(p) - \dv{P}(p) = \Proj_1 \circ \dt{P,\nabla_2} - \Proj_1 \circ \dt{P,\nabla_1} = \Proj_1 \circ \mathcal{N} \circ \dt{P,\nabla_1},
	\]
	so that \(\ran \mathcal{N} \subset \Omega^1_{P,\hor}\), while for all \(p \in P\) and \(b \in B\),
	\[
		(\id_P)_\ast(p \cdot \dt{B}(b)) = p \cdot \dt{B}(b) = \id(p \cdot \dt{B}(b)),
	\]
	so that \(\Omega^1_{P,\hor} \subset \ker \mathcal{N}\), and hence \(\mathcal{N} \circ \Pi_\oplus = \mathcal{N}\). Thus, for all \(p \in P\),
	\[
		(\nabla_2-\nabla_1)(p) = \Proj_2 \circ (\dt{P,\nabla_2}-\dt{P,\nabla,1}) (p) = \Proj_2 \circ \mathcal{N} \circ \Pi_{\oplus} \circ \dt{P,\nabla_1} = \Proj_2 \circ \mathcal{N} \circ \dv{P}(p),
	\]
	so that \(\nabla_2 - \nabla_1 \in \fr{at}[\Omega^1_H]\) with \(\omega[\nabla_2-\nabla_1] = \rest{\Proj_2 \circ \mathcal{N}\,}{\Omega^1_{P,\ver}}\).	
	
	On the other hand, suppose that \(\nabla_1,\nabla_2 \in \fr{At}\) satisfy \(\nabla_2 - \nabla_1 \in \fr{at}[\Omega^1_H]\). Let \[N \coloneqq \omega[\nabla_2-\nabla_1],\] and observe that \((N \circ \Pi_{\oplus})^2 = 0\). Then, for all \(p, \p{p},\pp{p} \in P\),
	\begin{align*}
		p \cdot \dt{P,\nabla_2}(\p{p}) \cdot \pp{p} &= p \cdot \left(\dv{P}(\p{p}),\nabla_2(\p{p})\right) \cdot \pp{p}	\\
		&= p \cdot \left(\dv{P}(\p{p}),\nabla_1(\p{p}) + N \circ \dv{P}(\p{p}) \right) \cdot \pp{p}\\
		&= p \cdot \left(\dt{P,\nabla_1}(\p{p}) + N \circ \Pi_{\oplus} \circ \dt{P,\nabla_1}(\p{p})\right) \cdot \pp{p}\\
		&= \left(\id + (N \circ \Pi_{\oplus})\right)\mleft(p \cdot \dt{P,\nabla_1}(\p{p}) \cdot \pp{p}\mright),
	\end{align*}
	so that \((\id_P)_\ast = \id + (N \circ \Pi_{\oplus})\) is a well-defined bijection satisfying
	\[
		\ver[\dt{P,\nabla_2}] \circ (\id_P)_\ast = \Proj_1 \circ \left(\id + (N \circ \Pi_{\oplus})\right) = \Proj_1 = (\id_P)_{\ast,\ver} \circ \ver[\dt{P,\nabla_1}].
	\]
	Hence, \(\id_P : (\Omega^1_{P,\oplus},\dt{P,\nabla_1}) \to (\Omega^1_{P,\oplus},\dt{P,\nabla_2})\) defines an arrow in \(\ker\mu[\Omega^1_H] \subset \cG[\Omega^1_H]\).
\end{proof}

\begin{remark}[cf.\ Zucca~\cite{Zucca}*{\S 8.4}]
	Let \((\Omega^1_P,\dt{P}) \in \Ob(\cG[\Omega^1_H])\). We can view the stabilizer subgroup of
	\(
		\tilde{\Sigma}[\Omega^1_H]^{-1}([(\Omega^1_P,\dt{P})]_{\ker\mu[\Omega^1_H]})
	\)
	in \(\fr{G}\) as the gauge group of the quantum principal \((H;\Omega^1_H,\dt{H})\)-bundle \((P;\Omega^1_P,\dt{P})\). Moreover, \(\Sigma[\Omega^1_H]\) restricts to a gauge-equivariant bijection from \(\fr{at}[\Omega^1_H]\) to the set of all strong bimodule connections on \((P;\Omega^1_P,\dt{P})\).
\end{remark}

\section{Gauge theory to second order}\label{sec3}

Our goal for this section is to develop and justify a conceptually economical notion of curvature of a principal connection on a quantum principal bundle that only involves differential calculus to second order and remains simultaneously compatible with both the theory of quantum principal \(H\)-bundles and strong bimodule connections \`{a} la Brzezi\'{n}ski--Majid and our notions of gauge transformations and (relative) gauge potentials. Our guiding principle will be prolongability of principal \fodc{} to second order in a manner compatible with decomposition into vertical and horizontal calculi.

\subsection{Deconstruction of quantum principal bundles to second order}

We begin by prolonging the standard theory of quantum principal \(H\)-bundles and strong bimodule connections \`{a} la Brzezi\'{n}ski--Majid to second order in the spirit of Beggs--Brzezi\'{n}ski~\cite{BB} and Beggs--Majid~\cite{BeM}*{\S 5.5}; in the process, we recover certain insights of \DJ{}ur\dj{}evi\'{c}~\cite{Dj97}*{\S 4}.

Let \((\Omega^1_H,\dt{H})\) be a bicovariant \textsc{fodc} on \(H\) with corresponding left crossed \(H\)-\(\ast\)-module \(\Lambda^1_H\) of right coinvariant \(1\)-forms and quantum Maurer--Cartan form \(\varpi_H\); let \((\Omega_H,\dt{H})\) be a prolongation of \((\Omega^1_H,\dt{H})\) to a bicovariant \(\ast\)-differential calculus on \(H\). We wish to consider differentiable quantum principal \(H\)-bundles compatible with \((\Omega_H,\dt{H})\) through degree \(2\). To this end, we will find it convenient to encode \((\Omega_H,\dt{H})\) in terms of the graded left crossed \(H\)-module \(\ast\)-algebra of right \(H\)-coinvariant forms
\begin{equation}
	\Lambda_H \coloneqq (\Omega_H)^{\operatorname{co}H}
\end{equation}
and the restricted differential \(\rest{\dt{H}}{\Lambda_H} : \Lambda_H \to \Lambda_H\). By results of Majid--Tao~\cite{MT}*{Prop.\ 3.3, 3.4}, it follows that \(\rest{\dt{H}}{
\Lambda^1_H} : \Lambda^1_H \to \Lambda^2_H\) is given by the Maurer--Cartan equation
\begin{equation}
	\forall h \in H, \quad \dt{H}(\varpi_H(h)) \coloneqq \varpi_H(\cm{h}{1}) \wedge \varpi_H(\cm{h}{2}).
\end{equation}
and satisfies
\begin{multline}\label{deg2eq}
	\forall h,k \in H, \\ \dt{H}\mleft(h \act \varpi_H(k)\mright) - h \act \dt{H}\varpi(k) = \cm{h}{1} \act \varpi_H(k) \wedge \varpi_H(\cm{h}{2}) + \varpi_H(\cm{h}{1}) \wedge \cm{h}{2}\act\varpi_H(k).
\end{multline}

The proof that \(\Omega_H\) can be recovered from \(\Lambda_H\) up to isomorphism~\cite{MT}*{Prop.\ 3.3}, \emph{mutatis mutandis}, lets us make the following definition, which we shall use to reconstruct graded \(\ast\)-algebras of (total) differential forms on a  quantum principal \(H\)-bundle from relevant graded \(\ast\)-algebras of vertical and horizontal forms, respectively. Note that we shall only be concerned with \(\ast\)-differential calculi and graded \(\ast\)-algebras through degree \(2\).

\begin{definition}[cf.\ \DJ{}ur\dj{}evi\'{c}~\cite{Dj97}*{Eqq.\ 4.49, 4.50}]
	Let \(\Omega\) be a graded left \(H\)-comodule \(\ast\)-algebra. We define the graded left \(H\)-comodule \(\ast\)-algebra \(\Lambda_H \dvatimes \Omega\) as follows:
	\begin{enumerate}
		\item \((\Lambda_H \dvatimes \Omega)^0 \coloneqq \Omega^0\) as a left \(H\)-comodule \(\ast\)-algebra;
		\item \((\Lambda_H \dvatimes \Omega)^1 \coloneqq (\Lambda^1_H \otimes \Omega^0) \oplus \Omega^1\) as a left \(H\)-comodule right \(\Omega^0\)-module together with the left \(\Omega^0\)-module structure defined by
		\begin{equation}
			\forall p,q \in \Omega^0, \, \forall \omega \in \Lambda^1_H, \, \forall \alpha \in \Omega^1, \quad q \cdot (\omega \otimes p,\alpha) \coloneqq (\ca{q}{-1}\act\omega \otimes \ca{q}{0}p,q\cdot\alpha)
		\end{equation}
		and the \(\ast\)-structure defined by
		\begin{equation}
			\forall p \in \Omega^0, \, \forall \omega \in \Lambda^1_H, \, \forall \alpha \in \Omega^1, \quad (\omega \otimes p,\alpha)^\ast \coloneqq (\ca{p}{-1}^\ast \act \omega^\ast \otimes \ca{p}{0}^\ast,\alpha^\ast);
		\end{equation}
		\item \((\Lambda_H \dvatimes \Omega)^1 \coloneqq (\Lambda^2_H \otimes \Omega^0) \oplus (\Lambda^1_H \otimes \Omega^1) \oplus \Omega^2\) as a left \(H\)-comodule right \(\Omega^0\)-module together with the left \(\Omega^0\)-structure defined by
		\begin{multline}
			\forall \omega \in \Lambda^1_H, \, \forall \mu \in \Lambda^2_H, \,\forall p,q \in \Omega^0,\,\forall \alpha \in \Omega^1,\,\forall \beta \in \Omega^2,\\
			q \cdot (\mu \otimes p,\omega \otimes \alpha,\beta) \coloneqq (\ca{q}{-1} \act \mu \otimes \ca{q}{0}p,\ca{q}{-1}\act\omega\otimes \ca{q}{0} \cdot \alpha,q \cdot \beta)
		\end{multline}
		and the \(\ast\)-structure defined by
		\begin{multline}
			\forall \omega \in \Lambda^1_H, \, \forall \mu \in \Lambda^2_H, \,\forall p \in \Omega^0,\,\forall \alpha \in \Omega^1,\,\forall \beta \in \Omega^2,\\
			(\mu \otimes p,\omega \otimes \alpha,\beta)^\ast \coloneqq (\ca{p}{-1}^\ast \act \mu^\ast \otimes \ca{p}{0}^\ast,\ca{p}{-1}^\ast \act\omega^\ast \otimes \ca{p}{0}^\ast,\beta^\ast);
		\end{multline}
		\item the wedge product \(\wedge : (\Lambda_H \dvatimes \Omega)^1 \otimes_{\Omega^0} (\Lambda_H \dvatimes \Omega)^1 \to (\Lambda_H \dvatimes \Omega)^2\) is defined by
		\begin{multline}
			\forall \omega,\p{\omega}\in\Lambda^1_H,\,\forall p,\p{p}\in \Omega^0, \, \forall \alpha,\p{\alpha} \in \Omega^1,\\
				(\omega \otimes p,\alpha) \wedge (\p{\omega}\otimes \p{p},\p{\alpha}) \coloneqq (\omega \wedge \ca{p}{-1}\act\p{\omega} \otimes \ca{p}{0}\p{p},\omega \otimes p\cdot\p{\alpha}+\ca{\alpha}{-1}\act\p{\omega}\otimes\ca{\omega}{0}\cdot\p{p},\alpha \wedge \p{\alpha});
		\end{multline}
		\item \((\Lambda_H \dvatimes \Omega)^k \coloneqq 0\) for \(k > 2\).
	\end{enumerate}
\end{definition}

For greater notational simplicity, given a graded left \(H\)-comodule \(\ast\)-algebra \(\Omega\), we shall view \(\Lambda_H \dvatimes \Omega\) as the graded left \(H\)-comodule \(\ast\)-algebra, truncated at degree \(2\), generated by the graded left \(H\)-subcomodule \(\ast\)-subalgebras \(\Lambda_H\) and \(\Omega\) subject the relation \(1_{\Lambda_H} = 1_\Omega\) and the braided graded commutation relations
\begin{equation}
	\forall \omega \in \Lambda_H, \, \forall \alpha \in \Omega, \quad \alpha \wedge \omega = (-1)^{\abs{\alpha}\abs{\omega}} \ca{\alpha}{-1} \act \omega \wedge \ca{\alpha}{0}.
\end{equation}

Now, let \(P\) once more be a principal left \(H\)-comodule \(\ast\)-algebra over \(\bC\) with \(B \coloneqq \coinv{H}{P}\). If the bicovariant \fodc{} \((\Omega^1_H,\dt{H})\) on \(H\) is locally freeing for \(P\), then the proof~\cite{MT}*{\S 3.1} that \((\Omega_H,\dt{H})\) can be recovered from the data \((\Lambda_H,\rest{\dt{H}}{\Lambda_H})\), \emph{mutatis mutandis}, lets us extend the induced (first-order) vertical calculus on \(P\) to a left \(H\)-covariant \textsc{sodc} as follows.

\begin{definition}[\DJ{}ur\dj{}evi\'{c}~\cite{Dj97}*{Lemma 3.1}]
	Suppose that the bicovariant \fodc{} \((\Omega^1_H,\dt{H})\) is locally freeing for \(P\). The \emph{second-order vertical calculus} of \(P\) is the extension of the vertical calculus \((\Omega^1_{P,\ver},\dv{P})\) to a left \(H\)-covariant \textsc{sodc} \((\Omega_{P,\ver},\dv{P})\) on \(P\) defined as follows:
	\begin{enumerate}
		\item \(\Omega_{P,\ver} \coloneqq \Lambda_H \dvatimes P\) as a graded left \(H\)-comodule \(\ast\)-algebra, where \(P\) is trivially extended to a graded left \(H\)-comodule \(\ast\)-algebra;
		\item the derivative \(\dv{P} : \Omega_{P,\ver} \to \Omega_{P,\ver}\) is given by
		\begin{equation}
			\forall \omega \in \Lambda_H, \, \forall p \in P, \quad \dv{P}(\omega \cdot p) \coloneqq \dt{H}(\omega) \cdot p+(-1)^{\abs{\omega}}\omega \cdot \dv{p}(p).
		\end{equation}
	\end{enumerate}
\end{definition}

We can now refine the standard definition of differentiable quantum principal \(H\)-bundle to account for differential calculus through degree \(2\). This can be viewed as a distillation of recent results of Beggs--Majid~\cite{BeM}*{\S 5.5} that specialise Beggs--Brzezi\'{n}ski's theory of noncommutative fibrations~\cite{BB} to the case of quantum principal bundles; it also echoes \DJ{}ur\dj{}evi\'{c}'s notion of differentiable quantum principal bundle~\cite{Dj97}*{\S 3}.

\begin{definition}[cf.\ Beggs--Majid~\cite{BeM}*{\S 5.5}]\label{sodc}
	Suppose that the bicovariant \fodc{} \((\Omega^1_H,\dt{H})\) on \(H\) is locally freeing for the principal left \(H\)-comodule \(\ast\)-algebra \(P\); hence, let \((\Omega_{P,\ver},\dv{P})\) be the second-order vertical calculus of \(P\) with respect to the bicovariant prolongation \((\Omega_H,\dt{H})\) of \((\Omega^1_H,\dt{H})\). Let \((\Omega_P,\dt{P})\) be a left \(H\)-covariant \textsc{sodc} on \(P\) . Define the restriction \((\Omega_B,\dt{B})\) of \((\Omega_P,\dt{P})\) to a \textsc{sodc} on \(B \coloneqq \coinv{H}{P}\) by
	\[
		\Omega^1_B \coloneqq B \cdot \dt{P}(B), \quad \Omega^2_B \coloneqq \Omega^1_B \wedge \Omega^1_B, \quad \dt{B} \coloneqq \rest{\dt{P}}{\Omega_B}.
	\]
	Then \((P;\Omega,\dt{P})\) defines a \emph{strong (second-order) quantum principal \((H;\Omega_H,\dt{H})\)-bundle} if and only if the following all hold:
	\begin{enumerate}
		\item\label{socd1} the pair \((\Omega^1_P,\dt{P})\) defines an \((H;\Omega^1_H,\dt{H})\)-principal \fodc{} on \(P\) whose vertical map \(\ver[\dt{P}] : \Omega^1_P \to \Omega^1_{P,\ver}\) satsfies
		\begin{equation}
			\ker \ver[\dt{P}] = P \cdot \Omega^1_B;
		\end{equation}
		\item\label{socd2} the map \(\ver^{2,2}[\dt{P}] : \Omega^2_P \to \Omega^2_{P,\ver}\) given by
		\begin{equation}
			\forall \alpha,\beta \in \Omega^1_P, \quad \ver^{2,2}[\dt{P}](\alpha \wedge \beta) \coloneqq \ver[\dt{P}](\alpha) \wedge \ver[\dt{P}](\beta)
		\end{equation}
		is well-defined, is surjective, and satisfies
		\begin{equation}
			\ker(\ver^{2,2}[\dt{P}]) = \Omega^1_P \wedge \Omega^1_B;
		\end{equation}
		\item\label{socd3} the map \(\ver^{2,1}[\dt{P}] : \Omega^2_P \to \Lambda^1_H \otimes \Omega^1_{P} \subset (\Lambda_H \dvatimes \Omega_P)^2\) given by
		\begin{equation}
			\forall \alpha,\beta \in \Omega^1_P, \quad \ver^{2,1}[\dt{P}](\alpha\wedge\beta) \coloneqq \ver[\dt{P}](\alpha) \wedge \beta + \alpha \wedge \ver[\dt{P}](\beta)
		\end{equation}
		is well-defined and satisfies
		\begin{align}
			\ker(\ver^{2,1}[\dt{P}]) \cap \ker(\ver^{2,2}[\dt{P}])&= P \cdot \Omega^2_B, \\ 
			\ver^{2,1}[\dt{P}](\ker(\ver^{2,2}[\dt{P}])) & = \Lambda^1_H \otimes P \cdot \Omega^1_B.
		\end{align}
	\end{enumerate}
	In this case, we call \((\Omega_P,\dt{P})\) a \emph{strongly \((H;\Omega_H,\dt{H})\)-principal} \sodc{} on \(P\), and we define the left \(H\)-covariant graded \(\ast\)-subalgebra \(\Omega_{P,\hor}\) of \emph{horizontal forms} in \(\Omega_P\) by
	\[
		\Omega^0_{P,\hor} \coloneqq P, \quad \Omega^1_{P,\hor} \coloneqq P \cdot \Omega^1_B, \quad \Omega^2_{P,\hor} \coloneqq P \cdot \Omega^2_B.
	\]
\end{definition}

\begin{remark}\label{verremark} 
	We can provide the following conceptual interpretation of Definition~\ref{sodc}. On the one hand, by condition~\ref{socd2}, the map \(\ver[\dt{P}] : \Omega^1_P \to \Omega^1_{P,\ver}\) extends via \(\ver^{2,2}[\dt{P}]\) to a left \(H\)-covariant graded \(\ast\)-epimorphism \(\ver[\dt{P}] : \Omega_P \to \Omega_{P,\ver}\), such that
		\[
			\rest{\ver[\dt{P}]}{P} = \id_P, \quad \ker\ver[\dt{P}] = \Omega_P \wedge \Omega^1_{P,\hor};
		\]
		following \DJ{}ur\dj{}evi\'{c}~\cite{Dj97}*{\S 3}, we interpret this extension of \(\ver[\dt{P}]\) as encoding restriction of differential forms to orbitwise differential forms. On the other hand, by condition~\ref{socd3}, the rescaled vertical map \(-\iu{}\,\ver[\dt{P}] : \Omega^1_P \to \Omega^1_{P,\ver}\) extends via \(-\iu{}\,\ver^{2,1}[\dt{P}]\) to a degree \(0\) left \(H\)-covariant \(\ast\)-derivation \(\Int[\dt{P}] : \Omega_P \to \Lambda_H \dvatimes \Omega_P\), such that
		\[
			\rest{\Int[\dt{P}]}{P} = 0, \quad \ker \Int[\dt{P}] \cap \ker\ver[\dt{P}] = \Omega^2_{P,\hor}, \quad \Int[\dt{P}](\ker\ver[\dt{P}]) = \Lambda^1_H \otimes \Omega^1_{P,\hor};
		\]
		we interpret the map \(\Int[\dt{P}] : \Omega_P \to \Lambda_H \dvatimes \Omega_P\) as encoding contraction of differential forms with fundamental vector fields. Proposition~\ref{strongprop} now implies that
		\[
			\ker \Int[\dt{P}] \cap \ker\ver[\dt{P}] = \Omega^1_{P,\hor} \oplus \Omega^2_{P,\hor}
		\]
		recovers the graded \(\ast\)-ideal of horizontal forms of positive degree, so that the equality \[\coinv{H}{(\ker \Int[\dt{P}] \cap \ker\ver[\dt{P}])} = \Omega^1_B \oplus \Omega^2_B\] recovers the basic differential forms as the \(H\)-coinvariant horizontal differential forms.
\end{remark}

\begin{remark}
	Suppose that \((\Omega_H,\dt{H})\) is Woronowicz's canonical prolongation~\cite{Woronowicz}*{\S\S 3--4} of the bicovariant \fodc{} \((\Omega^1_H,\dt{H})\) on \(H\), so that \(\Omega_H\) defines a graded super-Hopf \(\ast\)-algebra~\cite{Br93}. Suppose that \((\Omega_P,\dt{P})\) is a \(\ast\)-differential calculus on \(P\), such that the left \(H\)-coaction of \(H\) on \(P\) extends to a differentiable left \(\Omega_H\)-coaction \(\hat{\delta}_{\Omega_P}\) of the graded super-Hopf \(\ast\)-algebra \(\Omega_H\) on \(\Omega_P\) (cf.\ Beggs--Majid~\cite{BeM}*{\S 5.5}), so that \((P;\Omega_P,\dt{P})\) is a quantum principal \((H;\Omega_H,\dt{H})\)-bundle \`{a} la \DJ{}ur\dj{}evi\'{c}~\cite{Dj97}*{\S 3}, and hence the map \(\ver[\dt{P}] : \Omega_P \to \Omega_{P,\ver}\) of Remark~\ref{verremark} is a well-defined surjective morphism of left \(H\)- and \(\Omega_H\)-covariant \(\ast\)-\textsc{dga}~\cite{Dj97}*{Prop.\ 3.9}. Let \(\Omega_B\) be the graded \(\ast\)-subalgebra of left \(\Omega_H\)-coinvariants generated by \(B\) and \(\dt{P}(B)\). Then the degree \(2\) truncation \((\Omega_P^{\leq 2},\dt{P})\) of \((\Omega_P,\dt{P})\) is a strongly \((H;\Omega_H,\dt{H})\)-principal \sodc{} on \(P\) whenever
	\[
		\ker\ver[\dt{P}] = \Omega_P \wedge \Omega_B, \quad
		\set{\omega \in \Omega_P \given \hat{\delta}_{\Omega_P}(\omega) = \delta_{\Omega_P}(\omega)} = P \cdot \Omega_B, 
	\]
	in which case, \(\coinv{\Omega_H}{\Omega_P} = \Omega_B\) by Proposition \ref{strongprop}.
\end{remark}

\begin{remark}[Beggs--Majid~\cite{BeM}*{Lemma 5.60, Theorem 5.61}] In the context of Definition~\ref{sodc}, suppose that the \(B\)-\(\ast\)-bimodule of basic \(1\)-forms \(\Omega^1_B \coloneqq B \cdot \dt{P}(B)\) is flat as a left \(B\)-module (e.g., finitely generated and projective), that conditions~\ref{socd1} and~\ref{socd2} are satisfied, and that the map \(\ver^{2,1}[\dt{P}] : \Omega^2_P \to \Lambda^1_H \otimes \Omega^1_{P}\) of condition~\ref{socd3} is well-defined. Then \((\Omega_P,\dt{P})\) is a strongly \((H;\Omega_H,\dt{H})\)-principal \sodc{} on the principal left \(H\)-comodule \(\ast\)-algebra \(P\), and the inclusion \(\Omega_B \inj \Omega_P\) defines (to second order) a noncommutative fibration in the sense of Beggs--Brzezi\'{n}ski~\cite{BB}.
\end{remark}

We now refine the definition of strong bimodule connection to the context of strong second-order quantum principal bundles, thereby providing a straightforward notion of noncommutative principal connection compatible with non-universal \sodc{}. We shall soon see that this refinement is related to \DJ{}ur\dj{}evi\'{c}'s notion of multiplicative connection~\cite{Dj97}*{Def.\ 4.2}, and as such resolves (through degree \(2\)) the following question implicitly flagged by Hajac~\cite{Hajac}*{\S 4}: when does a connection  extend from \(\Omega^1_P\) to all of \(\Omega_P\)?

\begin{definition}\label{prolongcon}
	Suppose that the bicovariant \fodc{} \((\Omega^1_H,\dt{H})\) on \(H\) is locally freeing for \(P\) and that \((\Omega_P,\dt{P})\) is a strongly \((H;\Omega_H,\dt{H})\)-principal \sodc{} on the principal left \(H\)-comodule \(\ast\)-algebra \(P\). A connection \(\Pi\) on the quantum principal \((H;\Omega^1_H,\dt{H})\)-bundle \((P;\Omega^1_P,\dt{P})\) is called \emph{prolongable} if and only if it is a bimodule connection that satisfies both of the following:
	\begin{enumerate}
		\item\label{prolongcon1} the map \(\Pi \wedge \Pi : \Omega^2_P \to \Omega^2_P\) given by
		\begin{equation}
			\forall \alpha, \beta \in \Omega^1_P, \quad \Pi \wedge \Pi(\alpha \wedge \beta) \coloneqq \Pi(\alpha) \wedge \Pi(\beta)
		\end{equation}
		is well-defined and satisfies \(\ker(\Pi \wedge \Pi) = \Omega^1_P \wedge \Omega^1_{P,\hor}\).
		\item\label{prolongcon2} the map \(\Pi \wedge \id + \id{} \wedge \Pi : \Omega^2_P \to \Omega^2_P\) given by
		\begin{equation}
			\forall \alpha,\beta \in \Omega^1_P, \quad (\Pi \wedge \id + \id \wedge \Pi)(\alpha \wedge \beta) \coloneqq \Pi(\alpha) \wedge \beta + \alpha \wedge \Pi(\beta)
		\end{equation}
		is well-defined and satisfies \(\ker((\Pi \wedge \id + \id \wedge \Pi) - \Pi \wedge \Pi) = \Omega^2_{P,\hor}\).
	\end{enumerate}
\end{definition}

\begin{remark}
	Suppose that the bimodule connection \(\Pi\) satisfies Condition~\ref{prolongcon1}. Then \(\Pi\) satisfies condition~\ref{prolongcon2} if and only if the map \((\id-\Pi) \wedge (\id-\Pi) : \Omega^2_P \to \Omega^2_P\) given by
	\[
		\forall \alpha, \beta \in \Omega^1_P, \quad (\id-\Pi) \wedge (\id-\Pi)(\alpha \wedge \beta) \coloneqq (\id-\Pi)(\alpha) \wedge (\id-\Pi)(\beta)
	\]
	is well-defined and satisfies \(\ker(\id - (\id-\Pi)\wedge(\id-\Pi)) = \Omega^2_{P,\hor}\).
\end{remark}

Suppose that \((\Omega_P,\dt{P})\) is a strongly \((H;\Omega_H,\dt{H})\)-principal \sodc{} on the principal left \(H\)-comodule \(\ast\)-algebra \(P\). Recall that a bimodule connection \(\Pi\) on the quantum princnipal \((H;\Omega^1_H,\dt{H})\)-bundle \((P;\Omega^1_P,\dt{P})\) yields a resolution of the identity \(\set{\Pi,\id-\Pi}\) on \(\Omega^1_P\) realising \(\ker\ver[\dt{P}] = \Omega^1_{P,\hor}\) as a complementable left \(H\)-subcomodule \(P\)-\(\ast\)-subbimodule of \(\Omega^1_P\). Straightforward calculations now show that a prolongable bimodule connection \(\Pi\) yields a compatible resolution of the identity \(\set{\Pi_{2,2},\Pi_{2,1},\Pi_{2,0}}\) on \(\Omega^2_P\) that realises
\[
	\ker(\ver^{2,2}[\dt{P}]) = \Omega^1_P \wedge \Omega^1_{P,\hor}, \quad \ker(\ver^{2,1}[\dt{P}]) \cap \ker(\ver^{2,2}[\dt{P}]) = \Omega^2_{P,\hor}
\]
as complementable left \(H\)-subcomodule \(P\)-\(\ast\)-subbimodules of \(\Omega^2_P\).

\begin{proposition}
	Suppose that the bicovariant \fodc{} \((\Omega^1_H,\dt{H})\) on \(H\) is locally freeing for \(P\), that \((\Omega_P,\dt{P})\) is a strongly \((H;\Omega_H,\dt{H})\)-principal \sodc{} on \(P\), and that \(\Pi\) is a prolongable bimodule connection on the strong quantum principal \((H;\Omega_H,\dt{H})\)-bundle \((P;\Omega_P,\dt{P})\). 
	\begin{enumerate}
		\item The maps \(\Pi_{1,1},\Pi_{1,0} : \Omega^1_P \to \Omega^1_P\) given by
		\[
			\Pi_{1,1} \coloneqq \Pi, \quad \Pi_{1,0} \coloneqq \id-\Pi,
		\]
		respectively, define an orthogonal pair of idempotent left \(H\)-covariant morphisms of \(P\)-\(\ast\)-bimodules, such that \(\Pi_{1,1}+\Pi_{1,0}=\id\) and
		\[
			\ran \Pi_{1,0} = \ker \Pi_{1,1} = P \cdot \Omega^1_B.
		\]
		\item The maps \(\Pi_{2,2},\Pi_{2,1},\Pi_{2,0}: \Omega^2_P \to \Omega^2_P\) given by
		\begin{gather*}
			\Pi_{2,2} \coloneqq \Pi \wedge \Pi, \quad \Pi_{2,1} \coloneqq (\Pi\wedge\id{}+\id{}\wedge\Pi) -2(\Pi\wedge\Pi), \\ \Pi_{2,0} \coloneqq \id - (\Pi\wedge\id{}+\id{}\wedge\Pi) + \Pi \wedge \Pi, 
		\end{gather*}
		respectively, define a pairwise orthogonal triple of idempotent left \(H\)-covariant morphisms of \(P\)-\(\ast\)-bimodules, such that \(\Pi_{2,2}+\Pi_{2,1}+\Pi_{2,0}=\id\) and
		\[
			\ran \Pi_{2,0} = \ker(\Pi_{2,2}+\Pi_{2,1}) = \Omega^2_{P,\hor}, \quad \ran(\Pi_{2,1}+\Pi_{2,0}) = \ker(\Pi_{2,2}) = \Omega^1_P \wedge \Omega^1_{P,\hor};
		\]
		\item For all \(\alpha,\beta \in \Omega^1_P\),
		\begin{align*}
			\Pi_{2,2}(\alpha\wedge\beta) &= \Pi_{1,1}(\alpha)\wedge\Pi_{1,1}(\beta), \\ \Pi_{2,1}(\alpha\wedge\beta) &= \Pi_{1,1}(\alpha)\wedge \Pi_{1,0}(\beta) + \Pi_{1,0}(\alpha) \wedge \Pi_{1,1}(\beta), \\  \Pi_{2,0}(\alpha \wedge \beta) &= \Pi_{1,0}(\alpha)\wedge\Pi_{1,0}(\beta).
		\end{align*}
	\end{enumerate}
\end{proposition}

\begin{remark}\label{horremark}
	Recall the maps \(\ver[\dt{P}] : \Omega_P \to \Omega_{P,\ver}\) and \(\Int[\dt{P}] : \Omega_P \to \Lambda_H \dvatimes \Omega_P\) from Remark~\ref{verremark} induced by the vertical map \(\ver[\dt{P}] : \Omega^1_P \to \Omega^1_{P,\ver}\). On the one hand, the strong bimodule connection \(\Pi = \Pi_{1,1}\) extends via \(\Pi_{2,2}\) to a left \(H\)-covariant graded \(\ast\)-homomorphism \(\ver_\Pi : \Omega_P \to \Omega_P\), such that
	\[
		\rest{\ver_\Pi}{P} = \id, \quad (\ver_\Pi)^2 = \ver_\Pi, \quad \ker \ver_\Pi = \Omega^1_P \wedge \Omega^1_{P,\hor} = \ker \ver[\dt{P}];
	\]
	in particular, it follows that \(\ver[\dt{P}] \circ \ver_\Pi = \ver[\dt{P}]\). On the other, following an observation of \DJ{}ur\dj{}evi\'{c}~\cite{Dj97}*{Eq.\ 4.58}, we see that \(\id-\Pi = \Pi_{1,0}\) extends via \(\Pi_{2,0}\) to a left \(H\)-covariant graded \(\ast\)-homomorphism \(\hor_\Pi : \Omega_P \to \Omega_P\), such that
	\[
		\rest{\hor_\Pi}{P} = \id, \quad (\hor_\Pi)^2 = \hor_\Pi, \quad \ran \hor_\Pi = \Omega_{P,\hor} = P \oplus \ker \Int[\dt{P}] \cap \ker \ver[\dt{P}];
	\]
	in particular, it follows that \(\rest{\hor_\Pi}{\Omega_B} = \id\). Thus, in particular, the map \(\ver_\Pi\) yields the extension from \(\Omega^1_P\) to \(\Omega_P\)  of the strong bimodule connection \(\Pi\) required by Hajac for the discussion of curvature~\cite{Hajac}*{\S 4}.
\end{remark}

We now see that if \(\Pi\) is a prolongable bimodule projection on a strong quantum principal \((H;\Omega_H,\dt{H})\)-bundle \((P;\Omega_P,\dt{P})\), then \(\nabla_\Pi \coloneqq \Pi \circ \dt{P}\) extends to a lift of \(\dt{B} : \Omega_B \to \Omega_B\) to a degree \(1\) left \(H\)-covariant \(\ast\)-derivation on \(\Omega_{P,\hor}\), thereby addressing---if only through degree \(2\)---an open issue for connections \`{a} la Brzezi\'{n}ski--Majid that was first flagged by Hajac~\cite{Hajac}*{\S 4}. Furthermore, we shall also see that \(\nabla_\Pi^2\) directly yields a well-defined curvature \(2\)-form for \(\Pi\) \emph{qua} principal connection without invoking (absolute) connection \(1\)-forms in the sense of Brzezi\'{n}ski--Majid~\cite{BrM}*{Prop.\ 4.10} or \DJ{}ur\dj{}evi\'{c}~\cite{Dj97}*{Def.\ 4.1}.

\begin{propositiondefinition}[cf.\ Brzezi\'{n}ski--Majid~\cite{BrM}*{Appx.\ A and \S 3}, Hajac~\cite{Hajac}*{\S 4}, \DJ{}ur\dj{}evi\'{c}~\cite{Dj97}*{Def.\ 4.5 and Prop.\ 4.6}]
	Suppose that the bicovariant \fodc{} \((\Omega^1_H,\dt{H})\) on \(H\) is locally freeing for the principal left \(H\)-comodule \(\ast\)-algebra \(P\), that \((\Omega_P,\dt{P})\) is a strongly \((H;\Omega_H,\dt{H})\)-principal \sodc{} on \(P\), and that \(\Pi\) is a prolongable bimodule connection on the strong quantum principal \((H;\Omega_H,\dt{H})\)-bundle \((P;\Omega_P,\dt{P})\). 
	\begin{enumerate}
		\item The left \(H\)-covariant map \(\nabla_\Pi : \Omega_{P} \to \Omega_{P,\hor}\) given by
		\[
			\rest{\nabla_\Pi}{P} \coloneqq \Pi_{1,0} \circ \dt{P}, \quad \rest{\nabla_\Pi}{\Omega^1_P} \coloneqq \Pi_{2,0} \circ \dt{P}
		\]
		restricts to a degree \(1\) left \(H\)-covariant \(\ast\)-derivation \(\Omega_{P,\hor} \to \Omega_{P,\hor}\), such that
		\[
			\rest{\nabla_\Pi}{\Omega_B} = \dt{B};
		\]
		we call \(\nabla_\Pi\) the \emph{exterior covariant derivative} induced by \(\Pi\).
		\item The left \(H\)-covariant right \(P\)-linear map \(F_\Pi : \Omega^1_{P,\ver} \to \Omega^2_{P,\hor}\) given by
		\begin{equation}
			\rest{F_\Pi}{\Lambda^1_H} \coloneqq \iu{}\,\Pi_{2,0} \circ \dt{P} \circ \rest{(\rest{\ver[\dt{P}]}{\ran \Pi_{1,1}})^{-1}}{\Lambda^1_H}
		\end{equation}
		defines the unique left \(H\)-covariant morphism of \(P\)-\(\ast\)-bimodules, such that 
		\begin{equation}
			\rest{\nabla^2_\Pi}{P} = \iu{} \,F_\Pi \circ \dv{P};
		\end{equation}
		we call \(F_\Pi\) the \emph{curvature} of \(\Pi\).
	\end{enumerate}
\end{propositiondefinition}

\begin{proof}
	Let us first check that the map \(\nabla_\Pi\), which is left \(H\)-covariant and satisfies \[\forall \alpha \in \Omega_P, \quad \nabla_\Pi(\alpha^\ast) = -\nabla(\alpha)^\ast,\] restricts to a degree \(1\) \(\ast\)-derivation on \(\Omega_{P,\hor}\). Since \(\rest{\nabla_\Pi}{P}\) is the horizontal derivative of Proposition~\ref{analysisthm}, it remains to show that \(\rest{\nabla_\Pi}{\Omega^1_P}\) satisfies the correct graded Leibniz rules with respect to the products \(P \otimes \Omega^1_{P,\hor} \to \Omega^1_{P,\hor}\) and \(\Omega^1_{P,\hor} \otimes P \to \Omega^1_{P,\hor}\). Indeed, for all \(p \in P\) and \(\alpha \in \Omega^1_{P,\hor}\),
	\begin{gather*}
		\begin{multlined}\nabla_\Pi(p \cdot \alpha) = \Pi_{2,0} \circ \dt{P}(p \cdot \alpha) = \Pi_{2,0}(\dt{P}(p) \wedge \alpha + p \cdot \dt{P}(\alpha))\\ = \Pi_{1,0}(\dt{P}(p)) \wedge \Pi_{1,0}(\alpha) + p \cdot \Pi_{2,0}(\dt{P}(\alpha)) = \nabla_\Pi(p) \wedge \alpha + p \cdot \nabla_\Pi(\alpha),\end{multlined}\\
		\begin{multlined}
			\nabla_\Pi(\alpha \cdot p) = \Pi_{2,0} \circ \dt{P}(\alpha \cdot p) = \Pi_{2,0}(\dt{P}(\alpha) \cdot p - \alpha \wedge \dt{P}(p))\\ = \Pi_{2,0}(\dt{P}(\alpha)) \cdot p - \Pi_{1,0}(\alpha) \wedge \Pi_{1,0}(\dt{P}(p)) = \nabla_\Pi(\alpha) \cdot p - \alpha \wedge \nabla_\Pi(p)
		\end{multlined}
	\end{gather*}
	as required. Note that \(\rest{\nabla_\Pi}{\Omega_B} = \dt{B}\) because \(\rest{\Pi_{1,0}}{\Omega^1_B} = \id\) and \(\rest{\Pi_{2,0}}{\Omega^2_B} = \id\).
	
	Let us now check that \(F_\Pi\) is left \(P\)-linear and \(\ast\)-preserving. Set \(\theta \coloneqq (\rest{\ver[\dt{P}]}{\ran \Pi_{1,1}})^{-1}\). Then, for all \(p \in P\) and \(\omega \in \Lambda^1_H\),
	\begin{align*}
		p \cdot F_\Pi(\omega) &= p \cdot \left(\iu{}\, \Pi_{2,0} \circ \dt{P} \circ \theta(\omega)\right)\\
		&= \iu{}\, \Pi_{2,0}\mleft(\dt{P}(\theta(p \cdot \omega) - \dt{P}(p) \wedge \theta(\omega) \mright)\\
		&= \iu{}\, \Pi_{2,0} \mleft(\dt{P}(\theta(\ca{p}{-1} \act \omega) \cdot \ca{p}{0}) \mright)\\
		&= \iu{}\, \Pi_{2,0} \mleft(\dt{P}(\theta(\ca{p}{-1}) \act \omega) \cdot \ca{p}{0} - \theta(\ca{p}{-1}\act\omega) \wedge \dt{P}(\ca{p}{0})\mright)\\
		&= \iu{}\, \Pi_{2,0} \circ\dt{P}\mleft(\theta(\ca{p}{-1}\act\omega) \cdot \ca{p}{0}\mright)\\
		&= F_\Pi\mleft((\ca{p}{-1}\act\omega)\cdot\ca{p}{0}\mright),
	\end{align*}
	so that \(F_\Pi\) is left \(P\)-linear; that \(F_\Pi\) is \(\ast\)-preserving now follows from the observation that \(\rest{F_\Pi}{\Lambda^1_H} = -\iu{}\,\Pi_{2,0} \circ \dt{P} \circ \rest{\theta}{\Lambda^1_H}\) is \(\ast\)-preserving.
	
	Finally, let us check the relation between \(F_\Pi\) and \(\nabla^2\). For every \(p \in P\),
	\begin{align*}
		\nabla_\Pi^2(p) &= \nabla_\Pi(\dt{P}(p)-\theta \circ \dv{P}(p))\\
		&= \Pi_{2,0} \circ \dt{P}\mleft(\dt{P}(p) - \theta(\varpi_H(\ca{p}{-1})) \cdot \ca{p}{0} \mright)\\
		&= \Pi_{2,0}\mleft(-\dt{P} \circ \theta(\varpi_H(\ca{p}{-1})) \cdot \ca{p}{0} + \theta(\varpi_H(\ca{p}{-1})) \wedge \dt{P}(\ca{p}{0}) \mright)\\
		&= -\Pi_{2,0} \circ \dt{P} \circ \theta(\varpi_H(\ca{p}{-1})) \cdot \ca{p}{0}\\
		&= \iu{}\, F_\Pi \circ \dv{P}(p),
	\end{align*}
	so that \(\rest{\nabla_\Pi^2}{P} = \iu{} \,F_\Pi \circ \dv{P}\), as was claimed; since \(\Omega^1_{P,\ver} = P \cdot \dv{P}(P)\), it now follows that \(F_\Pi : \Omega^1_{P,\ver} \to \Omega^2_{P,\hor}\) is the unique such left \(H\)-covariant morphism of \(P\)-\(\ast\)-bimodules.
\end{proof}

\begin{remark}[cf.\ Brzezi\'{n}ski--Majid~\cite{BrM}*{Appx.\ A}, \DJ{}ur\dj{}evi\'{c}~\cite{Dj97}*{Prop.\ 4.16}]
	In terms of the map \(\hor_\Pi : \Omega_P \to \Omega_P\) defined in Remark~\ref{horremark}, 
	\[
		\nabla_\Pi = \hor_\Pi \circ \dt{P}, \quad \rest{F_\Pi}{\Lambda^1_H} = \hor_\Pi \circ \dt{P} \circ \rest{(\rest{\ver[\dt{P}]}{\ran \Pi})^{-1}}{\Lambda^1_H}.
	\]
\end{remark}

Given a prolongable connection \(\Pi\) on a strong quantum principal \((H;\Omega_H,\dt{H})\)-bundle \((P;\Omega_P,\dt{P})\), we can extend the decomposition of the \((H;\Omega^1_H,\dt{H})\)-principal \fodc{} \((\Omega^1_P,\dt{P})\) on \(P\) given by Proposition~\ref{analysisthm} to a left \(H\)-covariant isomorphism of graded \(\ast\)-alge\-bras \(\psi_\Pi : \Omega_P \to \Lambda_H \hotimes \Omega_{P,\hor}\), such that \(\psi_\Pi \circ \dt{P} \circ \psi_\Pi^{-1}\) can be explicitly expressed in terms of the quantum Maurer--Cartan form \(\varpi_H\) of the bicovariant \fodc{} \((\Omega^1_H,\dt{H})\) on \(H\), the gauge potential \(\nabla_\Pi\) induced by \(\Pi\), and the curvature \(F_\Pi\) of \(\Pi\). This will provide a roadmap for replacing prolongable connections with sufficiently well-behaved gauge potentials by the appropriate refinement of Proposition~\ref{equiv1}.

\begin{proposition}[cf.\ \DJ{}ur\dj{}evi\'{c}~\cite{Dj97}*{Thm.\ 4.12}]\label{isothm}
	Suppose that the bicovariant first-order differential calculus \((\Omega^1_H,\dt{H})\) on \(H\) is locally freeing for the principal left \(H\)-comodule \(\ast\)-algebra \(P\), that \((\Omega_P,\dt{P})\) is a strongly \((H;\Omega_H,\dt{H})\)-principal \sodc{} on \(P\), and that \(\Pi\) is a prolongable bimodule connection on the strong quantum principal \((H;\Omega_H,\dt{H})\)-bundle \((P;\Omega_P,\dt{P})\). The map \(\psi_\Pi : \Omega_P \to \Lambda_H \hotimes \Omega_{P,\hor}\) given by
	\[
		\rest{\psi_\Pi}{P} \coloneqq \id, \quad \rest{\psi_\Pi}{\Omega^1_P} \coloneqq \ver[\dt{P}] + \Pi_{1,0}, \quad \rest{\psi_\Pi}{\Omega^2_P} \coloneqq \ver^{2,2}[\dt{P}] + \ver^{2,1}[\dt{P}] \circ \Pi_{2,1} + \Pi_{2,0}
	\]
	defines a left \(H\)-covariant isomorphism of graded \(\ast\)-algebras, such that
	\begin{gather}
		\forall p \in P, \quad \psi_\Pi \circ \dt{P} \circ \psi_\Pi^{-1}(p) = \dv{P}(p) + \nabla_\Pi(p),\\
		\forall \omega \in \Lambda^1_H, \quad 	\psi_\Pi \circ \dt{P} \circ \psi_\Pi^{-1}(\omega) = \dt{H}(\omega) - \iu{}\,F_\Pi(\omega),\\
		\forall \alpha \in \Omega^1_{P,\hor}, \quad \psi_\Pi \circ \dt{P} \circ \psi_\Pi^{-1}(\alpha) = \varpi_H(\ca{\alpha}{-1}) \wedge \ca{\alpha}{0} +\nabla_\Pi(\alpha),
	\end{gather}
	and hence, in particular, \(\rest{\psi_\Pi \circ \dt{P} \circ \psi_\Pi^{-1}}{\Omega_B} = \dt{B}\).
\end{proposition}

\begin{proof}
	Before continuing, recall that the restriction \(\rest{\psi_\Pi}{\Omega^1_P}\) is a left \(H\)-covariant isomorphism of \(P\)-\(\ast\)-bimodules by Proposition~\ref{analysisthm}. First, observe the the left \(H\)-covariant graded \(\bC\)-linear map  \(\psi_\Pi : \Omega_P \to \Lambda_H \dvatimes \Omega_{P,\hor}\) is bijective: since \((P;\Omega_P,\dt{P})\) is a strong second-order quantum principal \((H;\Omega_H,\dt{H})\)-bundle and since \(\Pi\) is prolongable, it follows that
	\[
		\rest{\ver^{2,2}[\dt{P}]}{\ran \Pi_{2,2}} : \ran \Pi_{2,2} \to \Omega^2_{P,\ver}, \quad \rest{\ver^{2,1}}{\ran \Pi_{2,1}} : \ran \Pi_{2,1} \to \Lambda^1_H \otimes \Omega^1_{P,\hor},
	\]
	are both bijective, so that, in turn, the restriction
	\[
		\rest{\psi_\Pi}{\Omega^2_P} \coloneqq \ver^{2,2}[\dt{P}] + \ver^{2,1}[\dt{P}] \circ \Pi_{2,1} + \Pi_{2,0} = \ver^{2,2}[\dt{P}] \circ \Pi_{2,2} + \ver^{2,1}[\dt{P}] \circ \Pi_{2,1} + \Pi_{2,0}
	\]
	is also bijective. Next, to show that the map \(\psi_\Pi\) is a homomorphism, it suffices to show that it is multiplicative with respect to the product \(\Omega^1_P \otimes \Omega^1_P \to \Omega^2_P\); indeed, for all \(\alpha,\beta \in \Omega^1_P\),
	\begin{align*}
		\psi_\Pi(\alpha \wedge \beta) &= \ver^{2,2}[\dt{P}](\alpha\wedge\beta) +	 \ver^{2,1}[\dt{P}]\circ\Pi_{2,1}(\alpha\wedge\beta) + \Pi_{2,0}(\alpha\wedge\beta)\\
		&= \ver^{2,2}[\dt{P}](\alpha\wedge\beta) + \ver^{2,1}[\dt{P}]\mleft(\Pi_{1,1}(\alpha)\wedge\Pi_{1,0}(\beta)+\Pi_{1,0}(\alpha)\wedge\Pi_{1,1}(\beta)\mright)\\ &\quad\quad + \Pi_{1,0}(\alpha)\wedge\Pi_{1,0}(\beta)\\
		&= \ver[\dt{P}](\alpha)\wedge\ver[\dt{P}](\beta) + \ver[\dt{P}](\alpha)\wedge\Pi_{1,0}(\beta) + \Pi_{1,0}(\alpha)\wedge\ver[\dt{P}](\beta)\\ &\quad\quad+ \Pi_{1,0}(\alpha)\wedge\Pi_{1,0}(\beta)\\
		&=\left(\ver[\dt{P}](\alpha)+\Pi_{1,0}(\alpha)\right)\wedge\left(\ver[\dt{P}](\beta)+\Pi_{1,0}(\beta)\right)\\
		&=\psi_\Pi(\alpha)\wedge\psi_\Pi(\beta).
	\end{align*}
	Since \(\rest{\psi_\Pi}{\Omega^1_P}\) is \(\ast\)-preserving and \(\Omega^2_P = \Omega^1_P \wedge \Omega^1_P\), it now follows that \(\rest{\psi_\Pi}{\Omega^2_P}\) is \(\ast\)-preserving. Hence, the map \(\psi_\Pi\) is indeed a left \(H\)-covariant isomorphism of graded \(\ast\)-algebras.
	
	Let us now compute \(\psi_\Pi \circ \dt{P} \circ \psi_\Pi^{-1}\). By Proposition~\ref{analysisthm}, we know that \[\rest{\psi_\Pi \circ \dt{P} \circ \psi_\Pi^{-1}}{P} = \psi_\Pi \circ \dt{P} = \dv{P} + \nabla,\] so it remains to compute \(\rest{\psi_\Pi \circ \dt{P} \circ \psi_\Pi^{-1}}{\Omega^1_P}\). Observe that
	\begin{align*}
		\psi_\Pi \circ \dt{P} &= \left(\ver^{2,2}[\dt{P}] + \ver^{2,1}[\dt{P}] \circ \Pi_{2,1} + \Pi_{2,0}\right) \circ \dt{P}\\
		&= \ver^{2,2}[\dt{P}] \circ \dt{P} + \ver^{2,1}[\dt{P}] \circ \Pi_{2,1} \circ \dt{P} + \nabla_\Pi,
	\end{align*}
	On the one hand, for all \(p,q \in P\),
	\[
		\ver^{2,2}[\dt{P}] \circ \dt{P}(p \cdot \dt{P}q) = \ver^{2,2}[\dt{P}](\dt{P}(p) \wedge \dt{P}(q)) = \dv{P}(p) \wedge \dv{P}(q) = \dv{P}(p \cdot \dv{P}(q)),
	\]
	so that \(\ver^{2,2}[\dt{P}] \circ \dt{P} = \dv{P} \circ \ver[\dt{P}]\). On the other, for all \(p,q \in P\),
	\begin{align*}
		&\ver^{2,1}[\dt{P}] \circ \Pi_{2,1} \circ \dt{P}(p \cdot \dt{P}(q))\\ &= \ver^{2,1}[\dt{P}] \circ \Pi_{2,1}(\dt{P}(p) \wedge \dt{P}(q))\\
		&= \ver^{2,1}[\dt{P}]\mleft(\Pi_{1,1}(\dt{P}(p)) \wedge \Pi_{1,0}(\dt{P}(q)) + \Pi_{1,0}(\dt{P}(p)) \wedge \Pi_{1,1}(\dt{P}(q))\mright)\\
		&= (\ver[\dt{P}] \circ \Pi_{1,1})(\dt{P}(p)) \wedge \nabla_\Pi(q) + \nabla_\Pi(p) \wedge (\ver[\dt{P}] \circ \Pi_{1,1})(\dt{P}(q))\\
		&= \dv{P}(p) \wedge \nabla_\Pi(q) + \nabla_\Pi(p) \wedge \dv{P}(q)\\
		&= \varpi_H(\ca{p}{-1})\epsilon(\ca{q}{-1}) \otimes \ca{p}{0}\cdot\nabla_\Pi(\ca{q}{0}) - \ca{p}{-1} \act \varpi_H(\ca{q}{-1}) \otimes \nabla_\Pi(\ca{p}{0}) \cdot \ca{q}{0}\\
		&= \varpi_H(\ca{p}{-1}\ca{q}{-1}) \otimes \ca{p}{0}\cdot\nabla_\Pi(\ca{p}{0}) - \ca{p}{-1}\act\varpi_H(\ca{q}{-1}) \otimes \nabla_\Pi(\ca{p}{0}\ca{q}{0})\\
		&= \varpi_H(\ca{\Pi_{1,0}(p \cdot \dt{P}(q))}{-1}) \otimes \ca{\Pi_{1,0}(p \cdot \dt{P}(q))}{0} + (\id \hotimes \nabla_\Pi) \circ \ver[\dt{P}](p \cdot \dt{P}(q)),
	\end{align*}
	where \(\id \hotimes \nabla_\Pi : \Lambda^1_H \otimes P \to \Lambda^1_H \otimes \Omega^1_{P,\hor}\), yielding the projected Cartan's magic formula
	\[
		\forall \alpha \in \Omega^1_P, \quad \ver^{2,1}[\dt{P}] \circ \Pi_{2,1} \circ \dt{P}(\alpha) = \varpi_H(\ca{\Pi_{1,0}(\alpha)}{-1}) \otimes \ca{\Pi_{1,0}(\alpha)}{0} + (\id \hotimes \nabla_\Pi) \circ \ver[\dt{P}](\alpha).
	\]
	It therefore follows that
	\begin{align*}
		\forall \alpha \in \Omega^1_P, \quad \psi_\Pi \circ \dt{P}(\alpha) &= \dv{P} \circ \ver[\dt{P}](\alpha) + \varpi_H(\ca{\Pi_{1,0}(\alpha)}{-1}) \otimes \ca{\Pi_{1,0}(\alpha)}{0}\\&\quad\quad + (\id \hotimes \nabla_\Pi) \circ \ver[\dt{P}](\alpha) + \nabla_\Pi(\alpha).
	\end{align*}
	On the one hand, for all \(\omega \in \Lambda^1_H\), since \(\ver[\dt{P}] \circ \theta(\omega) = \omega\) and \(\Pi_{1,0} \circ \theta(\omega) = 0\), it follows that
	\[
		\psi_\Pi \circ \dt{P} \circ \psi_\Pi^{-1}(\omega) = \dt{H}(\omega) + \nabla_\Pi \circ \theta(\omega) = \dt{H}(\omega)-\iu{}\,F_\Pi(\omega),
	\]
	while on the other, for all \(\alpha \in \Omega^1_{P,\hor}\), since \(\ver[\dt{P}](\alpha) = 0\) and \(\Pi_{1,0}(\alpha)=\alpha\), it follows that
	\[
		\psi_\Pi \circ \dt{P} \circ \psi_\Pi^{-1}(\alpha) = \varpi_H(\ca{\alpha}{-1}) \otimes \ca{\alpha}{0} + \nabla_\Pi(\alpha) = \varpi_H(\ca{\alpha}{-1}) \wedge \ca{\alpha}{0} + \nabla_\Pi(\alpha).
	\]
	Finally, since \(\rest{\nabla_\Pi}{\Omega_B} = \dt{B}\) and \(\Omega_B = \coinv{H}{\Omega_{P,\hor}}\), it follows that \(\rest{\psi_\Pi \circ \dt{P} \circ \psi_\Pi^{-1}}{\Omega_B} = \dt{B}\).
\end{proof}

\begin{remark}
	The isomorphism \(\psi_\Pi\) was first constructed \DJ{}ur\dj{}evi\'{c}~\cite{Dj97}*{Thm.\ 4.12} for multiplicative regular connections in his sense~\cite{Dj97}*{Def.\ 4.2 and 4.3}. Since regular connections \`{a} la \DJ{}ur\dj{}evi\'{c} yield bimodule connections \`{a} la Beggs--Majid, this suggests that our notion of prolongable connection can be viewed as a variant of \DJ{}ur\dj{}evi\'{c}'s notion of multiplicative connection. Our computation of \(\psi_\Pi^{-1} \circ \dt{P} \circ \psi_\Pi\), however, seems to be novel.
\end{remark}

\begin{remark}
	\emph{Mutatis mutandis}, one can also prove the following (unprojected) Cartan's magic formula for a strong quantum principal \((H;\Omega_H,\dt{H})\)-bundle \((P;\Omega_P,\dt{P})\):
	\begin{equation}
		\forall \alpha \in \Omega^1_P, \quad \ver^{2,1}[\dt{P}] \circ \dt{P}(\alpha) = \varpi(\ca{\alpha}{-1}) \wedge \ca{\alpha}{0} + (\id \hotimes \dt{P}) \circ \ver[\dt{P}](\alpha).
	\end{equation}
\end{remark}

\subsection{Prolongability and field strength}

Having adapted the notions of quantum principal bundle and strong bimodule connection to the setting of second-order differential calculi, we now turn to the notions of gauge transformation and gauge potential with respect to a fixed horizontal calculus. In particular, we see that the standard noncommutative-geometric notion of curvature of a module connection can be adapted to gauge potentials \emph{qua} horizontal covariant derivatives. In contrast to \DJ{}ur\dj{}evi\'{c}~\cites{Dj98,Dj10}, we are explicitly concerned with the problem of extending constructions from degree \(1\) to degree \(2\).

From now on, let \(P\) be a principal \(H\)-comodule \(\ast\)-algebra, and let \(B \coloneqq \coinv{H}{P}\). We begin with the following refinement of the notion of horizontal calculus on \(P\), which encodes a choice of basic differential calculus (through degree \(2\)) together with compatible left \(H\)-covariant graded \(\ast\)-algebra over \(P\) of horizontal forms (through degree \(2\)).

\begin{definition}[cf.\ \DJ{}ur\dj{e}vi\'{c}~\cite{Dj10}*{\S 3.1}]
	A \emph{second-order horizontal calculus} on the principal left \(H\)-comodule \(\ast\)-algebra \(P\) is a quadruple \((\Omega_B,\dt{B};\Omega_{P,\hor},\iota)\), where:
	\begin{enumerate}
		\item \((\Omega_B,\dt{B})\) is a \sodc{} on \(B\);
		\item \(\Omega_{P,\hor}\) is a left \(H\)-covariant graded \(\ast\)-algebra generated by \(\Omega^1_{P,\hor}\) over \(\Omega^0_P = P\) and truncated at degree \(2\);
		\item \(\iota : \Omega_B \inj \coinv{H}{\Omega_{P,\hor}}\) is an injective morphism of graded \(\ast\)-algebras, such that the pair \((\Omega^1_{P,\hor},\rest{\iota}{\Omega^1_{P,\hor}})\) defines a projectable horizontal lift of the \(B\)-\(\ast\)-bimodule \(\Omega^1_B\).
	\end{enumerate}
\end{definition}

\begin{example}
	Let \((\Omega_H,\dt{H})\) be a bicovariant \sodc{} for \(H\), and suppose that the bicovariant \fodc{} \((\Omega^1_H,\dt{H})\) on \(H\) is locally freeing for \(P\). Suppose that \((\Omega_P,\dt{P})\) is a strongly \((H;\Omega_H,\dt{H})\)-principal \sodc{} on \(P\) admitting a prolongable bimodule connection; recall that \((\Omega_P,\dt{P})\) therefore restricts to the \sodc{} \((\Omega_B,\dt{B})\) on \(B \coloneqq \coinv{H}{P}\) given by
	\[
		\Omega^1_B \coloneqq B \cdot \dt{P}(B), \quad \Omega^2_B \coloneqq \Omega^1_B \wedge \Omega^1_B, \quad \dt{B} \coloneqq \rest{\dt{P}}{\Omega_B}.
	\]
	Finally, let \(\Omega_{P,\hor}\) be the left \(H\)-subcomodule graded \(\ast\)-subalgebra of \(\Omega_P\) generated by \(\Omega^1_B\) over \(P\), and let \(\iota : \Omega_B \inj \Omega_{P,\hor}\) be the inclusion map. Then, by Proposition~\ref{analysisthm} applied to the quantum principal \((H;\Omega^1_H,\dt{H})\)-bundle \((P;\Omega^1_P,\dt{P})\), the data
	\(
		(\Omega_B,\dt{B};\Omega_{P,\hor},\iota)
	\)
	define a second-order horizontal calculus on \(P\), which we can view as the \emph{canonical} second-order horizontal calculus on \(P\) induced by the strongly \((H;\Omega_H,\dt{H})\)-principal \sodc{} \((\Omega_P,\dt{P})\) on \(P\).
\end{example}

\begin{remark}
	If \((\Omega_B,\dt{B};\Omega_{P,\hor},\iota)\) is a second-order horizontal	calculus on \(P\), then the data \((\Omega^1_B,\dt{B};\Omega^1_{P,\hor},\rest{\iota}{\Omega^1_B})\) define a first-order horizontal calculus on \(P\).
\end{remark}

Although it is not obvious, in a second-order horizontal calculus \((\Omega_B,\dt{B};\Omega_{P,\hor},\iota)\) on \(P\), the left \(H\)-covariant graded \(\ast\)-algebra \(\Omega_{P,\hor}\) of horizontal forms defines a projectable horizontal lift of the entire graded \(\ast\)-algebra \(\Omega_B\) of basic forms (through degree \(2\)) on \(P\).

\begin{proposition}\label{projectprop}
	Suppose that \((\Omega_B,\dt{B};\Omega_{P,\hor},\iota)\) is a second-order horizontal calculus on \(P\); recall that \(B \coloneqq \coinv{H}{P}\). Then \((\Omega^2_{P,\hor},\rest{\iota}{\Omega^2_{P,\hor}})\) defines a projectable horizontal lift of the \(B\)-\(\ast\)-bimodule \(\Omega^2_B\), and hence \(\iota : \Omega_B \inj \coinv{H}{\Omega_{P,\hor}}\) is an isomorphism of graded \(\ast\)-algebras.
\end{proposition}

\begin{proof}
	First, by applying Proposition~\ref{strongprop} to the projectable horizontal lift \((\Omega^1_{P,\hor},\rest{\iota}{\Omega^1_{P,\hor}})\) of the \(B\)-\(\ast\)-bimodule \(\Omega^1_B\), we find that
	\[
		\Omega^2_{P,\hor} = \Omega^1_{P,\hor} \wedge \Omega^1_{P,\hor} = P \cdot \iota(\Omega^1_B) \wedge \iota(\Omega^1_B) \cdot P = P \cdot \iota(\Omega^2_B) \cdot P,
	\]
	so that \((\Omega^2_{P,\hor},\rest{\iota}{\Omega^2_{P,\hor}})\) is a horizontal lift of \(\Omega^2_B\).
	
	Now, let \(p,\p{p} \in P\) and \(\omega,\p{\omega} \in \Omega^1_B\). Then \(\iota(\p{\omega})\cdot \p{p} = \sum_i q_i \cdot \iota(\p{\omega}_i)\) for some \(q_i \in P\) and \(\p{\omega}_i \in \Omega^1_B\), and for each \(i\), \(\iota(\omega) \cdot q_i = \sum_j q_{ij} \cdot \iota(\omega_{ij})\) for some \(q_{ij} \in P\) and \(\omega_{ij} \in \Omega^1_B\), so that
	\[
		p \cdot \iota(\omega) \wedge \iota(\p{\omega}) \cdot \p{p} = \sum_i p \cdot \iota(\omega) \wedge q_i \cdot \iota(\p{\omega}_i) = \sum_{i,j} p q_{ij} \cdot \iota(\omega_{ij} \wedge \p{\omega}_i) \in P \cdot \iota(\Omega^2_B). 
	\]
	Hence, by Proposition~\ref{strongprop}, the horizontal lift \((\Omega^2_{P,\hor},\rest{\iota}{\Omega^2_{P,\hor}})\) of \(\Omega^2_B\) is projectable. In particular, it now follows that \(\iota(\Omega_B) = \coinv{H}{\Omega_{P,\hor}}\).
\end{proof}

Assume, therefore, that \(P\) admits a second-order horizontal calculus \((\Omega_B,\dt{B};\Omega_{P,\hor},\iota)\), which we now fix. To simplify notation, we suppress the inclusion map \(\iota\) and identify \(\Omega_B\) with its image in \(\Omega_{P,\hor}\); hence, where convenient, we denote the second-order horizontal calculus \((\Omega_B,\dt{B};\Omega_{P,\hor},\iota)\) by the triple \((\Omega_B,\dt{B};\Omega_{P,\hor})\). Since \((\Omega^1_B,\dt{B};\Omega^1_{P,\hor})\) is a first-order horizontal calculus on \(P\), we can define its gauge group \(\fr{G}\), its inner gauge group \(\Inn(\fr{G})\), and its Atiyah space \(\fr{At}\) with corresponding space of translations \(\fr{at}\).

We begin by characterizing those gauge transformations \(f \in \fr{G}\) that extend via the induced map \(f_{\ast} : \Omega^1_{P,\hor} \to \Omega^1_{P,\hor}\) to automorphisms of \(\Omega_{P,\hor}\).

\begin{definition}
	We say that a gauge transformation \(\phi \in \fr{G}\) is \emph{prolongable} with respect to the second-order horizontal calculus \((\Omega_B,\dt{B};\Omega_{P,\hor},\iota)\) on \(P\) whenever \(\phi\) is also an automorphism of the projectable horizontal lift  \((\Omega^2_{P,\hor},\rest{\iota}{\Omega^2_B})\) of \(\Omega^2_B\). Hence, we define the \emph{prolongable gauge group} of \(P\) with respect to \((\Omega_B,\dt{B};\Omega_{P,\hor},\iota)\) by
	\[
		\pr{\fr{G}} \coloneqq \fr{G} \cap \Aut\mleft(\Omega^2_{P,\hor},\rest{\iota}{\Omega^2_B}\mright),
	\]
	where \(\fr{G}\) is the gauge group of \(P\) with respect to \((\Omega^1_B,\dt{B};\Omega^1_{P,\hor},\rest{\iota}{\Omega^1_B})\) and \(\Aut(\Omega^2_{P,\hor},\rest{\iota}{\Omega^2_B})\) is the automorphism group of \((\Omega^2_{P,\hor},\rest{\iota}{\Omega^2_B})\). By abuse of notation, given \(\phi \in \dva{\fr{G}}\), we denote by \(\phi_{\ast}\) or by \(\phi_{\ast,\hor}\) the induced automorphism of the left \(H\)-comodule graded \(\ast\)-algebra \(\Omega_{P,\hor}\).
\end{definition}

\begin{remark}
	That a prolongable gauge transformation \(\phi \in \dva{\fr{G}}\) induces an automorphism of the entire left \(H\)-covariant graded \(\ast\)-algebra \(\Omega_{P,\hor}\), viz, that
	\[
		\forall \omega, \p{\omega} \in \Omega^1_{P,\hor}, \quad \phi_\ast(\omega \wedge \p{\omega}) = \phi_\ast(\omega) \wedge \phi_\ast(\p{\omega}),
	\]
	is a consequence of the proof of Proposition~\ref{projectprop}.
\end{remark}

As an example, all inner gauge transformations are automatically prolongable.

\begin{propositiondefinition}
	The prolongable gauge group \(\pr{\fr{G}}\) of \(P\) with respect to the second-order horizontal calculus \((\Omega_B,\dt{B};\Omega_{P,\hor})\) contains the inner gauge group \(\Inn(\fr{G})\) of \(P\) with respect to the first-order horizontal calculus \((\Omega^1_B,\dt{B};\Omega^1_{P,\hor})\) as a central subgroup. Hence, the \emph{outer prolongable gauge group} of \(P\) with respect to \((\Omega_B,\dt{B};\Omega_{P,\hor})\) is 
\[
	\Out(\pr{\fr{G}}) \coloneqq \pr{\fr{G}}/\Inn(\fr{G}).
\]
\end{propositiondefinition}

\begin{proof}
	Since \(\Omega^2_{P,\hor}\) is a horizontal lift of \(\Omega^2_B = \Omega^1_B \wedge \Omega^1_B\), it follows that
	\(
		\Cent_B(\Omega^2_B) \supseteq \Cent_B(\Omega^1_B)
	\), so that, by Proposition~\ref{innergaugeprop1},
	\[
		\Inn(\fr{G}) = \set{\Ad_v \given v \in \Unit(Z(B) \cap \Cent_B(\Omega^1_B)} \leq \set{\Ad_v \given v \in \Unit(Z(B) \cap \Cent_B(\Omega^2_B))} = \Inn(\Omega^2_{P,\hor},\iota),
	\]
	and hence, \(\Inn(\fr{G}) \leq \fr{G} \cap \Aut(\Omega^2_{P,\hor},\iota) \eqqcolon \pr{\fr{G}}\).
\end{proof}

We can now characterise those gauge potentials on \(P\) with respect to the first-order horizontal calculus \((\Omega^1_B,\dt{B};\Omega^1_{P,\hor})\) that extend to lifts of \(\dt{B} : \Omega_B \to \Omega_B\) to degree \(1\) left \(H\)-covariant \(\ast\)-derivations of \(\Omega_{P,\hor}\). Because we only work through degree \(2\), such extensions are completely determined by maps \(\Omega^1_{P,\hor} \to \Omega^2_{P,\hor}\) of the following form.

\begin{definition}
	Let \(\Omega\) be a left \(H\)-comodule graded \(\ast\)-algebra truncated at degree \(2\). Let \(\partial : \Omega^0 \to \Omega^1\) be a left \(H\)-comodule \(\ast\)-derivation on the \(\Omega^0\)-\(\ast\)-bimodule \(\Omega^1\). A \emph{second-order prolongation} of \(\partial\) is a left \(H\)-covariant \(\bC\)-linear map \(\partial^\prime : \Omega^1 \to \Omega^2\) satisfying 
	\begin{gather*}
		\forall a,b \in \Omega^0, \, \forall \omega \in \Omega^1, \quad \p{\partial}(a \cdot \omega \cdot b) = \nabla(a) \wedge \omega \cdot b + a \cdot \p{\partial}(\omega) \cdot b - a \cdot \omega \wedge \partial(b),\\
		\forall \alpha \in \Omega^1, \quad \p{\partial}(a)^\ast = -\p{\partial}(\alpha^\ast),
	\end{gather*}
	so that \(\partial : \Omega^0 \to \Omega^1\) extends via \(\partial^\prime : \Omega^1 \to \Omega^2\) and \(0 : \Omega^2 \to 0\) to a degree \(1\) left \(H\)-covariant \(\ast\)-derivation on the left \(H\)-comodule graded \(\ast\)-algebra \(\Omega\).
\end{definition}

We can characterise those gauge potentials on \(P\) with respect to \((\Omega^1_B,\dt{B};\Omega^1_{P,\hor})\) that suitably extend to all of \(\Omega_{P,\hor}\); this, in turn, yields a conceptually minimalistic notion of curvature compatible with the standard notion of curvature for module connections.

\begin{propositiondefinition}\label{canprol}
	We say that a gauge potential \(\nabla \in \fr{At}\) is \emph{prolongable} with respect to the second-order horizontal calculus \((\Omega_B,\dt{B};\Omega_{P,\hor})\) whenever its \emph{canonical prolongation}
	\[
		\dva{\nabla} : \Omega^1_{P,\hor} \to \Omega^2_{P,\hor}, \quad p \cdot \dt{B}(b) \cdot \p{p} \mapsto \nabla(p) \wedge \dt{B}(b) \cdot \p{p} - p \cdot \dt{B}(b) \wedge \nabla(\p{p})
	\]
	is well-defined, in which case:
	\begin{enumerate}
		\item \(\dva{\nabla}\) is the unique second-order prolongation of \(\nabla\), such that \(\rest{\pr{\nabla}}{\Omega^1_B} = \dt{B}\);
		\item the \emph{field strength} \(\bF[\nabla] : P \to \Omega^2_{P,\hor}\) of \(\nabla\) defined by
	\[
		\bF[\nabla] \coloneqq -\iu{}\,\dva{\nabla} \circ \nabla.
	\]
	is a left \(H\)-covariant \(\ast\)-derivation, such that \(\rest{\bF[\nabla]}{B} = 0\).
	\end{enumerate}
	Hence, we define the \emph{prolongable Atiyah space} \(\dva{\fr{At}}\) of \(P\) with respect to \((\Omega_B,\dt{B};\Omega_{P,\hor})\) to be the subset of all prolongable gauge potentials on \(P\) with respect to \((\Omega_B,\dt{B};\Omega_{P,\hor})\).
\end{propositiondefinition}

\begin{proof}
	Let us first show that \(\dva{\nabla}\) is a secord-order prolongation of \(\nabla\). On the one hand, \(\dva{\nabla}\) is left \(H\)-covariant since \(\Omega^1_B = \coinv{H}{\Omega^1_{P,\hor}}\) and \(\nabla\) is left \(H\)-covariant. On the other, for all \(p,q_1,q_2 \in P\) and \(b \in \dt{B}(b)\), we have
	\begin{align*}
		\dva{\nabla}\mleft(q_1 \cdot (p\cdot \dt{B}(b)) \cdot q_2\mright) 
		&= \nabla(q_1 \cdot p) \wedge \dt{B}(b)	\cdot q_2 - q_1 \cdot p \cdot \dt{B}(b) \wedge \nabla(q_2)\\
		&= \left(\nabla(q_1) \cdot p + q_1 \cdot \nabla(p)\right) \wedge \dt{B}(b)	\cdot q_2 - q_1 \cdot p \cdot \dt{B}(b) \wedge \nabla(q_2)\\
		&= \nabla(q_1) \wedge (p \cdot \dt{B}(b)) \cdot q_2 + q_1 \cdot \dva{\nabla}(p \cdot \dt{B}(b)) \cdot q_2\\&\quad\quad - q_1 \cdot (p\cdot\dt{B}(b)) \wedge \nabla(q_2),
	\end{align*}
	while for all \(p \in P\) and \(b \in \dt{B}(b)\),
	\[
		\dva{\nabla}(p \cdot \dt{B}(b))^\ast = \left(\nabla(p) \wedge \dt{B}(b)\right)^\ast = -\dt{B}(b^\ast) \wedge \nabla(p^\ast) = \dva{\nabla}(\dt{B}(b^\ast) \cdot p^\ast) = -\dva{\nabla}\mleft((p \cdot \dt{B}(b))^\ast\mright).
	\]
	
	Next, observe that \(\rest{\dva{\nabla}}{\Omega^1_B} = \dt{B}\), since for all \(b_1,b_2 \in B\),
	\[
		\dva{\nabla}(b_1 \cdot \dt{B}(b_2)) = \nabla(b_1) \wedge \dt{B}(b_2) = \dt{B}(b_1) \wedge \dt{B}(b_2) = \dt{B}(b_1 \cdot \dt{B}(b_2));
	\]
	in particular, it now follows that \(\bF[\nabla]\) vanishes on \(B\); thus, if \(\nabla^\prime\) is any second-order prolongation of \(\nabla\), then for all \(p,q \in P\) and \(\beta \in \Omega^1_B\),
	\begin{align*}
		\p{\nabla}(p \cdot \beta \cdot q) = \nabla(p) \wedge \dt{B}(b) \cdot q + p \cdot \dt{B}(\beta) \cdot q - p \cdot \dt{B}(b) \wedge \nabla(q) = \pr{\nabla}(p \cdot \beta \cdot q),
	\end{align*}
	so that \(\p{\nabla} = \pr{\nabla}\).

	Finally, observe that the left \(H\)-covariant map \(\bF[\nabla] : P \to \Omega^2_{P,\hor}\) is a \(\ast\)-derivation, since for all \(p,q \in P\),
	\begin{align*}
		\bF[\nabla](pq) &= -\iu{}\dva{\nabla}\mleft(\nabla(p) \cdot q + p \cdot \nabla(q)\mright)\\ &= \bF[\nabla](p) \cdot q + \iu{} \nabla(p) \wedge \nabla(q) - \iu{} \nabla(p) \wedge \nabla(q) + \bF[\nabla](q)\\ &= \bF[\nabla](p) \cdot q + p \cdot \bF[\nabla](q),
	\end{align*}
	while for all \(p \in P\),
	\[
		\bF[\nabla](p)^\ast = \left(-\iu{}\dva{\nabla}(\nabla(p))\right)^\ast = -\iu{}\dva{\nabla}(\nabla(p)^\ast) = \iu{}\dva{\nabla}(\nabla(p^\ast)) = -\bF[\nabla](p^\ast). \qedhere
	\]	
\end{proof}

\begin{remark}
	That the field strength of a prolongable gauge potential defines a left \(H\)-covariant \(\ast\)-derivation \(P \to \Omega^2_{P,\hor}\) was essentially first observed by \DJ{}ur\dj{}evi\'{c}~\cite{Dj98}*{p.\ 101}.
\end{remark}

We can now similarly characterise those relative gauge potentials on \(P\) with respect to \((\Omega^1_B,\dt{B};\Omega^1_{P,\hor})\) that suitably extend to all of \(\Omega_{P,\hor}\); this follows, \emph{mutatis mutandis}, from the proof of Proposition-Definition~\ref{canprol} .

\begin{propositiondefinition}
	We say that \(\bA \in \fr{at}\) is \emph{prolongable} with respect to the second-order horizontal calculus \((\Omega_B,\dt{B};\Omega_{P,\hor})\) whenever its \emph{canonical prolongation}
	\[
		\dva{\bA} : \Omega^1_{P,\hor} \to \Omega^2_{P,\hor}, \quad p \cdot \dt{B}(b) \cdot \p{p} \mapsto \bA(p) \wedge \dt{B}(b) \cdot \p{p} - p\cdot \dt{B}(b) \wedge \bA(\p{p})
	\]
	is well-defined, in which case, the canonical prolongation \(\pr{\bA}\) is the unique second-order prolongation of \(\bA\), such that \(\rest{\pr{\bA}}{\Omega^1_B} = 0\). We denote by \(\pr{\fr{at}}\) the subspace of all prolongable relative gauge potentials on \(P\) with respect to \((\Omega_B,\dt{B};\Omega_{P,\hor})\).
\end{propositiondefinition}

It now follows that the affine action of the gauge group \(\fr{G}\) on the Atiyah space \(\fr{At}\) restricts to an affine action of the subgroup \(\pr{\fr{G}}\) on the affine subspace \(\pr{\fr{At}}\) that is compatible with canonical prolongation.

\begin{proposition}\label{fieldstrengthprop}
	Suppose that the prolongable Atiyah space \(\pr{\fr{At}}\) of \(P\) with respect to the second-order horizontal calculus \((\Omega_B,\dt{B};\Omega^1_{P,\hor})\) is non-empty. Then \(\dva{\fr{At}}\) is a \(\dva{\fr{G}}\)-invariant affine subspace of the Atiyah space \(\fr{At}\) with space of translations \(\dva{\fr{at}}\). Moreover:
	\begin{enumerate}
		\item \((\nabla \mapsto \dva{\nabla}) : \dva{\fr{At}} \to \Hom_{\bC}(\Omega^1_{P,\hor},\Omega^2_{P,\hor})\) is \(\dva{\fr{G}}\)-equivariant and affine linear with \(\dva{\fr{G}}\)-equivariant linear part \((\bA \mapsto \dva{\bA}) : \dva{\fr{at}} \to \Hom_{\bC}(\Omega^1_{P,\hor},\Omega^2_{P,\hor})\);
		\item \(\bF \coloneqq (\nabla \mapsto \bF[\nabla]) : \dva{\fr{At}} \to \Der_P(\Omega^2_{P,\hor})\) is \(\dva{\fr{G}}\)-equivariant and affine quadratic, satisfying
			\[
				\forall \nabla \in \dva{\fr{At}}, \,\forall \bA \in \dva{\fr{at}}, \quad \bF[\nabla + \bA] - \bF[\nabla] = -\iu{}\left(\dva{\nabla} \circ \bA + \dva{\bA} \circ \nabla + \dva{\bA} \circ \bA\right).
			\]
	\end{enumerate}
\end{proposition}

\begin{proof}
	First, that \(\pr{\fr{At}}\) is an affine subspace of \(\fr{At}\) with space of translations \(\fr{at}\) follows from the observation that \(\pr{\fr{At}}\) is defined by an affine-linear condition whose linear part is defines \(\pr{\fr{at}}\). Next, let \(\phi \in \pr{\fr{G}}\), \(\nabla \in \pr{\fr{At}}\). Then, for all \(p, \p{p} \in P\) and \(b \in B\),
	\begin{align*}
		&(\phi \act \nabla)(p) \wedge \dt{B}(b) \cdot \p{p} - p \cdot \dt{B}(b) \wedge (\phi \act \nabla)(\p{p})\\
		&= (\phi_{\ast} \circ \nabla \circ \inv{\phi})(p) \wedge \dt{B}(b) \cdot \p{p} - p \cdot \dt{B}(b) \wedge (\phi_{\ast} \circ \nabla \circ \inv{\phi})(\p{p})\\
		&= \phi_\ast\mleft( \nabla(\phi^{-1}(p)) \wedge \phi^{-1}_{\ast}(\dt{B}(b) \cdot \p{p})  - \phi^{-1}_{\ast}(p \cdot \dt{B}(b)) \wedge \nabla(\phi^{-1}(\p{p}))\mright)\\
		&= \phi_\ast \circ \pr{\nabla} \circ \phi^{-1}_\ast \mleft(p \cdot \dt{B}(b) \cdot \p{p}\mright),
	\end{align*}
	so that \(\phi \act \nabla \in \pr{\fr{At}}\) with \(\pr{(\phi \act \nabla)} = \phi_\ast \circ \pr{\nabla} \circ \phi_\ast^{-1}\). Thus, \(\pr{\fr{At}}\) is \(\pr{\fr{G}}\)-invariant and the map \((\nabla \mapsto \pr{\nabla}) : \pr{\fr{At}} \to \Hom_{\bC}(\Omega^1_{P,\hor},\Omega^2_{P,\hor})\) is \(\pr{\fr{G}}\)-equivariant. A similary calculation shows that \(\pr{\fr{at}}\) is \(\pr{\fr{G}}\)-invariant and that \((\bA \mapsto \pr{\bA}) : \pr{\fr{at}} \to \Hom_{\bC}(\Omega^1_{P,\hor},\Omega^2_{P,\hor})\) is \(\pr{\fr{G}}\)-equivariant. Finally, our claims about \(\bF : \dva{\fr{At}} \to \Der_P^H(P,\Omega^2_{P,\hor})\) follow by straightforward calculation.
\end{proof}

While inner gauge transformations are automatically prolongable and yield inner automorphisms of the left \(H\)-comodule graded \(\ast\)-algebra \(\Omega_{P,\hor}\), the analogous statement is no longer generally true for inner relative gauge potentials. Instead, we shall consider those prolongable inner relative gauge potentials that yield inner derivations of \(\Omega_{P,\hor}\).

\begin{definition}
	Let \(\bA \in \pr{\fr{at}}\) be a prolongable gauge potential on \(P\) with respect to the second-order horizontal calculus \((\Omega_B,\dt{B};\Omega_{P,\hor})\). We say that \(\bA\) is \emph{inner} if and only if there exists a \(1\)-form \(\alpha \in (\Omega^1_B)_{\sa}\), such that
	\[
		\bA = \ad_\alpha \coloneqq (p \mapsto [\alpha,p]), \quad \pr{\bA} = \ad_\alpha \coloneqq (\omega \mapsto [\alpha,\omega]).
	\]
	We denote by \(\Inn(\pr{\fr{at}})\) the subspace of all inner prolongable relative gauge potentials on \(P\) with respect to \((\Omega_B,\dt{B};\Omega_{P,\hor})\). Thus, we define the \emph{outer prolongable Atiyah space} of \(P\) with respect to \((\Omega_B,\dt{B};\Omega_{P,\hor})\) to be the quotient affine space \(\Out(\pr{\fr{At}}) \coloneqq \pr{\fr{At}}/\Inn(\pr{\fr{at}})\) with space of translations \(\Out(\pr{\fr{at}}) \coloneqq \pr{\fr{at}}/\Inn(\pr{\fr{at}})\).
\end{definition}

We can now characterize those elements of \(\Omega^1_B\) that yield inner prolongable relative gauge potentials on \(P\) with respect to \((\Omega_B,\dt{B};\Omega_{P,\hor})\); as an upshot, we find that the field strength of a prolongable gauge potential varies \emph{linearly}---not quadratically---under translation by inner prolongable relative gauge potentials.

\begin{proposition}\label{innpot}
	Let \(\alpha \in (\Omega^1_B)_{\sa}\). Then the inner derivation \(\ad_\alpha\) defines an inner prolongable relative gauge potential on \(P\) with respect to \((\Omega_B,\dt{B};\Omega_{P,\hor})\) if and only if the \(1\)-form \(\alpha\) is central in \(\Omega_B\), in which case
	\begin{equation}\label{innpoteq}
		\forall \nabla \in \dva{\fr{At}}, \, \forall p \in P, \quad \quad \iu{}\left(\bF[\nabla+\ad_\alpha]-\bF[\nabla]\right)(p) = \ad_{\dt{B}(\alpha)}(p) \coloneqq [\dt{B}(\alpha),p].
	\end{equation}
	Thus, the map \((\alpha \mapsto \ad_\alpha) : (\Omega^1_B)_{\sa} \cap \Zent(\Omega_B) \surj \Inn(\pr{\fr{at}})\) yields a short exact sequence
\[
	0 \to (\Omega^1_B)_{\sa} \cap \Zent(\Omega_{P,\hor}) \to (\Omega^1_B)_{\sa} \cap \Zent(\Omega_B) \to \Inn(\pr{\fr{at}}) \to 0.
\]
\end{proposition}

\begin{proof}
	The first part of the claim follows from observing that for all \(p,\p{p} \in P\) and \(\beta \in \Omega^1_B\),
	\[
		[\alpha,p] \wedge \beta \cdot \p{p} - p\cdot \beta \wedge [\alpha,\p{p}] = [\alpha,p\cdot\beta\cdot\p{p}] - p \cdot [\alpha,\beta] \cdot \p{p}.
	\]
	Now, if \([\alpha,\Omega^1_B] = \set{0}\), then, by the above calculation, \(\ad_\alpha\) is prolongable with canonical prolongation \(\dva{(\ad_\alpha)} = [\alpha,\cdot]\); moreover, if \(\nabla \in \dva{\fr{At}}\), then for any \(p \in P\),
	\begin{align*}
		\iu{}\left(\bF[\nabla+\bA]-\bF[\nabla]\right)(p) &=(\dva{\nabla} \circ \ad_\alpha + \dva{(\ad_\alpha)} \circ \nabla)(p) + \dva{(\ad_\alpha)} \circ \ad_\alpha(p)\\
		&= \dva{\nabla}([\alpha,p]) + [\alpha,\nabla(p)] + [\alpha,[\alpha,p]]\\
		&= [\dt{B}(\alpha) + \alpha \wedge \alpha,p]\\
		&= [\dt{B}(\alpha),p]. \qedhere
	\end{align*}
\end{proof}

Just as in the first-order case, the affine action of the prolongable gauge group \(\fr{G}\) on the prolongable Atiyah space \(\pr{\fr{At}}\) descends further to an affine action of the outer prolongable gauge group \(\Out(\pr{\fr{G}})\) on the outer prolongable Atiyah space \(\Out(\pr{\fr{At}})\), which will yield a non-trivial invariant of the principal left \(H\)-comodule \(\ast\)-algebra \(P\) endowed with the second-order horizontal calculus \((\Omega_B,\dt{B};\Omega_{P,\hor})\). 

\begin{proposition}\label{innprop}
	The subspace \(\Inn(\pr{\fr{at}})\) of inner prolongable gauge potentials on \(P\) consists of \(\pr{\fr{G}}\)-invariant vectors, so that the affine action of \(\pr{\fr{G}}\) on \(\pr{\fr{At}}\) descends to an affine action of \(\pr{\fr{G}}\) on \(\Out(\pr{\fr{At}}) \coloneqq \pr{\fr{At}}/\Inn(\pr{\fr{at}})\). Moreover, the inner gauge group \(\Inn(\fr{G})\) acts trivially on \(\Out(\pr{\fr{at}})\), so that the affine action of \(\pr{\fr{G}}\) on \(\pr{\fr{At}}\) descends further to an affine action of \(\Out(\pr{\fr{G}})\) on \(\Out(\pr{\fr{At}})\).
\end{proposition}

\begin{proof}
	Since \(\Inn(\fr{at})\) consists of \(\fr{G}\)-invariant vectors, it follows \emph{a fortiori} that \(\Inn(\pr{\fr{at}})\) consists of \(\pr{\fr{G}}\)-invariant vectors. Now, let \(\phi \in \Inn(\fr{G})\), so that by Proposition~\ref{innergaugeprop1}, \(\phi = \Ad_v\) and \(\phi_\ast = \Ad_v\) for some \(v \in \Unit(Z(B) \cap \Cent_B(\Omega^1_B))\); let \(\nabla \in \pr{\fr{At}}\). By the proof of Proposition~\ref{inducedprop}, we find that \(\phi \act \nabla - \nabla = \ad_\alpha\) for \(\alpha \coloneqq - v^\ast \cdot \dt{B}(v)\), so that \(\phi \act \nabla - \nabla \in \Inn(\pr{\fr{at}})\) if and only if \(-v^\ast \cdot \dt{B}(v) \in Z(\Omega_B)\), if and only if \([\dt{B}(v),\dt{B}(B)] = \set{0}\). However, for all \(b \in B\),
	\[
		[\dt{B}(v),\dt{B}(b)] = \dt{B}(v) \wedge \dt{B}(b) + \dt{B}(b) \wedge \dt{B}(v) = \dt{B}(v \cdot \dt{B}(b)) - \dt{B}(\dt{B}(b) \cdot v) = 0,
	\]
	so that, indeed, \(\phi \act \nabla - \nabla \in \Inn(\pr{\fr{at}})\). Hence, \(\Inn(\fr{G})\) acts trivially on \(\Out(\pr{\fr{At}})\).
\end{proof}

Finally, by Proposition~\ref{innprop}, we can characterize the variation of field strength under translation by inner prolongable gauge potentials as follows---note that field strength, \emph{a priori}, is affine quadratic, not linear.

\begin{definition}
	Let \(\bA\) be an inner prolongable gauge potential on \(P\) with respect to the second-order horizontal calculus \((\Omega_B,\dt{B};\Omega_{P,\hor})\). The \emph{relative field strength} of \(\bA\) is the unique \(H\)-covariant \(\ast\)-derivation \(\bF_{\rel}[\bA] : P \to \Omega^2_{P,\hor}\) vanishing on \(B\), such that
	\[
		\forall \nabla \in \pr{\fr{At}}, \quad \bF[\nabla+\bA]-\bF[\nabla] = \bF_{\rel}[\bA].
	\]
\end{definition}

\begin{corollary}\label{relcor}
	Let \(\alpha \in (\Omega^1_B)_{\sa} \cap \Zent(\Omega_B)\), so that \(\ad_\alpha\) defines an inner prolongable gauge potential on \(P\) with respect to \((\Omega_B,\dt{B};\Omega_{P,\hor})\). Then
	\begin{equation}
		\bF_{\rel}[\ad_\alpha] = \ad_{-\iu{}\dt{B}(\alpha)}.
	\end{equation}
	Hence, the map \(\bF_{\rel} \coloneqq (\bA \mapsto \bF_{\rel}[\bA]) : \Inn(\pr{\fr{at}}) \to \Der^H(P,\Omega^2_{P,\hor})\) is \(\pr{\fr{G}}\)-equivariant, \(\Inn(\fr{G})\)-invariant, and \(\bR\)-linear.
\end{corollary}

\begin{example}\label{relex}
	Let \(\phi \in \Inn(\fr{G})\); we claim that
	\[
		\forall \nabla \in \pr{\fr{At}}, \quad \bF_{\rel}[\phi \act \nabla - \nabla] = 0.
	\]
	Indeed, by Proposition~\ref{innergaugeprop1}, choose \(v \in \Unit(\Cent_B(B\oplus\Omega^1_B))\), such that \(\phi = \Ad_v\) and \(\phi_\ast = \Ad_v\), and set \(\alpha \coloneqq -v^\ast \cdot \dt{B}(v)\);	by Proposition~\ref{inducedprop}, it follows that
	\[
		\forall \nabla \in \pr{\fr{At}}, \quad \phi \act \nabla - \nabla = \ad_{\alpha}.
	\]
	Hence, by Corollary~\ref{relcor}, it suffices to show that \(\dt{B}(\alpha) = 0\). Since \(v \in \Unit(\Cent_B(B\oplus\Omega^1_B))\),
	\[
		\dt{B}(\alpha) = \dt{B}(v^\ast \cdot \dt{B}(v)) = \dt{B}(v^\ast) \wedge \dt{B}(v) = - v^\ast \cdot \dt{B}(v) \cdot v^\ast \wedge \dt{B}(v) = 0.
	\]
\end{example}

\subsection{Reconstruction of quantum principal bundles to second order}

Recall that \(H\) is a Hopf \(\ast\)-algebra; let \(P\) be a principal left \(H\)-comodule \(\ast\)-algebra with \(\ast\)-subalgebra of coinvariants \(B \coloneqq \coinv{H}{P}\). Given a second-order horizontal calculus \((\Omega_B,\dt{B};\Omega_{P,\hor})\) on \(P\) and a bicovariant \(\ast\)-differential calculus \((\Omega_H,\dt{H})\) on \(H\) whose \fodc{} is locally freeing for \(P\), we consider \emph{all} possible strongly \((H;\Omega_H,\dt{H})\)-principal \sodc{} on \(P\) inducing the second-order horizontal calculus \((\Omega_B,\dt{B};\Omega_{P,\hor})\). This will permit us to justify our notions of prolongable gauge transformation, prolongable gauge potential, and field strength by refining the equivalence of groupoids of Proposition~\ref{equiv1}. Furthermore, when \((\Omega_H,\dt{H})\) is Woronowicz's canonical prolongation~\cite{Woronowicz}, this will yield a gauge-equivariant moduli space of strongly \((H;\Omega_H,\dt{H})\)-principal \sodc{} on \(P\) inducing \((\Omega_B,\dt{B};\Omega_{P,\hor})\). Again, for relevant definitions from the basic theory of groupoids, see Appendix~\ref{Groupoids}.

Let us now fix a second-order horizontal calculus \((\Omega_B,\dt{B};\Omega_{P,\hor},\iota)\) on the principal left \(H\)-comodule \(\ast\)-algebra \(P\). Given a bicovariant \(\ast\)-differential calculus \((\Omega_H,\dt{H})\) on \(H\) whose \fodc{} is locally freeing for \(P\), we construct the groupoid analogous to the groupoid \(\cG[\Omega^1_H]\) of Definition~\ref{qpbdef} whose objects are strongly \((H;\Omega_H,\dt{H})\)-principal \sodc{} on \(P\) and whose arrows give a abstract notion of gauge transformation adapted to general \sodc{}. In what follows, let \((\Omega^{\leq 2}_H,\dt{H})\) denote the truncation of \((\Omega_H,\dt{H})\) to a bicovariant \sodc{} on \(H\).

\begin{definition}
	Let \((\Omega_H,\dt{H})\) be a bicovariant \(\ast\)-differential calculus on \(H\) whose \fodc{} is locally freeing for \(P\). We define the \emph{prolongable abstract gauge groupoid} \(\cG[\Omega^{\leq 2}_H]\) with respect to \((\Omega_H,\dt{H})\) as follows:
	\begin{enumerate}
		\item an object is a strongly \((H;\Omega_H,\dt{H})\)-principal \textsc{sodc} \((\Omega_P,\dt{P})\) on \(P\), such that the resulting strong quantum principal \(H;\Omega_H,\dt{H})\)-principal bundle \((P;\Omega_P,\dt{P})\) admits prolongable bimodule connections, and
		\[
			(\ker\ver[\dt{P}], \dt{B}(b) \mapsto \dt{P}(b))
		\]
		defines a horizontal lift of \(\Omega^1_B\) admitting a (necessarily unique) left \(H\)-covariant isomorphism 
		\[
			C[\Omega_P] : P \oplus \ker\ver[\dt{P}] \oplus \left(\ker(\ver^{2,2}[\dt{P}]) \cap \ker(\ver^{2,1}[\dt{P}])\right) \to \Omega_{P,\hor}
		\]
		of graded \(\ast\)-algebras satisfying \(\rest{C[\Omega_P]}{P} = \id\) and \(C[\Omega_P] \circ \rest{\dt{P}}{B} = \iota \circ \dt{B}\).
		\item given objects \((\Omega_1,\dt{1})\) and \((\Omega_2,\dt{2})\), an arrow \(f : (\Omega_1,\dt{1}) \to (\Omega_2,\dt{2})\) consists of a left \(H\)-covariant \(\ast\)-automorphism \(f : P \to P\), such that \(\rest{f}{B}=\id\) and
		\[
			f_\ast : \Omega^1_1 \to \Omega^1_2, \quad p \cdot \dt{1}(\p{p}) \cdot \pp{p} \mapsto f(p) \cdot \dt{2}(f(\p{p})) \cdot f(\pp{p})
		\]
		is well-defined and extends multiplicatively to a bijection \(f_\ast : \Omega_1 \to \Omega_2\);
		\item  composition of arrows is induced by composition of \(\ast\)-automorphisms of \(P\), and the identity of an object \((\Omega,\dt{})\) is given by \(\id_{(\Omega,\dt{})} \coloneqq \left(\id_P : (\Omega,\dt{}) \to (\Omega,\dt{})\right)\).
	\end{enumerate}
	Moreover, we define the star-injective homomorphism \(\mu[\Omega^{\leq 2}_H] : \cG[\Omega^{\leq 2}_H] \to \Aut(P)\) by
	\[
		\forall \left(f : (\Omega_1,\dt{1}) \to (\Omega_2,\dt{2}) \right) \in \cG[\Omega^{\leq 2}_H], \quad \mu[\Omega^{\leq 2}_H]\mleft(f : (\Omega_1,\dt{1}) \to (\Omega_2,\dt{2}) \mright) \coloneqq f.
	\]
\end{definition}

\begin{remark}
	Thus, the canonical second-order horizontal calculus \((\Omega_P,\dt{P})\) of \(\cG[\Omega^{\leq 2}_H]\) induces the second-order horizontal calculus \((\Omega_B,\dt{B};\Omega_{P,\hor})\) up to the canonical isomorphism \(C[\Omega_P]\) of graded left \(H\)-comodule \(\ast\)-algebras.
\end{remark}

Just as in the first-order case, given a bicovariant \(\ast\)-differential calculus \((\Omega_H,\dt{H})\) on \(H\), we simultaneously consider prolongable bimodule connections on all strong quantum principal \((H;\Omega^1_H,\dt{H})\)-bundles induced from \(P\) by \sodc{} in \(\Ob(\cG[\Omega^{\leq 2}_H])\). Once more, it is straightforward to check that the prolongable abstract gauge groupoid admits a canonical action on this set of bimodule connections.

\begin{propositiondefinition}
	Let \((\Omega_H,\dt{H})\) be a bicovariant \(\ast\)-differential calculus on \(H\) whose \fodc{} is locally freeing for \(P\). Let \(\cA[\Omega^{\leq 2}_H])\) be the set of all triples \((\Omega_P,\dt{P};\Pi)\), where \((\Omega_P,\dt{P}) \in \Ob(\cG[\Omega^{\leq 2}_H]\) and \(\Pi\) is a prolongable bimodule connection on the strong quantum principal \((H;\Omega_H,\dt{H})\)-bundle \((P;\Omega_P,\dt{P})\); hence, let \(p[\Omega^{\leq 2}_H] : \cA[\Omega^{\leq 2}_H] \to \Ob(\cG[\Omega^{\leq 2}_H])\) be the canonical surjection given by
	\[
		\forall (\Omega_P,\dt{P};\Pi) \in \cA[\Omega^{\leq 2}_H], \quad p[\Omega^{\leq 2}_H](\Omega_P,\dt{P};\Pi) \coloneqq (\Omega_P,\dt{P}).
	\]
	Then the \emph{abstract gauge action} is the action of \(\cG[\Omega^{\leq 2}_H]\) on \(\cA[\Omega^{\leq 2}_H]\) via \(p[\Omega^{\leq 2}_H]\) defined by
	\begin{multline}
		\forall \left(f : (\Omega_1,\dt{1}) \to (\Omega_2,\dt{2})\right) \in \cG[\Omega^{\leq 2}_H], \, \forall (\Omega_1,\dt{1};\Pi) \in p[\Omega^1_H]^{-1}(\Omega_1,\dt{1}),\\ \left(f : (\Omega_1,\dt{1}) \to (\Omega_2,\dt{2})\right) \act (\Omega_1,\dt{1};\Pi) \coloneqq (\Omega_2,\dt{2};f_\ast \circ \Pi \circ f_\ast^{-1}).
	\end{multline}
	Hence, the canonical covering \(\pi[\Omega^{\leq 2}_H] : \cG[\Omega^{\leq 2}_H] \ltimes \cA[\Omega^{\leq 2}_H] \to \cG[\Omega^{\leq 2}_H]\) is given by
	\begin{multline}
		\forall \left(\left(f : (\Omega_1,\dt{1}) \to (\Omega_2,\dt{2})\right),(\Omega_1,\dt{1};\Pi)\right) \in \cG[\Omega^{\leq 2}_H] \ltimes \cA[\Omega^{\leq 2}_H], \\ \pi[\Omega^{\leq 2}_H]\mleft(\left(f : (\Omega_1,\dt{1}) \to (\Omega_2,\dt{2})\right),(\Omega_1,\dt{1};\Pi)\mright) \coloneqq \left(f : (\Omega_1,\dt{1}) \to (\Omega_2,\dt{2})\right).
	\end{multline}
\end{propositiondefinition}

Once more, as a convenient abuse of notation, we will denote an arrow
\[
	\left(\left(f : (\Omega_1,\dt{1}) \to (\Omega_2,\dt{2})\right),(\Omega_1,\dt{1};\Pi)\right)
\]
of the action groupoid \(\cG[\Omega^{\leq 2}_H] \ltimes \cA[\Omega^{\leq 2}_H]\) by
\[
	f: (\Omega_1,\dt{1};\Pi_1) \to (\Omega_2,\dt{2};\Pi_2),
\]
where \(\Pi_2 \coloneqq \rest{f_\ast \circ \Pi \circ f_\ast^{-1}}{\Omega^1_2}\), so that, in particular,
\[
	\pi[\Omega^{\leq 2}_H]\mleft(f: (\Omega_1,\dt{1};\Pi_1) \to (\Omega_2,\dt{2};\Pi_2)\mright) \coloneqq \left(f : (\Omega_1,\dt{1}) \to (\Omega_2,\dt{2})\right).
\]

\begin{remark}
There are an obvious star-injective homomorphism \(\cG[\Omega^{\leq 2}_H] \to \cG[\Omega^{1}_H]\) and an obvious surjection \(\cA[\Omega^{\leq 2}_H] \surj \cA[\Omega^1_H]\) defined by restricting \textsc{sodc} to \textsc{fodc}, which, in turn, yield the following commutative diagram of groupoid homomorphisms:
\[
	\begin{tikzcd}
		\cG[\Omega^{\leq 2}_H] \ltimes \cA[\Omega^{\leq 2}_H]\arrow[r] \arrow[d, "\pi\lbrack\Omega^{\leq 2}_H\rbrack"']   &\cG[\Omega^1_H] \ltimes \cA[\Omega^1_H] \arrow[d,"\pi\lbrack\Omega^1_H\rbrack"] \\
		\cG[\Omega^{\leq 2}_H] \arrow[r] \arrow[d, "\mu\lbrack\Omega^{\leq 2}_H\rbrack"']  &\cG[\Omega^1_H] \arrow[d, "\mu\lbrack\Omega^1_H\rbrack"] \\
		\Aut(P) \arrow[r, hookleftarrow] &\fr{G} 
	\end{tikzcd}
\]
here, \(\fr{G}\) is the gauge group of the principal left \(H\)-comodule \(\ast\)-algebra \(P\) with respect to the first-order horizontal calculus \((\Omega^1_B,\dt{B};\Omega^1_{P,\hor})\). 
Hence, it follows that \(\ran \mu[\Omega^{\leq 2}_H] \subseteq \fr{G}\).
\end{remark}

Now, let \(\pr{\fr{G}}\) and \(\pr{\fr{At}}\) respectively denote the prolongable gauge group and prolongable Atiyah space of the principal left \(H\)-comodule \(\ast\)-algebra \(P\) with respect to the second-order horizontal calculus \((\Omega_B,\dt{B};\Omega_{P,\hor})\). Our goal is to promote the isomorphism of Proposition~\ref{isothm} to an explicit equivalence of categories that suitably refines the equivalence of Proposition~\ref{equiv1}. This will first require a characterisation of those prolongable gauge potentials on \(P\) with respect to \((\Omega_B,\dt{B};\Omega_{P,\hor})\) that correctly induce elements of \(\cA[\Omega^{\leq 2}_H]\). In what follows, let \((\Omega_{P,\ver},\dv{P})\) denote the vertical calculus on \(P\) induced by a bicovariant \(\ast\)-differential calculus on \(H\) whose \fodc{} is locally freeing for \(P\).

\begin{definition}
	Let \((\Omega^1_H,\dt{H})\) be a bicovariant \fodc{} on \(H\) that is locally freeing for \(P\). We say that a prolongable gauge potential \(\nabla \in \dva{\fr{At}}\) on \(P\) with respect to \((\Omega_B,\dt{B};\Omega_{P,\hor})\) is \emph{\((\Omega^1_H,\dt{H})\)-adapted} whenever its field strength \(\bF[\nabla]\) is given by 
	\[\bF[\nabla] = F[\nabla] \circ \dv{P}\]
	for a (necessarily unique) left \(H\)-covariant morphism \(F[\nabla] : \Omega^1_{P,\ver} \to \Omega^2_{P,\hor}\) of \(P\)-\(\ast\)-bimod\-ules; in this case, we call \(F[\nabla]\) the \emph{curvature \(2\)-form} of \(\nabla\). We define the \emph{\((\Omega^1_H,\dt{H})\)-adapted prolongable Atiyah space} of \(P\) with respect to \((\Omega_B,\dt{B};\Omega_{P,\hor})\) to be the subset \(\pr{\fr{At}}[\Omega^1_H]\)of all \((\Omega^1_H,\dt{H})\)-adapted prolongable gauge potentials on \(P\).
\end{definition}

\begin{proposition}[cf.\ \DJ{}ur\dj{}evi\'{c}~\cite{Dj10}*{Prop. 26 and 27}]\label{sodciso}
	Let \((\Omega_H,\dt{H})\) be a bicovariant \(\ast\)-differen\-tial calculus on \(H\) whose \fodc{} is locally freeing for \(P\); let \(\Lambda_H\) be the resulting graded left crossed \(H\)-module \(\ast\)-algebra of right \(H\)-coinvariant forms, and set \(\Omega_{P,\oplus} \coloneqq \Lambda_H \dvatimes \Omega_{P,\hor}\). Let \(\nabla\) be a gauge potential on \(P\) with respect to \((\Omega^1_B,\dt{B};\Omega^1_{P,\hor})\). The following are equivalent:
	\begin{enumerate}
		\item\label{sodciso1} the left \(H\)-covariant \(\ast\)-derivation \(\dt{P,\nabla} : P \to \Omega^1_{P,\oplus}\) given by
		\[
			\forall p \in P, \quad \dt{P,\nabla}(p) \coloneqq \dv{P}(p) + \nabla(p)
		\]
		admits a (necessarily unique) extension to a left \(H\)-covariant degree \(1\) \(\ast\)-derivation \[\dt{P,\nabla} : \Omega_{P,\oplus} \to \Omega_{P,\oplus},\] such that \((\Omega_{P,\oplus},\dt{P,\nabla},\Pi_\oplus) \in \Ob(\cG[\Omega^{\leq 2}_H] \ltimes \cA[\Omega^{\leq 2}_H])\), where \(\Pi_\oplus : \Omega^1_{P,\oplus} \to \Omega^1_{P,\oplus}\) is the projection onto \(\Omega^1_{P,\ver}\) along \(\Omega^1_{P,\hor}\);
		\item\label{sodciso2} the gauge potential \(\nabla\) is prolongable with respect to \((\Omega_B,\dt{B};\Omega_{P,\hor})\) and \((\Omega^1_H,\dt{H})\)-adapted.
	\end{enumerate}
	If either (and hence both) of the above conditions are satisfied, the prolongable bimodule connection \(\Pi_{\oplus}\) on the strong quantum principal \((H;\Omega_H,\dt{H})\)-bundle \((P;\Omega_{P,\oplus},\dt{P,\oplus})\) satisfies
	\[
		\nabla_{\Pi_\oplus} = \nabla, \quad F_{\Pi_\oplus} = F[\nabla].
	\]
\end{proposition}

\begin{proof}
	First, suppose that \(\dt{P,\nabla}\) extends to \(\Omega_{P,\oplus}\) and that \((\Omega_{P,\oplus},\dt{P,\nabla};\Pi_\oplus)\) defines an object of \(\cG[\Omega^1_H] \ltimes \cA[\Omega^1_H]\), so that \(\nabla_{\Pi} = \nabla\) and
	\[
		\ver[\dt{P,\nabla}] = \Proj_1 : \Omega^1_{P,\oplus}\ = \Omega^1_{P,\ver} \oplus \Omega^1_{P,\hor} \to \Omega^1_{P,\ver}.
	\]
	For every \(p,q \in P\) and \(b \in B\),
	\begin{align*}
		\nabla_{\Pi_\oplus}(p \cdot \dt{B}(b) \cdot q) &= \nabla_{\Pi_\oplus}(p) \wedge \dt{B}(b) \cdot q + p \cdot \nabla_{\Pi_\oplus}(\dt{B}(b)) \cdot q - p\cdot\dt{B}(b) \wedge \nabla_{\Proj_1}(q)\\
		&= \nabla(p) \wedge \dt{B}(b) \cdot q - p \cdot \dt{B}(b) \wedge \nabla(q)
	\end{align*}
	so that \(\nabla\) is prolongable with \(\dva{\nabla} = \rest{\nabla_{\Pi_\oplus}}{\Omega^1_{P,\hor}}\) and \((\Omega^1_H,\dt{H})\)-adapted with \(F[\nabla] = F_{\Pi_\oplus}\).
	
	Now, suppose that \(\nabla\) is prolongable and \((\Omega^1_H,\dt{H})\)-adapted. We know that \((P;\Omega^1_{P,\oplus};\Pi_\oplus)\) defines an element of \(\cA[\Omega^1_H]\) with
	\[
		\ver[\dt{P,\nabla}] = \Proj_1 : \Omega^1_{P,\oplus} = \Omega^1_{P,\ver} \oplus \Omega^1_{P,\hor} \to \Omega^1_{P,\ver}.
	\]
	By mild abuse of notation, let \(\dt{P,\nabla}: \Omega^1_{P,\oplus} \to \Omega^2_{P,\oplus}\) be the left \(H\)-covariant degree \(1\) map defined by
	\begin{multline*}
		\forall \omega \in \Lambda^1_H, \, \forall p \in P, \, \forall \alpha \in \Omega^1_{P,\hor}, \\ \dt{P,\nabla}(\omega \cdot p + \alpha) \coloneqq  \dv{P}(\omega \cdot p) - \iu{}\,F[\nabla](\omega \cdot p)- \omega \wedge \nabla(p) + \varpi(\ca{\alpha}{-1}) \wedge \ca{\alpha}{0} + \pr{\nabla}(\alpha).
	\end{multline*}
	While a routine calculation using the properties of \(\dv{P}\), \(\nabla\) and \(\varpi\) shows that \(\dt{P,\nabla}\) as defined on \(P\) and \(\Omega^1_{P,\oplus}\), respectively, yields a \(\ast\)-derivation \(\dt{P,\nabla} : \Omega_{P,\oplus} \to \Omega_{P,\oplus}\) of degree \(1\), but it is less obvious that \(\dt{P,\nabla}^2 = 0\). However, for all \(p \in P\),
	\begin{align*}
		\dt{P,\nabla}^2(p) &= \dt{P,\nabla}(\varpi(\ca{p}{-1}) \cdot \ca{p}{0} + \nabla(p))\\
		&= \dv{P}(\varpi(\ca{p}{-1}) \cdot \ca{p}{0}) - \iu{}\,F[\nabla](\varpi(\ca{p}{-1}) \cdot p) - \varpi(\ca{p}{-1}) \wedge \nabla(\ca{p}{0})\\&\quad\quad + \varpi(\ca{\nabla(p)}{-1}) \wedge \ca{\nabla(p)}{0} + \pr{\nabla} \circ \nabla(p)\\
		&= \dv{P}(\dv{P}(p)) -\iu{}\,F[\nabla](\dv{P}(p))+\iu{}\,F[\nabla](\dv{P}(p)) = 0,
	\end{align*}
	as required. Given the explicit form of \(\ver[\dt{P,\nabla}]\) and  \(\rest{\dt{P,\nabla}}{\Omega^1_{P,\oplus}}\), one can now check that
	\begin{gather*}
		\ver^{2,2}[\dt{P,\nabla}] = \Proj_1 : \Omega^2_{P,\oplus} = \Omega^2_{P,\ver} \oplus \Lambda^1_H \otimes \Omega^1_{P,\hor} \oplus \Omega^2_{P,\hor} \to \Omega^2_{P,\ver},\\
		\ver^{2,1}[\dt{P,\nabla}] = \Proj_2 : \Omega^2_{P,\oplus} = \Omega^2_{P,\ver} \oplus \Lambda^1_H \otimes \Omega^1_{P,\hor} \oplus \Omega^2_{P,\hor} \to \Lambda^1_H \otimes \Omega^1_{P,\hor},
	\end{gather*}
	so that \((\Omega_{P,\oplus},\dt{P,\nabla})\) defines an object of \(\cG[\Omega^{\leq 2}_H]\); given the explicit form of \(\Pi_\oplus\), it now follows that the connection \(\Pi\) is totally prolongable with
	\[
		\Pi_\oplus \wedge \Pi_\oplus = \Proj_1 \oplus{} 0 \oplus{} 0, \quad \Pi_\oplus \wedge \id + \id \wedge \Pi_\oplus = \Proj_1 \oplus \Proj_2 \oplus{} 0
	\]
	with respect to the decomposition \(\Omega^2_{P,\oplus} =  \Omega^2_{P,\ver} \oplus \Lambda^1_H \otimes \Omega^1_{P,\hor} \oplus \Omega^2_{P,\hor}\).
\end{proof}

Given a bicovariant \fodc{} \((\Omega^1_H,\dt{H})\) on \(H\) that is locally freeing for \(P\), Propositions~\ref{fieldstrengthprop} and~\ref{astver} now guarantee that the \((\Omega^1_H,\dt{H})\)-adapted prolongable Atiyah space \(\pr{\fr{At}}[\Omega^{1}_H]\) is a \(\pr{\fr{G}}\)-invariant subset of the prolongable Atiyah space \(\pr{\fr{At}}\) on which the assignment of curvature \(2\)-forms defines a \(\pr{\fr{G}}\)-equivariant map.

\begin{proposition}
	Let \((\Omega^1_H,\dt{H})\) be a bicovariant \fodc{} on \(H\) that is locally freeing for \(P\). The \((\Omega^1_H,\dt{H})\)-adapted prolongable Atiyah space \(\pr{\fr{At}}[\Omega^{1}_H]\) is a \(\dva{\fr{G}}\)-invariant subset of the prolongable Atiyah space \(\pr{\fr{At}}\), and the assignment
	\[
		F \coloneqq (\nabla \mapsto F[\nabla]) : \pr{\fr{at}}_{\can}[\Omega^1_H] \to \Hom_P^H(\Omega^1_{P,\ver},\Omega^2_{P,\hor})
	\]
	defines a \(\dva{\fr{G}}\)-equivariant map.
\end{proposition}

\begin{remark}
	The \(\pr{\fr{G}}\)-invariant subset \(\pr{\fr{At}}[\Omega^{1}_H]\) of the affine space \(\pr{\fr{At}}\) is defined by an affine-quadratic constraint, and as such can be viewed as a \(\pr{\fr{G}}\)-invariant affine quadric subset of \(\pr{\fr{At}}\). In general, one should not expect \(\pr{\fr{At}}[\Omega^1_H]\) to be an affine-linear subspace of \(\pr{\fr{At}}\).
\end{remark}

\begin{remark}
	It follows that the resulting action groupoid \(\pr{\fr{G}} \ltimes \pr{\fr{At}}[\Omega^1_H]\) defines a subgroupoid of the action groupoid \(\fr{G} \ltimes \fr{At}\), where \(\fr{G}\) and \(\fr{At}\) are, respectively, the gauge group and Atiyah space of \(P\) with respect to the first-order horizontal calculus \((\Omega^1_B,\dt{B};\Omega^1_{P,\hor})\).
\end{remark}

Suppose that \((\Omega^1_H,\dt{H})\) is a bicovariant \fodc{} on \(H\) that is locally freeing for \(P\), and let \((\Omega_H,\dt{H})\) be any bicovariant prolongation of \((\Omega^1_H,\dt{H})\). We now promote Proposition~\ref{isothm} to an explicit equivalence of categories---refining the equivalence of Proposition~\ref{equiv1}---that realises the concrete action groupoid \(\pr{\fr{G}} \ltimes \pr{\fr{At}}[\Omega^1_H]\), which is independent of the choice of bicovariant prolongation \((\Omega_H,\dt{H})\), as a deformation retraction of the action groupoid \(\cG[\Omega^{\leq 2}_H] \ltimes \cA[\Omega^{\leq 2}_H]\) of the abstract gauge action on prolongable bimodule connections. This rigorously justifies identifying the action of \(\pr{\fr{G}}\) on \(\pr{\fr{At}}[\Omega^1_H]\) as the affine action of global gauge transformations on principal connections for the quantum principal \(H\)-bundle \(P\) with respect to the bicovariant \(\ast\)-differential calculus \((\Omega_H,\dt{H})\) and the second-order horizontal calculus \((\Omega_B,\dt{B};\Omega_{P,\hor})\). This follows from the proof of Proposition~\ref{equiv1}, \emph{mutatis mutandis}, together with Proposition~\ref{sodciso} and the proof of Proposition~\ref{isothm}.

\begin{proposition}\label{sodcequiv}
	Let \((\Omega_H,\dt{H})\) be a bicovariant \(\ast\)-differential calculus on \(H\) whole \fodc{} is locally freeing for \(P\). The groupoid homomorphism
	\[
		\Sigma[\Omega^{\leq 2}_H] : \pr{\fr{G}} \ltimes \pr{\fr{At}}[\Omega^{1}_H] \to \cG[\Omega^{\leq 2}_H] \ltimes \cA[\Omega^{\leq 2}_H]
	\]
	given by
	\begin{multline}
		\forall (\phi,\nabla) \in \pr{\fr{G}} \ltimes \pr{\fr{At}}[\Omega^{1}_H], \\ \Sigma[\Omega^{\leq 2}_H](\phi, \nabla) \coloneqq \left(\phi : (\Omega_{P,\oplus},\dt{P,\nabla},\Pi_{\oplus})  \to (\Omega_{P,\oplus},\dt{P,\phi\act\nabla},\Pi_{\oplus})\right)
	\end{multline}
	is an equivalence of groupoids with left inverse and homotopy inverse \(A[\Omega^{\leq 2}_H]\) given by
	\begin{multline}
		\forall \left(f : (\Omega_1,\dt{1};\Pi_1) \to (\Omega_2,\dt{2};\Pi_2)\right) \in \cG[\Omega^{\leq 2}_H] \ltimes \cA[\Omega^{\leq 2}_H],\\
				A[\Omega^{\leq 2}_H]\mleft(f : (\Omega_1,\dt{1};\Pi_1) \to (\Omega_2,\dt{2};\Pi_2)\mright) \coloneqq \left(f, C[\Omega_1] \circ \nabla_{\Pi_1}\right).
	\end{multline}
	In particular, there exists a homotopy \(\eta[\Omega^{\leq 2}_H] : \id_{\cG[\Omega^{\leq 2}_H] \ltimes \cA[\Omega^{\leq 2}_H]} \Rightarrow \Sigma[\Omega^{\leq 2}_H] \circ A[\Omega^{\leq 2}_H]\), which is necessarily unique, such that
	\begin{equation}
		\forall (\Omega_P,\dt{P};\Pi) \in \Ob(\cG[\Omega^{\leq 2}_H] \ltimes \cA[\Omega^{\leq 2}_H]), \quad \mu[\Omega^{\leq 2}_H] \circ \pi[\Omega^{\leq 2}_H]\mleft(\eta[\Omega^{\leq 2}_H]_{(\Omega_P,\dt{P};\Pi)}\mright) = \id_P.
	\end{equation}
\end{proposition}

This justifies the identification of the prolongable gauge group \(\pr{\fr{G}}\) as \emph{the} gauge group of the quantum principal \(H\)-bundle \(P\) with respect to the bicovariant \(\ast\)-differen\-tial calculus \((\Omega_H,\dt{H})\) on \(H\) and the second-order horizontal calculus \((\Omega_B,\dt{B};\Omega_{P,\hor})\).

\begin{corollary}
	Let \((\Omega_H,\dt{H})\) be a bicovariant \(\ast\)-differential calculus on \(H\) whole \fodc{} is locally freeing for \(P\). The star-injective groupoid homomorphism \(\mu[\Omega^{\leq 2}_H] : \cG[\Omega^{\leq 2}_H] \to \Aut(P)\) has range \(\pr{\fr{G}}\), so that, after restriction of codomain,
	\[
		\mu[\Omega^{\leq 2}_H] : \cG[\Omega^{\leq 2}_H] \to \pr{\fr{G}}, \quad \mu[\Omega^{\leq 2}_H] \circ \pi[\Omega^{\leq 2}_H] : \cG[\Omega^{\leq 2}_H] \ltimes \cA[\Omega^{\leq 2}_H] \to \pr{\fr{G}}
	\]
	both define coverings of groupoids.
\end{corollary}

\begin{remark}
	The groupoid equivalences \(\Sigma[\Omega^1_H]\), \(\Sigma[\Omega^{\leq 2}_H]\), \(A[\Omega^1_H]\), and \(A[\Omega^{\leq 2}_H]\), the subgroupoid inclusion \(\pr{\fr{G}} \ltimes \pr{\fr{At}}[\Omega^1_H] \inj \fr{G} \ltimes \fr{At}\), and the canonical star-injective groupoid homomorphisms \(\cG[\Omega^{\leq 2}_H] \ltimes \cA[\Omega^{\leq 2}_H] \to \cG[\Omega^1_H] \ltimes \cA[\Omega^1_H]\) and \(\cG[\Omega^{\leq 2}_H] \to \cG[\Omega^1_H]\) fit into the following commutative diagrams in the category of groupoids:
	\[
		\begin{tikzcd}[column sep=small]
			\pr{\fr{G}} \ltimes \pr{\fr{At}}[\Omega^{1}_H] \arrow[r, hookrightarrow]  \arrow[d,"\Sigma\lbrack\Omega^{\leq 2}_H\rbrack"'] & \fr{G} \ltimes \fr{At} \arrow[d,"\Sigma\lbrack\Omega^1_H\rbrack"] \\ \cG[\Omega^{\leq 2}_H] \ltimes \cA[\Omega^{\leq 2}_H] \arrow[r]  & \cG[\Omega^1_H] \ltimes \cA[\Omega^1_H]
		\end{tikzcd}\hfill
		\begin{tikzcd}[column sep=small]
			\pr{\fr{G}} \ltimes \pr{\fr{At}}[\Omega^1_H] \arrow[r, hookrightarrow]  \arrow[d,leftarrow,"A\lbrack\Omega^{\leq 2}_H\rbrack"'] & \fr{G} \ltimes \fr{At} \arrow[d,leftarrow,"A\lbrack\Omega^1_H\rbrack"]\\ \cG[\Omega^{\leq 2}_H] \ltimes \cA[\Omega^{\leq 2}_H] \arrow[r]  & \cG[\Omega^1_H] \ltimes \cA[\Omega^1_H]
		\end{tikzcd}
	\]
	Moreover, the homotopies \[\eta[\Omega^1_H] : \id_{\cG[\Omega^1_H] \ltimes \cA[\Omega^1_H]} \Rightarrow \Sigma[\Omega^1_H] \circ A[\Omega^1_H], \quad \eta[\Omega^{\leq 2}_H] : \id_{\cG[\Omega^{\leq 2}_H] \ltimes \cA[\Omega^{\leq 2}_H]} \Rightarrow \Sigma[\Omega^{\leq 2}_H] \circ A[\Omega^{\leq 2}_H]\]
	and the natural transformations
	\[
		\id_{\cG[\Omega^{\leq 2}_H] \ltimes \cA[\Omega^{\leq 2}_H]} \Rightarrow \id_{\cG[\Omega^1_H] \ltimes \cA[\Omega^1_H]}, \quad \Sigma[\Omega^{\leq 2}_H] \circ A[\Omega^{\leq 2}_H] \Rightarrow \Sigma[\Omega^1_H] \circ A[\Omega^1_H]
	\]
	induced by the subgroupoid inclusion \(\pr{\fr{G}} \ltimes \pr{\fr{At}}[\Omega^1_H] \inj \fr{G} \ltimes \fr{At}\) and the canonical star-injective groupoid homomorphisms \(\cG[\Omega^{\leq 2}_H] \ltimes \cA[\Omega^{\leq 2}_H] \to \cG[\Omega^1_H] \ltimes \cA[\Omega^1_H]\) fit into the following commutative diagram in the category of functors:
	\[
		\begin{tikzcd}
			\id_{\cG[\Omega^{\leq 2}_H] \ltimes \cA[\Omega^{\leq 2}_H]} \arrow[r, Rightarrow] \arrow[d,Rightarrow,"\eta\lbrack\Omega^{\leq 2}_H\rbrack"'] & \id_{\cG[\Omega^1_H] \ltimes \cA[\Omega^1_H]} \arrow[d,Rightarrow,"\eta\lbrack \Omega^1_H \rbrack"]\\
			\Sigma[\Omega^{\leq 2}_H] \circ A[\Omega^{\leq 2}_H] \arrow[r,Rightarrow] & \Sigma[\Omega^1_H] \circ A[\Omega^1_H]
		\end{tikzcd}
	\]
\end{remark}

Given a bicovariant \(\ast\)-differential calculus \((\Omega_H,\dt{H})\) on \(H\) whose \fodc{} is locally freeing for \(P\), we now use the weak equivalence \(\Sigma[\Omega^{\leq 2}_H]\) of Proposition~\ref{sodcequiv} to construct a \(\pr{\fr{G}}\)-equivariant moduli space of  strongly \((H;\Omega_H,\dt{H})\)-principal \sodc{} on \(P\) inducing the second-order horizontal calculus \((\Omega_B,\dt{B};\Omega_{P,\hor})\) on \(P\). To do so, we shall need \((\Omega_H,\dt{H})\) to be Woronowicz's \emph{canonical} prolongation~\cite{Woronowicz}*{\S\S 3--4} of its \fodc{} \((\Omega^1_H,\dt{H})\), so that
\begin{equation}\label{wor}
	\Lambda^2_H \coloneqq (\Omega^2_H)^{\operatorname{co}H} = \frac{\Lambda^1_H \otimes_{\mathbb{C}} \Lambda^1_H}{\Span\set{ \mu \otimes \nu + \ca{\mu}{-1} \act \nu \otimes \ca{\mu}{0} \given \mu,\nu \in \Lambda^1_H}}
\end{equation}
by~\cite{Majid17}*{\S 2}. For us, the distinguishing feature of this particular prolongation (as opposed to, e.g., the maximal prolongation) will be the \(H\)-equivariance of \(\dt{H}\) when restricted to  \(\Lambda \coloneqq \Omega^{\operatorname{co}H}\), which one can view as \(\Ad\)-equivariance of the dualised Lie bracket.

\begin{proposition}\label{worprop}
	Let \((\Omega^1_H,\dt{H})\) be a bicovariant \fodc{} on \(H\), let \((\Omega,\dt{})\) be a bicovariant prolongation of \((\Omega^1_H,\dt{H})\), and let \(\Lambda \coloneqq \Omega^{\operatorname{co}H}\) denote the corresponding graded left crossed \(H\)-module \(\ast\)-algebra of right coinvariant forms, where \(\Lambda^0 = \bC\). The following are equivalent:
	\begin{enumerate}
		\item\label{wor1} the restriction of \(\dt{}\) to \(\Lambda^1 = \Lambda^1_H\) satisfies
		\[
			\forall h \in H, \, \forall \mu \in \Lambda^1_H, \quad \dt{H}(h \act \mu) = h \act \dt{H}(\mu);
		\]
		\item\label{wor2} the product \(\Lambda^1 \otimes_{\bC} \Lambda^1 \to \Lambda^2\) satisfies the braided commutation relation
		\[
			\forall \mu,\nu \in \Lambda^1_H, \quad \mu \wedge \nu + \ca{\mu}{-1} \act \nu \wedge \ca{\mu}{0} = 0.
		\]
	\end{enumerate}
\end{proposition}

\begin{proof}
	On the one hand, suppose that Condition~\ref{wor1} holds. Then for all \(h,k \in H\), by \eqref{deg2eq},
	\begin{align*}
		\ca{\varpi(h)}{-1}\act\varpi(k) \wedge \ca{\varpi(h)}{0}
		&= \cm{h}{1}S(\cm{h}{3}) \act \varpi(k) \wedge \varpi(\cm{h}{2})\\
		&= - \cm{h}{1}S(\cm{h}{4}) \act \varpi(k) \wedge \cm{h}{2} \act \varpi(S(\cm{h}{3}))\\
		&= - \cm{h}{1} \act \left(S(\cm{h}{3}) \act \varpi(k) \wedge \varpi(S(\cm{h}{2}))\right)\\
		&= \cm{h}{1}\act \left(\varpi(S(\cm{h}{3})) \wedge S(\cm{h}{2}) \act \varpi(k)\right)\\
		&= \cm{h}{1} \act \varpi(S(\cm{h}{4})) \wedge \cm{h}{2} S(\cm{h}{3}) \act \varpi(k)\\
		&= - \varpi(h) \wedge \varpi(k),
	\end{align*}
	so that Condition~\ref{wor2} is satisfied. On the other hand, suppose that Condition~\ref{wor2} is satisfied. Then for all \(h, k \in H\), by \eqref{deg2eq},
	\begin{align*}
		\dt{H}(h \act \varpi(k)) 
		- h \act \dt{H}\varpi(k) &= \cm{h}{1}\act\varpi(k) \wedge\varpi(\cm{h}{2}) + \varpi(\cm{h}{1}) \wedge \cm{h}{2}\act\varpi(k)\\
		&= \cm{h}{1}\act\varpi(k) \wedge\varpi(\cm{h}{2}) - \cm{h}{1}S(\cm{h}{3})\cm{h}{4}\act\varpi(k)\wedge \varpi(\cm{h}{2})\\
		&= 0 
	\end{align*}
	so that Condition~\ref{wor1} is satisfied.
\end{proof}

In the case that \((\Omega_H,\dt{H})\) is the canonical prolongation of \((\Omega^1_H,\dt{H})\), we shall construct the aforementioned \(\pr{\fr{G}}\)-equivariant moduli space using the following distinguished set of relative gauge potentials, which will indeed turn out to be a \(\pr{\fr{G}}\)-invariant subspace of prolongable relative gauge potentials.

\begin{definition}
	Let \((\Omega^1_H,\dt{H})\) be a bicovariant \fodc{} on \(H\) that is locally freeing for \(P\). Let \(\bA\) be a prolongable relative gauge potential on \(P\) with respect to the second-order horizontal calculus \((\Omega_B,\dt{B};\Omega_{P,\hor})\). We say that \(\bA\) is \emph{canonically \((\Omega_H^{1},\dt{H})\)-adapted} if and only if it is \((\Omega^1_H,\dt{H})\)-adapted and satisfies both of the following:
	\begin{gather}
		\forall \mu \in \Lambda^1_H, \, \forall \alpha \in \Omega^1_B, \quad \left[\omega[\bA](\mu),\alpha\right] = 0,\label{toteq1}\\
		\forall \mu,\nu \in \Lambda^1_H, \quad \omega[\bA](\mu) \wedge \omega[\bA](\nu) + \omega[\bA](\ca{\mu}{-1}\act\nu) \wedge \omega[\bA](\ca{\mu}{0}) = 0.\label{toteq2}
	\end{gather}
	We denote by \(\pr{\fr{at}}_{\can}[\Omega^{1}_H]\) the subset of all canonically \((\Omega^{1}_H,\dt{H})\)-adapted relative gauge potentials on \(P\) with respect to \((\Omega_B,\dt{B};\Omega_{P,\hor})\).
\end{definition}

Now, let \((\Omega_1,\dt{1})\), \((\Omega_2,\dt{2}) \in \Ob(\cG[\Omega^{\leq 2}_H])\), where \((\Omega_H,\dt{H})\) is the canonical prolongation of a bicovariant \fodc{} \((\Omega^1_H,\dt{H})\) on \(H\) that is locally freeing for \(P\). Observe that the left \(H\)-covariant \textsc{sodc} \((\Omega_1,\dt{1})\) and \((\Omega_2,\dt{2})\) on the principal left \(H\)-comodule \(\ast\)-algebra \(P\) are isomorphic if and only if \(\id_P : (\Omega_1,\dt{1}) \to (\Omega_2,\dt{2})\) is an arrow in \(\cG[\Omega^{\leq 2}_H]\). Since the subgroupoid of such arrows is precisely \(\ker\mu[\Omega^{\leq 2}_H]\), it follows that \((\Omega_1,\dt{1})\) and \((\Omega_2,\dt{2})\) are isomorphic if and only if they define the same object in the quotient \(\cG[\Omega^{\leq 2}_H]/\ker\mu[\Omega^{\leq 2}_H]\), which will turn out to be well-defined and canonically isomorphic to the action groupoid \(\pr{\fr{G}} \ltimes \left(\pr{\fr{At}}[\Omega^1_H]/\pr{\fr{at}}_{\can}[\Omega^{1}_H]\right)\). Thus, the quadric subset \(\pr{\fr{At}}[\Omega^{\leq 1}_H]/\pr{\fr{at}}_{\can}[\Omega^{1}_H]\) of the quotient affine space \(\pr{\fr{At}}/\pr{\fr{at}}_{\can}[\Omega^{1}_H]\) yields the desired \(\pr{\fr{G}}\)-equivariant affine moduli space of strongly \((\Omega_H,\dt{H})\)-principal \sodc{} on \(P\) inducing \((\Omega_B,\dt{B};\Omega_{P,\hor})\).

\begin{theorem}\label{qpbthm}
	Let \((\Omega^1_H,\dt{H})\) be a bicovariant \fodc{} on \(H\) that is locally freeing for \(P\), and let \((\Omega_H,\dt{H})\) be its canonical prolongation. Suppose that \(\pr{\fr{At}}[\Omega^1_H]\) is non-empty. Let \(\pr{\fr{at}}\) be the space of prolongable gauge potentials on \(P\) with respect to \((\Omega_B,\dt{B};\Omega_{P,\hor})\), and let \(\fr{at}[\Omega^1_H]\) be subspace of \((\Omega^1_H,\dt{H})\)-adapted elements. The set \(\pr{\fr{at}}_{\can}[\Omega^{1}_H]\) defines a \(\pr{\fr{G}}\)-invariant subspace of \(\fr{at}[\Omega^1_H] \cap \pr{\fr{at}}\), such that \(\pr{\fr{At}}[\Omega_H]\) is invariant under translation by \(\pr{\fr{at}}_{\can}[\Omega^{1}_H]\); the subgroupoid \(\ker\mu[\Omega^{\leq 2}_H]\) of \(\cG[\Omega^{\leq 2}_H]\) is wide and has trivial isotropy groups, so that the quotient groupoid \(\cG[\Omega^{\leq 2}_H]/\ker\mu[\Omega^{\leq 2}_H]\) is well-defined; and there exists a unique isomorphism  \[\pr{\tilde{\Sigma}}[\Omega^{\leq 2}_H] : \pr{\fr{G}} \ltimes \left(\pr{\fr{At}}[\Omega^1_H]/\pr{\fr{at}}_{\can}[\Omega^{1}_H]\right) \to \cG[\Omega^{\leq 2}_H]/\ker\mu[\Omega^{\leq 2}_H],\] such that
	\begin{multline}
		\forall (\phi,\nabla) \in \pr{\fr{G}} \ltimes \pr{\fr{At}}[\Omega^1_H], \\ \pr{\tilde{\Sigma}}[\Omega^{\leq 2}_H]\mleft(\phi,\nabla+\pr{\fr{at}}_{\can}[\Omega^{1}_H]\mright) = \left[\pi[\Omega^{\leq 2}_H] \circ \Sigma[\Omega^{\leq 2}_H](\phi,\nabla)\right]_{\ker\mu[\Omega^{\leq 2}_H]}.
	\end{multline}
\end{theorem}

As a preliminary to the proof of this theorem, we shall prove the following sequence of lemmata. In what follows, we shall assume the hypotheses of Theorem~\ref{qpbthm}; in particular, \((\Omega_{P,\ver},\dv{P})\) will denote the second-order vertical calculus induced by \((\Omega_H,\dt{H})\).

\begin{lemma}\label{qpblem1}
	Let \(N : \Omega^1_{P,\ver} \to \Omega^1_{P,\hor}\) be a left \(H\)-covariant morphism of \(P\)-\(\ast\)-bimodules satisfying
	\[
		\forall \mu \in \Lambda^1_H, \, \forall \alpha \in \Omega^1_B, \quad [N(\mu),\alpha] = 0.
	\]
	Then \(N \circ \dv{P} \in \pr{\fr{at}}\). Moreover, for all \(\nabla \in \pr{\fr{At}}\), the map \(\nabla N : \Omega^1_{P,\ver} \to \Omega^2_{P,\hor}\) given by
	\[
		\forall \mu \in \Lambda^1_H, \, \forall p \in P, \quad (\nabla N)(\mu \cdot p) \coloneqq \nabla(N(\mu)) \cdot p
	\]
	defines a left \(H\)-covariant morphism of \(P\)-bimodules, such that \(-\iu{}\nabla N\) is \(\ast\)-preserving.
\end{lemma}

\begin{proof}
	Let us first show that \(\bA \coloneqq N \circ \dv{P} \in \fr{at}[\Omega^1_H]\) is prolongable. Indeed, for all \(p, q \in P\) and \(b \in B\), we have
	\begin{align*}
		\bA(p) \cdot \dt{B}(b) \cdot q - p \cdot \dt{B} \cdot \bA(q) &= N(\varpi(\ca{p}{-1})) \wedge \ca{p}{0} \dt{B}(b)	 \cdot q - p \cdot \dt{B}(b) \cdot \wedge \varpi(\ca{q}{-1}) \cdot \ca{q}{0}\\
		&= N\mleft(\varpi(\ca{p}{-1})\epsilon(\ca{q}{-1})+\ca{p}{-1}\act\varpi(\ca{q}{-1})\mright)\wedge \ca{p}{0}\dt{B}(b)\ca{q}{0}\\
		&= N\mleft(\varpi(\ca{p}{-1}\ca{q}{-1})\mright) \wedge \ca{p}{0} \cdot \dt{B}(b) \cdot \ca{q}{0}\\
		&= N\mleft(\varpi(\ca{(p \cdot \dt{B}(b) \cdot q)}{-1})\mright) \wedge \ca{(p \cdot \dt{B}(b) \cdot q)}{0},
	\end{align*}
	so that \(\bA\) is prolongable with \(\pr{\bA}\) given by
	\[
		\forall \alpha \in \Omega^1_{P,\hor}, \quad \pr{\bA}(\alpha) = N(\varpi(\ca{\alpha}{-1})) \wedge \ca{\alpha}{0}.
	\]
	Now, given \(\nabla \in \pr{\fr{At}}\), let us show that the left \(H\)-covariant right \(P\)-module map \(\nabla N\) is left \(P\)-linear and \(\ast\)-preserving. First, since \([N(\Lambda^1_H),\Omega^1_B]=\set{0}\), it follows that for all \(\mu \in \Lambda^1_H\), \(p \in P\), and \(\beta \in \Omega^1_B\),
	\begin{multline*}
		p \cdot \dt{B}(b) \wedge N(\mu) = - p \cdot N(\mu) \wedge \dt{B}(b) = -N(\ca{p}{-1}\act\mu) \cdot \ca{p}{0} \wedge \dt{B}(b)\\ = -N(\ca{(p \cdot \dt{B}(b))}{-1} \act \mu) \wedge \ca{(p \cdot\dt{B}(b))}{0},
	\end{multline*}
	so that
	\[
		\forall \mu \in \Lambda^1_H, \, \forall \alpha \in \Omega^1_{P,\hor}, \quad \alpha \wedge N(\mu) + N(\ca{\alpha}{-1}\act\mu) \wedge \ca{\alpha}{0} = 0;	
	\]
	Hence, for all \(p \in P\) and \(\mu \in \Lambda^1_H\),
	\begin{align*}
		p \cdot (\nabla N)(\mu) 
		&= \pr{\nabla}\mleft(p \cdot N(\mu)\mright) - \nabla(p) \wedge N(\mu)\\
		&= \pr{\nabla}\mleft(N(\ca{p}{-1}\act\mu)\mright) \cdot \ca{p}{0} - N(\ca{p}{-1}\act\mu) \wedge \nabla(\ca{p}{0}) - \nabla(p) \wedge N(\mu)\\
		&= (\nabla N)(p \cdot \mu) - N(\ca{p}{-1}\act\mu) \wedge \nabla(\ca{p}{0}) + N(\ca{\nabla(p)}{-1} \act \mu) \wedge \ca{\nabla(p)}{0}\\
		&= (\nabla N)(p \cdot \mu),
	\end{align*}
	which shows that \(\nabla N\) is left \(P\)-linear; a similar calculation then shows that the map \(-\iu{}\nabla N\) is also \(\ast\)-preserving.
\end{proof}

\begin{lemma}\label{qpblem2}
	Recall that \((\Omega_H,\dt{H})\) is the canonical prolongation of \((\Omega^1_H,\dt{H})\). Suppose that \(N : \Omega^1_{P,\ver} \to \Omega^1_{P,\hor}\) be a left \(H\)-covariant morphism of \(P\)-\(\ast\)-bimodules satisfying
	\begin{equation}\label{qpblem2eq}
		\forall \mu,\nu \in \Lambda^1_H, \quad N(\mu) \wedge N(\nu) + N(\ca{\mu}{-1}\act\nu) \wedge N(\ca{\mu}{0}) = 0,
	\end{equation}
	Then, the map \([N,N] : \Omega^1_{P,\ver} \to \Omega^2_{P,\hor}\) defined by
	\[
		\forall h \in H, \, \forall p \in P, \quad [N,N](\varpi(h) \cdot p) \coloneqq N(\varpi(\cm{h}{1})) \wedge N(\varpi(\cm{h}{2})) \cdot p
	\]
	is a left \(H\)-covariant morphism of \(P\)-bimodules, such that \(-\iu{}[N,N]\) is \(\ast\)-preserving.
\end{lemma}

\begin{proof}
	First, by \eqref{qpblem2eq} together with \eqref{wor}, the map
	\[
		(\mu \otimes \nu \mapsto N(\mu) \wedge N(\nu)) : \Lambda^1_H \otimes_{\bC} \Lambda^1_H \to \Omega^2_{P,\hor}
	\]
	descends to a map \(\tilde{N} : \Lambda^2_H \to \Omega^2_{P,\hor}\), so that \([N,N] : \Omega^1_{P,\ver} \to \Omega^2_{P,\hor}\) is well-defined as left \(H\)-covariant right \(P\)-linear map and given by
	\[
		\forall \mu \in \Lambda^1_H, \, \forall p \in P, \quad [N,N](\mu \cdot p) \coloneqq \tilde{N}(\dt{H}\mu) \cdot p.
	\]
	Now, for all \(p \in P\) and \(h \in H\),
	\begin{align*}
		p \cdot [N,N](\varpi(h)) &= p \cdot N(\varpi(\cm{h}{1})) \wedge N(\varpi(\cm{h}{2}))\\
		&= N(\ca{p}{-2} \act \varpi(\cm{h}{1})) \wedge N(\ca{p}{-1} \act \varpi(\cm{h}{2})) \cdot \ca{p}{0}\\
		&= \tilde{N}\mleft(\ca{p}{-2}\act\varpi(\cm{h}{1}) \wedge \ca{p}{-1} \act\varpi(\cm{h}{2})\mright) \cdot \ca{p}{0}\\
		&= \tilde{N}\mleft(\ca{p}{-1}\act\dt{H}\varpi(h)\mright) \cdot \ca{p}{0}\\
		&= \tilde{N}\mleft(\dt{H}(\ca{p}{-1}\act\varpi(h))\mright) \cdot \ca{p}{0}\\
		&= [N,N](p \cdot \varpi(h))
	\end{align*}
	by Proposition~\ref{wor}, so that \([N,N]\) is left \(P\)-linear. A qualitatively identical calculation now shows that \(-\iu{}[N,N]\) is \(\ast\)-preserving.
\end{proof}

\begin{lemma}\label{qpblem3}
	Recall that \((\Omega_H,\dt{H})\) is the canonical prolongation of \((\Omega^1_H,\dt{H})\) and that the subset \(\pr{\fr{At}}[\Omega^1_H]\) of \(\pr{\fr{At}}\) is assumed to be non-empty. Let \(\bA\) be a gauge potential on \(P\) with respect to the first-order horizontal calculus \((\Omega^1_B,\dt{B};\Omega^1_{P,\hor})\). The following are equivalent:
	\begin{enumerate}
		\item\label{qpblem3c} the relative gauge potential \(\bA\) is prolongable and canonically \((\Omega^{1}_H,\dt{H})\)-adapted;
		\item\label{qpblem3a} for every \(\nabla \in \pr{\fr{At}}[\Omega^{1}_H]\), the gauge potential \(\nabla + \bA\) is prolongable and \((\Omega^1_H,\dt{H})\)-adapted, and \(\id_P\) induces an arrow
	\[
		\left(\id_P : \pi[\Omega^{\leq 2}_H]\circ\Sigma[\Omega^{\leq 2}_H](\nabla) \to \pi[\Omega^{\leq 2}_H]\circ\Sigma[\Omega^{\leq 2}_H](\nabla+\bA)\right) \in \cG[\Omega^{\leq 2}_H];
	\]
		\item\label{qpblem3b} there exists \(\nabla \in \pr{\fr{At}}[\Omega^{1}_H]\), such that \(\nabla+\bA \in \pr{\fr{At}}[\Omega^{1}_H]\) and \(\id_P\) induces an arrow
	\[
		\left(\id_P : \pi[\Omega^{\leq 2}_H]\circ\Sigma[\Omega^{\leq 2}_H](\nabla) \to \pi[\Omega^{\leq 2}_H]\circ\Sigma[\Omega^{\leq 2}_H](\nabla+\bA)\right) \in \cG[\Omega^{\leq 2}_H].
	\]
	\end{enumerate}
	\end{lemma}

\begin{proof}
	First, suppose that condition~\ref{qpblem3c} holds; let \(N \coloneqq \omega[\bA]\), so that \(\bA = N \circ \dv{P}\). Let \(\nabla \in \pr{\fr{At}}[\Omega^1_H]\) be given. First, by Lemma~\ref{qpblem1}, \(\bA \in \dva{\fr{at}}\), so that \(\nabla + \bA \in \pr{\fr{At}}\). Next, by Lemma~\ref{qpblem1}, \(\pr{\bA}\) is given by
	\[
		\forall \alpha \in \Omega^1_{P,\hor}, \quad \pr{\bA}(\alpha) = N(\varpi(\ca{\alpha}{-1})) \wedge \ca{\alpha}{0},
	\]
	while by Lemmata~\ref{qpblem1} and~\ref{qpblem2}, respectively, \(-\iu{}\nabla N\) and \(-\iu{}[N,N]\) are well-defined left \(H\)-covariant morphisms of \(P\)-\(\ast\)-algebras; hence, for all \(p \in P\),
	\begin{align*}
		\iu{}\left(\bF[\nabla+\bA]-\bF[\nabla]\right)(p) &= \left(\pr{\nabla} \circ \bA + \pr{\bA} \circ \nabla + \pr{\bA} \circ \bA\right)(p)\\
		&= \pr{\nabla}\mleft(N(\varpi(\ca{p}{-1})) \cdot \ca{p}{0}\mright) + \pr{\bA}\mleft(\nabla(p)\mright) + \pr{\bA}\mleft(\bA(p)\mright)\\
		&\begin{multlined}= \pr{\nabla} \circ N\mleft(\varpi(\ca{p}{-1})\mright) \cdot \ca{p}{0} - N\mleft(\varpi(\ca{p}{-1})\mright) \wedge \nabla(\ca{p}{0}) \\+ N\mleft(\varpi(\ca{\nabla(p)}{-1})\mright) \wedge \ca{\nabla(p)}{0} + N\mleft(\varpi(\ca{\bA(p)}{-1})\mright) \wedge  \ca{\bA(p)}{0}\end{multlined}\\
		&= \nabla N\mleft(\varpi(\ca{p}{0})\mright) \cdot \ca{p}{0} + N(\varpi(\ca{p}{-2})) \wedge N(\varpi(\ca{p}{-1})) \cdot \ca{p}{0}\\
		&= \left(\nabla N + [N,N]\right) \circ \dv{P}(p),
	\end{align*}
	so that \(\nabla + \bA \in \pr{\fr{At}}[\Omega^{1}_H]\) with curvature \(2\)-form
	\[
		F[\nabla+\bA] = F[\nabla]-\iu{}\left(\nabla N + [N,N]\right).
	\]
	
	Let us now show that \(\id_P\) induces a morphism \[\left(\id_P : \pi[\Omega^{\leq 2}_H]\circ\Sigma[\Omega^{\leq 2}_H](\nabla) \to \pi[\Omega^{\leq 2}_H]\circ\Sigma[\Omega^{\leq 2}_H](\nabla+\bA)\right) \in \cG[\Omega^{\leq 2}_H].\] Since \(\Omega_{P,\oplus} \coloneqq \Lambda_H \dvatimes \Omega_{P,\hor}\) is generated as a left \(H\)-covariant graded \(\ast\)-algebra over \(P\) by \(\Lambda^1_H\) and \(\Omega^1_{P,\hor}\) subject to the relations
	\begin{gather*}
		\forall \mu \in \Lambda^1_H, \, \forall \alpha \in \Omega^1_{P,\hor}, \quad \alpha \wedge \mu + \ca{\alpha}{-1} \act \mu \wedge \ca{\alpha}{0} = 0,\\
		\forall \mu,\nu \in \Lambda^1_H, \quad \forall \mu \wedge \nu + \ca{\mu}{-1}\act\nu \wedge \ca{\mu}{0},
	\end{gather*}
	and since, by the proof of Lemma~\ref{qpblem2}, the map \(N : \Omega^1_{P,\ver} \to \Omega^1_{P,\hor}\) satisfies
	\begin{gather*}
		\forall \mu \in \Lambda^1_H, \, \forall \alpha \in \Omega^1_{P,\hor}, \quad \alpha \wedge N(\mu) + N(\ca{\alpha}{-1} \act \mu) \wedge \ca{\alpha}{0} = 0,\\
		\forall \mu,\nu \in \Lambda^1_H, \quad N(\mu) \wedge N(\nu) + N(\ca{\mu}{-1}\act\nu) \wedge N(\ca{\mu}{0}),
	\end{gather*}
	the left \(H\)-covariant morphisms \(\phi : \Omega^1_{P,\oplus} \to \Omega^1_{P,\oplus}\) and \(\psi : \Omega^1_{P,\oplus} \to \Omega^1_{P,\oplus}\) given by
	\begin{gather*}
		\forall \omega \in \Omega^1_{P,\ver}, \, \forall \alpha \in \Omega^1_{P,\hor}, \quad \phi(\omega+\alpha) \coloneqq \omega + N(\omega) + \alpha,\\
		\forall \omega \in \Omega^1_{P,\ver}, \, \forall \alpha \in \Omega^1_{P,\hor}, \quad \psi(\omega+\alpha) \coloneqq \omega - N(\omega) + \alpha
	\end{gather*}
	respectively, extend uniquely to left \(H\)-covariant graded \(\ast\)-endomorphisms of \(\Omega_{P,\oplus}\), such that \(\rest{\phi}{P} = \rest{\psi}{P} = \id_P\). Moreover, for all \(\omega \in \Omega^1_{P,\ver}\) and \(\alpha \in \Omega^1_{P,\hor}\),
	\begin{gather*}
		\psi \circ \phi(\omega+\alpha) = \psi(\omega+N(\omega)+\alpha) = \omega - N(\omega) + N(\omega) + \alpha = \omega + \alpha,\\
		\phi \circ \psi(\omega+\alpha) = \phi(\omega-N(\omega)+\alpha) = \omega+N(\omega)-N(\omega)+\alpha=\omega+\alpha,
	\end{gather*}
	so that \(\phi\) and \(\psi\) are automorphisms with \(\phi^{-1}=\psi\). Finally, on the one hand,
	\[
		\ver[\dt{P,\nabla+\bA}] \circ \phi = \Proj_1 \circ \phi = \Proj_1 = \ver[\dt{P,\nabla}],
	\]
	where \(\Proj_1 : \Omega^1_{P,\oplus} \to \Omega^1_{P,\ver}\) is the projection onto \(\Omega^1_{P,\ver}\) along \(\Omega^1_{P,\hor}\), while on the other, for all \(p \in P\),
	\begin{multline*}
		\phi \circ \dt{P,\nabla}(p) = \phi(\dv{P}+\nabla(p)) = \dv{P}(p) + N(\dv{P}(p)) + \nabla(p)\\ = \dv{P}(p) + \bA(p) + \nabla(p) = \dt{P,\nabla+\bA} \circ \phi(p).
	\end{multline*}
	Hence, \(\id_P\) induces an arrow \(\id_P : \pi[\Omega^{\leq 2}_H]\circ \Sigma[\Omega^{\leq 2}_H](\nabla) \to \pi[\Omega^{\leq 2}_H]\circ \Sigma[\Omega^{\leq 2}_H](\nabla+\bA)\) in \(\cG[\Omega^{\leq 2}_H]\) with \((\id_P)_\ast = \phi\), as required. Since \(\nabla \in \pr{\fr{At}}[\Omega^1_H]\) was arbitrary, condition~\ref{qpblem3a} is satsified. 
	
	Now, condition~\ref{qpblem3a} trivially implies condition~\ref{qpblem3b}, so suppose that condition~\ref{qpblem3b} is satisfied. Fix \(\nabla \in \pr{\fr{At}}[\Omega^{1}_H]\), such that \(\id_P\) induces an arrow
	\[
		\left(\id_P : \pi[\Omega^{\leq 2}_H]\circ\Sigma[\Omega^{\leq 2}_H](\nabla) \to \pi[\Omega^{\leq 2}_H]\circ\Sigma[\Omega^{\leq 2}_H](\nabla+\bA)\right) \in \cG[\Omega^{\leq 2}_H].
	\]
	Hence, let \(\Phi \coloneqq (\id_P)_\ast : \Omega_{P,\oplus} \to \Omega_{P,\oplus}\), so that \(\Phi\) is an automorphism of the graded \(H\)-comodule \(\ast\)-algebra \(\Omega_{P,\oplus}\), such that \(\rest{\Phi}{P} = \id_P\) and
	\begin{gather*}
		\Proj_1 \circ \rest{\Phi}{\Omega^1_{P,\oplus}} = \ver[\dt{P,\nabla+\bA}] \circ (\id_P)_\ast = (\id_P)_{\ast,\ver} \circ \ver[\dt{P,\nabla}] = \Proj_1,\\
		\dt{P,\nabla+\bA} = \dt{P,\nabla+\bA} \circ \id_P = (\id_P)_\ast \circ \dt{P,\nabla} = \Phi \circ \dt{P,\nabla},
	\end{gather*}
	where \(\Proj_1 : \Omega^1_{P,\oplus} \to \Omega^1_{P,\ver}\) is the projection onto \(\Omega^1_{P,\ver}\) along \(\Omega^1_{P,\hor}\). On the one hand, since \(\Proj_1 \circ \rest{\Phi}{\Omega^1_{P,\oplus}} = \Proj_1\), it follows that \((\id-\Phi)(\Omega^1_{P,\ver}) \subseteq \ker \Proj_1 = \Omega^1_{P,\hor}\). On the other hand, for all \(b \in B\), we see that
	\[
		\Phi(\dt{B}(b)) = \Phi \circ \dt{P,\nabla}(b) = \dt{P,\nabla+\bA}(b) = \dt{B}(b),
	\]
	so that, in turn,
	\(
		\rest{\Phi}{\Omega^1_{P,\hor}} = \id_{\Omega^1_{P,\hor}}
	\).
	Thus we can define a left \(H\)-covariant morphism \(N : \Omega^1_{P,\ver} \to \Omega^1_{P,\hor}\) of \(P\)-\(\ast\)-bimodules by
	\[
		\forall \omega \in \Omega^1_{P,\ver}, \quad N(\omega) \coloneqq \omega - \Phi(\omega);
	\]
	we claim that \(\bA = N \circ \dv{P} \in \pr{\fr{at}}[\Omega_H]\). First, for every \(p \in P\),
	\begin{multline*}
		N \circ \dv{P}(p) = \dv{P}(p) - \Phi(\dv{P}(p)) = \dv{P}(p) - \Phi(\dt{P,\nabla}(p)-\nabla(p))\\ = \dv{P}(p) - \dt{P,\nabla+\bA}(p) - \nabla(p) = \bA(p),
	\end{multline*}
	so that \(\bA = N \circ \dv{P} \in \fr{at}[\Omega^1_H]\). Next, for all \(\mu \in \Lambda^1_H\) and \(\alpha \in \Omega^1_B\),
	\[
		N(\mu) \wedge \alpha = (\mu - \Phi(\mu)) \wedge \alpha = \mu \wedge \alpha - \Phi(\mu \wedge \alpha) = -\alpha \wedge \mu + \Phi(\alpha \wedge \mu) = -\alpha \wedge (\id-\Phi)(\mu) = -\alpha \wedge N(\mu).
	\]
	Finally, for all \(\mu,\nu \in \Lambda^1_H\),
	\begin{align*}
		N(\mu) \wedge N(\nu) &= \mu \wedge \nu  - \Phi(\mu) \wedge \nu - \mu \wedge \Phi(\nu) + \Phi(\mu) \wedge \Phi(\nu)\\
		&= \mu \wedge \nu - \Phi(\mu) \wedge \nu - \Phi(\Phi^{-1}(\mu) \wedge \nu) + \Phi(\mu \wedge \nu)\\
		&= -\ca{\mu}{-1}\act\nu \wedge \ca{\mu}{0} + \ca{\Phi(\mu)}{-1}\act\nu \wedge \ca{\Phi(\mu)}{0} + \Phi\mleft(\ca{\Phi^{-1}(\mu)}{-1}\act\nu \wedge \ca{\Phi^{-1}(\mu)}{0}\mright)\\&\quad\quad + \Phi\mleft(\ca{\mu}{-1}\act\nu \wedge \ca{\mu}{0}\mright)\\
		&= -\ca{\mu}{-1}\act\nu \wedge \ca{\mu}{0} + \ca{\mu}{-1}\act \nu \wedge \Phi(\ca{\mu}{0}) + \Phi\mleft(\ca{\mu}{-1} \act \nu \wedge \inv{\Phi}(\ca{\mu}{0})\mright)\\&\quad\quad + \Phi\mleft(\ca{\mu}{-1}\act\nu \wedge \ca{\mu}{0}\mright)\\
		&= -N(\ca{\mu}{-1}\act\nu) \wedge N(\ca{\mu}{0}).
	\end{align*}
	Hence, \(\bA = N \circ \dv{P} \in \pr{\fr{at}}_{\can}[\Omega^{1}_H]\), as required. Thus, condition~\ref{qpblem3c} holds.
\end{proof}

\begin{proof}[Proof of Theorem~\ref{qpbthm}]

Let us first show that the subset \(\pr{\fr{at}}_{\can}[\Omega^{1}_H]\) is \(\pr{\fr{G}}\)-invariant. Let \(\bA \in \pr{\fr{at}}_{\can}[\Omega^{1}_H]\) and \(N \coloneqq \omega[\bA]\); let \(\phi \in \pr{\fr{G}}\), so that \(\phi \act \bA = (\phi \act N) \circ \dv{P} \in \fr{at}[\Omega^1_H]\) for 
\[
	\phi \act N \coloneqq \phi_\ast \circ N \circ (\id \otimes \inv{\phi}).
\]
Then, for all \(\mu \in \Lambda^1_H\) and \(\alpha \in \Omega^1_B\),
\[
	(\phi \act N)(\mu) \wedge \alpha + \alpha \wedge (\phi \act N)(\nu) = \phi_\ast \mleft(N(\mu)\wedge \alpha + \alpha \wedge N(\nu)\mright) = 0,
\]
while for all \(\mu,\nu \in \Lambda^1_H\),
\begin{multline*}
	(\phi \act N)(\mu) \wedge (\phi \act N)(\nu) + (\phi \act N)(\ca{\mu}{-1}\act\nu) \wedge (\phi \act N)(\ca{\mu}{0})\\ = \phi_\ast\mleft(N(\mu) \wedge N(\nu) + N(\ca{\mu}{-1}\act\nu) \wedge N(\ca{\mu}{0}) \mright) = 0,
\end{multline*}
so that \(\phi \act \bA \in \pr{\fr{at}}_{\can}[\Omega^{1}_H]\).

Next, let us show that \(\pr{\fr{at}}_{\can}[\Omega^1_H]\) is a subspace of \(\fr{at}[\Omega^1_H] \cap \pr{\fr{at}}\). First, by Lemma~\ref{qpblem1}, \(\pr{\fr{at}}_{\can}[\Omega^1_H] \subset \fr{at}[\Omega^1_H] \cap \pr{\fr{at}}\). Next, since \(0 \in \pr{\fr{at}}_{\can}[\Omega^1_H]\) and since equations \eqref{toteq1} and \eqref{toteq2} are homogeneous of degree \(1\) and \(2\) respectively, it follows that \(\pr{\fr{at}}_{\can}[\Omega^1_H]\) is non-empty and closed under scalar multiplication by \(\bR\).  Finally, let \(\bA_1\), \(\bA_2 \in \pr{\fr{at}}_{\can}[\Omega^1_H]\). Let \(\nabla \in \pr{\fr{At}}[\Omega^1_H]\) be given. By Lemma~\ref{qpblem3}, \(\nabla+\bA_1 \in \pr{\fr{At}}[\Omega^1_H]\), so that \[\left(\id_P : \pi[\Omega^{\leq 2}_H] \circ \Sigma[\Omega^{\leq 2}_H](\nabla) \to \pi[\Omega^{\leq 2}_H] \circ \Sigma[\Omega^{\leq 2}_H](\nabla+\bA_1)\right) \in \cG[\Omega^{\leq 2}_H];\] let \(\Phi \coloneqq (\id_P)_\ast\) be the unique left \(H\)-covariant automorphism of the \(P\)-\(\ast\)-bimodule \(\Omega^1_{P,\oplus}\), such that \(\Phi \circ \dt{P,\nabla} = \dt{P,\nabla+\bA_1}\). Hence, by Lemma~\ref{qpblem3}, \(\nabla+\bA_1 + \bA_2 \in \pr{\fr{At}}[\Omega^1_H]\), so that \[\left(\id_P : \pi[\Omega^{\leq 2}_H]\circ\Sigma[\Omega^{\leq 2}_H](\nabla+\bA_1) \to \pi[\Omega^{\leq 2}_H]\circ\Sigma[\Omega^{\leq 2}_H](\nabla+\bA_1+\bA_2)\right) \in \cG[\Omega^{\leq 2}_H];\] let \(\Psi \coloneqq (\id_P)_\ast\) be the unique left \(H\)-covariant automorphism of the \(P\)-\(\ast\)-bimodule \(\Omega^1_{P,\oplus}\), such that \(\Psi \circ \dt{P,\nabla+\bA_1} = \dt{P,\nabla+\bA_1+\bA_2}\). Then \(\Psi \circ \Phi\) is a left \(H\)-covariant automorphism of the \(P\)-\(\ast\)-bimodule \(\Omega^1_{P,\oplus}\) satisfying
\begin{gather*}
	(\Psi \circ \Phi) \circ \dt{P,\nabla} = \Phi \circ \dt{P,\nabla+\bA_1} = \dt{P,\nabla+\bA_1+\bA_2},\\
	\ver[\dt{P,\nabla+\bA_1+\bA_2}] \circ (\Psi \circ \Phi) = \ver[\dt{P,\nabla+\bA_1}] \circ \Phi = \ver[\dt{P,\nabla}],
\end{gather*}
so that \(\id_P : \pi[\Omega^{\leq 2}_H] \circ \Sigma[\Omega^{\leq 2}_H](\nabla) \to \pi[\Omega^{\leq 2}_H]\circ\Sigma[\Omega^{\leq 2}_H](\nabla+\bA_1+\bA_2)\) defines an arrow in \(\cG[\Omega^{\leq 2}_H]\). Hence, by Lemma~\ref{qpblem3}, \(\bA_1+\bA_2 \in \pr{\fr{at}}_{\can}[\Omega^1_H]\). Thus, \(\pr{\fr{at}}_{\can}[\Omega^1_H]\) is a subspace of \(\pr{\fr{at}} \cap \fr{at}[\Omega^1_H]\). Note that \(\pr{\fr{At}}[\Omega^1_H] \subset \pr{\fr{At}}\) is invariant under translation by the subspace \(\pr{\fr{at}}_{\can}[\Omega^1_H] \subset \pr{\fr{at}}\) by Lemma~\ref{qpblem3}. 

Finally, by Proposition~\ref{sodcequiv}, the covering \(\pi[\Omega^{\leq 2}_H] : \cG[\Omega^{\leq 2}_H] \ltimes \cA[\Omega^{\leq 2}_H] \to \cG[\Omega^{\leq 2}_H]\), the star-injective groupoid homomorphism \(\mu[\Omega^{\leq 2}_H] : \cG[\Omega^{\leq 2}_H] \to \Aut(P)\), the injective groupoid homomorphism \(\Sigma[\Omega^{\leq 2}_H] : \pr{\fr{G}} \ltimes \pr{\fr{At}}[\Omega^1_H] \to \cG[\Omega^{\leq 2}_H] \ltimes \cA[\Omega^{\leq 2}_H]\), and the left inverse \(A[\Omega^{\leq 2}_H]\) of \(\Sigma[\Omega^{\leq 2}_H]\)  satisfy the hypotheses of Lemma~\ref{groupoidlemma}. Thus, the equivalence kernel \(\sim\) of the map
\[
	\left(\nabla \mapsto \left[\pi[\Omega^{\leq 2}_H] \circ \Sigma[\Omega^{\leq 2}_H](\id_P,\nabla)\right]_{\ker\mu[\Omega^{\leq 2}_H]}\right) : \pr{\fr{At}}[\Omega^1_H] \to \Ob(\cG[\Omega^{\leq 2}_H]/\ker\mu[\Omega^{\leq 2}_H])
\]
is a \(\pr{\fr{G}}\)-invariant equivalence relation, the subgroupoid \(\ker\mu[\Omega^{\leq 2}_H]\) is wide and has trivial isotropy groups, the quotient groupoid \(\cG[\Omega^{\leq 2}_H]/\ker\mu[\Omega^{\leq 2}_H]\) is well-defined, and there exists a unique isomorphism \(\pr{\tilde{\Sigma}}[\Omega^{\leq 2}_H] : \pr{\fr{G}} \ltimes \pr{\fr{At}}[\Omega^1_H]/\sim \, \iso \cG[\Omega^{\leq 2}_H]/\ker\mu[\Omega^{\leq 2}_H]\), such that
\[
	\forall (\phi,\nabla) \in \pr{\fr{G}} \ltimes \pr{\fr{At}}[\Omega^1_H], \quad \pr{\tilde{\Sigma}}[\Omega^{\leq 2}_H](\phi,[\nabla]_\sim) = \left[\pi[\Omega^{\leq 2}_H] \circ \Sigma[\Omega^{\leq 2}_H](\phi,\nabla)\right]_{\ker\mu[\Omega^{\leq 2}_H]}.
\]
Lemma~\ref{qpblem3} now implies that \(\sim\) is the orbit equivalence relation with respect to the translation action of \(\pr{\fr{at}}_{\can}[\Omega^1_H] \leq \pr{\fr{at}}\) on \(\pr{\fr{At}}[\Omega^1_H] \subset \pr{\fr{At}}\).
\end{proof}

Finally, observe that the subspace \(\Inn(\pr{\fr{at}})\) of inner prolongable gauge potentials acts by translation on the affine space \(\pr{\fr{At}}\). Given a bicovariant \fodc{} \((\Omega^1_H,\dt{H})\) on \(H\) that is locally freeing for \(P\), we can now characterise the stabiliser subgroup in \(\Inn(\pr{\fr{at}})\) of the affine quadric subset \(\pr{\fr{At}}[\Omega^1_H]\) of \(\pr{\fr{At}}\). This will turn out to be the \(\pr{\fr{G}}\)-invariant \(\bR\)-linear subspace of all inner prolongable relative gauge potentials of the following form.

\begin{definition}
	Let \((\Omega^1_H,\dt{H})\) be a bicovariant \fodc{} on \(H\) that is locally freeing for \(P\). Let \(\bA\) be an inner prolongable relative gauge potential on \(P\) with respect to the second-order horizontal calculus \((\Omega_B,\dt{B};\Omega_{P,\hor})\). We say that \(\bA\) is \emph{\((\Omega^1_H,\dt{H})\)-semi-adapted} whenever
	\[
		\bF_{\rel}[\bA] = F_{\rel}[\bA] \circ \dv{P},
	\]
	for some (necessarily unique) left \(H\)-covariant morphism \(F_{\rel}[\bA] : \Omega^1_{P,\ver} \to \Omega^2_{P,\hor}\) of \(P\)-\(\ast\)-bimodules, in which case, we call \(F_{\rel}[\bA]\) the \emph{relative curvature \(2\)-form} of \(\bA\). We denote by \(\Inn(\pr{\fr{at}};\Omega^1_H)\) the subspace of all \((\Omega^1_H,\dt{H})\)-semi-adapted inner prolongable gauge potentials on \(P\) with respect to \((\Omega_B,\dt{B};\Omega_{P,\hor})\).
\end{definition}

\begin{example}
	By Example~\ref{relex}, for every \(\phi \in \Inn(\fr{G})\) and \(\nabla \in \pr{\fr{At}}\), we have
	\[
		\phi \act \nabla -\nabla \in \Inn(\pr{\fr{at}};\Omega^1_H), \quad F_{\rel}[\phi \act \nabla-\nabla] = 0.
	\] 
\end{example}

\begin{proposition}\label{semiadapt}
	Let \((\Omega^1_H,\dt{H})\) be a bicovariant \fodc{} on \(H\) that is locally freeing for \(P\). Suppose that the \((\Omega^1_H,\dt{H})\)-adapted Atiyah space \(\pr{\fr{At}}[\Omega^1_H]\) is non-empty. Let \(\bA\) be an inner prolongable relative gauge potential on \(P\) with respect to \((\Omega_B,\dt{B};\Omega_{P,\hor})\). The following are equivalent:
	\begin{enumerate}
		\item the inner prolongable relative gauge potential \(\bA\) is \((\Omega^1_H,\dt{H})\)-semi-adapted;
		\item for all \(\nabla \in \pr{\fr{At}}[\Omega^1_H]\), the prolongable gauge potential \(\nabla+\bA\) is also \((\Omega^1_H,\dt{H})\)-adapted;
		\item there exists \(\nabla \in \pr{\fr{At}}[\Omega^1_H]\), such that \(\nabla+\bA \in \pr{\fr{At}}\) is also \((\Omega^1_H,\dt{H})\)-adapted.
	\end{enumerate}
	Thus, in particular, the affine quadric subset \(\pr{\fr{At}}[\Omega^1_H]\) of the affine space \(\pr{\fr{At}}\) is invariant under translation by \(\Inn(\pr{\fr{at}};\Omega^1_H)\).
\end{proposition}

\begin{proof}
	First, suppose that \(\bA \in \Inn(\pr{\fr{at}};\Omega^1_H)\). Let \(\nabla \in \pr{\fr{At}}[\Omega^1_H]\). Then \(\nabla + \bA \in \pr{\fr{At}}\) with
	\[
		\bF[\nabla+\bA] = \bF[\nabla] + \bF_{\rel}[\bA] = F[\nabla] \circ \dv{P} + F_{\rel}[\bA] \circ \dv{P} = \left(F[\nabla] + F_{\rel}[\bA]\right) \circ \dv{P},
	\]
	so that \(\nabla + \bA \in \pr{\fr{At}}[\Omega^1_H]\) with \(F[\nabla+\bA] = F[\nabla]+F_{\rel}[\bA]\). Now, suppose that there exists \(\nabla \in \pr{\fr{At}}[\Omega^1_H]\), such that \(\nabla+\bA \in \pr{\fr{At}}[\Omega^1_H]\). Then
	\[
		\bF_{\rel}[\bA] = \bF[\nabla+\bA] - \bF[\nabla] = F[\nabla+\bA] \circ \dv{P} - F[\nabla] \circ \dv{P} = \left(F[\nabla+\bA]-F[\nabla]\right)\circ\dv{P},
	\]
	so that \(\bA \in \Inn(\pr{\fr{at}};\Omega^1_H)\) with \(F_{\rel}[\bA] = F[\nabla+\bA]-F[\nabla]\).
\end{proof}

We now summarise the basic properties of the subspace \(\Inn(\pr{\fr{at}};\Omega^1_H)\) and observe that the action of \(\pr{\fr{G}}\) on \(\pr{\fr{At}}[\Omega^1_H]\) descends to a second-order analogue of the action of \(\Out(\fr{G})\) on \(\Out(\fr{At})\); this all follows, \emph{mutatis mutandis}, from the proof of Proposition~\ref{innprop}.

\begin{proposition}
	Let \((\Omega^1_H,\dt{H})\) be a bicovariant \fodc{} on \(H\) that is locally freeing for \(P\). Suppose that \(\pr{\fr{At}}[\Omega^1_H]\) is non-empty. The subspace \(\Inn(\pr{\fr{at}};\Omega^1_H)\) of \(\Inn(\pr{\fr{at}})\) consists of \(\pr{\fr{G}}\)-invariant vectors, so that the affine action of \(\pr{\fr{G}}\) on the affine space \(\pr{\fr{At}}\) descends to an affine action of \(\pr{\fr{G}}\) on the quotient affine space
	\(
		\pr{\fr{At}}/\Inn(\pr{\fr{at}};\Omega^1_H)
	\).
	Furthermore, the inner gauge group \(\Inn(\fr{G})\) acts trivially on \(\pr{\fr{At}}/\Inn(\pr{\fr{at}};\Omega^1_H)\), so that the action of \(\pr{\fr{G}}\) on \(\pr{\fr{At}}/\Inn(\pr{\fr{at}};\Omega^1_H)\) further descends to an affine action of the outer prolongable gauge group \(\Out(\pr{\fr{G}})\) on \(\pr{\fr{At}}/\Inn(\pr{\fr{at}};\Omega^1_H)\) that restricts, in turn, to an action on	 the  quadric subset\[
		\Out(\pr{\fr{At}}[\Omega^1_H]) \coloneqq \pr{\fr{At}}[\Omega^1_H]/\Inn(\pr{\fr{at}};\Omega^1_H).
	\]
\end{proposition}

Finally, we record the basic properties of the relative curvature \(2\)-form.

\begin{corollary}
	Let \((\Omega^1_H,\dt{H})\) be a bicovariant \fodc{} on \(H\) that is locally freeing for \(P\). The map \(F_{\rel} \coloneqq (\bA \mapsto F_{\rel}[\bA]) : \Inn(\pr{\fr{at}};\Omega^1_H) \to \Hom_P^H(\Omega^1_{P,\ver},\Omega^2_{P,\hor})\) is linear with range contained in \(\Hom_P^H(\Omega^1_{P,\ver},\Omega^2_{P,\hor})^{\pr{\fr{G}}}\). In particular, it satisfies
	\[
		\forall \nabla \in \pr{\fr{At}}[\Omega^1_H], \, \forall \bA \in \Inn(\pr{\fr{at}};\Omega^1_H), \quad F[\nabla+\bA] = F[\nabla] + F_{\rel}[\bA].
	\]
\end{corollary}

\begin{proof}
	The map \(F_{\rel}\) is \(\bR\)-linear and \(\pr{\fr{G}}\)-equivariant by Corollary~\ref{relcor} and \(\pr{\fr{G}}\)-equivar\-iance of \(\dv{P}\). Since \(\pr{\fr{G}}\) acts trivially on \(\Inn(\pr{\fr{at}};\Omega^1_H)\), it follows that the range of \(F_{\rel}\) is contained in \(\Hom_P^H(\Omega^1_{P,\ver},\Omega^2_{P,\hor})^{\pr{\fr{G}}}\). The relation between the maps \(F\) on \(\pr{\fr{At}}[\Omega^1_H]\) and \(F_\rel\) on \(\Inn(\pr{\fr{at}};\Omega^1_H)\) now follows by the proof of Proposition~\ref{semiadapt}.
\end{proof}

\section{Gauge theory on crossed products as lazy cohomology}\label{sec4}

Unlike in the commutative case, trivial quantum principal bundles can encode non-trivial dynamical information. As \'{C}a\'{c}i\'{c}--Mesland have already observed in the context of spectral triples~\cite{CaMe}*{\S 3.4}, the noncommutative \(\bT^m\)-gauge theory of crossed products by \(\bZ^m\) can be related to the degree \(1\) group cohomology of \(\bZ^m\) with certain geometrically meaningful coefficients. In light of the groupoid equivalences of Proposition~\ref{equiv1} and~\ref{sodcequiv}, we shall now similarly relate the gauge theory of (non-twisted) crossed product algebras to certain generalisations of group cohomology in degree \(1\).

\subsection{Cohomological preliminaries}

In effect, \'{C}a\'{c}i\'{c}--Mesland  computed the gauge group \(\fr{G}\) and Atiyah space \(\fr{At}\) of a crossed product by \(\bZ^m\) in terms of the degree \(1\) group cohomology of \(\bZ^m\) with coefficients in a certain group of unitaries and a certain \(\mathbf{R}[\bZ^m]\)-module of noncommutative \(1\)-forms, respectively~\cite{CaMe}*{Thm.\ 3.36}. In this section, we construct suitable generalisations of these two distinct but interrelated cases of degree \(1\) group cohomology to arbitrary Hopf \(\ast\)-algebras; in the process, we provide an \emph{ad hoc} generalisation of degree \(1\) Sweedler cohomology to not-necessarily-cocommutative Hopf \(\ast\)-algebras and non-trivial coefficients. From now on, let \(H\) be a Hopf \(\ast\)-algebra.

We begin by recalling more-or-less standard constructions of convolution algebras and bimodules on \(H\) with suitable coefficients.

\begin{definition}
	Let \(B\) be a [graded] \(\ast\)-algebra. The \emph{\(B\)-valued convolution algebra on \(H\)} is the [graded] unital \(\ast\)-algebra \(\cC(H;B)\) defined by endowing \(\Hom_{\bC}(H;B)\) with the product and \(\ast\)-structure defined by
	\begin{gather}
		\forall f,g \in \cC(H;B), \, \forall h \in H, \quad (f \star g)(h) \coloneqq f(\cm{h}{1})g(\cm{h}{2}),\\
		\forall f \in \cC(H;B),\, \forall h \in H, \quad f^\ast(h) \coloneqq f(S(h)^\ast)^\ast,
	\end{gather}
	respectively, and the unit \(1_{\cC(H,B)} \coloneqq \epsilon(\cdot)1_B\).
\end{definition}

\begin{definition}
	Let \(B\) be a \(\ast\)-algebra, and let \(M\) be a \(B\)-\(\ast\)-bimodule. The \emph{\(M\)-valued convolution bimodule on \(H\)} is the \(\cC(H,B)\)-\(\ast\)-bimodule \(\cC(H;M)\) defined by endowing the \(\bC\)-vector space \(\Hom_{\bC}(H,M)\) with the left \(\cC(H,B)\)-module structure, right \(\cC(H,B)\)-module structure, and \(\ast\)-structure defined, respectively, by
	\begin{gather}
		\forall f \in \cC(H,B), \, \forall \mu \in \cC(H,M), \, \forall h \in H, \quad (f \star \mu)(h) \coloneqq f(\cm{h}{1}) \cdot \mu(\cm{h}{2}),\\
		\forall f \in \cC(H,B), \, \forall \mu \in \cC(H,M), \, \forall h \in H, \quad (\mu \star f)(h) \coloneqq \mu(\cm{h}{1}) \cdot f(\cm{h}{2}),\\
		\forall \mu \in \cC(H,M), \, \forall h \in H, \quad \mu^\ast(h) \coloneqq \mu(S(h)^\ast)^\ast.
	\end{gather}
\end{definition}

In the case of \(H\)-module \(\ast\)-algebras and \(H\)-equivariant \(\ast\)-bimodules, one can canonically embed coefficient algebras and bimodules into the resulting convolution algebras and bimodules, respectively, in a manner respecting \(\ast\)-bimodule structures. When \(H\) is no longer a group algebra, this will yield subtler notions of commutation than pointwise commutation.

\begin{proposition}
	Let \(B\) be a [graded] \(H\)-module \(\ast\)-algebra. Then \(\rho_B : B \to \cC(H;B)\) defined by
	\[
		\forall b \in B, \, \forall h \in H, \quad \rho_B(b)(h) \coloneqq b \ract h
	\]	
	is an injective [graded] \(\ast\)-homomorphism.
\end{proposition}

\begin{proposition}
	Let \(B\) be a  right \(H\)-module \(\ast\)-algebra, and let \(M\) be a right \(H\)-equivariant \(B\)-\(\ast\)-bimodule. Then the map \(\rho_M : M \to \cC(H;M)\) defined by
	\[
		\forall m \in M, \, \forall h \in H, \quad \rho_M(m)(h) \coloneqq m \ract h
	\]
	is an injective \(\ast\)-preserving \(\bC\)-linear map satisfying
	\[
		\forall a,b \in B, \, \forall m \in M, \quad \rho_B(a) \star \rho_M(m) \star \rho_B(b) = \rho_M(a \cdot m \cdot b).
	\]
	In particular, \(\cC(H;M)\) defines a \(B\)-\(\ast\)-bimodule with respect to \(\rho_B\).
\end{proposition}

We can now provide a suitable generalisation of the degree \(1\) group cohomology of a group \(\Gamma\) with coefficients in the unitary group of a commutative \(\Gamma\)-\(\ast\)-algebra. It can be viewed as an \emph{ad hoc} generalisation of degree \(1\) Sweedler cohomology~\cite{Sweedler} in the spirit of Bichon--Carnovale's `lazy' cohomology~\cite{BC} to the case of non-cocommutative Hopf algebras and non-commutative coefficient algebras. We shall use it to compute the gauge group of a crossed product by \(H\).

\begin{propositiondefinition}[cf.\ Sweedler~\cite{Sweedler}]\label{lazysweedler}
	Let \(B\) be a right \(H\)-module \(\ast\)-algebra, and let \(M\) be a right \(H\)-equivariant \(B\)-\(\ast\)-bimodule.
	\begin{enumerate}
		\item A \emph{lazy \((B,M)\)-valued Sweedler \(0\)-cochain on \(H\)} is an element of \[\CS_\ell^0(H;B,M) \coloneqq \Unit(\Cent_B(B\oplus M)) \leq \Unit(B).\]
		\item A \emph{lazy \((B,M)\)-valued Sweedler \(1\)-cocycle on \(H\)} is \(\sigma \in \Unit\mleft(\Cent_{\cC(H,B)}(\rho_B(B) \oplus \rho_M(M))\mright)\), such that \(\sigma(1) = 1\) and
		\begin{equation}\label{scocycle}
			\forall h,k \in H, \quad \sigma(hk) = \left(\sigma(h) \ract \cm{k}{1}\right)\sigma(\cm{k}{2});
		\end{equation}
		we denote by \(\ZS_\ell^1(H;B,M)\) the set of all lazy \((B,M)\)-valued Sweedler \(1\)-cocycles on \(H\), which defines a subgroup of \(\Unit(\cC(H,B))\).
		\item The \emph{coboundary map} is the homomorphism \(D: \ZS_\ell^0(H;B,M) \to \ZS_\ell^1(H;B,M)\) defined by
		\begin{equation}
			\forall \upsilon \in \CS_\ell^0(H;B,M), \, \forall h \in H, \quad D \upsilon(h) \coloneqq (\upsilon \ract h) \cdot \upsilon^\ast;
		\end{equation}
		thus, a \emph{lazy \((B,M)\)-valued Sweedler \(1\)-coboundary on \(H\)} is an element of
		\[
			\BS_\ell^1(H;B,M) \coloneqq D\mleft(\CS_\ell^0(H;B,M)\mright) \leq \Zent\mleft(\ZS^1_\ell(H;B,M)\mright).
		\]
		\item The \emph{lazy degree \(1\) Sweedler cohomology of \(H\) with coefficients in \((B,M)\)} is the group
		\[
			\HS_\ell^1(H;B,M) \coloneqq \ZS_\ell^1(H;B,M)/\BS_\ell^1(H;B,M).
		\]
	\end{enumerate}
\end{propositiondefinition}

The proof of this result---and several others besides---will require the following straightforward technical lemma.

\begin{lemma}\label{centrallem}
	Let \(B\) be a  right \(H\)-module \(\ast\)-algebra, and let \(M\) be a right \(H\)-equivariant \(B\)-\(\ast\)-bimodule.
	\begin{enumerate}
		\item Let \(b \in B\) and let \(\mu \in \cC(H,M)\). Then \([\rho_B(b),\mu] = 0\) if and only if
		\[
			\forall h \in H, \quad (b \ract \cm{h}{1}) \cdot \mu(\cm{h}{2}) = \mu(\cm{h}{1}) \cdot (b \ract \cm{h}{2}).
		\]
		\item Let \(\beta \in \cC(H,B)\) and let \(m \in M\). Then \([\beta,\rho_M(m)] = 0\) if and only if
		\[
			\forall h \in H, \quad \beta(\cm{h}{1}) \cdot (m \ract \cm{h}{2}) = (m \ract \cm{h}{1}) \cdot \beta(\cm{h}{2}).
		\]
	\end{enumerate}
\end{lemma}

\begin{proof}
	On the one hand, for all \(b \in B\), \(\mu \in \cC(H,M)\), and \(h \in H\),
	\begin{align*}
		[\rho_B(b),\mu](h) &= \rho_B(b)(\cm{h}{1}) \cdot \mu(\cm{h}{2}) - \mu(\cm{h}{1}) \cdot \rho_B(b)(\cm{h}{2})\\ &= (b \ract \cm{h}{1}) \cdot \mu(\cm{h}{2}) - \mu(\cm{h}{1}) \cdot (b \ract \cm{h}{2}).
	\end{align*}
	On the other hand, for all \(\beta \in \cC(H,B)\), \(m \in M\), and \(h \in H\),
	\begin{align*}
		[\beta,\rho_M(m)](h) &= \beta(\cm{h}{1}) \cdot \rho_M(m)(\cm{h}{2}) - \rho_M(m)(\cm{h}{1}) \cdot \beta(\cm{h}{2})\\
		&= 	\beta(\cm{h}{1}) \cdot (m \ract \cm{h}{2}) - (m \ract \cm{h}{1}) \cdot \beta(\cm{h}{2}). \qedhere
	\end{align*}
\end{proof}

\begin{proof}[Proof of Proposition~\ref{lazysweedler}]
	Before continuing, note that by Lemma~\ref{centrallem}, \(f \in \cC(H,B)\) centralises the subset \(\rho_B(B) \oplus \rho_M(B)\) of  \(\cC(H;B\oplus M)\) if and only if
	\begin{gather*}
		\forall b \in B,\,\forall h \in H, \quad f(\cm{h}{1}) \cdot (b \ract \cm{h}{2}) = (b \ract \cm{h}{1}) \cdot f(\cm{h}{2}),\\
		\forall m \in M,\,\forall h \in H, \quad f(\cm{h}{1}) \cdot (m \ract \cm{h}{2}) = (m \ract \cm{h}{1}) \cdot f(\cm{h}{2}).
	\end{gather*}

	We first show that \(\ZS_\ell^1(H;B,M)\) is a subgroup of \(\Unit\mleft(\Cent_{\cC(H,B)}(\rho_B(B) \oplus \rho_M(M))\mright)\). First, observe that \(\ZS_\ell^1(H;B,M) \ni 1_{\cC(H,B)}\). Next, let \(\sigma_1,\sigma_2 \in \ZS_\ell^1(H;B,M)\). Then for all \(h, k \in H\),
	\begin{align*}
		\sigma_1 \star \sigma_2(hk) &= \sigma_1(\cm{h}{1}\cm{k}{1}) \cdot \sigma_2(\cm{h}{2}\cm{k}{2})\\
		&= (\sigma_1(\cm{h}{1}) \ract \cm{k}{1}) \cdot \sigma_1(\cm{k}{2}) \cdot (\sigma_2(\cm{h}{2}) \ract \cm{k}{3}) \cdot \sigma_2(\cm{k}{4})\\
		&= (\sigma_1(\cm{h}{1}) \ract \cm{k}{1}) \cdot (\sigma_2(\cm{h}{2}) \ract \cm{k}{2}) \cdot \sigma_1(\cm{k}{3}) \cdot \sigma_2(\cm{k}{4})\\
		&= (\sigma_1 \star \sigma_2(h) \ract \cm{k}{1}) \cdot \sigma_1\star\sigma_2(\cm{k}{2}),
	\end{align*}
	while \(\sigma_1 \star \sigma_2(1) = \sigma_1(1) \cdot \sigma_2(1)=1\), so that, indeed, \(\sigma_1 \star \sigma_2 \in \ZS^1_\ell(H;B,M)\). Finally, let \(\sigma \in \ZS^1_\ell(H;B,M)\); note that \(\sigma^{-1} \in \Unit\mleft(\Cent_{\cC(H,B)}(\rho_B(B) \oplus \rho_M(M))\mright)\). Then, for all \(h,k \in H\),
	\begin{align*}
		\sigma^{-1}(hk) &= \sigma(S(h)^\ast S(k)^\ast)^\ast\\
		&= \left((\sigma(S(h)^\ast) \ract \cm{(S(k)^\ast)}{1}) \cdot \sigma(\cm{(S(k)^\ast)}{2}) \right)^\ast\\
		&=\left(\sigma(\cm{(S(k)^\ast)}{1}) \cdot (\sigma(S(h)^\ast) \ract \cm{(S(k)^\ast)}{2})\right)^\ast\\
		&= (\sigma^\ast(h) \ract \cm{k}{1})\cdot \sigma^\ast(\cm{k}{2}),
	\end{align*}
	while \(\sigma^\ast(1) = \sigma(1)^\ast = 1\), so that \(\sigma^{-1} \in \ZS^1_\ell(H;B,M)\).
	
	Next, we show that \(D: \ZS_\ell^0(H;B,M) \to \ZS_\ell^1(H;B,M)\) is well-defined as a function. Let \(\upsilon \in \CS^0_\ell(H;B,M)\), and set \(\sigma \coloneqq D\upsilon\). First, \(\sigma\) is unitary, since for all \(h \in H\),
	\begin{gather*}
		\sigma \star \sigma^\ast (h) = (\upsilon \ract \cm{h}{1}) \cdot \upsilon^\ast \cdot \left((\upsilon \ract S(\cm{h}{2})^\ast) \cdot \upsilon^\ast \right)^\ast  = (\upsilon \ract \cm{h}{1}) \cdot \upsilon^\ast \cdot \upsilon \cdot (\upsilon^\ast \ract \cm{h}{2}) = \epsilon(h) 1_B,\\
		\sigma^\ast \star \sigma(h) = \left((\upsilon \ract S(\cm{h}{1})^\ast) \cdot \upsilon^\ast \right)^\ast \cdot (\upsilon \ract \cm{h}{2}) \cdot \upsilon^\ast = \upsilon \cdot (\upsilon^\ast \ract \cm{h}{1}) \cdot (\upsilon \ract \cm{h}{2}) \cdot \upsilon^\ast = \epsilon(h) 1_B;
	\end{gather*}
	note, moreover, that \(\sigma(1) = (\upsilon \ract 1) \cdot \upsilon^\ast  = \upsilon \cdot \upsilon^\ast  = 1\). Next, \(\sigma \in \Cent_{\cC(H,B)}(\rho_B(B) \oplus \rho_M(M))\) since for all \(b \in B\), \(m \in M\), and \(h \in H\),
	\begin{align*}
		\sigma(\cm{h}{1}) \cdot (b \ract \cm{h}{2}) - (b \ract \cm{h}{1}) \cdot \sigma(\cm{h}{2}) &= (\upsilon \ract \cm{h}{1}) \cdot \upsilon^\ast  \cdot (b \ract \cm{h}{2}) - (b \ract \cm{h}{1}) \cdot (\upsilon \ract \cm{h}{2}) \cdot \upsilon^\ast\\
		&= \left((\upsilon \cdot b - b \cdot \upsilon) \ract h \right)\cdot \upsilon^\ast\\
		&= 0,
	\end{align*}
	and hence, \emph{mutatis mutandis}, \(\sigma(\cm{h}{1})\cdot(m\ract\cm{h}{2}) - (m \ract \cm{h}{1}) \cdot \sigma(\cm{h}{2})\). Finally, for all \(h,k\in H\),
	\[
		\sigma(hk) = (\upsilon \ract hk) \cdot \upsilon^\ast = \left((\upsilon \ract h) \ract \cm{k}{1}\right) \cdot (\upsilon^\ast \ract \cm{k}{2}) \cdot (\upsilon \ract \cm{k}{3}) \cdot \upsilon^\ast = (\sigma(h) \ract \cm{k}{1}) \cdot \sigma(\cm{k}{2}).
	\]
	Hence, \(D\upsilon \eqqcolon \sigma \in \ZS^1_{\ell}(H;B,M)\).
	
	Next, we show that \(D\) is a group homomorphism. First, note that 
	\[
		D1_B = \left(h \mapsto (1_B \ract h) \cdot 1_B\right) = \epsilon(\cdot)1_B = 1_{\cC(H,B)}.
	\]
	Now, let \(\upsilon_1,\upsilon_2 \in \CS^0_\ell(H,B;M)\). Then for all \(h \in H\),
	\begin{align*}
		D(\upsilon_1 \cdot \upsilon_2)(h) &= \left((\upsilon_1 \cdot \upsilon_2) \ract h\right) \cdot (\upsilon_1 \cdot \upsilon_2)^\ast \\
		&= (\upsilon_1 \ract \cm{h}{1}) \cdot (\upsilon_2 \ract \cm{h}{2}) \cdot \upsilon_2^\ast \upsilon_1^\ast \\
		&= D\upsilon_1(\cm{h}{1}) \cdot D\upsilon_2(\cm{h}{2}),
	\end{align*}
	so that, indeed, \(D(\upsilon_1 \cdot \upsilon_2) = D\upsilon_1 \star D\upsilon_2\).
	
	Finally, we show that \(\BS_\ell^1(H;B,M)\) is central in \(\ZS_\ell^1(H;B,M)\). Let \(\upsilon \in \CS^0_\ell(H;B,M)\) and \(\sigma \in \ZS^1_\ell(H;B,M)\). Then, for all \(h \in H\),
	\begin{multline*}
		\sigma \star D\upsilon \star \sigma^{-1} (h) = \sigma(\cm{h}{1}) \cdot (\upsilon \ract \cm{h}{2}) \cdot \upsilon^\ast \cdot \sigma^{-1}(\cm{h}{3}) = (\upsilon \ract \cm{h}{1}) \cdot \upsilon^\ast \cdot \sigma(\cm{h}{2}) \cdot \sigma^{-1}(\cm{h}{3}) = D\upsilon(h),
	\end{multline*}
	so that, indeed, \(\sigma \star D\upsilon \star \sigma^{-1} = D\upsilon\).
\end{proof}

We now provide a suitable generalisation of the degree \(1\) group cohomology of a group \(\Gamma\) with coefficients in an \(\bR\)-linear representation of \(\Gamma\). That degree \(1\) group cohomology can be generalised using the appropriate Hochschild cohomology was first observed by Sch\"{u}rmann~\cite{Schurmann}*{\S 3}; this observation has been used, for instance, to characterise the Haagerup property on locally compact quantum groups~\cite{DFSW}. We shall use the following variation on this theme to compute the Atiyah space of a crossed product by \(H\).

\begin{propositiondefinition}
	Let \(B\) be a right \(H\)-module \(\ast\)-algebra, and let \(M\) be a right \(H\)-equivariant \(B\)-\(\ast\)-bimodule.
	\begin{enumerate}
		\item A \emph{lazy \(M\)-valued Hochschild \(0\)-cochain on \(H\)} is an element of 
		\[
			\CH^0_\ell(H;M) \coloneqq \Zent_B(M)_{\sa}.
		\]
		\item A \emph{lazy \(M\)-valued Hochschild \(1\)-cocycle on \(H\)} is
		\(
			\mu \in \Zent_B\mleft(\cC(H,M)\mright)_{\sa}
		\),
		such that 
		\begin{equation}\label{hcocycle}
			\forall h,k \in H, \quad \mu(hk) = \mu(h) \ract k + \epsilon(h)\mu(k)
		\end{equation}
		and \(\mu(1) = 0\); we denote by \(\ZH^1_\ell(H;M)\) the set of all lazy \(M\)-valued Hochschild \(1\)-cocycles on \(H\), which defines a \(\bR\)-subspace of \(\cC(H,M)\).
		\item The \emph{coboundary map} is the homomorphism \(D : \CH^0_\ell(H;M) \to \ZH^1_\ell(H;M)\) with
		\begin{equation}
			\forall m \in \CH^0_\ell(H;M), \, \forall h \in H, \quad D(m)(h) \coloneqq m \ract h - m \epsilon(h);
		\end{equation}
		thus, a \emph{lazy \(M\)-valued Hochschild \(1\)-coboundary on \(H\)} is an element of
		\[
			\BH^1_\ell(H;M) \coloneqq D(\CH^0_\ell(H;M)) \leq \ZH^1_\ell(H;M).
		\]
		\item The \emph{lazy degree \(1\) Hochschild cohomology of \(H\) with coefficients in \(M\)} is the group
		\[
			\HH^1_\ell(H;M) \coloneqq \ZH^1_\ell(H;M)/\BH^1_\ell(H;M).
		\]
	\end{enumerate}	
\end{propositiondefinition}

\begin{proof}
	Note that the right \(H\)-module \(M\) becomes an \(H\)-bimodule with respect to the trivial left \(H\)-action. From this perspective, the set \(\ZH^1_\ell(H;B,M)\) is an \(\bR\)-subspace of the \(\bC\)-vector space of \(M\)-valued Hochschild \(1\)-cocycles on \(H\), the set \(\BH^1_\ell(H;B,M)\) is an \(\bR\)-subspace of the \(\bC\)-vector space of \(M\)-valued Hochschild \(1\)-coboundaries on \(H\), and the map \(D\) is the restriction of the usual Hochschild coboundary operator. Thus, it suffices to check the inclusion
	\(
		\BH^1_\ell(H;B,M) \subseteq \Zent_B(\cC(H,M))_{\sa}
	\).
	Let \(m \in \CH^0_\ell(H;M)\), and set \(\mu \coloneqq Dm\). Then, for all \(b \in B\) and \(h \in H\),
	\begin{align*}
		[\mu,\rho_B(b)](h) &= (m \ract \cm{h}{1} - m \epsilon(\cm{h}{1})) \cdot (b \ract \cm{h}{2}) - (b \ract \cm{h}{1}) \cdot 	(m \ract \cm{h}{2} - m \epsilon(\cm{h}{2})) \\
		&= [m,b] \ract h - [m,b \ract h] =0.
	\end{align*}
	Hence, indeed, \([\mu,\rho_B(b)] = 0\).
\end{proof}

We now show that the lazy Sweedler cohomology of \(H\) admits a canonical \(\bR\)-linear action on the lazy Hochschild cohomology of \(H\) with relevant coefficients, with respect to which \(H\)-equvariant \(\ast\)-derivations canonically yield \(1\)-cocycles in the sense of ordinary group cohomology. We shall compute the gauge action of gauge transformations on gauge potentials for a crossed product by \(H\) in terms of the resulting \(\bR\)-affine actions of lazy Sweedler cohomology on lazy Hochschild cohomology.

\begin{theorem}\label{lazymcthm}
	Let \(B\) be a right \(H\)-module \(\ast\)-algebra, and let \(M\) be a right \(H\)-equivariant \(B\)-\(\ast\)-bimodule. 
	\begin{enumerate}
		\item The group \(\ZS^1_\ell(H;B,M)\) acts \(\bR\)-linearly on \(\ZH^1_\ell(H;M)\) via conjugation, i.e.,
		\begin{equation}
			\forall \sigma \in \ZS^1_\ell(H;B,M), \, \forall \mu \in \ZH^1_\ell(H;M), \quad \sigma \act \mu \coloneqq \sigma \star \mu \star \sigma^{-1}.
		\end{equation}
		\item The subgroup \(\BS^1_\ell(H;B,M)\) acts trivially on \(\ZH^1_\ell(H;M)\) and the subspace \(\BH^1_\ell(H;M)\) consists of \(\ZS^1_\ell(H;B,M)\)-invariant vectors, so that the linear action of \(\ZS^1_\ell(H;B,M)\) on \(\ZH^1_\ell(H;M)\) descends to a linear action of \(\HS^1_\ell(H;B,M)\) on \(\HH^1_\ell(H;M)\).
		\item Let \(\partial : B \to M\) be an \(H\)-equivariant \(\ast\)-derivation. The action of \(\ZS^1_{\ell}(H;B,M)\) on \(\ZH^1_\ell(H;M)\) admits the \(1\)-cocycle \(\MC[\partial] : \ZS^1_\ell(H;B,M) \to \ZH^1_\ell(H;M)\) given by
	\begin{equation}
		\forall \sigma \in \ZS^1_\ell(H;B,M), \, \forall h \in H, \quad \MC[\partial](\sigma)(h) \coloneqq -\partial\sigma(\cm{h}{1}) \cdot \sigma^\ast(\cm{h}{2}).
	\end{equation}
	Furthermore, the \(1\)-cocycle \(\MC[\partial]\) satisfies
	\begin{equation}\label{coboundaryeq}
		\forall \upsilon \in \CS^0_\ell(H;B,M), \quad \MC[\partial](D\upsilon) = D(-\partial(\upsilon)\upsilon^\ast),
	\end{equation}
	and hence descends to a \(1\)-cocycle \(\tilde{\MC}[\partial] : \HS^1_\ell(H;B,M) \to \HH^1_\ell(H;M)\) for the induced action of \(\HS^1_{\ell}(H;B,M)\) on \(\HH^1_\ell(H;M)\).
	\end{enumerate}
\end{theorem}

\begin{proof}
	Let us first show that \(\ZH^1_\ell(H;M)\) is invariant under conjugation by \(\ZS^1_\ell(H;B,M)\); note that
	\(
		\set{\sigma \act \mu \given \sigma \in \ZS^1_\ell(H;B,M), \, \mu \in \ZH^1_\ell(H;M)} \subset Z_B(M)
	\) since \(\ZS^1_\ell(H,M)\) centralises \(\rho_B(B)\). Let \(\sigma \in \ZS^1_\ell(H;B,M)\) and let \(\mu \in \ZH^1_\ell(H;M)\). First,
	\[
		(\sigma \act \mu)(1) = \sigma(1) \cdot \mu(1) \cdot \inv{\sigma}(1) = 0.
	\] Next, since \(\sigma\) is unitary in the \(\ast\)-algebra \(\cC(H;B)\) and \(\mu\) is self-adjoint in the \(\cC(H;B)\)-\(\ast\)-bimodule \(\cC(H;M)\), it follows that \(\sigma \act \mu = \sigma \star \mu \star \sigma^\ast\) is self-adjoint. Finally, for all \(h,k \in H\),
	\begin{align*}
		&(\sigma \act \mu)(hk)\\ &\quad= \sigma(\cm{h}{1}\cm{k}{1}) \cdot \mu(\cm{h}{2}\cm{k}{2}) \cdot \inv{\sigma}(\cm{h}{3}\cm{k}{3})\\
		&\quad= (\sigma(\cm{h}{1}) \ract \cm{k}{1}) \cdot \sigma(\cm{k}{2}) \cdot \left(\mu(\cm{h}{2}) \ract \cm{k}{3} - \epsilon(\cm{h}{2})\mu(\cm{k}{3}) \right) \cdot (\inv{\sigma}(\cm{h}{3}) \ract \cm{k}{4}) \cdot \inv{\sigma}(\cm{k}{5})\\
		&\quad=  (\sigma(\cm{h}{1}) \ract \cm{k}{1}) \cdot (\mu(\cm{h}{2}) \ract \cm{k}{2}) \cdot \sigma(\cm{k}{3}) \cdot (\inv{\sigma}(\cm{h}{3})\ract\cm{k}{4}) \cdot \inv{\sigma}(\cm{k}{5})\\
		&\quad\quad\quad\quad -(\sigma(\cm{h}{1}) \ract \cm{k}{1}) \cdot \sigma(\cm{k}{2}) \cdot (\inv{\sigma}(\cm{h}{2} \ract \cm{k}{3}) \cdot \mu(\cm{k}{4}) \cdot \inv{\sigma}(\cm{k}{5})\\
		&\quad= (\sigma(\cm{h}{1}) \ract \cm{k}{1}) \cdot (\mu(\cm{h}{2}) \ract \cm{k}{2}) \cdot (\inv{\sigma}(\cm{h}{3})\ract\cm{k}{3}) \cdot \sigma(\cm{k}{4}) \cdot \inv{\sigma}(\cm{k}{5})\\
		&\quad\quad\quad\quad - (\sigma(\cm{h}{1}) \ract \cm{k}{1}) \cdot (\inv{\sigma}(\cm{h}{2} \ract \cm{k}{2}) \cdot \sigma(\cm{k}{3}) \cdot \mu(\cm{k}{4}) \cdot \inv{\sigma}(\cm{k}{5})\\
		&\quad= (\sigma\act\mu)(h)\ract k - \epsilon(h)(\sigma\act\mu)(k)
	\end{align*}
	by repeated application of Lemma~\ref{centrallem}. Hence, \(\sigma \act \mu \in \ZH^1_\ell(H;M)\).
	
	Next, let us show that \(\BS^1_\ell(H;B,M)\) acts trivially by conjugation on \(\ZH^1_\ell(H;M)\) and that \(\BH^1_\ell(H;M)\) consists of \(\ZS^1_\ell(H;B,M)\)-invariant vectors. First, let \(\upsilon \in \CS^0_\ell(H;B,M)\) and let \(\mu \in \ZH^1_\ell(H;M)\). For all \(h \in H\), by Lemma~\ref{centrallem},
	\begin{multline*}
		D\upsilon \act \mu(h) = (\upsilon \ract \cm{h}{1})\upsilon^\ast \cdot \mu(\cm{h}{2}) \cdot (\upsilon^\ast \ract \cm{h}{3})\upsilon
		= (\upsilon \ract \cm{h}{1})\upsilon^\ast (\upsilon^\ast \ract \cm{h}{2}) \cdot \mu(\cm{h}{3})\cdot\upsilon \\
		= (\upsilon \ract \cm{h}{1})(\upsilon^\ast \ract \cm{h}{2}) \cdot \mu(\cm{h}{3}) \cdot \upsilon^\ast\upsilon = \epsilon(\cm{h}{1})\mu(\cm{h}{2}) = \mu(h).
	\end{multline*}
	so that, indeed, \(D\upsilon \act \mu = \mu\). Now, let \(m \in \CH^0_\ell(H;M)\) and \(\sigma \in \ZS^1_\ell(H;B,M)\). Then for all \(h \in H\), by Lemma~\ref{centrallem},
	\begin{multline*}
		(\sigma \act Dm)(h) = \sigma(\cm{h}{1}) \cdot (m \ract \cm{h}{2} - \epsilon(\cm{h}{2})m) \cdot \inv{\sigma}(\cm{h}{3})\\
		= (m \ract \cm{h}{1}) \cdot \sigma(\cm{h}{2})\inv{\sigma}(\cm{h}{3}) - \sigma(\cm{h}{1})\inv{\sigma}(\cm{h}{2}) \cdot m = m \ract h - \epsilon(h)m = Dm(h). 
	\end{multline*}
	
	From now on, let \(\partial : B \to M\) be an \(H\)-equivariant \(\ast\)-derivation. Let us now show that \(\MC[\partial]\) is a well-defined function. Let \(\sigma \in \ZS^1_\ell(H;B,M)\), and let \(\mu \coloneqq \MC[\partial](\sigma) \in \cC(H,M)\). First, \(\mu(1) = -\partial\sigma(1) \cdot \sigma^\ast(1) = 0\). Next, for all \(b \in B\) and \(h \in H\), by Lemma~\ref{centrallem},
	\begin{align*}
		&(b \ract \cm{h}{1}) \cdot \mu(\cm{h}{2})\\ &\quad= - (b \ract \cm{h}{1}) \cdot \partial\sigma(\cm{h}{2}) \cdot \sigma^\ast(\cm{h}{3})\\
		&\quad= -\partial\mleft((b\ract\cm{h}{1}) \cdot \sigma(\cm{h}{2})\mright) \cdot \sigma^\ast(\cm{h}{3}) + \partial(b \ract \cm{h}{1}) \cdot \sigma(\cm{h}{2})\sigma^\ast(\cm{h}{3})\\
		&\quad= -\partial\mleft(\sigma(\cm{h}{1}) \cdot (b\ract\cm{h}{2})\mright) \cdot \sigma^\ast(\cm{h}{3}) + \partial(b) \ract h\\
		&\quad= -\partial\sigma(\cm{h}{1}) \cdot (b \ract \cm{h}{2})\sigma^\ast(\cm{h}{3}) -\sigma(\cm{h}{1}) \cdot \left(\partial(b) \ract \cm{h}{2}\right) \cdot \sigma^\ast(\cm{h}{3}) + \partial(b) \ract h\\
		&\quad= - \partial\sigma(\cm{h}{1}) \cdot \sigma^\ast(\cm{h}{2}) (b \ract \cm{h}{3}) - \sigma(\cm{h}{1})\sigma^\ast(\cm{h}{2}) \cdot (\partial(b) \ract h) +\partial(b)\ract h\\
		&\quad= \mu(\cm{h}{1}) \cdot (b \ract \cm{h}{2}),
	\end{align*}
	so that \(\mu \in \Zent_B(\cC(H,M))\). Next, for all \(h \in H\), since \(\sigma(\cm{h}{1})\sigma^\ast(\cm{h}{2}) = \epsilon(h)1_B\),
	\[
		\mu^\ast(h) = \left(-\partial\sigma(\cm{S(h)}{1}^\ast) \cdot \sigma(S(\cm{S(h)}{2}^\ast)^\ast)^\ast \right)^\ast	
		= \sigma(\cm{h}{1}) \cdot \partial\sigma^\ast(\cm{h}{2}) = -\partial \sigma(\cm{h}{1}) \cdot \sigma^\ast(\cm{h}{2}) = \mu(h),
	\]
	so that \(\mu^\ast = \mu\). Finally, for all \(h, k \in H\), by Lemma~\ref{centrallem} and the fact that \(\sigma \in \Unit(\cC(H,B))\),
	\begin{align*}
		&\mu(hk)\\ &\quad= -\partial\sigma(\cm{h}{1}\cm{k}{1}) \cdot \sigma^\ast(\cm{h}{2}\cm{k}{2})\\
		&\quad= -\partial\mleft((\sigma(\cm{h}{1}) \ract \cm{k}{1}) \cdot \sigma(\cm{k}{2})\mright) \cdot (\sigma^\ast(\cm{h}{2}) \ract \cm{k}{3}) \sigma^\ast(\cm{k}{4})\\
		&\quad= -\partial\mleft((\sigma(\cm{h}{1}) \ract \cm{k}{1}) \cdot \sigma(\cm{k}{2})\mright) \cdot \sigma^\ast(\cm{k}{3}) (\sigma^\ast(\cm{h}{2}) \ract \cm{k}{4})\\
		&\quad= -(\partial\sigma(\cm{h}{1}) \ract \cm{k}{1}) \cdot \sigma(\cm{k}{2}) \sigma^\ast(\cm{k}{3}) (\sigma^\ast(\cm{h}{2}) \ract \cm{k}{4}) \\
		&\quad\quad\quad\quad - (\sigma(\cm{h}{1})\ract \cm{k}{1}) \cdot \partial\sigma(\cm{k}{2}) \cdot \sigma^\ast(\cm{k}{3}) (\sigma^\ast(\cm{h}{2}) \ract \cm{k}{4})\\
		&\quad= -(\partial\sigma(\cm{h}{1}) \ract \cm{k}{1}) \cdot \epsilon(\cm{k}{2})(\sigma^\ast(\cm{h}{2}) \ract \cm{k}{3}) - (\sigma(\cm{h}{1}) \ract \cm{k}{1}) \cdot\mu(\cm{k}{2}) \cdot (\sigma^\ast(\cm{h}{2}) \ract \cm{k}{3})\\
		&\quad= -(\partial\sigma(\cm{h}{1}) \ract \cm{k}{1}) (\sigma^\ast(\cm{h}{2}) \ract \cm{k}{2}) - (\sigma(\cm{h}{1}) \ract \cm{k}{1})(\sigma^\ast(\cm{h}{2}) \ract \cm{k}{2})\cdot \mu(\cm{k}{3})\\
		&\quad = \mu(h)\ract k + \epsilon(h)\mu(k).
	\end{align*}
	Hence, \(\MC[\partial](\sigma) \eqqcolon \mu \in \ZH^1_\ell(H;M)\).
	
	Next, let us show that \(\MC[\partial]\) is a \(1\)-cocycle for the conjugation action of \(\ZS^1_\ell(H;B,M)\) on \(\ZH^1_\ell(H;M)\). Let \(\sigma,\tau \in \ZS^1_\ell(H;B,M)\). Then, for all \(h \in H\),
	\begin{align*}
		\MC[\partial](\sigma \star \tau)	(h) &= -\partial\mleft(\sigma(\cm{h}{1})\tau(\cm{h}{2})\mright) \cdot \tau^\ast(\cm{h}{3})\sigma^\ast(\cm{h}{4})\\
		&= -\partial\sigma(\cm{h}{1}) \cdot \tau(\cm{h}{2})\tau^\ast(\cm{h}{3})\sigma^\ast(\cm{h}{4}) - \sigma(\cm{h}{1}) \cdot \partial\tau(\cm{h}{2}) \cdot \tau^\ast(\cm{h}{3})\sigma^\ast(\cm{h}{4})\\
		&= \MC[\partial](\sigma)(h) + \sigma \act \MC[\partial](\tau)(h),
	\end{align*}
	so that, indeed, \(\MC[\partial](\sigma \star \tau) = \MC[\partial](\sigma) + \sigma \act \MC[\partial](\tau)\).

	Finally, let us show that \(\MC[\partial]\) satisfies \eqref{coboundaryeq}. First, observe that \(\partial(\Cent_B(M)) \subset \Zent_B(M)\), since for all \(b \in \Cent_B(M)\) and \(c \in B\),
	\[
		\partial(b) \cdot c = \partial(bc) - b \cdot \partial(c) = \partial(cb) - b \cdot \partial(c) = \partial(c)\cdot b + c \cdot \partial(b) - b \cdot \partial(c) = c \cdot \partial(b).
	\]
	Now, let \(\upsilon \in \CS^0_\ell(H;B,M)\). Then, for all \(h \in H\), by the above observation,
	\begin{align*}
		\MC[D\upsilon](h) &= - \partial\mleft((\upsilon \ract \cm{h}{1}) \cdot \upsilon^\ast\mright) \cdot \left((\upsilon^\ast \ract \cm{h}{2}) \cdot \upsilon\right)\\
		&= -(\partial(\upsilon) \ract \cm{h}{1}) \cdot \upsilon^\ast (\upsilon^\ast \ract \cm{h}{2}) \cdot \upsilon - (\upsilon \ract \cm{h}{1}) \cdot \partial(\upsilon^\ast) \cdot (\upsilon^\ast \ract \cm{h}{2}) \cdot \upsilon\\
		&= -(\partial(\upsilon) \cdot \upsilon^\ast) \ract h - \epsilon(h)\upsilon\cdot\partial(\upsilon^\ast)\\
		&= -(\partial(\upsilon) \cdot \upsilon^\ast) \ract h + \epsilon(h) \partial(\upsilon) \cdot \upsilon^\ast\\
		&= D(-\partial(\upsilon) \cdot \upsilon^\ast)(h),
	\end{align*}
	as was claimed.
\end{proof}

\begin{definition}
	Let \(B\) be a right \(H\)-module \(\ast\)-algebra, let \(M\) be a right \(H\)-equivariant \(B\)-\(\ast\)-bimodule, and let \(\partial : B \to M\) be a right \(H\)-equivariant \(\ast\)-derivation. We define the \emph{Maurer--Cartan \(1\)-cocycle} of \(\partial\) to be the \(1\)-cocycle \(\MC[\partial] : \ZS^1(H;B,M) \to \ZH^1(H;M)\) of Theorem~\ref{lazymcthm} induced by \(\partial\). 
\end{definition}

When \(H\) is a group algebra, our constructions reduce to degree \(1\) group cohomology in a manner compatible, e.g., with the  results of \'{C}a\'{c}i\'{c}--Mesland~\cite{CaMe}*{Thm.\ 3.11}.

\begin{proposition}[cf.\ Sweedler~\cite{Sweedler}*{Thm.\ 3.1}]\label{lazygrp}
	Suppose that \(H = \bC[\Gamma]\) for \(\Gamma\) a group. Let \(B\) be a right \(\Gamma\)-\(\ast\)-algebra, and let \(M\) be a right \(\Gamma\)-equivariant \(B\)-\(\ast\)-bimodule, so that \(\Unit(\Cent_B(B \oplus M))\) defines a multiplicative \(\Gamma\)-module and \(Z_B(M)_{\sa}\) defines a \(\bR\)-linear representation of \(\Gamma\). 
	\begin{enumerate}
		\item The map
		\(
			r_B \coloneqq (f \mapsto \rest{f}{\Gamma}) : \ZS^1_\ell(H;B,M) \to \operatorname{Z}^1\mleft(\Gamma,\Unit(\Cent_B(B \oplus M))\mright)
		\)
		is a group isomorphism, such that \(r_B \circ D = \dt{\Gamma}\) and
		\(
			r_B\mleft(\BS^1_\ell(H;B,M)\mright) = \operatorname{B}^1\mleft(\Gamma,\Unit(\Cent_B(B \oplus M))\mright);
		\)
		hence, in particular, it descends to an isomorphism of Abelian groups
		\[
			\HS^1_\ell(H;M) \iso \operatorname{H}^1\mleft(\Gamma,\Unit(\Cent_B(B\oplus M))\mright).
		\]
		\item The map
		\(
			r_M \coloneqq (\mu \mapsto \rest{\mu}{\Gamma}) : \ZH^1_\ell(H;M) \to \operatorname{Z}^1(\Gamma,\Zent_B(M)_\sa)
		\)
		defines an isomorphism of \(\bR\)-vector spaces, such that \(r_M \circ D = \dt{\Gamma}\) and 
		\(
			 r_M\mleft(\BH^1_\ell(H;M)\mright) = \operatorname{B}^1(\Gamma,\Zent_B(M)_\sa);
		\)
		hence, in particular, it descends to an isomorphism of Abelian groups
		\[
			\HH^1_\ell(H;M) \iso \operatorname{H}^1(\Gamma,\Zent_B(M)_\sa).
		\]
		\item The conjugation action of \(\ZS^1_\ell(H;B,M)\) on \(\ZH^1_\ell(H;M)\) is trivial. Thus, for every \(\Gamma\)-equivariant \(\ast\)-derivation \(\partial : B \to M\), the induced map
		\[
			\MC_\Gamma \coloneqq r_M \circ \MC[\partial] \circ r_B^{-1} : \operatorname{Z}^1\mleft(\Gamma,\Unit(\Cent_B(B \oplus M))\mright)
 \to \operatorname{Z}^1(\Gamma,\Zent_B(M)_\sa)
		\]
		is a group homomorphism satisfying
		\begin{gather*}
			\forall \sigma \in 	\operatorname{Z}^1\mleft(\Gamma,\Unit(\Cent_B(B \oplus M))\mright), \,\forall \gamma \in \Gamma, \quad \MC_\Gamma(\sigma)(\gamma) = -\partial \sigma(\gamma) \cdot \sigma(\gamma)^\ast,\\
			\forall \upsilon \in \Unit(\Cent_B(B \oplus M)), \quad 
			\MC_\Gamma(\dt{\Gamma}\upsilon) = \dt{\Gamma}(-\partial(\upsilon) \cdot \upsilon^\ast),
		\end{gather*}
		so that, in particular, it descends to a group homomorphism \[\operatorname{H}^1\mleft(\Gamma,\Unit(\Cent_B(B \oplus M))\mright)
 \to \operatorname{H}^1(\Gamma,\Zent_B(M)_\sa).\]
	\end{enumerate}
\end{proposition}

\begin{proof}
Given sets \(X\) and \(Y\), let \(\cF(X,Y)\) be the set of all functions from \(X\) to \(Y\). It follows that \(\cF(\Gamma,B)\) is a unital \(\ast\)-algebra with respect to the pointwise operations, that any unital \(\ast\)-subalgebra \(A\) of \(B\) yields a unital \(\ast\)-subalgebra \(\cF(\Gamma,A)\) of \(\cF(\Gamma,B)\), that \(\cF(\Gamma,M)\) defines a \(\cF(\Gamma,B)\)-\(\ast\)-bimodule with respect to pointwise operations, and that any \(B\)-\(\ast\)-sub-bimodule \(L\) of \(M\) yields a \(\cF(\Gamma,B)\)-\(\ast\)-sub-bimodule of \(\cF(\Gamma,M)\). By abuse of notation, let
\[
	r_B \coloneqq (f \mapsto \rest{f}{\Gamma}) : \cC(H,B) \to \cF(\Gamma,B), \quad r_M \coloneqq (\mu \mapsto \rest{\mu}{\Gamma}) : \cC(H,M) \to \cF(\Gamma,M).
\]
Since \(\Gamma\) is a basis for \(H\), it follows that \(r_B\) and \(r_M\) are bijections, and since \(\Gamma\) consists of group-like elements and satisfies \(\rest{S \circ \ast}{\Gamma} = \id_\Gamma\), it follows that \(r_B\) is a \(\ast\)-isomorphism and that \(r_M\) is a \(\ast\)-preserving \(\bC\)-linear map satisfying
\[
	\forall \sigma,\tau \in \cC(H,B), \, \forall \mu \in \cC(H,M), \quad r_M(\sigma \star \mu \star \tau) = r_B(\sigma) \cdot r_M(\mu) \cdot r_B(\tau).
\]
The fact that \(\Gamma\) consists of group-like elements now implies that
\begin{gather*}
	r_B\mleft(\Unit(\Cent_{\cC(H,B)}(\rho_B(B) \oplus \rho_M(M)))\mright) = \cF(\Gamma,\Unit(\Cent_B(B\oplus M))),\\
	r_M\mleft(\Zent_B(\cC(H,M))_\sa\mright) = \cF(\Gamma,\Zent_B(M)_\sa),
\end{gather*}
from which our claims now follow by routine calculations.
\end{proof}

Similarly, when \(H\) is the universal enveloping algebra of a real Lie algebra, our constructions reduce to standard degree \(1\) Lie cohomology (with a small caveat related to lazy Sweedler cohomology) when the coefficient algebra satisfies the following condition.

\begin{definition}
	Let \(\fg\) be a real Lie algebra, and let \(A\) be a commutative right \(\fg\)-module \(\ast\)-algebra. We say that \(A\) \emph{has \(\fg\)-equivariant exponentials} if for every \(a \in A_\sa\) there exists a unitary \(\upsilon_a \in \Unit(A)\), such that
	\[
		\forall X \in \fg, \quad (\upsilon_a) \ract X = \iu{} (a \ract X) \cdot \upsilon_a,
	\]
	in which case we call \(\upsilon_a\) a \emph{\(\fg\)-equivariant exponential} of \(a\).
\end{definition}

\begin{proposition}[cf.\ Sweedler~\cite{Sweedler}*{Prop.\ 4.2}]
	Suppose that \(H = \cU(\fg)\) for \(\fg\) a real Lie algebra. Let \(B\) be a right \(\fg\)-module \(\ast\)-algebra, and let \(M\) be a right \(\fg\)-equivariant \(B\)-\(\ast\)-bimodule, so that \(\Cent_B(B \oplus M)_\sa\) and \(Z_B(M)_\sa\) both define right \(\fg\)-modules over \(\bR\).
	\begin{enumerate}
		\item The map
		\(
			r_B \coloneqq (f \mapsto \rest{-\iu{}f}{\fg}) : \ZS^1_\ell(H;B,M) \to \operatorname{Z}^1\mleft(\fg,\Cent_B(B \oplus M)_{\sa}\mright)
		\)
		defines an isomorphism of Abelian groups, such that 
		\[
			r_B\mleft(\BS^1_\ell(H;B,M)\mright) = \set{(X \mapsto -\iu{}(\upsilon \ract X)\upsilon^\ast) \given \upsilon \in \Unit(\Cent_B(B \oplus M))}.
		\]
		Moreover, if \(\Cent_B(B \oplus M)\) has \(\fg\)-equivariant exponentials, then for every \(b \in \Cent_B(B \oplus M)_\sa\) and every \(\fg\)-equivariant exponential \(\upsilon_b \in \Unit(\Cent_B(B \oplus M))\) of \(b\),
		\[
			r_B \circ D(\upsilon_b) = \dt{\fg}(b),
		\]
		so that \(B^1\mleft(\fg,\Cent_B(B \oplus M)_{\sa}\mright) \leq r_B\mleft(\BS^1_\ell(H;B,M)\mright)\), and 
		hence, \(r_B\) induces a short exact sequence of Abelian groups
		\[
			0 \to \frac{r_B\mleft(\BS^1_\ell(H;B,M)\mright)}{B^1\mleft(\fg,\Cent_B(B \oplus M)_{\sa}\mright)} \to \operatorname{H}^1\mleft(\fg,\Cent_B(B \oplus M)_{\sa}\mright) \to \HS^1_\ell(H;B,M) \to 0.
		\]
		\item The map
		\(
			r_M \coloneqq (\mu \mapsto \rest{\mu}{\fg}) : \ZH^1_\ell(H;M) \to \operatorname{Z}^1(\fg,\Zent_B(M)_\sa)
		\)
		is an isomorphism of \(\bR\)-vector spaces, such that \(r_M \circ D = \dt{\fg}\) and
		\(
			r_M\mleft(\BH^1_\ell(H;M)\mright) = \operatorname{B}^1(\fg,\Zent_B(M)_\sa)
		\);
		hence, in particular, \(r_M\) descends to an isomorphism of Abelian groups
		\[
			\HH^1_\ell(H;M) \iso \operatorname{H}^1(\fg,\Zent_B(M)_\sa).
		\]
		\item The conjugation action of \(\ZS^1_\ell(H;B,M)\) on \(\ZH^1_\ell(H;M)\) is trivial. Thus, for every \(\fg\)-equivariant \(\ast\)-derivation \(\partial : B \to M\), the induced map
		\[
			\MC_\fg \coloneqq r_M \circ \MC[\partial] \circ r_B^{-1} :  \operatorname{Z}^1\mleft(\fg,\Cent_B(B \oplus M)_{\sa}\mright) \to \operatorname{Z}^1(\fg,\Zent_B(M)_\sa)
		\]
		is a group homomorphism satisfying
		\begin{gather*}
			\forall c \in 	 \operatorname{Z}^1\mleft(\fg,\Cent_B(B \oplus M)_{\sa}\mright), \,\forall X \in \fg, \quad \MC_\fg(\sigma)(X) = -\iu{}\partial c(X),\\
			\forall b \in \Cent_B(B \oplus M)_\sa, \, \quad \MC_\fg(\dt{\fg}b) = \dt{\fg}(-\iu{}\partial b),
		\end{gather*}
		so that, in particular, it descends to a group homomorphism \[\operatorname{H}^1\mleft(\fg,\Cent_B(B \oplus M)_{\sa}\mright) \to \operatorname{H}^1(\fg,\Zent_B(M)_\sa).\]

	\end{enumerate}
	
\end{proposition}

\begin{proof}
	Recall that \(H = \cU(\fg)\) inherits a filtration from the complexified tensor algebra of \(\fg\).
	Let \(r_B \coloneqq (\sigma \mapsto \rest{-\iu{}\sigma}{\fg}) : \ZS^1_\ell(H;B,M) \to \Hom(\fg,B)\). First, given \(\sigma \in \ZS^1_\ell(H;B,M)\), for all \(k \in \bN \cup \set{0}\), \(h \in H_k\), and \(X \in \fg\), 
	\[
		\sigma(h \cdot X) = \left(\sigma(h) \ract \cm{X}{1}\right)\sigma(\cm{X}{2}) = \left(\sigma(h) \ract X\right)\sigma(1) + \left(\sigma(h) \ract 1\right)\sigma(X) = \sigma(h) \ract X + \sigma(h)\sigma(X),
	\]
	so that by induction on \(k \in \bN \cup \set{0}\), the map \(\sigma\) is uniquely determined on \(H = \bigcup_{k \geq 0} H_k\) by \(r_B(\sigma)\); hence, \(r_B\) is injective. Next, note that for all \(\sigma, \tau \in \ZS^1_\ell(H;B,M)\) and \(X \in \fg\),
	\[
		(\sigma \star \tau)(X) = \sigma(\cm{X}{1})\tau(\cm{X}{2}) = \sigma(X)\tau(1) + \sigma(1)\tau(X) = \sigma(X) + \tau(X),
	\]
	and that \(\rest{1_{\cC(H,B)}}{\fg} = \rest{-\iu{}\epsilon(\cdot)1_B}{\fg} = 0\); hence, \(r_B\) is an injective group homomorphism from \(\ZS^1_\ell(H;B,M)\) to the additive group \(\Hom(\fg,B)\). Next, for \(\sigma \in \ZS^1_\ell(H;B,M)\) and \(X \in \fg\),
	\[
		0 = \epsilon(X)1_B = \sigma(\cm{X}{1})\sigma^\ast(\cm{X}{2}) = \sigma(X)\sigma(S(1)^\ast)^\ast + \sigma(1)\sigma(S(X)^\ast)^\ast = \sigma(X) + \sigma(X)^\ast;
	\]
	hence, \(r_B\) maps into \(\Hom(\fg,B_\sa)\). Next, for \(\sigma \in \ZS^1_\ell(H;B,M)\), \(b \in B\), and \(X \in \fg\),
	\begin{align*}
		0 &= [\sigma,\rho_B(b)](X)\\
		&= \sigma(\cm{X}{1})(b \ract \cm{X}{2}) - (b \ract \cm{X}{1})\sigma(\cm{X}{2})\\
		&= \sigma(X)(b \ract 1) - (b \ract X)\sigma(1) +\sigma(1)(b \ract X) - (b \ract 1)\sigma(X)\\
		&= [\sigma(X),b],
	\end{align*}
	so that \(r_B\) maps into \(\Hom(\fg,\Zent(B)_\sa)\); in fact, \emph{mutatis mutandis}, this shows that \(r_B\) actually maps into \(\Hom(\fg,\Cent_B(B \oplus M)_\sa)\). Finally, for all \(\sigma \in \ZS^1_\ell(H;B,M)\) and \(X,Y \in \fg\),
	\begin{align*}
		0 &= \sigma(XY-YX) - \sigma([X,Y])\\
		&= \left(\sigma(X) \ract \cm{Y}{1}\right)\sigma(\cm{Y}{2}) -\left(\sigma(Y) \ract \cm{X}{1}\right)\sigma(\cm{X}{2}) - \sigma([X,Y])\\
		&= \left(\sigma(X)\ract Y\right)\sigma(1) + \left(\sigma(X)\ract 1\right)\sigma(Y) - \left(\sigma(Y) \ract X\right)\sigma(1) - \left(\sigma(Y) \ract 1\right)\sigma(X) - \sigma([X,Y])\\
		&= \sigma(X) \ract Y - \sigma(Y) \ract X - \sigma([X,Y])\\
		&= \iu{}\dt{\fg}(r_B(\sigma))
	\end{align*}
	so that \(r_B\) maps into \(\operatorname{Z}^1(\fg,\Cent_B(B \oplus M)_\sa\). Conversely, given \(c \in \operatorname{Z}^1(\fg,\Cent_B(B \oplus M)_\sa\), one can construct \(r_B^{-1}(c) \in \ZS^1_\ell(H;B,M)\) inductively by setting \(r_B^{-1}(c)(1) \coloneqq 1\) and setting
	\[
		\forall k \in \bN \cup \set{0}, \, \forall h \in H_k, \, \forall X \in \fg, \quad r_B^{-1}(c)(h \cdot X) \coloneqq r_B^{-1}(c)(h) \ract X+ \iu{} r_B^{-1}(c)(h) \cdot c(X);
	\]
	in particular, that \(r_B^{-1}(c)\) is well-defined on \(H = \cU(\fg)\) follows from the \(1\)-cocycle identity
	\[
		\forall X,Y \in \fg, \quad c(X) \ract Y - c(Y) \ract X - c([X,Y]) = 0
	\]
	together with the universal property of \(\cU(\fg)\). Thus, \(r_B : \ZS^1_\ell(H;B,M) \to \operatorname{Z}^1(\fg,\Cent_B(B\oplus M)_\sa)\) is a group isomorphism.
	
	Next, by definition of \(D : \CS^0_\ell(H;B,M) \to \ZS^1_\ell(H;B,M)\) and of \(r_B\), it follows that
	\[
		r_B(\BS^1_\ell(H;B,M)) = \set{(X \mapsto -\iu{}(\upsilon \ract X)\upsilon^\ast \given \upsilon \in \Unit(\Cent_B(B\oplus M)) } \leq \operatorname{Z}(\fg,\Cent_B(B\oplus M)_\sa).
	\]
	Suppose, now, that \(B\) has \(\fg\)-equivariant exponentials. Given \(b \in \Cent_B(B \oplus M)_\sa\), for every \(\fg\)-equivariant exponential \(\upsilon_b \in \Unit(\Cent_B(B\oplus M))\) of \(b\) and every \(X \in \fg\),
	\begin{align*}
		D\upsilon_b(X) = (\upsilon_b \ract X)\upsilon_b^\ast = \iu{}(b \ract X)\upsilon_b\upsilon_b^\ast = \iu{}(b\ract X) = \iu{}\dt{\fg}(b)(X),
	\end{align*}
	so that \(\dt{\fg}(b) = r_B \circ D(\upsilon_b) \in r_B(\ZS^1_\ell(H;B,M))\); hence, as additive groups,
	\[
		\operatorname{B}^1(\fg;\Cent_B(B \oplus M)_\sa) \leq r_B\mleft(\BS^1_\ell(H;B,M)\mright).
	\]

	Now, the above proof that \(r_B : \ZS^1_\ell(H;B,M) \to \operatorname{Z}^1(\fg,\Cent_B(B\oplus M)_\sa)\) is a group isomorphism implies, \emph{mutatis mutandis}, that 
	\[
		r_M \coloneqq (\mu \mapsto \rest{-\iu{}\mu}{\fg}) : \ZH^1_\ell(H;M) \to \operatorname{Z}^1(\fg,\Zent_B(M)_\sa)
	\] 
	is an isomorphism of \(\bR\)-vector spaces; indeed, for all \(m \in \Zent_B(M)_\sa\) and \(X \in \fg\),
	\[
		Dm(X) = m \ract X - \epsilon(X)m = m \ract X = \dt{\fg}(m)(X),
	\]
	so that \(r_M \circ D = \dt{\fg}\), and hence \(\operatorname{B}^1(\fg,\Zent_B(M)_\sa) = r_M\mleft(\BH^1_\ell(H;M)\mright)\). Moreover, \(\ZS^1_\ell(H;B,M)\) acts trivially on \(\ZH^1_\ell(H;M)\), since for all \(\sigma \in \ZS^1_\ell(H;B,M)\), \(\mu \in \ZH^1_\ell(H;M)\), and \(X \in \fg\),
	\begin{align*}
		\sigma \star \mu \star \inv{\sigma}(X) &= \sigma(\cm{X}{1})\mu(\cm{X}{2})\inv{\sigma}(\cm{X}{3})\\
		&= \sigma(X)\mu(1)\inv{\sigma}(1) + \sigma(1)\mu(X)\inv{\sigma}(1) + \sigma(1)\mu(1)\inv{\sigma}(X)\\
		&= \mu(X).
	\end{align*}
	
	Finally, let \(\partial : B \to M\) be a \(\fg\)-equivariant \(\ast\)-derivation. For all \(c \in \operatorname{Z}^1(\fg,\Zent_B(M)_\sa)\)  and \(X \in \fg\), we find that
	\begin{align*}
		\MC_\fg[\partial](c)(X) &= -\partial\mleft(r_B^{-1}(c)(\cm{X}{1})\mright) \cdot r_B^{-1}(c)^\ast(\cm{X}{2})\\
		&= -\partial\mleft(r_B^{-1}(c)(X)\mright) \cdot r_B^{-1}(c)^\ast(1) -\partial\mleft(r_B^{-1}(c)(1)\mright) \cdot r_B^{-1}(c)^\ast(X)\\
		&= -\iu{}\partial c(X).
	\end{align*}
	Hence, in particular, for all \(b \in B\) and \(X \in \fg\),
	\[
		\MC_\fg[\partial](\dt{\fg}(b))(X) = -\iu{}\partial(b \ract X) = \partial(-\iu{}b) \ract X = \dt{\fg}(-\iu{}\partial b). \qedhere
	\]
\end{proof}

We now refine our construction of lazy Hochschild cohomology in the manner that, as we shall see, will encode prolongability of (relative) gauge potentials.

\begin{propositiondefinition}
	Let \(B\) be a right \(H\)-module \(\ast\)-algebra, let \(M\) be a right \(H\)-equivariant \(B\)-\(\ast\)-module, and let \(\Omega\) be a right \(H\)-module graded \(\ast\)-algebra \(\Omega\) with \(\Omega^0 = B\) and \(\Omega^1 = M\) that is generated over \(B\) by \(M\). Then
	\begin{gather*}
		\ZH^1_\ell(H;M,\Omega) \coloneqq \ZH^1_\ell(H;M) \cap \Cent_{\cC(H;\Omega)}(\rho_\Omega(\Omega)),\\
		\BH^1_\ell(H;M,\Omega) \coloneqq D(\Zent_{B}(M)_\sa \cap \Zent(\Omega)) \subset \BH^1_\ell(H;M) \cap \Cent_{\cC(H;\Omega)}(\rho_\Omega(\Omega)),\\
		\HH^1_\ell(H;M,\Omega) \coloneqq \ZH^1_\ell(H;M,\Omega)/\BH^1_\ell(H;M,\Omega).
	\end{gather*}
	Moreover, \(\ZH^1_\ell(H;M,\Omega)\) is a \(\ZS^1_\ell(H;B,M)\)-invariant subspace of \(\ZH^1_\ell(H;M)\), and hence the conjugation action of \(\ZS^1_\ell(H;B,M)\) on \(\ZH^1_\ell(H;M,\Omega)\) descends to a \(\bR\)-linear action on the quotient space \(\HH^1_\ell(H;M,\Omega)\).
\end{propositiondefinition}

\begin{proof}
	The first non-trivial claim to check is the inclusion \(\BH^1_\ell(H;M,\Omega) \subset \Cent_{\cC(H;\Omega)}(\rho_\Omega(\Omega))\). Let \(m \in \Cent_\Omega(M)_\sa\). Then for all \(\p{m} \in M\) and \(h \in H\),
	\begin{align*}
		[Dm,\rho_M(\p{m})](h) &= (m \ract \cm{h}{1} - \epsilon(\cm{h}{1})m) \wedge (\p{m} \ract \cm{h}{2}) + (\p{m} \ract \cm{h}{1}) \wedge (m \ract \cm{h}{2} - \epsilon(\cm{h}{2})m)\\
		&= [m,\p{m}] \ract h - [m, \p{m}\ract h]= 0,
	\end{align*}
	so that, indeed, \(Dm \in \Cent_{\cC(H;\Omega)}(\rho_\Omega(\Omega))\).
	
	The other non-trivial claim to check is \(\ZS^1_\ell(H;B,M)\)-invariance of \(\ZH^1_\ell(H;M,\Omega)\). Let \(\sigma \in \ZS^1_\ell(H;B,M)\) and let \(\mu \in \ZH^1_\ell(H;M,\Omega)\). Then, for all \(m \in M\) and \(h \in H\),
	\begin{align*}
		[\sigma \act \mu,\rho_M(m)](h) &= \sigma(\cm{h}{1}) \cdot \mu(\cm{h}{2}) \cdot \inv{\sigma}(\cm{h}{3}) \wedge m \ract \cm{h}{4}\\
		&\quad\quad\quad - m \ract \cm{h}{1} \wedge \sigma(\cm{h}{2}) \cdot \mu(\cm{h}{3}) \cdot \inv{\sigma}(\cm{h}{4})\\
		&= \sigma(\cm{h}{1}) \cdot \mu(\cm{h}{2}) \wedge m \ract \cm{h}{3} \cdot \inv{\sigma}(\cm{h}{4})\\
		&\quad\quad\quad - \sigma(\cm{h}{1}) \cdot m \ract \cm{h}{2} \wedge \mu(\cm{h}{3}) \cdot \inv{\sigma}(\cm{h}{4})\\
		&= \sigma \star [\mu,\rho_M(m)] \star \inv{\sigma}(h)=0,
	\end{align*}
	so that, indeed, \(\sigma \act \mu \in \Cent_{\cC(H;\Omega)}(\rho_\Omega(\Omega))\).
\end{proof}

We conclude by observing that Maurer--Cartan \(1\)-cocycles are automatically compatible with the above refinement.

\begin{corollary}
	Let \(B\) be a right \(H\)-module \(\ast\)-algebra, and let \((\Omega_B,\dt{B})\) be an \(H\)-equivariant \textsc{sodc} over \(B\). Then
	\[
		\MC[\dt{B}]\mleft(\ZS^1_\ell(H;B,\Omega^1_B)\mright) \subset \ZH^1_\ell(H;\Omega^1_B,\Omega_B), \quad \MC[\dt{B}]\mleft(\BS^1_\ell(H;B,\Omega^1_B)\mright) \subset \BH^1_\ell(H;\Omega^1_B,\Omega_B),
	\]
	so that \(\MC[\dt{B}]\) is a \(1\)-cocycle for the restricted action of \(\ZS^1_\ell(H;B,\Omega^1_B)\) on \(\ZH^1_\ell(H;\Omega^1_B,\Omega_B)\), and hence descends to a \(1\)-cocycle \(\widetilde{\MC}[\dt{B},\Omega_B] : \HS^1_\ell(H;B,\Omega^1_B) \to \HH^1_\ell(H;\Omega^1_B,\Omega_B)\) for the action of \(\HS^1_\ell(H;B,\Omega^1_B)\) on \(\HH^1_\ell(H;\Omega^1_B,\Omega_B)\).
\end{corollary}

\begin{proof}
	First, let \(\sigma \in \ZS^1_\ell(H;B,\Omega^1_B)\). Then, for all \(b \in B\) and \(h \in H\),
	\begin{align*}
		&[\MC[\dt{B}](\sigma),\rho_M(\dt{B} b)](h)\\ &\quad= -\dt{B} \sigma(\cm{h}{1}) \cdot \inv{\sigma}(\cm{h}{2}) \wedge (\dt{B}(b) \ract \cm{h}{3}) - (\dt{B}(b) \ract \cm{h}{1}) \wedge \dt{B} \sigma(\cm{h}{2}) \cdot \inv{\sigma}(\cm{h}{3})\\
		&\quad = -\dt{B}\sigma(\cm{h}{1}) \wedge \dt{B}(b \ract \cm{h}{2}) \cdot \inv{\sigma}(\cm{h}{3})- \dt{B}(b \ract \cm{h}{1}) \wedge \dt{B} \sigma(\cm{h}{2}) \cdot \inv{\sigma}(\cm{h}{3})\\
		&\quad= - \dt{B}\mleft(\sigma(\cm{h}{1}) \cdot (\dt{B}(b) \ract \cm{h}{2})-(\dt{B}(b)\ract\cm{h}{1}) \cdot \sigma(\cm{h}{2})\mright) \cdot \inv{\sigma}(\cm{h}{3})\\
		&\quad=0,
	\end{align*}
	so that, indeed, \(\MC[\dt{B}](\sigma) \in \Cent_{\cC(H;\Omega_B)}(\rho_{\Omega_B}(\Omega_B))\).
	
	Now, let \(\upsilon \in \CS^0_\ell(H,B,\Omega^1_B) = \Unit(\Cent_B(B\oplus \Omega^1_B))\), so that \(\MC[\dt{B}](D\upsilon) = D(-\dt{B}(\upsilon)\upsilon^\ast)\) by Theorem~\ref{lazymcthm}. Then, for all \(b \in B\),
	\begin{align*}
		[-\dt{B}(\upsilon)\upsilon^\ast,\dt{B}(b)] &= -\dt{B}(\upsilon)\upsilon^\ast \wedge \dt{B}(b) - \dt{B}(b) \wedge 	\dt{B}(\upsilon)\upsilon^\ast \\
		&= -\left(\dt{B}(\upsilon) \wedge \dt{B}(b) + \dt{B}(b) \wedge \dt{B}(\upsilon) \right)\upsilon^\ast\\
		&= -\dt{B}\mleft(\upsilon \cdot \dt{B}(b) - \dt{B}(b)\cdot\upsilon \mright)\upsilon^\ast = 0,
	\end{align*}
	so that, indeed, \(-\dt{B}(\upsilon)\upsilon^\ast \in \Zent(\Omega)\).
\end{proof}

\subsection{Gauge transformations and (relative) gauge potentials}

In this section, we use lazy Sweedler cohomology, lazy Hochschild cohomology, and Maurer--Cartan \(1\)-cocycles to compute the gauge group \(\fr{G}\), the Atiyah space \(\fr{At}\), and the affine action of \(\fr{G}\) on \(\fr{At}\) in the case of a crossed product by the Hopf \(\ast\)-algebra \(H\). This yields a far-reaching generalisation of the computation by \'{C}a\'{c}i\'{c}--Mesland~\cite{CaMe}*{\S 3.4} of the noncommutative \(\bT^m\)-gauge theory in this sense of a crossed product by \(\bZ^m\).

Let \(B\) be a right \(H\)-module \(\ast\)-algebra with fixed right \(H\)-invariant \textsc{sodc} \((\Omega_B,\dt{B})\). Let \[P \coloneqq B \rtimes H,\] so that \(P = H \otimes_\bC B\) together with the multiplication, \(\ast\)-structure, and left \(H\)-coaction defined, respectively, by
\begin{gather*}
	\forall h,\p{h} \in H, \, \forall b,\p{b} \in B, \quad (h \otimes b)(\p{h} \otimes \p{b}) \coloneqq	h\cm{h}{1}^\prime \otimes (b \ract \cm{h}{2}^\prime)\p{b},\\
	\forall h \in H, \, \forall b \in B, \quad (h \otimes b)^\ast \coloneqq \cm{h}{1}^\ast \otimes b^\ast \ract \cm{h}{2}^\ast,\\
	\forall h \in H, \, \forall b \in B, \quad \delta_P(h \otimes b) \coloneqq \cm{h}{1} \otimes \left(\cm{h}{2} \otimes b\right);
\end{gather*}
it follows that \(P\) defines a principal left \(H\)-comodule \(\ast\)-algebra, such that \(\coinv{H}{P} = B\); in fact, \(P\) can be viewed as a globally trivial (\emph{cleft}) principal \(H\)-comodule algebra with trivialisation given by the injective left \(H\)-covariant \(\ast\)-homomorphism
\(
	(h \mapsto h \otimes 1_B) : H \inj P
\). 
There is a canonical second-order horizontal calculus on \(P\), which we shall use exclusively from now on; it can be straightforwardly constructed as follows.

\begin{proposition}
	Let \(\Omega_{P,\hor} \coloneqq \Omega_B \rtimes H\), so that \(\Omega_{P,\hor} = H \otimes_{\bC} \Omega_B\) with the grading 
	\[
		\forall k \in \set{0,1,2}, \quad \Omega^k_{P,\hor} = H \otimes_{\bC} \Omega_B^k
	\] 
	and the multiplication, \(\ast\)-structure, and left \(H\)-coaction defined, respectively, by
	\begin{gather*}
		\forall h,\p{h} \in H,\, \forall \beta,\p{\beta} \in \Omega_B, \quad (h \otimes \beta) \wedge (\p{h} \otimes \p{\beta}) \coloneqq h\cm{h}{1}^\prime \otimes (\beta \ract \cm{h}{2}^\prime)\p{\beta},\\
		\forall h \in H, \, \forall \beta \in \Omega_B, \quad (h \otimes \beta)^\ast \coloneqq \cm{h}{1}^\ast \otimes \beta^\ast \ract \cm{h}{2}^\ast,\\
		\forall h \in H, \, \forall \beta \in B, \quad \delta_{\Omega_{P,\hor}}(h \otimes \beta) \coloneqq \cm{h}{1} \otimes \left(\cm{h}{2} \otimes \beta\right).
	\end{gather*}
	Then \((\Omega_B,\dt{B};\Omega_{P,\hor},\beta\mapsto 1\otimes\beta)\) defines a second-order horizontal calculus on \(P \coloneqq B \rtimes H\).
\end{proposition}

In what follows, we  will find it notationally convenient to view \(\Omega_{P,\hor}\) as the graded left \(H\)-comodule \(\ast\)-algebra generated by the graded \(\ast\)-subal\-gebra \(\Omega_B\) of \(H\)-coinvariants together with the left \(H\)-subcomodule \(\ast\)-subalgebra \(H\) in degree \(0\), subject to the relation \(1_H = 1_{\Omega_B}\) and the braided supercommutation relation
\[
	\forall h \in H, \, \forall \beta \in \Omega_B, \quad \beta h = \cm{h}{1} (\beta \ract \cm{h}{2}).
\]

We begin by expressing the gauge group \(\fr{G}\) of \(P\) with respect to the canonical first-order horizontal calculus \((\Omega^1_B,\dt{B};\Omega^1_{P,\hor})\) in terms of lazy Sweedler cohomology; note that inner gauge transformations will correspond precisely to lazy Sweedler \(1\)-coboundaries.

\begin{proposition}[cf.\ Brzezi\'{n}ski~\cite{Br96}*{Thm.\ 5.4}, \'{C}a\'{c}i\'{c}--Mesland~\cite{CaMe}*{Thm.\ 3.36.(2)}]\label{trivialgaugetrans}
	The function \(\Op : \ZS^1_\ell(H;B,\Omega^1_B) \to \fr{G}\) given by
		\begin{equation}
			\forall \sigma \in \ZS^1_\ell(H;B,\Omega^1_B), \, \forall h \in H, \, \forall b \in B, \quad \Op(\sigma)(hb) \coloneqq \cm{h}{1}\sigma(\cm{h}{2})b
		\end{equation}
		is a group isomorphism. Moreover,
		\begin{equation}
			\forall \upsilon \in \CS^0_\ell(H;B,\Omega^1_B), \quad \Op(D\upsilon) = \Ad_\upsilon,
		\end{equation}
		so that \(\Op(\BS^1_\ell(H;B,\Omega^1_B)) = \Inn(\fr{G})\) is the central subgroup of inner gauge transformations on \(P\) with respect to \((\Omega^1_B,\dt{B};\Omega^1_{P,\hor})\), and hence \(\Op\) descends to a group isomorphism
		\[
			\widetilde{\Op}: \HS^1_\ell(H;B,\Omega^1_B) \iso \fr{G}/\Inn(\fr{G}) \eqqcolon \Out(\fr{G}).
		\]
\end{proposition}

\begin{proof}
	Let \(\Aut_B(P)\) denote the group of all automorphisms \(f : P \to P\) of \(P\) as a left \(H\)-comodule right \(B\)-module, such that \(\rest{f}{B} =\id_B\), and let \(\cA(P)\) denote the group of all invertible elements \(\sigma \in \cC(H,B)\), such that \(\sigma(1)=1\); observe that
	\[
		\Inn(\fr{G}) \leq \fr{G} \leq \Aut_B(P), \quad \BS^1_\ell(H;B,\Omega^1_B) \leq \ZS^1_\ell(H;B,\Omega^1_B) \leq \cA(P).
	\]
	By~\cite{Br96}*{Thm.\ 5.4}, \emph{mutatis mutandis}, the map \(\Op : \cA(P) \to \Aut_B(P)\) defined by
	\[
		\forall \sigma \in \cA(P), \, \forall h \in H, \, \forall b \in B, \quad \Op(\sigma)(hb) \coloneqq \cm{h}{1}\sigma(\cm{h}{2})b
	\]
	is a well-defined group isomorphism with inverse given by
	\[
		\forall f \in \Aut_B(P), \, \forall h \in H, \quad \inv{\Op}(f)(h) \coloneqq S(\cm{h}{1})f(\cm{h}{2}).
	\]
	Note, moreover, that for all \(\sigma \in \cA(P)\), the inverse \(\inv{\sigma}\) is given by
	\[
		\forall h \in H, \quad \inv{\sigma}(h) \coloneqq \sigma(S(\cm{h}{1})) \ract \cm{h}{2},
	\]
	since, for all \(h \in H\),
	\[
		\epsilon(h)1_B = \sigma(S(\cm{h}{1})\cm{h}{2}) = \left(\sigma(S(\cm{h}{1}) \ract \cm{h}{2}\right)\cdot\sigma(\cm{h}{3}).
	\]
	
	Let us first show that \(\Op^{-1}(\fr{G}) = \ZS^1_\ell(H;B,\Omega^1_B)\). Let \(\sigma \in \cA(P)\) and set \(f \coloneqq \Op(\sigma)\). First, for all \(h,k \in H\),
	\begin{align*}
		f(hk)-f(h)f(k) &= \cm{h}{1}\cm{k}{1}\sigma(\cm{h}{2}\cm{k}{2}) - \cm{h}{1}\sigma(\cm{h}{2})\cm{k}{1}\sigma(\cm{k}{2})\\ 
		&= \cm{h}{1}\cm{k}{1}\left(\sigma(\cm{h}{2}\cm{k}{2})\epsilon(\cm{k}{3}) - (\sigma(\cm{h}{2}) \ract \cm{k}{2})\sigma(\cm{k}{3})\right),
	\end{align*}
	while for all \(h \in H\) and \(b \in B\), since \(bh = \cm{h}{1}(b \ract \cm{h}{2})\),
	\begin{align*}
		f(bh) - f(b)f(h) &=	\cm{h}{1}\sigma(\cm{h}{2})(b \ract \cm{h}{3}) - b \cm{h}{1}\sigma(\cm{h}{2})\\ &= \cm{h}{1}\left(\sigma(\cm{h}{2})(b \ract \cm{h}{3}) - (b \ract \cm{h}{2})\sigma(\cm{h}{3})\right),
	\end{align*}
	so that \(f \in \Aut_B(P)\) is an algebra automorphism if and only if \(\sigma \in \Cent_{\cC(H,B)}(\rho_B(B))\) and
	\[
		\forall h,k \in H, \quad \sigma(hk) = (\sigma(h) \ract \cm{k}{1}) \sigma(\cm{k}{2}).
	\]
	In particular, if \(f\) is an algebra automorphism, then \(\inv{\sigma}\) is given by
	\[
		\forall h \in H, \quad \inv{\sigma}(h) \coloneqq \sigma(S(\cm{h}{1})) \ract \cm{h}{2},
	\]
	since, for all \(h \in H\),
	\[
		\epsilon(h)1_B = \sigma(S(\cm{h}{1})\cm{h}{2}) = \left(\sigma(S(\cm{h}{1}) \ract \cm{h}{2}\right)\cdot\sigma(\cm{h}{3}).
	\]

	Now, suppose that \(f\) is an algebra automorphism. Then, for all \(h \in H\),
	\begin{align*}
		f(S(h)^\ast)^\ast - f(S(h)) 
		&= \left(S(\cm{h}{2})^\ast \sigma(S(\cm{h}{1})^\ast)\right)^\ast - S(\cm{h}{2}) \sigma(S(\cm{h}{1}))\\
		&= \sigma(S(\cm{h}{1})^\ast)^\ast S(\cm{h}{2})^\ast - \left(\sigma(S(\cm{h}{1}))^\ast S(\cm{h}{2})^\ast\right)^\ast\\
		&=	\sigma(S(\cm{h}{1})^\ast)^\ast S(\cm{h}{2})^\ast - \left(S(\cm{h}{3})^\ast (\sigma(S(\cm{h}{1}))^\ast \ract S(\cm{h}{2})^\ast)\right)^\ast\\
		&= \left(\sigma(S(\cm{h}{1})^\ast)^\ast \epsilon(\cm{h}{2}) - \sigma(S(\cm{h}{1})) \ract \cm{h}{2}\right)S(\cm{h}{2})^\ast\\
		&= \left(\sigma^\ast(\cm{h}{1}) - \inv{\sigma}(\cm{h}{1})\right) S(\cm{h}{2})^\ast,
	\end{align*}
	so that \(f\) is a \(\ast\)-automorphism if and only if \(\sigma\) is unitary. Finally, suppose that \(f\) is a \(\ast\)-automorphism. Then, for all \(h \in H\) and \(\beta \in B\), since \(\beta h = \cm{h}{1} \cdot (\beta \ract \cm{h}{2})\),
	\begin{align*}
		f(\cm{h}{1}) \cdot (\beta \ract \cm{h}{2}) - \beta \cdot f(h) &= \cm{h}{1}\sigma(\cm{h}{2}) \cdot (\beta \ract \cm{h}{2}) - \beta \cdot \cm{h}{1}\sigma(\cm{h}{2})\\
		&= \cm{h}{1} \cdot \left(\sigma(\cm{h}{2}) \cdot (\beta \ract \cm{h}{2}) - (\beta \ract \cm{h}{2}) \cdot \sigma(\cm{h}{3})\right),
	\end{align*}
	so that \(f \in \fr{G}\) if and only if \(\sigma \in \Cent_{\cC(H,B)}(\rho_{\Omega^1_B}(\Omega^1_B))\).

	Let us now show that \(\Op \circ D = (\upsilon \mapsto \Ad_\upsilon)\) on \(\CS^0_\ell(H;B,\Omega^1_B) = \Unit(\Cent_B(B\oplus\Omega^1_B))\), which will imply the rest of the claim. Let \(\upsilon \in \Unit(\Cent_B(B\oplus\Omega^1_B))\). Then, for all \(h \in H\) and \(b \in B\),
	\[
		\Op(D\upsilon)(hb) = \cm{h}{1}(\upsilon \ract \cm{h}{2})\upsilon^\ast b = \upsilon hb \upsilon^\ast = \Ad_\upsilon(hb).	\qedhere
	\]
\end{proof}

Next, we express the Atiyah space \(\fr{At}\) of \(P\) with respect to the canonical first-order horizontal calculus \((\Omega^1_B,\dt{B};\Omega^1_{P,\hor})\) and its space of translations \(\fr{at}\) in terms of lazy Hochschild cohomology; note that inner relative gauge potentials will correspond precisely to lazy Hoch\-schild \(1\)-coboundaries.

\begin{proposition}[cf.\ \'{C}a\'{c}i\'{c}--Mesland~\cite{CaMe}*{Thm.\ 3.36.(1)}]\label{trivialgaugepot}
	We have an isomorphism of \(\bR\)-affine spaces \(\Op : \ZH^1_\ell(H;\Omega^1_B) \to \fr{At}\) given by
		\begin{equation}
			\forall \mu \in \ZH^1_\ell(H;\Omega^1_B)\, \forall h \in H, \, \forall b \in B, \quad \Op(\mu)(hb) \coloneqq h \cdot \dt{B}(b) + \cm{h}{1} \cdot \mu(\cm{h}{2}) \cdot b,
		\end{equation}
		whose linear part \(\Op_0 : \ZH^1_\ell(H;\Omega^1_B) \iso \fr{at}\) is given by
		\begin{equation}
			\forall \mu \in \ZH^1_\ell(H;\Omega^1_B),\, \forall h \in H, \, \forall b \in B, \quad \Op_0(\mu)(hb) \coloneqq \cm{h}{1} \cdot \mu(\cm{h}{2}) \cdot b.
		\end{equation}
		Moreover,
		\begin{equation}
			\forall \alpha \in \CH^0_\ell(H;\Omega^1_B), \quad 	\Op_0(D\alpha) = \ad_\alpha,
		\end{equation}
		so that \(\Op_0(\BH^1_\ell(H;\Omega^1_B) = \Inn(\fr{at})\) is the subspace of inner relative gauge potentials on \(P\) with respect to \((\Omega^1_B,\dt{B};\Omega^1_{P,\hor})\), and hence \(\Op\) descends to an isomorphism of \(\bR\)-affine spaces
		\[
			\widetilde{\Op} : \HH^1_\ell(H;\Omega^1_B) \iso \fr{At}/\Inn(\fr{at}) \eqqcolon \Out(\fr{At})
		\]
		with linear part \(\widetilde{\Op_0} : \HH^1_\ell(H;\Omega^1_B) \iso \fr{at}/\Inn(\fr{at}) \eqqcolon \Out(\fr{at})\) descending from \(\Op_0\). 
\end{proposition}

\begin{proof}
	Let \(\Hom_B(P,\Omega^1_{P,\hor})_0\) denote the \(\bR\)-vector space of all morphisms \(f : P \to \Omega^1_{P,\hor}\) of left \(H\)-comodule right \(B\)-modules, such that \(\rest{f}{B} = 0\), and let \(\cC(H,\Omega^1_B)_0\) denote the \(\bR\)-vector space of all \(\mu \in \cC(H,\Omega^1_B)\), such that \(\mu(1) = 0\); observe that
	\[
		\Inn(\fr{at}) \leq \fr{At} \leq \Hom_B(P,\Omega^1_{P,\hor})_0, \quad \BH^1_\ell(H;\Omega^1_B) \leq \ZH^1_\ell(H;\Omega^1_B) \leq \cC(H,\Omega^1_B)_0.
	\]
	By~\cite{Br96}*{Thm.\ 5.4}, \emph{mutatis mutandis}, the map \(\Op_0 : \cC(H,\Omega^1_B)_0 \to \Hom_B(P,\Omega^1_{P,\hor})_0\) defined by
	\[
		\forall \mu \in \cC(H,\Omega^1_B)_0, \, \forall h \in H, \,\forall b \in B, \quad \Op_0(\mu)(hb) \coloneqq \cm{h}{1}\mu(\cm{h}{2}) \cdot b
	\]
	is a well-defined \(\bR\)-linear isomorphism with inverse given by
	\[
		\forall \bA \in \Hom_B(P,\Omega^1_{P,\hor})_0,\, \forall h \in H, \quad \Op_0^{-1}(\bA)(h) \coloneqq S(\cm{h}{1})\bA(\cm{h}{2}).
	\]
	
	Let us first show that \(\Op_0^{-1}(\fr{at}) = \ZH^1_\ell(H;\Omega^1_B)\) and that \(\Op_0 \circ D = (\alpha \mapsto \ad_\alpha)\) on the domain \(\CH^0_\ell(H;\Omega^1_B) = \Zent_B(\Omega^1_B)_\sa\). Let \(\mu \in \cC(H,\Omega^1_B)_0\) and set \(\bA \coloneqq \Op(\mu)\). By the proof of Proposition~\ref{trivialgaugetrans}, \emph{mutatis mutandis}, it follows that \(\bA \in \Hom_B(P,\Omega^1_{P,\hor})_0\) is a \(P\)-bimodule derivation if and only if \(\mu \in \Zent_B(\cC(H,\Omega^1_B))\) and
	\[
		\forall h, k \in H, \quad \mu(hk) = \mu(h) \ract k + \epsilon(h)\mu(k),
	\]
	in which case, for all \(h \in H\),
	\[
		0 = \mu(S(\cm{h}{1})\cm{h}{2}) = \mu(S(\cm{h}{1})) \ract \cm{h}{2} + \epsilon(S(\cm{h}{1}))\mu(\cm{h}{2}) = \mu(S(\cm{h}{1})) \ract \cm{h}{2} + \mu(h).
	\]
	Thus, if \(\bA\) is a \(P\)-bimodule derivation, then, by the proof of Proposition~\ref{trivialgaugetrans}, 
	\[
		\forall h \in H, \quad \bA(S(h)^\ast)^\ast + \bA(S(h)) = \left(\mu^\ast(\cm{h}{1})-\mu(\cm{h}{1})\right)S(\cm{h}{2})^\ast,
	\]
	so that the \(\bA\) is a \(\ast\)-derivation if and only if \(\mu = \mu^\ast\). The rest now follows from the proof of Proposition~\ref{trivialgaugetrans}, \emph{mutatis mutandis}.
	
	Let us now show that \(\Op : \ZH^1_\ell(H;\Omega^1_B) \to \fr{At}\) is a well-defined isomorphism of \(\bR\)-affine spaces with linear part \(\Op_0\); it suffices to show that the map \(\nabla_0 : P \to \Omega^1_{P,\hor}\) defined by
	\[
		\forall h \in H, \, \forall b \in B, \quad \nabla_0(hb) \coloneqq h \cdot \dt{B}(b)
	\]
	is a gauge potential, so that \(\Op(0) = \nabla_0\) is well-defined. First, for all \(h,k \in H\) and \(b,c \in B\), since \(hbkc = h\cm{k}{1}(b\ract\cm{h}{2})c\),
	\begin{align*}
		\nabla_0(hbkc) - \nabla_0(hb)kc - hb\nabla_0(kc) &= h\cm{k}{1}\dt{B}((b \ract \cm{k}{2})c) - h\dt{B}(b)kc - hbk\dt{B}(c)\\
		&= h\cm{k}{1}\cdot\left(\dt{B}((b \ract \cm{k}{2})c) - \dt{B}(b \ract \cm{k}{1}) \cdot (b \ract \cm{k}{1})\cdot\dt{B}(c)\right)\\
		&=0.
	\end{align*}
	Next, for all \(h \in H\) and \(b \in B\), since \((hb)^\ast = \cm{h}{1}^\ast (b^\ast \ract \cm{h}{2}^\ast)\) and \((h\dt{B}(b))^\ast = \cm{h}{1}^\ast \cdot \dt{B}(b)^\ast \ract \cm{h}{2}^\ast\),
	\[
		\nabla_0((hb)^\ast) + \nabla_0(hb)^\ast = \cm{h}{1}^\ast \dt{B}(b^\ast \ract \cm{h}{2}^\ast) +\left(h\dt{B}(b)\right)^\ast 
		= \cm{h}{1}^\ast \cdot \left(\dt{B}(b^\ast) +\dt{B}(b)^\ast\right) \ract \cm{h}{2}^\ast = 0.
	\]
	Finally, by construction, \(\rest{\nabla_0}{B} = \dt{B}\).
\end{proof}

We now relate the affine action of the gauge group \(\fr{G}\) on the Atiyah space \(\fr{At}\) to the affine action of lazy Sweedler cohomology with coefficients in \((B,\Omega^1_B)\) on lazy Hochschild cohomology with coefficients in \(\Omega^1_B\) induced by the Maurer--Cartan \(1\)-cocycle \(\MC[\dt{B}]\) of the derivation \(\dt{B} : B \to \Omega^1_B\).

\begin{proposition}[cf.\ Brzezi\'{n}ski--Majid~\cite{BrM}*{Prop.\ 4.8}, \'{C}a\'{c}i\'{c}--Mesland~\cite{CaMe}*{Thm.\ 3.36.(3)}]\label{trivialgaugepotprop}
	For every \(\sigma \in \ZS^1_\ell(H;B,\Omega^1_B)\) and every \(\mu \in \ZH^1_\ell(H;\Omega^1_B)\), we have
	\begin{gather}
		\Op(\sigma) \act \Op(\mu) = \Op\mleft(\sigma \act \mu + \MC[\dt{B}](\sigma)\mright),\\
		\Op(\sigma) \act \Op_0(\mu) = \Op_0(\sigma \act \mu).
	\end{gather}
	and hence, at the level of cohomology,
	\begin{gather}
		\widetilde{\Op}([\sigma]) \act \widetilde{\Op}([\mu]) = \widetilde{\Op}\mleft([\sigma] \act [\mu] + \widetilde{\MC}[\dt{B}]([\sigma])\mright),\\
		\widetilde{\Op}([\sigma]) \act \widetilde{\Op_0}([\mu]) = \widetilde{\Op_0}([\sigma] \act [\mu]).
	\end{gather}
	Thus, the maps \(\Op \times \Op\) and \(\widetilde{\Op} \times \widetilde{\Op}\) define groupoid isomorphisms
	\begin{gather*}
		\ZS^1_\ell(H;B,\Omega^1_B) \ltimes \ZH^1_\ell(H;\Omega^1_B) \to \fr{G} \ltimes \fr{At}, \\ \HS^1_\ell(H;B,\Omega^1_B) \ltimes \HH^1_\ell(H;\Omega^1_B) \to \Out(\fr{G}) \ltimes \Out(\fr{At}),
	\end{gather*}
	respectively, where \(\ZS^1_\ell(H;B,\Omega^1_B)\) acts affine-linearly on \(\ZH^1_\ell(H;\Omega^1_B)\) with \(1\)-cocycle \(\MC[\dt{B}]\) and \(\HS^1_\ell(H;B,\Omega^1_B)\) acts affine-linearly on \(\HH^1_\ell(H;\Omega^1_B)\)  with \(1\)-cocycle \(\widetilde{\MC}[\dt{B}]\).
\end{proposition}

\begin{proof}
	Let \(\sigma \in \ZS^1_\ell(H;B,\Omega^1_B)\) and \(\mu \in \ZH^1_\ell(H;\Omega^1_B)\); note that
	\[
		\Op(\sigma) \act \Op(\mu) = \Op(\sigma) \act \left(\Op(0) + \Op_0(\mu)\right) = \Op(\sigma) \act \Op(0) + \Op(\sigma) \act \Op_0(\mu).
	\]
	On the one hand, for all \(h \in H\) and \(b \in B\), since \(\sigma \star \sigma^\ast = \epsilon(\cdot)1_B\),
	\begin{align*}
		\Op(\sigma) \act \Op(0)(hb) &= \Op(\sigma)_\ast \circ \Op(0) \circ \Op(\sigma^\ast)(hb)\\
		&= \cm{h}{1}\sigma(\cm{h}{2}) \cdot \left(\dt{B}\sigma^\ast(\cm{h}{3}) \cdot b + \sigma^\ast(\cm{h}{3})\dt{B}(b)\right)\\
		&= -\cm{h}{1} \cdot \dt{B}\sigma(\cm{h}{2}) \cdot \sigma^\ast(\cm{h}{3})b + \cm{h}{1}\sigma(\cm{h}{2})\sigma^\ast(\cm{h}{3}) \cdot \dt{B}(b)\\
		&=  \cm{h}{1} \cdot \MC[\dt{B}](\sigma)(\cm{h}{2}) \cdot b + h \cdot \dt{B}(b)\\
		&= \Op(\MC[\dt{B}](\sigma))(hb),
	\end{align*}
	so that \(\Op(\sigma) \act \Op(0) = \Op(\MC[\dt{B}](\sigma))\). On the other hand, for all \(h \in H\) and \(b \in B\),
	\begin{align*}
		\Op(\sigma) \act \Op_0(\mu)(hb) &= \Op(\sigma)_\ast \circ \Op_0(\mu) \circ \Op(\sigma^\ast)(hb)\\
		&= \cm{h}{1}\sigma(\cm{h}{2}) \cdot \mu(\cm{h}{3}) \cdot \sigma^\ast(\cm{h}{4})b\\
		&= \cm{h}{1} \cdot (\sigma \act \mu)(\cm{h}{2}) \cdot b\\
		&= \Op_0(\sigma \act \mu)(hb),	
	\end{align*}
	so that \(\Op(\sigma) \act \Op_0(\mu) = \Op_0(\sigma \act \mu)\).
\end{proof}

All of the above results yield straightforward characterisations of prolongable gauge transformations and prolongable (relative) gauge potentials in terms of lazy Sweedler cohomology nd lazy Hochschild cohomology, respectively; in particular, we shall find that every gauge transformation is automatically prolongable.

\begin{corollary}\label{trivialprcor}
	Let \(\pr{\fr{G}}\) be the prolongable gauge group of \(P\) with respect to the canonical second-order horizontal calculus \((\Omega_B,\dt{B};\Omega_{P,\hor})\). Let \(\pr{\fr{At}}\) be the prolongable Atiyah space of \(P\) with respect to \((\Omega_B,\dt{B};\Omega_{P,\hor})\), let \(\pr{\fr{at}}\) be its translation space, and let \(\Inn(\pr{\fr{at}}) \subset \pr{\fr{at}}\) be the subspace of all inner prolongable gauge potentials. Then \(\pr{\fr{G}} = \fr{G}\) and
	\begin{gather*}
		\inv{\Op}\mleft(\pr{\fr{At}}\mright) = \Op_0^{-1}\mleft(\pr{\fr{at}}\mright) = \ZH^1_\ell(H;\Omega^1_B,\Omega_B),\\
		\Op_0^{-1}\mleft(\Inn(\pr{\fr{at}})\mright) = \BH^1_\ell(H;\Omega^1_B,\Omega_B),
	\end{gather*}
	hence \(\Op\) induces an isomorphism \(\pr{\widetilde{\Op}} : \HH^1_\ell(H;\Omega^1_B,\Omega_B) \iso \pr{\fr{At}}/\Inn(\pr{\fr{at}}) \eqqcolon \Out(\pr{\fr{At}})\) of \(\bR\)-affine spaces with linear part \(\pr{\widetilde{\Op_0}} : \HH^1_\ell(H;\Omega^1_B,\Omega_B) \iso \pr{\fr{at}}/\Inn(\pr{\fr{at}}) \eqqcolon \Out(\pr{\fr{at}})\) induced by \(\Op_0\). Moreover, for all \(\sigma \in \ZS^1_\ell(H;B,\Omega^1_B,\Omega_B)\) and \(\mu \in \ZH^1_\ell(H;\Omega^1_B,\Omega_B)\),
	\begin{gather}
		\widetilde{\Op}([\sigma]) \act \pr{\widetilde{\Op}}([\mu]) = \pr{\widetilde{\Op}}\mleft([\sigma] \act [\mu] + \widetilde{\MC}[\dt{B},\Omega_B]([\sigma])\mright),\\
		\widetilde{\Op}([\sigma]) \act \pr{\widetilde{\Op_0}}([\mu]) = \pr{\widetilde{\Op_0}}([\sigma] \act [\mu]).
	\end{gather}
	Thus, the maps \(\Op \times \rest{\Op}{\ZH^1_\ell(H;\Omega^1_B,\Omega_B)}\) and \(\widetilde{\Op} \times \pr{\widetilde{\Op}}\) define groupoid isomorphisms
	\begin{gather*}
		\ZS^1_\ell(H;B,\Omega^1_B) \ltimes \ZH^1_\ell(H;\Omega^1_B,\Omega_B) \to \pr{\fr{G}} \ltimes \pr{\fr{At}}, \\ \HS^1_\ell(H;B,\Omega^1_B) \ltimes \HH^1_\ell(H;\Omega^1_B,\Omega_B) \to \Out(\pr{\fr{G}}) \ltimes \Out(\pr{\fr{At}}),
	\end{gather*}
	respectively, where \(\ZS^1_\ell(H;B,\Omega^1_B)\) acts affinely on \(\ZH^1_\ell(H;\Omega^1_B,\Omega_B)\) with \(1\)-cocycle \(\MC[\dt{B}]\) and \(\HS^1_\ell(H;B,\Omega^1_B)\) acts affinely on \(\HH^1_\ell(H;\Omega^1_B,\Omega_B)\) with \(1\)-cocycle \(\widetilde{\MC}[\dt{B}]\).
\end{corollary}

\begin{proof}
	Before continuing, note that for all \(h,k \in H\) and \(b,\p{b},c \in B\),
	\[
		hb \cdot \dt{B}(\p{b}) \cdot kc = h\cm{k}{1}(b \ract \cm{k}{2})\cdot \dt{B}(\p{b} \ract \cm{k}{3}) \cdot c.
	\]

	Let us first show that \(\fr{G} = \pr{\fr{G}}\). Let \(\sigma \in \ZS^1_\ell(H;B,\Omega^1_B)\) be given. Then, for all \(h,k \in H\), \(b,c \in B\), and \(\alpha,\beta \in \Omega^1_B\),
	\begin{align*}
		&\Op(\sigma)(hb) \cdot \alpha \wedge \beta \cdot \Op(\sigma)(kc)	\\
		&\quad= \cm{h}{1} \sigma(\cm{h}{2}) b \cdot\alpha \wedge \beta \cdot \cm{k}{1}\sigma(\cm{k}{2})c\\
		&\quad= \cm{h}{1}\cm{k}{1} (\sigma(\cm{h}{2}) \ract \cm{k}{2})(b \ract \cm{k}{3}) \cdot (\alpha \ract \cm{k}{4}) \wedge (\beta \ract \cm{k}{5}) \cdot \sigma(\cm{k}{6})c\\
		&\quad= \cm{h}{1}\cm{k}{1} (\sigma(\cm{h}{2}) \ract \cm{k}{2})\sigma(\cm{k}{3})(b \ract \cm{k}{4}) \cdot (\alpha \ract \cm{k}{5}) \wedge (\beta \ract \cm{k}{6}) \cdot c\\
		&\quad= \cm{h}{1}\cm{k}{1} \sigma(\cm{h}{2}\cm{k}{2})(b \ract \cm{k}{3}) \cdot (\alpha \ract \cm{k}{4}) \wedge (\beta \ract \cm{k}{5}) \cdot c\\
		&\quad= \Op(\sigma)(h\cm{k}{1}) \cdot \left((b \ract \cm{k}{2})\cdot (\alpha \ract \cm{k}{3}) \wedge (\beta \ract \cm{k}{4}) \cdot c\right),
	\end{align*}
	where \(hb \cdot \alpha \wedge \beta \cdot kc = h\cm{k}{1} \cdot \left((b \ract \cm{k}{2})\cdot (\alpha \ract \cm{k}{3}) \wedge (\beta \ract \cm{k}{4}) \cdot c\right)\), so that \(\Op(\sigma) \in \pr{\fr{G}}\) with
	\[
		\forall h \in H, \, \forall \eta \in \Omega^2_B, \quad \Op(\sigma)_\ast(h\eta) = \cm{h}{1}\sigma(\cm{h}{2}) \cdot \eta.
	\]
	
	Next, observe that \(\Op(0) \in \pr{\fr{At}}\). Indeed, for all \(h,k \in H\) and \(b,\p{b},c \in B\),
	\begin{align*}
			&\Op(0)(hb) \wedge \dt{B}(\p{b}) \cdot kc - hb \cdot \dt{B}(\p{b}) \wedge \Op(0)(kc)\\
			&= h \cdot \dt{B}(b) \wedge \dt{B}(\p{b}) \cdot kc - hb \cdot \dt{B}(\p{b}) \wedge k \cdot \dt{B}(c)\\
			&= h \cm{k}{1} \cdot \dt{B}\mleft((b \ract \cm{k}{2}) \cdot \dt{B}(\p{b} \ract \cm{k}{3}) \cdot c  \mright),
	\end{align*}
	where \(hb \cdot \dt{B}(\p{b}) \cdot kc = h\cm{k}{1}(b \ract \cm{k}{2})\cdot \dt{B}(\p{b} \ract \cm{k}{3}) \cdot c\), so that \(\Op(0) \in \pr{\fr{At}}\) with
	\[
		\forall h \in H, \, \forall \beta \in \Omega^1_B, \quad \pr{\Op(0)}(h \cdot \beta) = h \cdot \dt{B}(\beta).
	\]	
	
	Next, let us show that \(\Op_0^{-1}(\pr{\fr{at}}) = \ZH^1_\ell(H;\Omega^1_B,\Omega_B)\); since \(\Op(0) \in \pr{\fr{at}}\), this will also imply that \(\Op^{-1}(\pr{\fr{At}}) = \ZH^1_\ell(H;\Omega^1_B,\Omega_B)\). Let \(\mu \in \ZH^1_\ell(H;\Omega^1_B)\) be given. Then, for all \(h \in H\) and \(\beta \in \Omega^1_B\), so that \(\beta \cdot h = \cm{h}{1} \cdot \beta \ract \cm{h}{2}\), 
	\begin{align*}
		\Op_0(\mu)(\cm{h}{1}) \wedge \beta \ract \cm{h}{2} + \beta \wedge \Op_0(\mu)(h) &= \cm{h}{1}\cdot \mu(\cm{h}{2}) \wedge \beta \ract \cm{h}{3} + \beta \wedge \cm{h}{1} \cdot \mu(\cm{h}{2})\\
		&= \cm{h}{1} \cdot \left(\mu(\cm{h}{2}) \wedge \beta \ract \cm{h}{3} + \beta \ract \cm{h}{2} \wedge \mu(\cm{h}{3})\right)\\
		&= \cm{h}{1} \cdot [\mu,\rho_{\Omega^1_B}(\beta)](\cm{h}{2}).
	\end{align*}
	Thus, \(\Op_0(\mu) \in \pr{\fr{at}}\) if and only if \(\mu \in \ZH^1_\ell(H;\Omega^1_B,\Omega_B)\), in which case
	\[
		\forall h \in H, \, \forall \beta \in \Omega^1_B, \quad \pr{\Op_0(\mu)}(h \cdot \beta) \coloneqq \cm{h}{1} \cdot \mu(\cm{h}{2}) \wedge \beta.
	\]
	
	Finally, let us show that \(\Op_0^{-1}(\Inn(\pr{\fr{at}})) = \BH^1_\ell(H;\Omega^1_B,\Omega_B)\). Let \(\alpha \in \Zent_B(\Omega^1_B)_\sa\). Then, for all \(h \in H\) and \(\beta \in \Omega^1_B\),
	\begin{align*}
		\pr{\Op_0(D\alpha)}(h \cdot \beta) - [\alpha,h \cdot \beta] &= \cm{h}{1} \cdot (\alpha \ract \cm{h}{2} - \epsilon(\cm{h}{2})\alpha) \wedge \beta - \alpha \wedge h \cdot \beta - h \cdot \beta \wedge \alpha\\
		&= - h \cdot [\alpha,\beta]_{\cC(H,\Omega_B)}, 
	\end{align*}
	so that \(\Op_0(D\alpha) \in \Inn(\pr{\fr{at}})\) if and only if \(\alpha \in (\Omega^1_B)_{\sa} \cap \Zent(\Omega_B) = \CH^0_\ell(H;\Omega^1_B,\Omega_B)\).
\end{proof}

We conclude by observing that field strength as a map on the prolongable Atiyah space \(\pr{\fr{At}}\) admits a straightforward reinterpretation as an \(\bR\)-affine quadratic map between spaces of lazy Hochschild \(1\)-cocycles that, in turn, descends to a map between lazy Hochschild cohomology groups. Given a choice of bicovariant \fodc{} \((\Omega^1_H,\dt{H})\) on \(H\), this reinterpretation will prove crucial to characterising \((\Omega^1_H,\dt{H})\)-adapted prolongable gauge potentials in terms of lazy Hochschild cohomology.

\begin{proposition}\label{curveprop}
	The map \(\cF : \ZH^1_\ell(H;\Omega^1_B,\Omega_B) \to \ZH^1_\ell(H;\Omega^2_B)\) defined by
	\begin{equation*}
		\forall h \in H, \, \forall \mu \in \ZH^1_\ell(H;\Omega^1_B,\Omega_B), \quad \cF[\mu](h) \coloneqq S(\cm{h}{1}) \cdot \bF[\Op(\mu)](h)
	\end{equation*}
	is an \(\bR\)-affine quadratic map, satisfying
	\begin{gather*}
		\forall \mu \in \ZH^1_\ell(H;\Omega^1_B,\Omega_B), \quad \cF[\mu] = -\iu{}\left(\dt{B}\mu(\cdot) + \tfrac{1}{2}[\mu,\mu]_{\cC(H,\Omega_B)}\right),\\
		\forall \mu,\nu \in \ZH^1_\ell(H;\Omega^1_B,\Omega_B), \quad \cF[\mu+\nu] - \cF[\mu] -\cF[\nu] = -\iu{}[\mu,\nu]_{\cC(H,\Omega_B)}, \\
		\forall \alpha \in (\Omega^1_B)_\sa \cap \Zent(\Omega_B), \quad \cF[D\alpha] = D(-\iu{}\dt{B}\alpha), \label{coboundcurve}\\
		\forall \sigma \in \ZS^1_\ell(H;B,\Omega^1_B), \, \forall \mu \in \ZH^1_\ell(H;\Omega^1_B,\Omega_B), \quad \cF[\sigma \act \mu + \MC[\dt{B}](\sigma)] = \sigma \act \cF[\mu].
	\end{gather*}
	Hence, \(\cF\) descends to an \(\bR\)-affine quadratic map \(\widetilde{\cF} : \HH^1_\ell(H;\Omega^1_B,\Omega_B) \to \HH^1_\ell(H;
	\Omega^2_B)\), satisfying 
	\begin{multline*}
		\forall \sigma \in \ZS^1_\ell(H;B,\Omega^1_B), \, \forall \mu \in \ZH^1_\ell(H;\Omega^1_B,\Omega_B), \\
		\widetilde{\cF}\mleft[[\sigma] \act [\mu] + \widetilde{\MC}[\dt{B},\Omega_B]([\sigma])\mright] = [\sigma] \act \widetilde{\cF}\mleft[[\mu]\mright].
	\end{multline*}
\end{proposition}

\begin{proof}
	First, note that \(\cF : \ZH^1_\ell(H;\Omega^1_B,\Omega_B) \to \ZH^1_\ell(H;\Omega^2_B)\) is well-defined by the proof of Proposition~\ref{trivialgaugepot}, \emph{mutatis mutandis}, hence \(\bR\)-affine quadratic by Proposition~\ref{fieldstrengthprop}.
	
	Next, let \(\mu \in \ZH^1_\ell(H;\Omega^1_B,\Omega_B)\). Then, for all \(h \in H\),
	\begin{align*}
		\iu{}\cF[\mu](h) &= S(\cm{h}{1}) \cdot \pr{\Op(\mu)} \circ \Op(\mu)(\cm{h}{2})\\
		&= S(\cm{h}{1}) \cdot \pr{\Op(\mu)}(\cm{h}{2}\cdot\mu(\cm{h}{3})\\
		&= S(\cm{h}{1}) \cdot \left(\cm{h}{2} \cdot \mu(\cm{h}{3}) \wedge \mu(\cm{h}{4})+\cm{h}{2} \cdot \dt{B}\mu(\cm{h}{3})\right)\\
		&= \dt{B}\mu(h) +\mu(\cm{h}{1})\wedge\mu(\cm{h}{2}),
	\end{align*}
	so that \(\cF[\mu] = -\iu{}\left(\dt{B}\mu(\cdot) + \tfrac{1}{2}[\mu,\mu]_{\cC(H,\Omega_B)}\right)\). Thus, for all \(\mu,\nu \in \ZH^1_\ell(H;\Omega^1_B,\Omega_B)\),
	\[
		\cF[\mu+\nu] = \dt{B}(\mu+\nu)(\cdot)+\tfrac{1}{2}[\mu+\nu,\mu+\nu]_{\cC(H,\Omega_B)} = \cF[\mu] + \cF[\nu] +[\mu,\nu]_{\cC(H,\Omega_B)}.
	\]
	
	Now, let \(\alpha \in (\Omega^1_B)_\sa \cap \Zent(\Omega_B)\); note that \(-\iu{}\dt{B}\alpha \in \Zent_B(\Omega^2_B)_\sa\), since for all \(b \in B\),
	\[
		b \cdot \dt{B}\alpha - \dt{B}\alpha \cdot b = \dt{B}(b \cdot \alpha)-\dt{B}b\wedge\alpha-\dt{B}(\alpha \cdot b)-\alpha \wedge \dt{B}b = 0.
	\]
	Then, for all \(h \in H\),
	\begin{align*}
		\iu{}\cF[D\alpha](h) &= \dt{B}(\alpha \ract h+\epsilon(h)\alpha) + (\alpha \ract \cm{h}{1}+\epsilon(\cm{h}{1})\alpha) \wedge (\alpha \ract \cm{h}{2}+\epsilon(\cm{h}{2})\alpha) \\
		&= \left(\dt{B}\alpha + \tfrac{1}{2}[\alpha,\alpha]_{\Omega_B}\right) \ract h + \epsilon(h)\left(\dt{B}\alpha + \tfrac{1}{2}[\alpha,\alpha]_{\Omega_B}\right)  + [\alpha,\alpha \ract h]_{\Omega_B}\\
		&= \dt{B}(\alpha) \ract h + \epsilon(h)\dt{B}(\alpha),
	\end{align*}
	so that \(\cF[D\alpha] = D(-\iu{}\dt{B}\alpha) \in \BH^1_\ell(H;\Omega^2_B)\).

	Finally, let \(\sigma \in \ZS^1_\ell(H;B,\Omega^1_B)\) and \(\mu \in \ZH^1_\ell(H;\Omega^1_B,\Omega_B)\). Then, for all \(h \in H\),
	\begin{align*}
		\cF[\sigma \act \mu + \MC[\dt{B}](\sigma)](h) &= S(\cm{h}{1}) \cdot \bF[\Op(\sigma) \act \Op(\mu)](\cm{h}{2})\\
		&= S(\cm{h}{1}) \cdot \Op(\sigma)_\ast \circ \bF[\Op(\mu)] \circ \Op(\sigma^\ast)(\cm{h}{2})\\
		&= S(\cm{h}{1}) \cdot \Op(\sigma)_\ast \circ \bF[\Op(\mu)]\mleft(\cm{h}{2}\sigma^\ast(\cm{h}{3})\mright)\\
		&= S(\cm{h}{1}) \cdot \Op(\sigma)_\ast\mleft(\cm{h}{2} \cdot \cF[\mu](\cm{h}{3}) \cdot \sigma^\ast(\cm{h}{4})\mright)\\
		&= S(\cm{h}{1}) \cm{h}{2}\sigma(\cm{h}{3}) \cdot \cF[\mu](\cm{h}{4}) \cdot \sigma^\ast(\cm{h}{5})\\
		&= \sigma(\cm{h}{1}) \cdot \cF[\mu](\cm{h}{2}) \cdot \sigma^\ast(\cm{h}{3})\\
		&= \sigma \act \cF[\mu](h),
	\end{align*}
	so that, indeed, \(\cF[\sigma \act \mu + \MC[\dt{B}](\sigma)](h) = \sigma \act \cF[\mu]\).
\end{proof}

\subsection{Reconstruction of total calculi}

As we have just seen, all purely horizontal aspects of noncommutative gauge theory on the trivial quantum principal \(H\)-bundle \(B \rtimes H\) can be computed solely in terms of lazy Sweedler cohomology and lazy Hochschild cohomology on \(H\) with coefficients arising from the basic calculus \((\Omega_B,\dt{B})\). We now turn to those aspects that depend on a choice of bicovariant \fodc{} \((\Omega^1_H,\dt{H})\) on \(H\), e.g., the \(\fr{G}\)-equivariant moduli space \(\fr{At}/\fr{at}[\Omega^1_H] \cong \Ob\mleft(\cG[\Omega^1_H]/\ker\mu[\Omega^1_H]\mright)\) of total \fodc{} on \(P\) and the \(\pr{\fr{G}}\)-invariant quadratic subset \(\pr{\fr{At}}[\Omega^1_H]\) of \(\Omega^1_H\)-adapted prolongable gauge potentials.

Let us fix a bicovariant \textsc{fodc} \((\Omega^1_H,\dt{H})\) on \(H\) with left crossed \(H\)-\(\ast\)-module \(\Lambda^1_H\) of right \(H\)-covariant \(1\)-forms and quantum Maurer--Cartan form \(\varpi : H \to \Lambda^1_H\). Let \((\Omega_H,\dt{H})\) denote its canonical prolongation to a bicovariant \textsc{sodc} on \(H\), let \((\Omega_{P,\ver},\dt{P,\ver})\) denote the resulting second-order vertical calculus of \(P\) , and let \(\Omega_{P,\oplus} \coloneqq \Lambda_H \hotimes^{\leq 2} \Omega_{P,\hor}\), which therefore contains \(\Omega_{P,\ver}\) as a left \(H\)-subcomodule graded \(\ast\)-algebra. 

Using the multiplication map \(
		(\omega \otimes h \mapsto \omega \cdot h) : \Lambda_H \otimes H \to \Omega_H
	\), we can identify \(\Omega_H\) with the graded \(\ast\)-subalgebra of \(\Omega_{P,\oplus}\) generated by \(H \subset P\) and \(\Lambda^1_H \hotimes 1_P \subset \Omega^1_{P,\ver}\). Hence, we can view \(\Omega_{P,\oplus}\) as the graded left \(H\)-comodule \(\ast\)-algebra, truncated at degree \(2\), generated by the graded left \(H\)-subcomodule \(\ast\)-subalgebras \(\Omega_H\) and \(\Omega_B\) subject to the relation \(1_{\Omega_H} = 1_{\Omega_B}\) and the braided graded commutation relations
\begin{equation}
	\forall \omega \in \Omega_H, \, \forall \alpha \in \Omega_B, \quad \alpha \wedge \omega \coloneqq (-1)^{\abs{\alpha}\abs{\omega}} \ca{\omega}{0} \wedge \alpha \ract \ca{\omega}{1};
\end{equation}
indeed, \(\Omega_{P,\oplus} = \Omega_H \cdot \Omega_B\) is freely generated as a graded right \(\Omega_B\)-module by \(\Omega_H\). In particular, the graded left \(H\)-subcomodule \(\ast\)-subalgebra \(\Omega_{P,\ver} = \Omega_H \cdot B\) is freely generated as a right \(B\)-module by \(\Omega_H\), and
\begin{equation}
	\forall \omega \in \Omega_H, \, \forall b \in B, \quad \dv{P}(\omega \cdot b) = \dt{H}(\omega) \cdot b,
\end{equation}
so that \((\Omega^1_H,\dt{H})\) is necessarily locally freeing for \(P\).

We can now characterise the space \(\fr{at}[\Omega^1_H]\) of \((\Omega^1_H,\dt{H})\)-adapted relative gauge potentials on \(P\) with respect to the canonical second-order horizontal calculus \((\Omega_B,\dt{B};\Omega_{P,\hor})\), and hence relate the \(\fr{G}\)-equivariant moduli space \(\fr{At}/\fr{at}[\Omega^1_H] \cong \Ob\mleft(\cG[\Omega^1_H]/\ker\mu[\Omega^1_H]\mright)\) of strongly \((H;\Omega^1_H,\dt{H})\)-principal \sodc{} on \(P\) inducing \((\Omega_B,\dt{B};\Omega_{P,\hor})\) to lazy Hochschild cohomology on \(H\).

\begin{proposition}\label{trivialadaptprop}
	The groupoid isomorphism \(\ZS^1_\ell(H;B,\Omega^1_B) \ltimes \ZH^1_\ell(H;\Omega^1_B) \iso \fr{G} \ltimes \fr{At}\) of Proposition~\ref{trivialgaugepotprop} descends to a groupoid isomorphism
	\[
		\ZS^1_\ell(H;B,\Omega^1_B) \ltimes  \left(\frac{\ZH^1_\ell(H;\Omega^1_B)}{\Op_0^{-1}\mleft(\fr{at}[\Omega^1_H]\mright)} \right) \iso \fr{G} \ltimes \left(\frac{\fr{At}}{\fr{at}[\Omega^1_H]}\right),
	\]
	where
	\[
		\Op_0^{-1}\mleft(\fr{at}[\Omega^1_H]\mright) = \set{\mu \in \ZH^1_\ell(H;\Omega^1_B) \given \ker(\mu\circ\inv{S}) \supseteq \ker\varpi}.
	\]
\end{proposition}

\begin{proof}
	Before continuing, recall that \(\varpi : H \to \Lambda^1_H\) satisfies
	\[
		\forall h \in H, \quad \ca{\varpi(h)}{-1} \otimes \ca{\varpi(h)}{0} = \cm{h}{1}S(\cm{h}{3}) \otimes \varpi(\cm{h}{2}).
	\]
	so that that
	\(
		\Delta(\ker\varpi) \subseteq \ker(\id_H \otimes \varpi) = H \otimes \ker\varpi
	\). Let \(\mu \in \ZH^1_\ell(H;\Omega^1_B)\) and set \(\bA \coloneqq \Op_0(\mu)\).

	First, suppose that \(\bA \in \fr{at}[\Omega^1_H]\). Then for all \(h \in H\),
	\begin{align*}
		\mu \circ \inv{S}(h) &= S(\cm{\inv{S}(h)}{1}) \cdot \bA(\cm{\inv{S}(h)}{2})\\
		&= \cm{h}{2} \cdot \bA(\inv{S}(\cm{h}{1}))\\
		&= \bA(\inv{S}(\cm{h}{1}S(\cm{h}{2})) - \bA(\cm{h}{2})\cdot\inv{S}(\cm{h}{1})\\
		&= -\omega[\bA]\mleft(\varpi(\cm{h}{2})\cdot\cm{h}{3}\mright) \cdot \inv{S}(\cm{h}{1})\\
		&= -\omega[\bA]\mleft(\ca{\varpi(h)}{0} \cdot \inv{S}(\ca{\varpi(h)}{-1}) \mright),
	\end{align*}
	so that, indeed, \(\ker (\mu \circ \inv{S}) \supseteq \ker\varpi\).
	
	Now, suppose that \(\ker (\mu \circ \inv{S}) \supseteq \ker\varpi\). Then, for all \(h\in H\), 
	\begin{multline*}
		\bA(\cm{h}{1}) \cdot S(\cm{h}{2}) = -\bA(\cm{h}{1}^\ast)^\ast \cdot S(\cm{h}{2})
		= -\left(\cm{h}{1}^\ast \cdot \mu (\cm{h}{2})^\ast\right)^\ast \cdot S(\cm{h}{3})
		= -\mu(\cm{h}{2}^\ast)^\ast \cdot \cm{h}{1}S(\cm{h}{3})\\
		= -\mu(S(\cm{h}{2}^\ast)^\ast) \cdot \cm{h}{1}S(\cm{h}{3})
		= -\mu(\inv{S}(\cm{h}{2})) \cdot \cm{h}{1}S(\cm{h}{3});
	\end{multline*}
	hence, if \(h \in \ker\varpi\), then \(\cm{h}{1}S(\cm{h}{3}) \otimes \cm{h}{2} \in H \otimes \ker\varpi\), so that \(\bA(\cm{h}{1}) \cdot S(\cm{h}{2}) = 0\). Since \(\varpi : H \to \Lambda^1_H\) is surjective, \(\bA = N \circ \dv{P}\) for the morphism of left \(H\)-comodule right \(P\)-modules \(N : \Omega^1_{P,\ver} \to \Omega^1_{P,\hor}\) defined by
	\[
		\forall h,k \in H, \, \forall b \in B, \quad N(\varpi(k) \cdot hb) \coloneqq \bA(\cm{k}{1}) \cdot S(\cm{k}{2})hb;
	\]
	it therefore suffices to show that \(N\) is left \(P\)-linear and \(\ast\)-preserving, for then \(\bA \in \fr{at}[\Omega^1_H]\) with relative connection \(1\)-form \(\omega[\bA] = N\). Let \(h,k \in H\) and \(b \in B\), so that
	\begin{align*}
		hb \cdot \varpi(k) &= h \cdot \varpi(k) \cdot \left(b \ract 1\right)\\
		&= \left(\varpi(\cm{h}{1}k) - \epsilon(k)\varpi(\cm{h}{1})\right) \cdot \cm{h}{2}b\\
		&= \varpi(\cm{h}{1}k) \cdot \cm{h}{2} b - \varpi(\cm{h}{1}) \cdot \cm{h}{2}\epsilon(k)b.
	\end{align*}
	On the one hand,
	\begin{align*}
	&N\mleft(\varpi(\cm{h}{1}k) \cdot \cm{h}{2} b - \varpi(\cm{h}{1}) \cdot \cm{h}{2}\epsilon(k)b\mright)\\
	&= \bA(\cm{h}{1}\cm{k}{1}) \cdot S(\cm{h}{2}\cm{k}{2}) \cdot \cm{h}{3}b -\bA(\cm{h}{1}) \cdot S(\cm{h}{2})\cm{h}{3}\epsilon(k)b \\
	&= \bA(\cm{h}{1}\cm{k}{1}) \cdot S(\cm{k}{2})S(\cm{h}{2})\cm{h}{3}b - \epsilon(k)\bA(h)b\\
	&= \bA(h)\cm{k}{1} \cdot S(\cm{k}{2})b +h\bA(\cm{k}{1}) \cdot S(\cm{k}{2})b - \epsilon(k)\bA(h)b\\
	&= h \bA(\cm{k}{1}) \cdot S(\cm{k}{2}) b,
	\end{align*}
	while on the other,
	\begin{align*}
		hb N(\varpi(k)) &= hb \cdot \bA(\cm{k}{1})S(\cm{k}{2})\\
		&= -hb\cm{k}{1})\bA(S(\cm{k}{2}))\\
		&= -h\cm{k}{1}S(\cm{k}{5})\left((b\ract\cm{k}{2})\ract S(\cm{k}{4})\right) \cdot \mu(S(\cm{k}{3}))\\
		&= -h\cm{k}{1}S(\cm{k}{5}) \cdot \mu(S(\cm{k}{4})) \cdot \left((b\ract\cm{k}{2})\ract S(\cm{k}{3})\right)\\
		&= -h\cm{k}{1}S(\cm{k}{3})\mu(S(\cm{k}{2}))b\\
		&= -h\cm{k}{1} \cdot \bA(S(\cm{k}{2})) \cdot b\\
		&= h \cdot \bA(\cm{k}{1}) \cdot S(\cm{k}{2})b,
	\end{align*}
	so that
	\(
		hb N(\varpi(k)) = N\mleft(\varpi(\cm{h}{1}k) \cdot \cm{h}{2} b - \varpi(\cm{h}{1}) \cdot \cm{h}{2}\epsilon(k)b\mright)
	\).
	Hence, \(N\) is left \(P\)-linear. Similar calculations now show that \(N\) is also \(\ast\)-preserving.
\end{proof}

The techniques of this proof can also be used to characterise the quadric set \(\pr{\fr{At}}[\Omega^1_H]\) of \((\Omega^1_H,\dt{H})\)-adapted prolongable gauge potentials on \(P\) with respect to  \((\Omega_B,\dt{B};\Omega_{P,\hor})\)---hence, also, the \(\pr{\fr{G}}\)-equivariant moduli space \[\pr{\fr{At}}[\Omega^1_H]/\pr{\fr{at}}_{\can}[\Omega^1_H] \cong \Ob\mleft(\cG[\Omega^{\leq 2}_H]/\ker\mu[\Omega^{\leq 2}_H]\mright)\] of  strongly \((H;\Omega_H,\dt{H})\)-principal \sodc{} on \(P\) with respect to \((\Omega_B,\dt{B};\Omega_{P,\hor})\)---in terms of lazy Hochschild cohomology. Recall that \(\pr{\fr{at}}_{\can}[\Omega^1_H]\) denotes the space of all \((\Omega_H,\dt{H})\)-adapted prolongable relative gauge potentials on \(P\) with respect to \((\Omega_B,\dt{B};\Omega_{P,\hor})\) and that \(\Inn(\pr{\fr{at}};\Omega^1_H)\) denotes the space of all \((\Omega^1_H,\dt{H})\)-semi-adapted inner prolongable relative gauge potentials on \(P\) with respect to \((\Omega_B,\dt{B};\Omega_{P,\hor})\), so that \[\Out(\pr{\fr{at}}[\Omega^1_H]) \coloneqq \pr{\fr{At}}[\Omega^1_H]/\Inn(\pr{\fr{at}};\Omega^1_H).\]

\begin{theorem}\label{trivialthm}
	The groupoid isomorphism \(\ZS^1_\ell(H;B,\Omega^1_B) \ltimes \ZH^1_\ell(H;\Omega^1_B,\Omega_B) \iso \pr{\fr{G}} \ltimes \pr{\fr{At}}\) of Corollary~\ref{trivialprcor} restricts to a groupoid isomorphism
	\[
		\ZS^1_\ell(H;B,\Omega^1_B) \ltimes \Op^{-1}\mleft(\pr{\fr{At}}[\Omega^1_H]\mright) \iso \pr{\fr{G}} \ltimes \pr{\fr{At}}[\Omega^1_H],
	\]
	where
	\begin{equation}\label{triveq1}
		\Op^{-1}\mleft(\pr{\fr{At}}[\Omega^1_H]\mright) = \set{\mu \in \ZH^1_\ell(H;\Omega^1_B,\Omega_B) \given \ker(\cF[\mu] \circ \inv{S}) \supseteq \ker\varpi}.
	\end{equation}
	In turn, this restricted groupoid isomorphism descends to respective groupoid isomorphisms
	\begin{gather*}
			\ZS^1_\ell(H;B,\Omega^1_B) \ltimes \left(\frac{\Op^{-1}\mleft(\pr{\fr{At}}[\Omega_H]\mright)}{\Op_0^{-1}\mleft(\pr{\fr{at}}_{\can}[\Omega^1_H]\mright)}\right) \iso \pr{\fr{G}} \ltimes \left(\frac{\pr{\fr{At}}[\Omega^1_H]}{\pr{\fr{at}}_{\can}[\Omega^1_H]}\right),\label{triveq2}\\
			\HS^1_\ell(H;B,\Omega^1_B) \ltimes \left(\frac{\Op^{-1}\mleft(\pr{\fr{At}}[\Omega^1_H]\mright)}{\Op_0^{-1}\mleft(\Inn(\pr{\fr{at}};\Omega^1_H)\mright)}\right) \iso \Out(\pr{\fr{G}}) \ltimes \Out(\pr{\fr{At}}[\Omega^1_H]),\label{triveq3}
	\end{gather*}
	where
	\begin{gather}
		\Op_0^{-1}\mleft(\pr{\fr{at}}_{\can}[\Omega^1_H]\mright) = \set{\mu \in \ZH^1_\ell(H;\Omega^1_B,\Omega_B) \given \ker(\mu \circ \inv{S}) \supseteq \ker\varpi},\\
		\Op_0^{-1}\mleft(\Inn(\pr{\fr{at}};\Omega^1_H)\mright) = \set{D\beta \given \beta \in Z_B(\Omega^1_B)_\sa \cap Z(\Omega_B), \,\ker(D(-\iu{}\dt{B}\beta) \circ \inv{S}) \supseteq \ker\varpi}.
	\end{gather}
	In particular, \(\pr{\fr{at}}_{\can}[\Omega^1_H] = \pr{\fr{at}} \cap \fr{at}[\Omega^1_H]\).
\end{theorem}

\begin{proof}
	Let us first prove \eqref{triveq1}. Let \(\mu \in \ZH^1_\ell(H;\Omega^1_B,\Omega_B)\), so that
	\[
		\forall h \in H, \, \forall b \in B, \quad \bF[\Op(\mu)](hb) = \cm{h}{1}\cdot\cF[\mu](\cm{h}{2})\cdot b.
	\]
	By applying the proof of Proposition~\ref{trivialadaptprop}, \emph{mutatis mutandis}, to the left \(H\)-covariant \(\ast\)-deriva\-tion \(\bF[\Op(\mu)] : P \to \Omega^2_{P,\hor}\), it follows that \(\Op(\mu) \in \tot{\fr{At}}\) if and only if
	\[
		\ker(\cF[\mu] \circ \inv{S}) \supseteq \ker\varpi.
	\]
	This argument now also proves \eqref{triveq3}, since for all \(\beta \in Z_B(\Omega^1_B)_\sa \cap Z(\Omega_B)\) and \(h \in H\)
	\[
		D(-\iu{}\dt{B}\beta)(h) = \cF[D\alpha](h) = S(\cm{h}{1}) \cdot \bF[\Op(\alpha)](\cm{h}{2}) = S(\cm{h}{1}) \cdot \bF_\rel[\Op_0(\alpha)](\cm{h}{2})
	\]
	by \eqref{coboundcurve} and the fact that \(\bF[\Op(0)] = 0\). Hence, it remains to prove \eqref{triveq3}.

	On the one hand, \(\pr{\fr{At}}[\Omega_H] \subseteq \pr{\fr{at}} \cap \fr{at}[\Omega^1_H]\) by Theorem~\ref{qpbthm}; on the other hand, by Corollary~\ref{trivialprcor} and Proposition~\ref{trivialadaptprop},
	\[
		\Op_0^{-1}(\pr{\fr{at}} \cap \fr{at}[\Omega^1_H]) = \set{ \mu \in \ZH^1_\ell(H;\Omega^1_B,\Omega_B) \given \ker (\mu \circ \inv{S}) \supseteq \ker\varpi} \eqqcolon \ZH^1_\ell(H;\Omega^1_B,\Omega_B)[\Omega^1_H] .
	\]
	Hence, it suffices to show that \(\pr{\fr{At}}[\Omega_H] \supseteq \pr{\fr{at}} \cap \fr{at}[\Omega^1_H]\).
	
	Let \(\bA \in \pr{\fr{at}} \cap \fr{at}[\Omega^1_H]\) be given; set \(\mu \coloneqq \Op_0^{-1}(\bA)\) and \(N \coloneqq \omega[\bA]\), so that for all \(h \in H\),
	\[
		N(\varpi(h)) = \bA(\cm{h}{1}) \cdot S(\cm{h}{2}) = -\cm{h}{1} \cdot \bA(S(\cm{h}{2})) = -\cm{h}{1}S(\cm{h}{3}) \cdot \mu(S(\cm{h}{2}))
	\]
	by the proof of Proposition~\ref{trivialadaptprop}. We will show that \(\bA \in \pr{\fr{At}}[\Omega_H]\) by showing that \(N\) satisfies both \eqref{toteq1} and \eqref{toteq2}.
	
	First, let \(h \in H\) and \(\beta \in \Omega^1_B\). Since \(\mu \in \ZH^1_\ell(H;\Omega^1_B,\Omega_B)\), it follows that
	\begin{align*}
		\beta \wedge N(\varpi(h)) &= -\beta \wedge \cm{h}{1}S(\cm{h}{3}) \cdot \mu(S(\cm{h}{2}))\\
		&= -\cm{h}{1}S(\cm{h}{5}))\left((\beta \ract \cm{h}{2}) \ract S(\cm{h}{4})\right) \wedge \mu(S(\cm{h}{3}))\\
		&=  \cm{h}{1}S(\cm{h}{5})) \cdot \mu(S(\cm{h}{4})) \wedge \left((\beta \ract \cm{h}{2}) \ract S(\cm{h}{3})\right) \\
		&=  \cm{h}{1}S(\cm{h}{3}) \cdot \mu(S(\cm{h}{3})) \cdot \beta\\
		&= -N(\varpi(h)) \wedge \beta
	\end{align*}
	by Lemma~\ref{centrallem}. Hence, \(N\) satisfies \eqref{toteq1}.
	
	Let us now check that \(N\) satisfies \eqref{toteq2}. Let \(h, k \in H\); note that
	\begin{align*}
		\ca{\varpi(k)}{-1} \act\varpi(h) \otimes \ca{\varpi(k)}{0} &= \cm{k}{1}S(\cm{h}{k}) \act\varpi(h) \otimes \varpi(\cm{k}{2})\\
		&= \varpi(\cm{h}{1}S(\cm{h}{3})k) \otimes \varpi(\cm{h}{2}) - \epsilon(k)\varpi(\cm{h}{1}S(\cm{h}{3})) \otimes \varpi(\cm{h}{2}),
	\end{align*}
	so that it suffices to show that
	\[
		N(\varpi(h)) \wedge N(\varpi(k)) + N(\varpi(\cm{h}{1}S(\cm{h}{3})k)) \wedge N(\varpi(\cm{h}{2})) - \epsilon(k)N(\varpi(\cm{h}{1}S(\cm{h}{3})))N(\varpi(\cm{h}{2})) = 0.
	\]
	First, by repeated applications of Lemma~\ref{centrallem} together with the fact that \(\varpi\) is a \(1\)-cocycle valued in the left crossed \(H\)-\(\ast\)-module \(\Lambda^1_H\), we find that
	\begin{align*}
		N(\varpi(h)) \wedge N(\varpi(k)) &= \bA(\cm{h}{1}) \cdot S(\cm{h}{2}) \wedge \bA(\cm{k}{1}) \cdot S(\cm{k}{2})\\
		&= \cm{h}{1} \cdot \bA(S(\cm{h}{2})) \wedge \cm{k}{1} \cdot \bA(S(\cm{k}{2}))\\
		&= \cm{h}{1}\cdot \bA(S(\cm{h}{2})\cm{k}{1}) \wedge \bA(\cm{k}{2}) - \cm{h}{1}S(\cm{h}{2}) \cdot \bA(\cm{k}{1}) \wedge \bA(S(\cm{k}{2}))\\
		&= \cm{h}{1}S(\cm{h}{3})\cm{k}{1} \cdot \mu(S(\cm{h}{2})\cm{k}{2}) \wedge S(\cm{k}{4}) \cdot \mu(S(\cm{k}{3}))\\
		&\quad\quad\quad-\epsilon(h)\cm{k}{1}\cdot\mu(\cm{k}{2})\wedge S(\cm{k}{4})\cdot\mu(S(\cm{k}{3}))\\
		&= \cm{h}{1}S(\cm{h}{3})\cm{k}{1}S(\cm{k}{5}) \cdot \mu(S(\cm{h}{2})\cm{k}{2}) \ract S(\cm{k}{4}) \wedge \mu(S(\cm{k}{3}))\\
		&\quad\quad\quad-\epsilon(h)\cm{k}{1}S(\cm{k}{5})\cdot\mu(\cm{k}{2})\ract S(\cm{k}{4})\wedge\mu(S(\cm{k}{3}))\\
		&= - \cm{h}{1}S(\cm{h}{3})\cm{k}{1}S(\cm{k}{5}) \cdot \mu(S(\cm{k}{4})) \wedge \mu(S(\cm{h}{2})\cm{k}{2}) \ract S(\cm{k}{3})\\
		&\quad\quad\quad+\epsilon(h)\cm{k}{1}S(\cm{k}{5})\cdot\mu(S(\cm{k}{4})) \wedge \mu(\cm{k}{2})\ract S(\cm{k}{3})\\
		&=- \cm{h}{1}S(\cm{h}{3})\cm{k}{1}S(\cm{k}{3}) \cdot \mu(S(\cm{k}{2})) \wedge \mu(S(\cm{h}{2}))\\
		&\quad\quad\quad+\epsilon(h)\cm{k}{1}S(\cm{k}{4}) \cdot \mu(S(\cm{k}{3}))\wedge\mu(S(\cm{k}{2}))\\
		&\quad\quad\quad-\epsilon(h)\cm{k}{1}S(\cm{k}{4}) \cdot \mu(S(\cm{k}{3}))\wedge\mu(S(\cm{k}{2}))\\
		&= -\cm{h}{1}S(\cm{h}{3})\cm{k}{1}S(\cm{k}{3}) \cdot \mu(S(\cm{k}{2}) \wedge \mu(S(\cm{h}{2})).
	\end{align*}
	Next, by applying the last calculation, \emph{mutatis mutandis}, we find that
	\begin{align*}
		&N(\varpi(\cm{h}{1}S(\cm{h}{3})k)) \wedge N(\varpi(\cm{h}{2}))\\
		&=- \cm{h}{1}S(\cm{h}{9})\cm{k}{1}S(\cm{k}{3})S^2(\cm{h}{7})S(\cm{h}{3})\cm{h}{4}S(\cm{h}{6}) \cdot\mu(S(\cm{h}{5})\wedge \mu(S(\cm{k}{2})S^2(\cm{h}{8})S(\cm{h}{2}))\\
		&= -\cm{h}{1}S(\cm{h}{5})\cm{k}{1}S(\cm{k}{3})\cdot\mu(S(\cm{h}{3}))\wedge\mu(S(\cm{k}{2})S^2(\cm{h}{4})S(\cm{h}{2}))\\
		&= -\cm{h}{1}S(\cm{h}{5})\cm{k}{1}S(\cm{k}{3})\cdot\mu(S(\cm{h}{3}))\wedge\left(\mu(S(\cm{k}{2})S^2(\cm{h}{4})) \ract S(\cm{h}{2}) + \epsilon(\cm{k}{2})\epsilon(\cm{h}{4})\mu(S(\cm{h}{2}) \right) \\
		&= \cm{h}{1}S(\cm{h}{5})\cm{k}{1}S(\cm{k}{3}) \cdot \mu(S(\cm{k}{2})S^2(\cm{h}{4})) \ract S(\cm{h}{3}) \wedge \mu(S(\cm{h}{2})) \\
		&\quad\quad\quad-\epsilon(k)\cm{h}{1}S(\cm{h}{4}) \cdot \mu(S(\cm{h}{3})) \wedge \mu(S(\cm{h}{2}))\\
		&= \cm{h}{1}S(\cm{h}{3})\cm{k}{1}S(\cm{k}{3}) \cdot \mu(S(\cm{k}{2}) \wedge \mu(S(\cm{h}{2})) - \epsilon(k)\cm{h}{1}S(\cm{h}{4}) \cdot \mu(S(\cm{h}{3})) \wedge \mu(S(\cm{h}{2}))\\
		&\quad\quad\quad-\epsilon(k)\cm{h}{1}S(\cm{h}{4}) \cdot \mu(S(\cm{h}{3})) \wedge \mu(S(\cm{h}{2}))\\
		&= -N(\varpi(h)) \wedge N(\varpi(k)) - 2\epsilon(k)\cm{h}{1}S(\cm{h}{4}) \cdot \mu(S(\cm{h}{3})) \wedge \mu(S(\cm{h}{2})).
	\end{align*}
	Finally, by applying the last calculation, \emph{mutatis mutandis}, we find that
	\[
		N(\varpi(\cm{h}{1}S(\cm{h}{3}))) \wedge N(\varpi(\cm{h}{2})) = -2\cm{h}{1}S(\cm{h}{4}) \cdot \mu(S(\cm{h}{3})) \wedge \mu(S(\cm{h}{2})),
	\]
	which we can substitute into the calculation of \(N(\varpi(\cm{h}{1}S(\cm{h}{3})k)) \wedge N(\varpi(\cm{h}{2}))\) to obtain
	\[
		N(\varpi(\cm{h}{1}S(\cm{h}{3})k)) \wedge N(\varpi(\cm{h}{2})) = -N(\varpi(h)) \wedge N(\varpi(k)) + 2\epsilon(k)N(\varpi(\cm{h}{1}S(\cm{h}{3}))) \wedge N(\varpi(\cm{h}{2})),
	\]
	which, in turn, yields our claim.
\end{proof}

\begin{remark}
Note, in particular, that \(\pr{\fr{At}}[\Omega^1_H]\) necessarily contains \(\Op(0)\), which corresponds to the trivial flat connection.	
\end{remark}

The proof of the above result almost immediately yields the following expression for the curvature \(2\)-form of an \((\Omega^1_H,\dt{H})\)-adapted prolongable gauge potential.

\begin{corollary}\label{trivcurvcor}
	Let \(\mu \in \Op^{-1}\mleft(\pr{\fr{At}}[\Omega^1_H]\mright)\). Then
	\[
		\forall h \in H, \quad F[\Op(\mu)](\varpi(h)) = \cm{h}{1}S(\cm{h}{4}) \cdot \iu{}\left(\dt{B}\mu(S(\cm{h}{3}))\epsilon(\cm{h}{2}) + \mu(S(\cm{h}{3})) \wedge \mu(S(\cm{h}{2}))\right)
	\]
\end{corollary}

\section{\texorpdfstring{\(q\)}{q}-Monopoles over real multiplication noncommutative \texorpdfstring{\(2\)}{2}-tori}\label{sec5}

Let \(\theta \in \bR \setminus \mathbf{Q}\) be a quadratic irrationality, and let \(\cA_\theta\) be the corresponding smooth noncommutative \(2\)-torus. In this section, we show how to subsume Connes's constant curvature connections~\cite{Connes80} and Polishchuk--Schwarz's holomorphic structures~\cite{PolishchukSchwarz} on the self-Morita equivalence bimodules amongst the basic Heisenberg modules over \(\cA_\theta\) into gauge theory on a certain canonical non-trivial principal \(\cO(\Unit(1))\)-module algebra \(P\) over \(\cA_\theta\) implicit in Manin's `Alterstraum'~\cite{Manin}. In the process, we will encounter striking formal similarities with the \(q\)-deformed complex Hopf fibration; in this case, however, there is a canonical value for the parameter \(q\) arising from the algebraic number theory of the real quadratic irrationality \(\theta\).

\subsection{Number-theoretic preliminaries} We begin by recalling relevant folklore about real quadratic number fields; our primary reference is the monograph of Halter--Koch~\cite{HK}. Let \(\theta \in \bR \setminus \mathbf{Q}\) be a quadratic irrationality. Recall that \(\SL(2,\bZ)\) acts on \(\bR \setminus \mathbf{Q}\) by fractional linear transformations, i.e., by
\[
	\forall g \in \SL(2,\bZ), \, \forall \xi \in \bR \setminus \mathbf{Q}, \quad g \act \xi \coloneqq \frac{g_{11}\xi+g_{12}}{g_{21}\xi+g_{22}},
\]
and observe that this action descends to an action of \(\PSL(2,\bZ)\) on \(\bR \setminus \mathbf{Q}\); denote the stabilizers of \(\theta\) in \(\SL(2,\bZ)\) and \(\PSL(2,\bZ)\) by \(\SL(2,\bZ)_\theta\) and \(\PSL(2,\bZ)_\theta\). Our goal is to compute \(\SL(2,\bZ)_\theta\) (and hence \(\PSL(2,\bZ)_\theta\)) in terms of the solutions of the positive Pell's equation associated with the real quadratic irrationality \(\theta\).

First, the \emph{type} of \(\theta\) is the unique coprime triple \((a,b,c) \in (\bZ \setminus \set{0}) \times \bZ^2\), such that
\[
	\theta = \frac{b+\sqrt{b^2-4ac}}{2a}
\]
and \(b^2-4ac\) is not a square, so that \(\Delta \coloneqq b^2-4ac\) is the \emph{discriminant} of \(\theta\). On the one hand, then, the \emph{norm} on the real quadratic number field \(\mathbf{Q}[\theta] = \mathbf{Q}[\sqrt{\Delta}]\) is the multiplicative unit-preserving map \(\mathcal{N} : \mathbf{Q}[\theta] \to \mathbf{Q}\) defined by
\[
	\forall r,s \in \mathbf{Q}, \quad \mathcal{N}(r+s\sqrt{\Delta}) \coloneqq r^2-\Delta s^2;
\]
on the other hand, the \emph{quadratic order} of discriminant \(\Delta\) in \(\mathbf{Q}[\theta] = \mathbf{Q}[\!\sqrt{\Delta}]\) is the subring
\[
	\mathcal{O}_\Delta = \set{\tfrac{u+v\sqrt{\Delta}}{2} \given u,v \in \bZ, \, u \equiv v \Delta \bmod 2}.
\]
Thus, we can define the multiplicative group of \emph{norm-positive} units in \(\cO_\Delta\) by
\[
	\mathcal{O}_\Delta^{\times,+} \coloneqq \set{\epsilon \in \mathcal{O}_\Delta^\times \given \mathcal{N}(\epsilon) > 0} = \set{\tfrac{u+v\sqrt{\Delta}}{2} \given u,v \in \bZ, \, u^2 - \Delta v^2 = 4},
\]
which can therefore be viewed as the group of positive solutions of the Pell's equation 
\[x^2 - \Delta y^2 = 4.\]

By applying Dirichlet's unit theorem to the quadratic order \(\cO_\Delta\) of the real quadratic field \(\mathbf{Q}[\theta]\), one finds that \(\cO_\Delta^\times \cap \bR_{>0}\) is infinite cyclic and generated by the \emph{fundamental unit}
\[
	\epsilon_\Delta \coloneqq \min\set{\epsilon \in \cO_\Delta^\times \cap \bR_{>0} \given \epsilon > 1} = \min\set{\epsilon \in \cO_\Delta^\times \given \epsilon > 1},
\]
so that, in turn, the subgroup \(\cO_\Delta^{\times,+} \cap \bR_{>0}\) is also infinite cyclic and generated by the \emph{norm-positive fundamental unit} (or \emph{Pell's unit})
\[
	\epsilon_\Delta^+ \coloneqq \begin{cases} \epsilon_\Delta &\text{if \(\mathcal{N}(\epsilon_\Delta)=1\),}\\ \epsilon_\Delta^2 &\text{if \(\mathcal{N}(\epsilon_\Delta) = -1\).} \end{cases}
\]
The following folkloric result now gives the desired number-theoretic characterization of the stabilizer groups \(\SL(2,\bZ)_\theta\) and \(\PSL(2,\bZ)_\theta\).

\begin{proposition}[Folklore~\cite{HK}*{Thm.\ 5.2.10}]\label{stabthm}
	The map \(\Phi : \cO_\Delta^{\times,+} \to \SL(2,\bZ)_\theta\) given by
	\begin{equation}
		\forall \tfrac{u+v\sqrt{\Delta}}{2} \in \cO_\Delta^{\times,+}, \quad \Phi\left(\tfrac{u+v\sqrt{\Delta}}{2}\right) \coloneqq \begin{pmatrix} \tfrac{u+bv}{2} & -cv \\ av & \tfrac{u-bv}{2} \end{pmatrix}
	\end{equation}
	defines a group isomorphism satisfying \(\Phi(-1) = -I\) and
	\begin{equation}\label{stabthmeq}
		\forall g \in \SL(2,\bZ)_\theta, \quad \inv{\Phi}(g) \coloneqq g_{21}\theta+g_{22}.
	\end{equation}
	Thus, in particular, \(\SL(2,\bZ)_\theta\) decomposes as the internal direct product
	\[
		\SL(2,\bZ)_\theta = \Phi\mleft(\cO^{\times,+}_\Delta \cap \bR_{>0}\mright) \times \set{\pm I},
	\]
	where \(\Phi\mleft(\cO^{\times,+}_\Delta \cap \bR_{>0}\mright)\) is infinite cyclic with canonical generator \(\Phi(\epsilon_\Delta^+)\), so that \(\PSL(2,\bZ)_\theta\) is infinite cyclic with canonical generator \([\Phi(\epsilon_\Delta^+)]\).
\end{proposition}

\subsection{Basic Heisenberg modules over irrational noncommutative \texorpdfstring{\(2\)}{2}-tori} In this section, we recall the construction of basic Heisenberg modules over an irrational noncommutative \(2\)-torus \(C^\infty(\bT^2_\theta)\); when \(\theta \in \bR \setminus \mathbf{Q}\) is a quadratic irrationality, we assemble the self-Morita equivalence bimodules of positive rank among the basic Heisenberg modules into a canonical non-trivial principal \(\cO(\Unit(1))\)-comodule algebra \(P\) over \(C^\infty(\bT^2_\theta)\).

Let us first recall the relevant algebraic notions of (Morita) equivalence bimodule and self-Morita equivalence bimodule.

\begin{definition}
	Let \(A\) and \(B\) be unital \(\bC\)-algebras. An \emph{\((A,B)\)-equivalence bimodule} is an \((A,B)\)-bimodule \({}_A E_B\) satisfying the following:
	\begin{enumerate}
		\item the right \(B\)-module \(E_B\) is finitely generated, projective, and full in the sense that
		\[
			B = \Span_{\bC}\set{f(e) \given e \in E, f \in \Hom_B(E_B,B_B)};
		\]
		\item the left \(A\)-module structure on \({}_A E_B\) is an algebra isomorphism \(A \iso \End_B(E_B)\).
	\end{enumerate}
	In particular, a \emph{self-Morita equivalence bimodule} over \(B\) is an \((B,B)\)-equiva\-lence bimodule.
\end{definition}

\begin{example}
	Let \(B\) be a unital \(\bC\)-algebra. The \emph{trivial} self-Morita equivalence bimodule over \(B\) is \(B\) itself equipped with the left and right \(B\)-module structures induced by multiplication in the algebra \(B\).
\end{example}

Next, recall that the \emph{smooth noncommutative \(2\)-torus} with deformation parameter \(\theta \in \bR\) is the unital \(\ast\)-algebra \(\cA_\theta\) of rapidly decaying Laurent series in two unitary generators \(U_\theta\) and \(V_\theta\) satisfying the commutation relation
\begin{equation}
	V_{\theta} U_\theta = \eu^{2\pi\iu{}\theta} U_\theta V_\theta;
\end{equation}
note that \(\cA_\theta\) can be canonically topologised as a Fr\'{e}chet pre-\(C^\ast\)-algebra, whose \(C^\ast\)-algebraic completion can be identified with the rotation algebra \(A_\theta \coloneqq C(\Unit(1)) \rtimes_\theta \bZ\). As Rieffel famously observed~\cite{Rieffel83}*{Thm.\ 1.1}, the action of \(\SL(2,\bZ)\) on \(\bR \setminus \mathbf{Q}\) by fractional linear transformations manifests itself as distinguished family of equivalence bimodules, the \emph{basic Heisenberg modules}, for smooth noncommutative \(2\)-tori with irrational deformation parameter. These bimodules had already been constructed \emph{qua} right modules by Connes~\cite{Connes80} as the very first examples of noncommutative smooth vector bundles in noncommutative differential geometry.

\begin{theoremdefinition}[Connes~\cite{Connes80}, Rieffel~\citelist{\cite{Rieffel81}\cite{Rieffel83}*{Thm.\ 1.1}\cite{Rieffel88}*{\S\S 2--3, 5}}]
	Let \(\theta \in \bR \setminus \mathbf{Q}\). For every \(g \in \SL(2,\bZ)\), we can construct a \(\left(\cA_{g \act \theta},\cA_\theta\right)\)-equivalence bimodule \(\cE(g,\theta)\) as follows:
	\begin{enumerate}
		\item if \(g_{21} = 0\), so that \(\cA_{g \act \theta} = \cA_{\theta}\) and \(g_{22} \neq 0\), let \(\cE(g,\theta) \coloneqq \cA_{\theta}\) be the trivial self-Morita equivalence bimodule over \(\cA_\theta\);
		\item if \(g_{22} \neq 0\), let \(\cE(g,\theta) \coloneqq \cS(\bR) \otimes \bC[\bZ_{g_{21}}]\) with the unique right \(\cA_{\theta}\)-module structure satisfying
		\begin{align*}
			\forall f \in \cE(g,\theta), \, \forall (x,k) \in \bR \times \bZ_{g_{21}}, \quad (f \cdot U_\theta)(x,k) &= \exp\mleft(2\pi\iu{}\left(x-\tfrac{k g_{22}}{g_{21}}\right)\mright) f(x,k),\\
			\forall f \in \cE(g,\theta), \, \forall (x,k) \in \bR \times \bZ_{g_{21}}, \quad (f \cdot V_\theta)(x,k) &= f\mleft(x-\tfrac{g_{21}\theta+g_{22}}{g_{21}},k-1\mright),
		\end{align*}
		and the unique left \(\cA_{g\act\theta}\)-module structure satisfying
		\begin{align*}
			\forall f \in \cE(g,\theta), \, \forall (x,k) \in \bR \times \bZ_{g_{21}}, \quad (U_{g\act\theta} \cdot f)(x,k) &= \exp\mleft(2\pi\iu{}\left(\tfrac{x}{g_{21}\theta+g_{22}} - \tfrac{k}{g_{21}}\right)\mright) f(x,k),\\
			\forall f \in \cE(g,\theta), \, \forall (x,k) \in \bR \times \bZ_{g_{21}}, \quad (V_{g\act\theta} \cdot f)(x,k) &= f\mleft(x-\tfrac{1}{g_{21}},k-g_{11}\mright).
		\end{align*}
	\end{enumerate}
	We call \(\cE(g,\theta)\) the \emph{basic Heisenberg module} over \(\cA_\theta\) with \emph{rank} \(g_{21}\theta+g_{22}\) and \emph{degree} \(g_{21}\).
\end{theoremdefinition}

Now, given \(\theta \in \bR \setminus \mathbf{Q}\) and \(g \in \SL(2,\bZ)\), we see that \(g \act \theta = \theta\) if and only if \(\theta\) is a quadratic irrationality and \(g \in \SL(2,\bZ)_\theta\). Thus, in light of Proposition~\ref{stabthm}, if \(\theta\) is a real quadratic irrationality, then the family of basic Heisenberg bimodules over \(\cA_\theta\) contains a canonical family of self-Morita equivalence bimodules, which we can now assemble into a principal \(\cO(\Unit(1))\)-module \(\ast\)-algebra \(P\) with \(\coinv{\cO(\Unit(1))}{P} = \cA_\theta\).

\begin{theorem}[Schwarz~\cite{Schwarz98}*{\S 3}, Dieng--Schwarz~\cite{DiengSchwarz}, Polishchuk--Schwarz~\cite{PolishchukSchwarz}*{\S 1.3}, Po\-li\-shchuk \cite{Polishchuk}*{\S 2.2}, Vlasenko~\cite{Vlasenko}*{Thm.\ 6.1}]\label{quadprincipalthm} Let \(\theta \in \bR \setminus \mathbf{Q}\) be a quadratic irrationality with norm-positive fundamental unit \(\e\), and let \(\Phi : \langle \epsilon \rangle \times \set{\pm 1} \to \SL(2,\bZ)_\theta\) be the isomorphism of Proposition~\ref{stabthm}; hence, given \(m \in \bZ\), let
\[
	a_m \coloneqq \Phi(\e^m)_{11}, \quad b_m \coloneqq \Phi(\e^m)_{12}, \quad c_m \coloneqq \Phi(\e^m)_{21}, \quad d_m \coloneqq \Phi(\e^m)_{22}.
\]
For each \(n \in \bZ\), let \(P_n \coloneqq \cE(\Phi(\e^n),\theta)\), and define a \(\cO(\Unit(1))\)-comodule \(\cA_\theta\)-bimodule 
\[
	P \coloneqq \bigoplus_{n \in \bZ} P_n,
\]
where the corepresentation of \(\cO(\Unit(1))\) is induced by the obvious \(\bZ\)-grading. Then \(P\) defines a principal \(\cO(\Unit(1))\)-comodule \(\ast\)-algebra over \(\coinv{\cO(\Unit(1))}{P} = P_0 = \cA_\theta\) when endowed with the \(\ast\)-operation defined by
\[
	\forall m \in \bZ, \, \forall f \in P_m, \, \forall (x,k) \in \bR \times \bZ_{c_m}, \quad f^\ast(x,k) \coloneqq \overline{f(\e^m x,-a_mk)}
\]
and the multiplication \(P \otimes_\bC P \to P\) defined as follows:
\begin{enumerate}
	\item the restrictions of the multiplication to \(P_0 \otimes_{\bC} P = \cA_\theta \otimes_{\bC} P\) and \(P \otimes_{\bC} P_0 = P \otimes_{\bC} \cA_\theta\) are given by the left and right \(\cA_\theta\)-module structures on \(P\), respectively;
	\item for all non-zero \(m \in \bZ\), for all \(f \in P_{-m}\), and for all \(g \in P_m\),
	\[
		f \cdot g \coloneqq \sum_{(n_1,n_2) \in \bZ^2} U^{n_1} V^{n_2} \sum_{k \in \bZ_{c_m}} \int_{\bR} (V_\theta^{-n_2}U_\theta^{-n_1} \cdot f)\mleft(\tfrac{x}{\e^m},k\mright) g\mleft(x,-a_m k\mright) \,\mathrm{d}x;
	\]
	\item for all non-zero \(m,n \in \bZ\) with \(m+n \neq 0\), for all \(f \in P_m\), and for all \(g \in P_n\),
	\begin{multline*}
		\forall (x,k) \in \bR \times \bZ_{c_{m+n}},\\ 
		(f\cdot g)(x,k) \coloneqq \sum_{j\in\bZ} f\mleft(\tfrac{x}{\e^{n}} + \e^m\left(\tfrac{d_{m+n}k}{c_{m+n}}-\tfrac{j}{c_m}\right),a_{m}d_{m+n}k-j\mright) \cdot g\mleft(x-\left(\tfrac{d_{m+n}k}{c_{m+n}}-\tfrac{j}{c_n}\right),a_n j \mright).
	\end{multline*}
\end{enumerate}
\end{theorem}

\begin{proof}
	First, observe that the results of Schwarz~\cite{Schwarz98}*{\S 3}, Dieng--Schwarz~\cite{DiengSchwarz}, and Poli\-shchuk--Schwarz~\cite{PolishchukSchwarz}*{\S 1.3} specialise to yield the multiplication \(P \otimes P \to P\) that makes \(P\) into a unital \(\bC\)-algebra. Next, the results of Po\-li\-shchuk~\cite{Polishchuk}*{\S 2.2} specialise to show that \(\ast : P \to P\) makes \(P\) into a \(B\)-\(\ast\)-bimodule, while results of Vlasenko~\cite{Vlasenko}*{Thm.\ 6.1}, with our notation and conventions, specialise to show that
	\[
		\forall m,n \in \bZ, \, \forall p_1,q_1 \in P_m, \, \forall p_2,q_2 \in P_n, \quad (p_1 \cdot p_2) \cdot (q_1 \cdot q_2)^\ast = (p_1 \cdot (p_2 \cdot q_2^\ast)) \cdot q_1^\ast,
	\]
	so that, by associativity of the multiplication on \(P\), 
	\begin{equation}\label{vlasenkoeq}
		\forall m,n \in \bZ, \, \forall p_1,q_1 \in P_m, \, \forall p_2,q_2 \in P_n, \quad (p_1 \cdot p_2) \cdot (q_2^\ast \cdot q_1).
	\end{equation}
	Finally, by observations of Dieng--Schwarz~\cite{DiengSchwarz}*{\S 2} and Polishchuk--Schwarz~\cite{PolishchukSchwarz}*{Proof of Prop.\ 1.2}, for all \(m,n \in \bZ\), multiplication in \(P\) yields an isomorphism \(P_m \otimes_B P_n \iso P_{m+n}\). On the one hand, for all \(m,n \in \bZ\), \(q_1 \in P_m\), and \(q_2 \in P_n\), since \(P_{m+n} = P_m \cdot P_n\) and since the multiplication on \(P\) restricts to an isomorphism of \(B\)-bimodules \(P_{m+n} \otimes_B P_{-m-n} \iso B\), Equation~\ref{vlasenkoeq} implies that \((q_1 \cdot q_2)^\ast = q_2^\ast \cdot q_1\), so that  \(\ast : P \to P\) makes \(P\) into a unital \(\ast\)-algebra. On the other hand, since \(P_m \cdot P_n = P_{m+n}\) for all \(m,n \in \bZ\), 	it follows that the \(\ast\)-algebra \(P\) defines a principal \(\cO(\Unit(1))\)-comodule \(\ast\)-algebra~\cite{ADL}*{Thm.\ 4.4}.
\end{proof}

Now, if \(B\) is a \(\cO(\Unit(1))\)-module \(\ast\)-algebra, then the trivial principal \(\cO(\Unit(1))\)-comodule \(\ast\)-algebra \(B \rtimes \cO(\Unit(1))\) admits the left \(\cO(\Unit(1))\)-covariant \(\ast\)-homomorphism
\[
	(h \mapsto h \otimes 1_P) : \cO(\Unit(1)) \to B \rtimes \cO(\Unit(1)),
\]
which one views as a global trivialisation of \(B \rtimes \cO(\Unit(1))\). Indeed, more generally, a principal \(\cO(\Unit(1))\)-comodule \(\ast\)-algebra \(P\) is left \(\cO(\Unit(1))\)-covariant \(\ast\)-isomorphic to such a trivial principal \(\cO(\Unit(1))\)-comodule \(\ast\)-algebra if and only if it admits a left \(\cO(\Unit(1))\)-covariant \(\ast\)-homomorphism \(\cO(\Unit(1)) : \to P\). We now show that Theorem~\ref{quadprincipalthm} yields non-trivial principal \(\cO(\Unit(1))\)-comodule \(\ast\)-algebra in this technically precise sense.

\begin{proposition}
	Let \(\theta \in \bR \setminus \mathbf{Q}\) be a quadratic irrationality, and let \(P\) be the resulting principal \(\cO(\Unit(1))\)-comodule \(\ast\)-algebra of Theorem~\ref{quadprincipalthm}. There does not exist a left \(H\)-covariant unital \(\ast\)-homomorphism \(\cO(\Unit(1)) \to P\).
\end{proposition}

\begin{proof}
	Let \(\epsilon\) be the norm-positive fundamental unit of \(\mathbf{Q}[\theta]\) induced by the real quadratic irrationality \(\theta\), and let \(\Phi : \langle \e \rangle \times \set{\pm 1} \to \SL(2,\bZ)_\theta\) is the isomorphism of Proposition~\ref{stabthm}; note, in particular, that \(\Phi_{21}(\e)\e + \Phi_{22}(\e) = \e \in \bR \setminus \mathbf{Q}\), where \(\Phi_{21}(\e),\Phi_{22}(\e)\in\bZ\), so that, necessarily, \(\Phi_{21}(\e) \neq 0\). Thus, by a result of Connes~\cite{Connes80}*{Thm.\ 7}, it follows \(P_1 \coloneqq \cE(\Phi(\e),\theta)\) is not free as a right \(\cA_\theta\)-module. But now, if \(\psi : \cO(\Unit(1)) \to P\) were a left \(H\)-covariant \(\ast\)-homomor\-phism, then \(P_1\) would be freely generated as a right \(\cA_\theta\)-module by the unitary element \(\psi((z \mapsto z)) \in P_1 \cap \Unit(P)\).
\end{proof}

\begin{remark}
	In fact, this shows that \(P\) is not cleft, i.e., that there does not exist a unital convolution-invertible left \(\cO(\Unit(1))\)-covariant map \(\cO(\Unit(1)) \to P\). Indeed, \(P_1\) is the quantum vector bundle (in the sense of Brzezi\'{n}ski--Majid~\cite{BrM}*{Def.\ A.3}) associated to \(P\) induced by the standard irreducible corepresentation of \(\cO(\Unit(1))\), so that \(P_1\) would be free as a right \(\cA_\theta\)-module if \(P\) were cleft~\cite{BrM}*{Appx.\ A}.
\end{remark}

From now on, we will be concerned with the \(\cO(\Unit(1))\)-gauge theory of the canonical non-trivial principal left \(\cO(\Unit(1))\)-comodule \(\ast\)-algebra of Theorem~\ref{quadprincipalthm} induced by a real quadratic irrationality.

\subsection{Constant curvature connections as \texorpdfstring{\(q\)}{q}-monopoles}

From now on, let \(\theta \in \bR \setminus \mathbf{Q}\) be a real quadratic irrationality, let \(\e \in \mathbf{Q}[\theta]^\times \cap (1,\infty)\) be its norm-positive fundamental unit, let \(B \coloneqq \cA_\theta\) be the smooth noncommutative \(2\)-torus with deformation parameter \(\theta\), and let \(P\) be the non-trivial principal \(\cO(\Unit(1))\)-comodule \(\ast\)-algebra with \(\coinv{\cO(\Unit(1))}{P} = \cA_\theta\) of Theorem~\ref{quadprincipalthm}. We shall now show that Connes's constant curvature connections~\cite{Connes80} on the isotypical subspaces of \(P\) \emph{qua} basic Heisenberg modules on \(\cA_\theta\) combine to yield a \(\e^{-1}\)-monopole closely analogous to Brzezi\'{n}ski--Majid's \(q\)-monopole on the \(q\)-deformed complex Hopf fibration~\cite{BrM}*{\S 5.2}. In fact, we shall see that Polishchuk--Schwarz's holomorphic structures~\cite{PolishchukSchwarz} exhaust the \(\cO(\Unit(1))\)-gauge theory of \(P\) in a manner that demonstrates the necessity of considering \emph{all} principal \fodc{} on \(P\) compatible with given vertical and horizontal calculi. 

Before continuing, let us fix notation. On the one hand, define a left \(\cO(\Unit(1))\)-covariant unital \(\bC\)-algebra automorphism \(\sigma : P \to P\) by
\begin{equation}
	\forall m \in \bZ, \, \forall p \in P_m, \quad \sigma(p) \coloneqq \e^{-m}p.
\end{equation}
Although neither \(\sigma\) nor \(\sigma^2\) is a \(\ast\)-automorphism, they do respectively satisfy
\[
	(\sigma \circ \ast)^2 = \id_P, \quad (\sigma^2 \circ \ast)^2 = \id_P.
\]
On the other hand, given \(m \in \bZ\), define \(a_m,b_m,c_m,d_m \in \bZ\) by
\begin{equation}
	\begin{pmatrix}
		a_m & b_m \\ c_m & d_m	
	\end{pmatrix}
	\coloneqq \Phi(\e^m),
\end{equation}
where \(\Phi : \langle \e \rangle \times \set{\pm 1} \iso \SL(2,\bZ)_\theta\) is the isomorphism of Proposition~\ref{stabthm}. Thus, by \eqref{stabthmeq},
\[
	\forall m \in \bZ, \quad c_m \e + d_m = \e^m,
\]
so that \(\set{m \in \bZ \given c_m = 0} = \set{0}\). Moreover, by  an observation of Polishchuk--Schwarz~\cite{PolishchukSchwarz}*{Eq.\ 1.2}, the map \((m \mapsto c_m) : \bZ \to \bZ\) satisfies
\begin{equation}
	\forall m, n \in \bZ,\quad c_{m+n} = c_m\e^{-n} + \e^m c_n
\end{equation}

Let us now recall the construction of the standard \(\ast\)-differential calculus on the smooth noncommutative \(2\)-torus \(B \coloneqq \cA_\theta\), which consists of rapidly decaying Laurent series in unitary generators \(U_\theta\) and \(V_\theta\) satisfying \(V_\theta U_\theta = \e^{2\pi\iu{}\theta}U_\theta V_\theta\). Recall that \(B\) admits a commuting pair of \(\ast\)-derivations \(\delta_1,\delta_2 : B \to B\) uniquely determined by
\[
	\delta_1(U) \coloneqq 2\pi U, \quad \delta_1(V) \coloneqq 0, \quad \delta_2(U) \coloneqq 0, \quad \delta_2(V) \coloneqq 2\pi V.
\]
Hence, the canonical \(\ast\)-differential calculus \((\Omega_B,\dt{B})\) on \(B\) is given by
\[
	\forall k \in \bZ_{\geq 0}, \quad \Omega_B^k \coloneqq \begin{cases} B &\text{if \(k=0,2\)},\\ B^{\oplus 2} &\text{if \(k=1\),}\\ 0 &\text{else,} \end{cases}
\]
where \(B\) is viewed as the trivial \(B\)-\(\ast\)-bimodule,  the non-trivial product  \(\wedge : \Omega^1_B \otimes_B \Omega^1_B \to \Omega^2_B\) is given by
\[
	\forall (b_1,b_2), (c_1,c_2) \in \Omega^1_B, \quad (b_1,b_2) \wedge (c_1,c_2) \coloneqq b_2 c_1 - b_1 c_2,
\]
and the exterior derivative \(\dt{B}\) given by
\begin{gather*}
	\forall b \in \Omega^0_B, \quad \dt{B}(b) \coloneqq (\delta_1(b),\delta_2(b)),\\
	\forall (b_1,b_2) \in \Omega^1_B, \quad \dt{B}(b_1,b_2) \coloneqq \delta_2(b_1)-\delta_1(b_2).
\end{gather*}
For notational convenience, let
\[
	\du{\tau}^1 \coloneqq \frac{1}{2\pi\iu{}} U_\theta^\ast \dt{B}(U_\theta) = (-\iu,0), \quad \du{\tau}^2 \coloneqq \frac{1}{2\pi\iu{}} V_\theta^\ast \dt{B}(V_\theta) = (0,-\iu), \quad \vol_B \coloneqq \du{\tau}^1 \wedge \du{\tau}^2 = 1,
\]
so that \(\du{\tau}^1\) and \(\du{\tau^2}\) are central in the graded \(\ast\)-algebra \(\Omega_B\) and skew-adjoint, \(\vol_B\) is central in \(\Omega_B\) and self-adjoint, \(\set{\du{\tau}^1,\du{\tau}^2}\) is a basis for \(\Omega^1_B\) as both a left and right \(B\)-module, and \(\set{\vol_B}\) is a basis for \(\Omega^2_B\) as both a left and right \(B\)-module. In particular, can write
\begin{gather*}
	\forall b \in B, \quad \dt{B}(b) = 	\iu{}\partial_1(b)\du{\tau}^1 +\iu{}\partial_2(b)\du{\tau}^2,\\
	\forall (b_1,b_2) \in B^{\oplus 2}, \quad \dt{B}(b_1 \du{\tau}^1 + b_2 \du{\tau}^2) = -\iu{}(\partial_2(b_1)-\partial_1(b_2))\vol_B
\end{gather*}

Now, suppose that we have completed \((\Omega_B,\dt{B})\) to a second-order horizontal calculus \((\Omega_B,\dt{B};\Omega_{P,\hor})\) on the principal left \(H\)-comodule \(\ast\)-algebra \(P\). Since \(P\) is principal and since \(\Omega_{P,\hor}^1\) and \(\Omega_{P,\hor}^2\) are projectable horizontal lifts for \(\Omega^1_B\) and \(\Omega^2_B\), respectively, it follows by the generalised Hopf module lemma~\cite{BeM}*{Lemma 5.29} that \(\Omega^1_{P,\hor}\) and \(\Omega^2_{P,\hor}\) are free as left \(P\)-modules with respective bases \(\set{\du{\tau}^1,\du{\tau}^2}\) and \(\set{\vol_B}\). Thus, for every prolongable gauge potential \(\nabla\) on \(P\) with respect to \((\Omega_B,\dt{B};\Omega_{P,\hor})\), there necessarily exist unique left \(H\)-covariant maps \(\partial_1,\partial_2 : P \to P\) extending \(\delta_1,\delta_2 : B \to B\), respectively, such that
\begin{equation}\label{connectioneq1}
	\forall p \in P, \quad \nabla(p) = \iu{}\partial_1(p)\du{\tau}^1+\iu{}\partial_2(p)\du{\tau}^2;
\end{equation}
in this case, it follows that the canonical prolongation \(\pr{\nabla}\) of \(\nabla\) is given by
\begin{equation}\label{connectioneq2}
	\forall (p_1,p_2) \in P^{\oplus 2}, \quad \nabla(p_1\du{\tau}^1+p_2\du{\tau}^2) = -\iu{}(\partial_2(p_1)-\partial_1(p_2)) \vol_B,
\end{equation}
and hence that the field strength \(\bF[\nabla]\) of \(\nabla\) is given by
\begin{equation}\label{connectioneq3}
	\forall p \in P, \quad \bF[\nabla](p) = \iu{}[\partial_1,\partial_2](p)\vol_B.
\end{equation}
As it turns out, we have the following canonical candidate for \(\set{\partial_1,\partial_2}\) due originally to Connes~\cite{Connes80}*{Thm.\ 7}; our immediate goal, then, will be to reverse-engineer \(\Omega_{P,\hor}\) so that \eqref{connectioneq1} defines a prolongable gauge potential \(\nabla\) on \(P\) with respect to \((\Omega_B,\dt{B};\Omega_{P,\hor})\).

\begin{theorem}[Connes~\cite{Connes80}*{Thm.\ 7}, Polishchuk--Schwarz~\cite{PolishchukSchwarz}*{Propp.\ 2.1, 2.2}]\label{cpsthm}
	Define \(\cO(\Unit(1))\)-covariant \(\bC\)-linear extensions \(\partial_1,\partial_2 : P \to P\) of \(\delta_1\) and \(\delta_2\), respectively, by
	\begin{gather}
		\forall m \in \bZ \setminus \set{0}, \, \forall f \in P_m, \, \forall (x,k) \in \bR \times \bZ_{c_m}, \quad \partial_1f(x,k) \coloneqq -\iu{}\ddt{f}(x,k), \quad \\
		\forall m \in \bZ \setminus \set{0}, \, \forall f \in P_m, \, \forall (x,k) \in \bR \times \bZ_{c_m}, \quad \partial_2f(x,k) \coloneqq 2\pi \e^{-n}c_n x f(x,\alpha).
	\end{gather}
	where \(\ddt{}\) denotes differentiation on \(\cS(\bR) \otimes \bC[\bZ_{c_m}]\) with respect to the continuous variable.
	Then the maps \(\partial_1\) and \(\partial_2\) satisfy the following relations:
	\begin{gather}
		\forall j \in \set{1,2}, \forall p, \p{p} \in P, \quad \partial_j(p\p{p}) = \partial_j(p)\cdot \sigma(\p{p}) + p \cdot \partial_j(\p{p}),\label{twist1}\\
		\forall j \in \set{1,2}, \, \forall p \in P, \quad \partial_j(p^\ast) = -\sigma(\partial_j(p)^\ast).	\label{twist2}
	\end{gather}
	Moreover, the commutator \([\partial_1,\partial_2] \coloneqq \partial_1 \circ \partial_2 - \partial_2 \circ \partial_1\) satisfies
	\begin{equation}\label{twist3}
		\forall m \in \bZ, \, \forall f \in P_m, \quad [\partial_1,\partial_2](f) = -2\pi\iu{}\e^{-m}c_m f.
	\end{equation}
\end{theorem}

\begin{remark}[Polishchuk~\cite{PolishchukJGP}]
 Let \(\tau \in \set{z \in \bC \given \Im z < 0}\); let \(g \coloneqq \Phi(\e)\). Then \[B_g(\theta,\tau) \coloneqq\ker(\partial_1 + \tau \partial_2)\] is a unital (non-\(\ast\)-closed)  subalgebra of \(P\), which can be interpreted as the homogeneous coordinate ring of \(\cA_\theta\) with the complex structure \(\tau\) \emph{qua} noncommutative projective variety.
\end{remark}

We can now use \eqref{twist1} and \eqref{twist2} to reverse-engineer a canonical second-order horizontal calculus on \(P\) of the form \((\Omega_B,\dt{B};\Omega_{P,\hor})\) encoding the canonical twisted derivations \(\partial_1\) and \(\partial_2\) on \(P\) as a prolongable gauge potential. This will require the following definition.

\begin{definition}
	Let \(H\) be a Hopf \(\ast\)-algebra, let \(A\) be a left \(H\)-comodule \(\ast\)-algebra, and let \(E\) be a left \(H\)-covariant \(A\)-\(\ast\)-bimodule. Let \(\phi : A \to A\) be a left \(H\)-covariant unital \(\bC\)-algebra automorphism, such that \((\phi \circ \ast)^2 = \id_A\). We define the left \(H\)-covariant \(A\)-\(\ast\)-bimodule \(E_\phi\) to be the left \(H\)-covariant left \(A\)-module \(E\) together with right \(A\)-module structure \(\cdot_\phi : E \otimes_\bC A \to E\) and \(\ast\)-structure \(\ast_\phi : E \to E\) defined, respectively, by
	\begin{gather*}
		\forall e \in E, \, \forall a \in A, \quad e \cdot_\phi a \coloneqq e \cdot \phi(a),\\
		\forall e \in E, \quad e^{\ast_\phi} \coloneqq \phi(e^\ast).	
	\end{gather*}
\end{definition}

\begin{proposition}\label{constantcurvatureprop}
	Define a \(\mathbf{Z}_{\geq 0}\)-graded left \(\cO(\Unit(1))\)-comodule \(P\)-\(\ast\)-bimodule \(\Omega_{P,\hor}\) by
	\[
		\forall k \in \mathbf{Z}_{\geq 0}, \quad \Omega^k_{P,\hor} \coloneqq \begin{cases} P &\text{if \(k = 0\),} \\ (P_\sigma)^{\oplus 2} &\text{if \(k = 1\),} \\ P_{\sigma^2} & \text{if \(k=2\),}\\0 &\text{else;}\end{cases}
	\]
	let \(\iota : \Omega_B \inj \Omega_{P,\hor}\) be the map induced by the inclusion \(B \inj P\). Then \((\Omega_B,\dt{B};\Omega_{P,\hor},\iota)\) defines a second-order horizontal calculus on \(P\) when \(\Omega_{P,\hor}\) is endowed with the extension of the \(P\)-bimodule structure to a multiplication \(\Omega_{P,\hor} \otimes \Omega_{P,\hor} \to \Omega_{P,\hor}\) given by
	\begin{multline*}
		\forall (p_1,p_2), (\p{p}_1,\p{p}_2) \in \Omega^1_{P,\hor}, \\ (p_1 \cdot \du{\tau}^1+p_2 \cdot \du{\tau}^2) \wedge (\p{p}_1 \cdot \du{\tau}^1 + \p{p}_2 \cdot \du{\tau}^2) \coloneqq \left(p_1 \sigma(\p{p}_2) - p_2 \sigma(\p{p}_1)\right) \cdot \vol_B.
	\end{multline*}
	Moreover, the left \(\cO(\Unit(1))\)-covariant map \(\nabla_0 : P \to \Omega^1_{P,\hor}\) given by
	\begin{equation}
		\forall p \in P, \quad \nabla_0(p) = \iu{}\partial_1(p)\cdot\du{\tau}^1+\iu{}\partial_2(p)\cdot\du{\tau}^2,
	\end{equation}
	where \(\partial_1,\partial_2 : P \to P\) are the maps constructed in Theorem~\ref{cpsthm}, defines a prolongable gauge potential \(\nabla_0\) on \(P\) with respect to \((\Omega_B,\dt{B};\Omega_{P,\hor})\) with canonical prolongation \(\pr{\nabla}_0\) given by
	\begin{equation}
		\forall (p_1,p_2) \in P^{\oplus 2}, \quad \pr{\nabla}_0(p_1\cdot\du{\tau}^1+p_2\cdot\du{\tau}^2) = -\iu{}(\partial_2(p_1)-\partial_1(p_2)) \cdot\vol_B,
	\end{equation}
	and field strength \(\bF[\nabla_0]\) given by
	\begin{equation}
		\forall p \in P, \quad \bF[\nabla_0](p) = \left[2\pi\frac{\e c_1}{\e^2-1}\vol_B,p\right].
	\end{equation}
\end{proposition}

\begin{proof}
	Let us first show that \((\Omega_B,\dt{B};\Omega_{P,\hor},\iota)\) correctly defines a second-order horizontal calculus on \(P\); since \(\rest{\sigma}{B} = \id_B\), the only non-trivial point is that the \(\ast\)-structure on \(\Omega_{P,\hor}\) is graded-antimultiplicative with respect to the multiplication \(\Omega^1_{P,\hor} \times \Omega^1_{P,\hor} \to \Omega^2_{P,\hor}\). Indeed, for all \((p_1,p_2),(\p{p}_1,\p{p}_2) \in P^{\oplus 2}\),
	\begin{align*}
		&\left((p_1 \cdot \du{\tau}^1 + p_2 \cdot \du{\tau}^2) \wedge (\p{p}_1 \cdot \du{\tau}^1 + \p{p}_2 \cdot \du{\tau}^2) \right)^\ast	\\
		&\quad\quad = \left((p_2\sigma(\p{p}_1)-p_1\sigma(\p{p}_2))\cdot\vol_B\right)^\ast\\
		&\quad\quad = \vol_B \cdot (\sigma(\p{p}_1)^\ast p_2^\ast - \sigma(\p{p}_2)^\ast p_1^\ast)\\
		&\quad\quad = \left(\sigma^2(\sigma(\p{p}_1)^\ast) \sigma^2(p_2^\ast)-\sigma^2(\sigma(\p{p}_2)^\ast)\sigma^2(p_1^\ast)\right) \cdot \vol_B\\
		&\quad\quad = -\left(\sigma((\p{p}_1)^\ast) \cdot \du{\tau}^1 + \sigma((\p{p}_2)^\ast) \cdot \du{\tau}^2\right) \wedge \left(\sigma(p_1^\ast) \cdot\du{\tau}^1 + \sigma(p_2^\ast) \cdot \du{\tau}^2\right)\\
		&\quad\quad = -\left(\du{\tau}^1 \cdot (\p{p}_1)^\ast + \du{\tau}^2 \cdot (\p{p}_2)^\ast\right) \wedge \left( \du{\tau}^1 \cdot p_1^\ast + \du{\tau}^2 \cdot p_2^\ast\right)\\
		&\quad\quad = -(\p{p}_1 \cdot \du{\tau}^1 + \p{p}_2 \cdot \du{\tau}^2)^\ast \wedge (p_1 \cdot \du{\tau}^1 + p_2 \cdot \du{\tau}^2)^\ast.
	\end{align*}
	
	Let us now show that the map \(\nabla_0\) defines a prolongable gauge potential on \(P\) with respect to \((\Omega_B,\dt{B};\Omega_{P,\hor})\) with the correct canonical prolongation and field strength. Before continuing, note that \(\delta_1\), \(\delta_2\), and \(\sigma\) are all block-diagonal with respect to the decomposition \(P = \bigoplus_{m\in\bZ}P_m\) and that \(\sigma\) acts as multiplication by a scalar on each isotypical subspace \(P_m\), so that \(\sigma \circ \delta_1 = \delta_1 \circ \sigma\) and \(\sigma \circ \delta_2 = \delta_2 \circ \sigma\). 
	
	First, equations \eqref{twist1} and \eqref{twist2} together with the construction of \(\Omega_{P,\hor}\) imply that the left \(\cO(\Unit(1))\)-covariant map \(\nabla_0 : P \to \Omega^1_{P,\hor}\) is a \(\ast\)-derivation, while the fact that \(\rest{\partial_j}{B} = \delta_j\) for \(j=1,2\) implies that \(\rest{\nabla_0}{B} = \dt{B}\). Hence, \(\nabla_0\) defines a gauge potential with respect to the first-order horizontal calculus \((\Omega^1_B,\dt{B};\Omega^1_{P,\hor})\) induced by \((\Omega_B,\dt{B};\Omega_{P,\hor})\).
	
	Next, for all \(p,q \in P\) and \(b \in B\), 
	\begin{align*}
		&\nabla_0(p) \wedge \dt{B}(b) \cdot q - p \cdot \dt{B}(b) \wedge \nabla_0(q)	\\
		&\quad\quad = \iu{}\left(\partial_1(p)\cdot\du{\tau}^1 + \partial_2(p)\cdot\du{\tau}^2\right) \wedge \iu{}\left(\delta_1(b) \cdot \du{\tau}^1+\delta_2(b)\cdot\du{\tau}^2\right) \cdot q\\
		&\quad\quad\quad\quad - p \cdot \iu{}\left(\delta_1(b) \cdot \du{\tau}^1+\delta_2(b)\cdot\du{\tau}^2\right) \wedge \iu{}\left(\partial_1(q)\cdot\du{\tau}^1 + \partial_2(q)\cdot\du{\tau}^2\right)\\
		&\quad\quad = -\left(\partial_1(p) \delta_2(b) - \delta_2(p)\delta_1(b)\right)\sigma^2(q)\cdot \vol_B + p \left(\delta_1(b)\sigma(\partial_2(q))-\delta_2(b)\sigma(\partial_1(q))\right) \cdot \vol_B\\
		&\quad\quad = \left(\partial_2(p\delta_1(b)\sigma(q))-\partial_1(p\delta_2(b)\sigma(q))\right) \cdot \vol_B,
	\end{align*}
	where
	\[
		p \cdot \dt{B}(b) \cdot q = p \cdot \iu{}\left(\delta_1(b) \cdot \du{\tau}^1+\delta_2(b)\cdot\du{\tau}^2\right) \cdot q = \iu{}p \delta_1(b)\sigma(q) \cdot \du{\tau}^1 + \iu{} p\delta_2(b)\sigma(q) \cdot \du{\tau}^2.
	\]
	Hence, the gauge potential \(\nabla_0\) is prolongable with canonical prolongation \(\pr{\nabla}_0\) given by
	\[
		\forall (p_1,p_2) \in P^{\oplus 2}, \quad \pr{\nabla}_0(p_1 \cdot \du{\tau}^1+ p_2 \cdot \du{\tau}^2) = -\iu{}\left(\partial_2(p_1)-\partial_1(p_2)\right) \cdot \vol_B.
	\]
	
	Finally, by Equation~\ref{twist3}, we see that for all \(m \in \bZ\) and \(p \in P_m\),
	\begin{align*}
		\bF[\nabla_0](p) &= -\iu{}\nabla_0(\iu{}\partial_1(p)\cdot\du{\tau}^1 + \iu{}\partial_2(p) \cdot \du{\tau}^2)\\
		&= -\iu{}\left(\partial_2(\partial_1(p))-\partial_1(\partial_2(p))\right) \cdot \vol_B
		= 2\pi \e^{-m} c_m p \cdot \vol_B.
	\end{align*}
	But now, since
	\[
		\forall m,n \in \bZ, \quad \e^{-(m+n)}c_{m+n} = \e^{-m-n}(q^{n}c_m+q^{-m}c_n) = (\e^{-m}c_m) + (\e^{-2})^m (\e^{-n}c_n),
	\]
	it follows that for all \(m \in \bZ\),
	\[
		\e^{-m} c_m = \frac{1-\e^{-2m}}{1-\e^{-2}}\e^{-1}c_1 = -(1-\e^{-2m})\frac{\e c_1}{\e^2-1},
	\]
	so that, in turn, for all \(m \in \bZ\) and \(p \in P_m\),
	\[
		F[\nabla_0](p) = 2\pi\e^{-m}c_m p \cdot  \vol_B = 2\pi (1-\e^{-2m}) \frac{\e c_1}{\e^2-1} p \cdot \vol_B = \left[2\pi\frac{\e c_1}{\e^2-1}\vol_B,p\right],
	\]
	as was claimed.
\end{proof}

We can now readily compute the Atiyah space \(\fr{At}\) and gauge group \(\fr{G}\) of \(P\) with respect to \((\Omega^1_B,\dt{B};\Omega^1_{P,\hor})\) as well as their prolongable analogues with respect to \((\Omega_B,\dt{B};\Omega_{P,\hor})\). In fact, we shall see that every gauge potential and gauge transformation is automatically prolongable and that the group \(\fr{G} \cong \Unit(1)\) acts trivially on the affine space \(\fr{At} \cong \bR^2\).

\begin{proposition}\label{heisthm}
Let \((\Omega_B,\dt{B};\Omega_{P,\hor})\) be the second-order horizontal calculus on \(P\) and let \(\nabla_0\) be the prolongable gauge potential on \(P\) with respect to \((\Omega_B,\dt{B};\Omega_{P,\hor})\) of Proposition~\ref{constantcurvatureprop}.
\begin{enumerate}
	\item\label{heis1} Let \(\fr{at}\) be the space of relative gauge potentials on \(P\) with respect to the first-order horizontal calculus \((\Omega^1_B,\dt{B};\Omega^1_{P,\hor})\). The map \(\psi_{\fr{at}} : \bR^2 \to \fr{at}\) defined by
	\[
		\forall (s_1,s_2) \in \bR^2, \, \forall p \in P, \quad \psi_{\fr{at}}(s_1,s_2)(p) \coloneqq \left[\iu{}(s_1\du{\tau}^1+s_2\du{\tau}^2),p\right]
	\]
	is an isomorphism of \(\bR\)-vector spaces; hence, in particular, the subspace \(\Inn(\fr{at})\) of inner relative gauge potentials satisfies \(\Inn(\fr{at}) = \fr{at}\).
	\item\label{heis2} Let \(\pr{\fr{at}}\) be the space of prolongable relative gauge potentials on \(P\) with respect to the second-order horizontal calculus \((\Omega_B,\dt{B};\Omega_{P,\hor})\), and let \(\Inn(\pr{\fr{at}})\) be the subspace of  inner prolongable gauge potentials. Then \(\pr{\fr{at}} = \Inn(\pr{\fr{at}}) = \fr{at}\), where
	\[
		\forall (s_1,s_2) \in \bR^2, \quad \bF[\nabla_0 + \sigma(s_1,s_2)] = \bF[\nabla_0].
	\]
	\item\label{heis3} Let \(\fr{G}\) be the gauge group of \(P\) with respect to \((\Omega^1_B,\dt{B};\Omega^1_{P,\hor})\). Then the function \(\psi_{\fr{G}} : \Unit(1) \to \fr{G}\) defined by
	\[
		\forall \zeta \in \Unit(1), \, \forall m \in \bZ, \, \forall p \in P_m, \quad \psi_{\fr{G}}(\zeta)(p) \coloneqq \zeta^m p
	\]
	is a group isomorphism; hence, in particular, the subgroup \(\Inn(\fr{G})\) of inner gauge transformations and the subgroup \(\pr{\fr{G}}\) of prolongable gauge transformations with respect to \((\Omega_B,\dt{B};\Omega_{P,\hor})\) satisfy \(\Inn(\fr{G}) = \set{\id_P}\) and \(\pr{\fr{G}} = \fr{G}\). Moreover, the group \(\fr{G}\) acts trivially on the Atiyah space \(\fr{At}\) of \(P\) with respect to \((\Omega^1_B,\dt{B};\Omega^1_{P,\hor})\).
\end{enumerate}

\end{proposition}

\begin{proof}

	Let us begin by proving part~\ref{heis1}. Note that \(\psi_{\fr{at}}\) is a well-defined \(\bR\)-linear map with range contained in \(\Inn(\fr{at})\). Note, moreover, that \(\psi_{\fr{at}}\) is injective: indeed, if \((s_1,s_2) \in \ker\psi_{\fr{at}}\), then for all \(p \in P_1\),
	\[
		0 = \psi_{\fr{at}}(s_1,s_2)(p) = \iu{}(s_1\du{\tau}^1+s_2\du{\tau}^2) \cdot p - p \cdot \iu{}(s_1\du{\tau}^1+s_2\du{\tau}^2) = \iu{}(\e^{-1}-1)(s_1 p \cdot \du{\tau}^1 + s_2 p \cdot \du{\tau}^2),
	\]
	so that \((s_1,s_2) = (0,0)\) as was claimed. Thus, it suffices to show that \(\psi_{\fr{at}}\) is surjective.
	
	Let \(\bA \in \fr{at}\) be given. Given \(m \in \bZ\), since \(\bA\) is left \(\cO(\Unit(1))\)-covariant and right \(B\)-linear, since the left \(P\)-module \(\Omega^1_{P,\hor}\) is free with basis \(\set{\du{\tau}^1,\du{\tau}^2} \subset \Zent_B(\Omega^1_B)\), and since \(P_m\) is a self-Morita equivalence bimodule for \(B\), there exist unique \(\alpha_{m,1},\alpha_{m,2} \in B\), such that
	\[
		\forall p \in P_m, \quad \bA(p) = \alpha_{m,1} p \cdot \du{\tau}^1 + \alpha_{m,2}p \cdot \du{\tau}^2;
	\]
	note that \(\alpha_{0,1} = \alpha_{0,2} = 0\) since \(\rest{\bA}{P_0} = 0\). First, for all \(m \in \bZ\), we have \(\alpha_{m,1},\alpha_{m,2} \in \bC\); indeed, given \(b \in B\), for all \(p \in P_m\),
	\begin{align*}
	 0 &= \bA(b p) - b \cdot \bA(p)\\
	 &= \alpha_{m,1}bp \cdot \du{\tau}^1 + \alpha_{m,2}bp \cdot \du{\tau}^2 -b\alpha_{m,1}p \cdot \du{\tau}^1 - b\alpha_{m,2}p \cdot \du{\tau}^2\\
	 &= [\alpha_{m,1},b]p \cdot \du{\tau}^1+ [\alpha_{m,2},b]p \cdot \du{\tau}^2,
	\end{align*}
	so that \([\alpha_{m,1},b] = [\alpha_{m,2},b] = 0\) since \(P_{m}\) is a self-Morita equivalence bimodule for \(B\), and hence \(\alpha_{m,1},\alpha_{m,2} \in \bC\) since the algebra \(B \coloneqq \cA_\theta\) is central. Next, given \(m,n \in \bZ\), we see that for all \(p \in P_m\) and \(q \in P_n\),
	\begin{align*}
		0 &= \bA(pq) - \bA(p)\cdot q - p \cdot \bA(q)\\
		&=(\alpha_{m+n,1}pq \cdot \du{\tau}^1+\alpha_{m+n,2}pq \cdot \du{\tau}^2) - (\alpha_{m,1}p\cdot\du{\tau}^1 + \alpha_{m,2}p \cdot \du{\tau}^2) \cdot q\\
		&\quad\quad-p \cdot (\alpha_{n,1}q \cdot \du{\tau}^1 + \alpha_{n,2}q \cdot \du{\tau}^2)\\
		&= (\alpha_{m+n,1}-\e^{-n}\alpha_{m,1}-\alpha_{n,1})p \cdot \du{\tau}^1 + (\alpha_{m+n,2}-\e^{-n}\alpha_{m,2}-\alpha_{n,2})p \cdot \du{\tau}^2;
	\end{align*}
	again, since \(P_{m+n} = P_m \cdot P_n\) is a self-Morita equivalence bimodule, it follows that 
	\begin{equation}\label{cocycle}
		\forall j \in \set{1,2}, \quad \alpha_{m+n,j} = \e^{-n}\alpha_{m,j} + \alpha_{n,j}
	\end{equation}
	On the one hand, given \(m \in \bZ\), for all \(p \in P_m\),
	\begin{align*}
		0 &= \bA(p) + \bA(p^\ast)^\ast\\
		&= (\alpha_{m,1}p \cdot \du{\tau}^1+\alpha_{m,2}p\cdot\du{\tau}^2) + (\alpha_{-m,1}p^\ast \cdot \du{\tau}^1 + \alpha_{-m,2}p^\ast \cdot\du{\tau}^2)^\ast\\
		&= (\alpha_{m,1}p \cdot \du{\tau}^1+\alpha_{m,2}p\cdot\du{\tau}^2) + (-\e^{-m}\alpha_{m,1}p^\ast \cdot \du{\tau}^1 - \e^{-m}\alpha_{m,2}p^\ast \cdot\du{\tau}^2)^\ast\\
		&= (\alpha_{m,1}p \cdot \du{\tau}^1+\alpha_{m,2}p\cdot\du{\tau}^2) + (\alpha_{m,1} \du{\tau}^1 \cdot p^\ast + \alpha_{m,2}\du{\tau}^2 \cdot p^\ast)^\ast \\
		&= (\alpha_{m,1}+\overline{\alpha_{m,1}})p \cdot \du{\tau}^1 + (\alpha_{m,2}+\overline{\alpha_{m,2}})p \cdot \du{\tau}^2,
	\end{align*}
	so that again, since \(P_m\) is a self-Morita equivalence bimodule, \(\alpha_{m,1},\alpha_{m,2} \in \iu{}\bR\). On the other hand, by induction together with the \(1\)-cocycle identity \eqref{cocycle}, for all \(m \in \bZ\) and \(j \in \set{1,2}\),
	\[
		\alpha_{m,j} = \frac{1-\e^{-m}}{1-\e^{-1}}\alpha_{1,j} = (1-\e^{-m}) \frac{\alpha_{1,j}}{1-\e^{-1}}.
	\]
	Hence, for all \(m \in \bZ\) and \(p \in P_m\),
	\[
		\bA(p) = \sum_{j=1}^2 \alpha_{m,j}p \cdot \du{\tau}^j = \sum_{j=1}^2(1-\e^{-m}) \frac{\alpha_{1,j}}{1-\e^{-1}}p \cdot \du{\tau}^j = \left[\sum_{j=1}^2\frac{\alpha_{1,j}}{1-\e^{-1}}\du{\tau}^j,p\right],
	\]
	so that \(\bA = \sigma(s_1,s_2)\) for \(s = (-\iu{}(1-\e^{-1})^{-1}\alpha_{1,1},-\iu{}(1-\e^{-1})^{-1}\alpha_{1,2}) \in \bR^2\).
	
	Let us now turn to part~\ref{heis2}. First, since \(\du{\tau}^1,\du{\tau}^2 \in \Zent_B(\Omega^1_B)\), it follows that
	\[
		\fr{at} = \ran \psi_{\fr{at}} \subseteq \Inn(\pr{\fr{at}}),
	\]
	so that, indeed, \(\pr{\fr{at}} = \Inn(\pr{\fr{at}}) = \fr{at}\). Now, given \((s_1,s_2) \in \bR^2\), it  follows by Corollary~\ref{relcor} that for all \(p \in P\),
	\[
		\left(\bF[\nabla_0+\sigma(s_1,s_2)]-\bF[\nabla_0]\right)(p) = \left[-\iu{}\du_B\mleft(\iu{}(s_1\du{\tau}^2+s_2\du{\tau}^2)\mright),p\right] = 0.
	\]
	
	Finally, let us consider part~\ref{heis3}. Before continuing, note that \(\Inn(\fr{G}) = \set{\id_P}\) since \(B\) is central. By construction, the map \(\psi_{\fr{G}} : \Unit(1) \to \fr{G}\) is a well-defined injective group homomorphism. Furthermore, for every \(\zeta \in \Unit(1)\), the gauge transformation \(\psi_{\fr{G}}(\zeta)\) acts diagonally on \(P = \bigoplus_{m\in\bZ}P_m\) by scalar multiplication; hence, the range of \(\psi_{\fr{G}}\) is contained in \(\pr{\fr{G}}\) and acts trivially on \(\fr{At}\). The proof that \(\psi_{\fr{at}} : \bR^2 \to \fr{at}\) is an \(\bR\)-vector space isomorphism, \emph{mutatis mutandis}, now shows that \(\psi_{\fr{G}}\) is surjective.
\end{proof}

Now, recall that the usual de Rham calculus on \(\cO(\Unit(1))\) fits into a canonical \(1\)-parameter family of \(1\)-dimensional bicovariant \(\ast\)-differential calculi on \(\cO(\Unit(1))\) defined as follows. Let \(q \in \bR^\times\) be given, and recall that the corresponding \emph{\(q\)-numbers} are defined by
\[
	\forall n \in \bZ, \quad [n]_q \coloneqq \begin{cases} \frac{1-q^n}{1-q} &\text{if \(q \neq 1\),}\\ n &\text{if \(q	= 1\).}\end{cases}
\]
We define \((\Omega_q,\dt{q})\) to be the unique \(\ast\)-differential calculus on \(\cO(\Unit(1))\) with \(\ast\)-closed left \(\cO(\Unit(1))\)-comodule \(\Lambda^1_q \eqqcolon (\Omega^1_q)^{\cO(\Unit(1))}\) of right \(\cO(\Unit(1))\)-covariant \(1\)-forms and quantum Maurer--Cartan form \(\varpi_q : \cO(\Unit(1)) \to \Lambda^1_q\) defined as follows:
\begin{enumerate}
	\item the left crossed \(\cO(\Unit(1))\)-\(\ast\)-module \(\Lambda^1_q\) is defined to be \(\bC\) with the left \(\cO(\Unit(1))\)-module structure given by
	\[
		\forall m \in \bZ, \, \forall \mu \in \Lambda^1_q, \quad (z \mapsto z^m) \act \mu \coloneqq q^m\mu
	\]
	and the \(\ast\)-structure given by complex conjugation;
	\item the \(\Lambda^1_q\)-valued \(\Ad\)-invariant \(1\)-cocycle \(\varpi_q : \cO(\Unit(1)) \to \Lambda^1_q\) is given by
	\[
		\forall m \in \bZ, \quad \varpi_q((z\mapsto z^m)) \coloneqq 2\pi[m]_{q}.
	\]
\end{enumerate}
Since elements of \(\Lambda^1_q\) are also bicoinvariant and since \(\Lambda^1_q\) is spanned by the image under \(\varpi_q\) of the group-like unitary \((z \mapsto z) \in \cO(\Unit(1))\), it follows \(\Omega^k_q = 0\) for \(k \geq 2\). Note that \(\Omega^1_q\) is freely generated as both a left and right \(\cO(\Unit(1))\)-module by the skew-adjoint bicoinvariant element \(\du_q t \coloneqq (2\pi\iu{})^{-1}\varpi_q((z \mapsto z)) = -\iu{}\);
note also that setting \(q = 1\) recovers the de Rham calculus on \(\cO(\Unit(1))\). We now show that there exists a unique value of \(q\), such that  \(P\) admits \((\Omega^1_q,\dt{q})\)-adapted prolongable gauge potentials with respect to the second-order horizontal calculus \((\Omega_B,\dt{B};\Omega_{P,\hor})\).

\begin{theorem}
	Let \(q \in \bR^\times\). 
	\begin{enumerate}
		\item\label{heiscor1} The quadric subset \(\pr{\fr{At}}[\Omega^1_q]\) of all \((\Omega^1_q,\dt{q})\)-adapted prolongable gauge potentials on \(P\) with respect to \((\Omega_B,\dt{B};\Omega_{P,\hor})\) satisfies
		\[
			\pr{\fr{At}}[\Omega^1_q] = \begin{cases}\fr{At},&\text{if \(q = \e^{2}\),}\\ \emptyset, &\text{else.}\end{cases}
		\]
		In particular, for every \(\nabla \in \fr{At} = \pr{\fr{At}}[\Omega^1_{\e^2}]\), the curvature \(2\)-form \(F[\nabla]\) is non-zero, independent of \(\nabla\), and uniquely determined by
		\[
			F[\nabla](\du_{\e^2}t) = -\iu{}\e c_1 \vol_B.
		\]
		\item\label{heiscor2} The space \(\fr{at}[\Omega^1_q]\) of all \((\Omega^1_q,\du_q t)\)-adapted relative gauge potentials on \(P\) with respect to \((\Omega^1_B,\dt{B};\Omega^1_{P,\hor})\) and the subspace \(\pr{\fr{at}}[\Omega^{\leq 2}_q]\) of all \((\Omega_q,\dt{q})\)-adapted prolongable relative gauge potentials on \(P\) with respect to \((\Omega_B,\dt{B};\Omega_{P,\hor})\) satisfy
		\[
			\fr{at}[\Omega^1_q] = \pr{\fr{at}}[\Omega^{\leq 2}_q] = \begin{cases}\fr{at},&\text{if \(q = \e\),}\\0 &\text{else;}\end{cases}
		\]
		hence, in particular, it follows that
		\[
			\fr{At}/\pr{\fr{at}}[\Omega^1_{\e^{2}}] = \pr{\fr{At}}[\Omega^1_{\e^{2}}]/\pr{\fr{at}}[\Omega^{\leq 2}_{\e^{2}}] = \fr{At}.
		\]
		\item\label{heiscor3} The space \(\Inn(\pr{\fr{at}};\Omega^1_q)\) of all \((\Omega^1_q,\dt{q})\)-semi-adapted inner prolongable gauge potentials on \(P\) with respect to \((\Omega_B,\dt{B};\Omega_{P,\hor})\) satisfies \(\Inn(\pr{\fr{at}};\Omega^1_q) = \fr{at}\); hence, in particular,
		\[
			\Out(\pr{\fr{At}}[\Omega^1_{\e^{2}}]) \coloneqq \pr{\fr{At}}[\Omega^1_{\e^{2}}]/\Inn(\pr{\fr{at}};\Omega^1_{\e^{2}}) = \fr{At}/\fr{at} =\set{\nabla_0 + \fr{at}}.
		\]
	\end{enumerate}
\end{theorem}

\begin{proof}
	Let \((\Omega_{P,\ver},\dv{P})\) be the second-order vertical calculus on \(P\) induced by the unique bicovariant prolongation \((\Omega_q,\dt{q})\) of \((\Omega^1_q,\dt{q})\). Since \(\Omega^1_{P,\ver}\) and \(\Omega^2_{P,\hor}\) are free as left \(P\)-modules with respective bases \(\set{\du_q t} \subset \coinv{\cO(\Unit(1))}{\Omega^1_{P,\ver}}\) and \(\set{\vol_B} \subset \coinv{\cO(\Unit(1))}{\Omega^2_{P,\hor}} = \Omega^2_B\), it follows that a left \(\cO(\Unit(1))\)-covariant left \(P\)-linear map \(\phi : \Omega^1_{P,\ver} \to \Omega^2_{P,\hor}\) is completely determined by the unique element \(c \in \coinv{\cO(\Unit(1))}{P} = B\), such that \(\phi(\du_q t) = c\vol_B\), and \emph{vice versa}. Thus, given \(c \in B\) with corresponding map \(\phi\), for all \(b \in B\), \(m \in \bZ\), and \(p \in P_m\),
	\begin{gather*}
		\phi(\du_q t \cdot b) - \phi(\du_q t) \cdot b = \phi(b\cdot \du_q t) - c\vol_B \cdot b = (b c - c b) \vol_B,\\
		\phi(\du_q t \cdot p) - \phi(\du_q t) \cdot p = \phi(q^{-m} p \cdot \du_q t) - c\vol_B \cdot p = (q^{-m} p c - c \e^{-2m} p) \vol_B,
	\end{gather*}
	so that \(\phi\) is a \(P\)-bimodule map if and only if \(c = 0\) or \(c \in \bC \setminus \set{0}\) and \(q = \e^{2}\). We can now proceed with the proof of this corollary.
	
	Let us first check part \ref{heiscor1}. On the one hand, suppose that there exists \(\nabla \in \pr{\fr{At}}[\Omega^1_q]\); recall that \(F[\nabla]\) denotes its curvature \(2\)-form. Since \(\bF[\nabla] = \bF[\nabla_0] \neq 0\) by Propositions~\ref{constantcurvatureprop} and~\ref{heisthm}, it follows that \(F[\nabla] : \Omega^1_{P,\ver} \to \Omega^2_{P,\hor}\) is a non-zero left \(\cO(\Unit(1))\)-covariant morphism of \(P\)-bimodules, so that \(q = \e^{2}\) by the above discussion. On the other hand, let \(\nabla \in \fr{At}\) be given. By the proof of Proposition~\ref{constantcurvatureprop} together with Proposition~\ref{heisthm}, for all \(m \in \bZ\) and  \(p \in P_m\), 
	\begin{gather*}
		\bF[\nabla](p) = \bF[\nabla_0](p) = 2\pi [m]_{\e^{-2}} \e^{-1}c_1 p \cdot \vol_B, \\ \dv{P}(p) = 2\pi\iu{}[m]_q\,\du_q t \cdot p = 2\pi \iu{}[m]_q q^{-m}p\cdot\du_q t = 2\pi \iu{} q^{-1} [m]_{-q} p \cdot \du_q t
	\end{gather*}
	so that \(\nabla \in \pr{\fr{At}}[\Omega^1_{\e^2}]\) with \(F[\nabla] : \Omega^1_{P,\ver} \to \Omega^2_{P,\hor}\) uniquely defined by
	\[
		F[\nabla](\du_{\e^2} t) \coloneqq -\iu{} \e c_1\vol_B.
	\]
	
	Let us now turn to parts \ref{heiscor2} and \ref{heiscor3}. Note that \(\Omega^1_{P,\hor}\) is free as a left \(P\)-module with basis given by \(\set{\du\tau^1,\du\tau^2} \subset \coinv{\cO(\Unit(1))}{\Omega^1_{P,\hor}} = \Omega^1_B\). Hence, the above argument, \emph{mutatis mutandis}, shows that \(\fr{at}[\Omega^1_q] = \fr{at}\) when \(q = \e\) and \(\fr{at}[\Omega^1_q] = 0\) otherwise; indeed, for all \((s_1,s_2) \in \bR^2\), the relative connection \(1\)-form \(\omega[\sigma(s_1,s_2)]\) of \(\sigma(s_1,s_2) \in \fr{at}[\Omega^1_\e]\) is given by
	\[
		\omega[\sigma(s_1,s_2)](\du_\e t) \coloneqq - \frac{\e-1}{2\pi}(s_1\du\tau^1+s_2\du\tau^2) \in \Span_{\bR}\set{\du\tau^1,\du\tau^2}.
	\]
	Since \(\pr{\fr{at}} = \fr{at}\), since \(\Omega^2_q	 = 0\), and since \(\du\tau^1\wedge\du\tau^1 = \du\tau^1 \wedge \du\tau^2+\du\tau^2 \wedge\du\tau^1 = \du\tau^2\wedge\du\tau^2 = 0\), it now follows that \(\pr{\fr{at}}[\Omega^{\leq 2}_q] = \fr{at}[\Omega^1_q]\) in each case. Since \(\bF[\nabla] = \bF[\nabla_0]\) for all \(\nabla \in \pr{\fr{At}} = \fr{At}\), it now follows by Proposition~\ref{heisthm} that \(\Inn(\pr{\fr{at}};\Omega^1_q) = \fr{at}\) in general.
\end{proof}

It therefore follows that the norm-positive fundamental unit \(\e\) of the real quadratic field \(\mathbf{Q}[\theta]\) induced by \(\theta\) is the unique value of \(q \in \bR^\times\) for which \(P\) admits \((\Omega^1_{q^2},\dt{q^2})\)-adapted prolongable gauge potentials. In this case, every gauge potential is prolongable \((\Omega^1_{q^2},\dt{q^2})\)-adapted with the same non-zero constant curvature \(2\)-form
\[
	\du_{q^2} t \mapsto -\iu{}\e c_1 \vol_B
\] 
thereby defining a \(q\)-monopole analogous to the \(q\)-monopole of Brzezi\'{n}ski--Majid~\cite{BrM}*{\S 5.2} on the \(q\)-deformed complex Hopf fibration; mdistinct gauge potentials are gauge-inequivalent and yield non-isomorphic \((\cO(\Unit(1));\Omega^1_{q^2},\dt{q^2})\)-principal \fodc{} on \(P\).

\appendix\section{Groupoids}\label{Groupoids}

In this appendix, we fix notation and terminology related to groupoids, and we state and prove a key technical lemma used in the proofs of Theorems~\ref{firstorderclassification} and~\ref{qpbthm}.

Recall that a \emph{groupoid} is a small category \(\cG\) whose morphisms are all invertible; thus, in particular, a group is a groupoid with a single object. By abuse of notation, we conflate \(\cG\) with the set of all morphisms in \(\cG\), and we further identify the set \(\Ob(\cG)\) of all objects in \(\cG\) with a subset of \(\cG\) via the injection \(e \mapsto \id_e\). If \(f \in \cG\) is an \emph{arrow} (morphism) in \(\cG\), we denote its \emph{source} (domain) by \(s(f)\) and its \emph{target} (codomain) by \(t(f)\). Given \(f,g \in \cG\) satisfying \(t(f) = s(g)\), we denote the composition \(g \circ f\) by \(gf\) or \(g \cdot f\). Given \(e \in \Ob(\cG)\), the \emph{isotropy group} of \(e\) is the automorphism group \(\cG(e)\) of \(e\) in the category \(\cG\) and the \emph{star} of \(e\) is the set \(\operatorname{St}_{\cG}(e)\) of all arrows in \(\cG\) with source \(e\); given \(e_1,e_2 \in \Ob(\cG)\), we denote by \(\cG(e_1,e_2)\) the set of all arrows in \(\cG\) with source \(e_1\) and target \(e_2\).

A \emph{subgroupoid} of a groupoid \(\cG\) is a subcategory \(\cH\) of \(\cG\) whose morphisms are all invertible; we say that \(\cH\) is \emph{wide} whenever \(\Ob(\cG) = \Ob(\cH)\), and we say that \(\cH\) \emph{has trivial isotropy groups} whenever \(\cG(e) = \set{\id_e}\) for all \(e \in \Ob(\cG)\). Given a groupoid \(\cG\) and a wide subgroupoid \(\cH\) of \(\cG\) that has trivial isotropy groups, the \emph{quotient} of \(\cG\) by \(\cH\) is the groupoid \(\cG/\cH\) constructed as follows. Define an equivalence relation
\[
	\sim_{\cH} \coloneqq \set{(f,g) \in \cG^2 \given \exists (a,b) \in \cH^2, \, g=afb}
\]
on \(\cG\), and denote:
\begin{gather*}
	\forall e \in \Ob(\cG), \quad [e]_{\cH} \coloneqq [\id_e]_{\sim_{\cH}}; \quad
	\forall f \in \cG, \quad [f]_{\cH} \coloneqq [f]_{\sim_{\cH}}.	
\end{gather*}
Then \(\cG/\cH \coloneqq \cG/\sim_{\cH}\) as a set of arrows with set of objects \(\Ob(\cG/\cH) \coloneqq \set{[e]_\cH \given e \in \Ob(\cG)}\) and with the unique composition of arrows, such that:
\begin{gather*}
	\forall e \in \Ob(\cG), \quad \id_{[e]_{\cH}} \coloneqq [\id_e]_\cH = [e]_{\cH},\\
	\forall e_1,e_2,e_3 \in \Ob(\cG), \, \forall f \in \cG(e_1,e_2), \, \forall g \in \cG(e_2,e_3), \quad [g]_{\cH} \cdot [f]_{\cH} \coloneqq [g \cdot f]_{\cH}.
\end{gather*}

A \emph{homomorphism} from a groupoid \(\cG_1\) to a groupoid \(\cG_2\) is a functor \(F : \cG_1 \to \cG_2\); the \emph{kernel} of \(F\) is the subgroupoid
\[
	\ker F \coloneqq \set{f \in \cG_1 \given F(f) \in \Ob(\cG_2)}
\]
of \(\cG_2\). A homomorphism \(F : \cG_1 \to \cG_2\) is \emph{star-injective} if \(\rest{F}{\operatorname{St}_{\cG}(e)} : \operatorname{St}_{\cG_1}(e) \to \operatorname{St}_{\cG_2}(F(e))\) is injective for all \(e \in \Ob(\cG_1)\); it is a\emph{covering} if, in addition, \(\rest{F}{\Ob(\cG_1)} : \Ob(\cG_1) \to \Ob(\cG_2)\) is surjective; it is an \emph{isomorphism} if it is an isomorphism of categories; and it is an \emph{equivalence} if it is an equivalence of categories, in which case a \emph{homotopy inverse} of \(F\) is a weak inverse of \(F\). A \emph{homotopy} from a homomorphism \(F_1 : \cG_1 \to \cG_2\) to a homomorphism \(F_2 : \cG_1 \to \cG_2\) is a natural isomorphism \(\eta : F_1 \Rightarrow F_2\).

Finally, given a groupoid \(\cG\), a set \(X\), and a function \(p : X \to \Ob(\cG)\), an \emph{action} of \(\cG\) on \(X\) via \(p\) is a function
\[
	\set{(g,x) \in \cG \times X \given p(x) = s(g)} \to X, \quad (g,x) \mapsto g \act x
\]
satisfying the following conditions:
\begin{gather*}
	\forall e \in \Ob(\cG), \, \forall x \in \inv{p}(e), \quad \id_e \act x = x,\\
	\begin{multlined}\forall (e_1,e_2,e_3) \in \Ob(\cG), \, \forall f \in \cG(e_1,e_2), \, \forall g \in \cG(e_2,e_3), \, \forall x \in \inv{p}(e_1),\\ g \act (f \act x) = (g \cdot f) \act x.\end{multlined}
\end{gather*}
In this case, the \emph{action groupoid} of the action of \(\cG\) on \(X\) via \(p\) is the groupoid \(\cG \act X\) with set of objects \(\Ob(e) \times X\) and set of arrows \(\set{(g,x) \in \cG \times X \given p(x) = s(g)}\), with the following definitions:
\begin{gather*}
	\forall (g,x) \in \cG \ltimes X, \quad s(g,x) \coloneqq (s(g),x), \quad t(g,x) \coloneqq (t(g),g\act x);\\
	\forall (g_1,x) \in \cG \ltimes X, \, \forall g_2 \in \operatorname{St}_{\cG}(t(g_1)), \quad  (g_2,g_1 \act x) \cdot (g_1,x) \coloneqq (g_2 \cdot g_1,x);\\
	\forall (e,x) \in \Ob(\cG \ltimes X), \quad \id_{(e,x)} \coloneqq (\id_e,x).
\end{gather*}
When \(\cG\) is a group, this recovers the usual notions of group action and action groupoid.

We can now state and prove our technical lemma, which we shall need for the proofs of Theorems~\ref{firstorderclassification} and~\ref{qpbthm}.

\begin{lemma}\label{groupoidlemma}
	Let \(\pi : \cG \to \cH\) be surjective covering of groupoids, and suppose that we have a star-injective homomorphism \(\mu : \cH \to H\) from \(\cH\) to a group \(H\). Suppose that we are given a homomorphism \(S : G \ltimes X \to \cG\) and a left inverse \(T : \cG \to G \ltimes X\) of \(S\), such that:
	\begin{enumerate}
		\item there exists a natural isomorphism \(\eta : \id_{\cG} \Rightarrow S \circ T\) satisfying
		\[
			\forall e \in \Ob(\cG), \quad \mu \circ \pi(\eta_e) = \id_H;
		\]
		\item we have commutative diagrams
			\[
				\begin{tikzcd}
				G \ltimes X \arrow[d,twoheadrightarrow] \arrow[r,"S", hookrightarrow]  & \cG \arrow[d,"\mu\circ\pi"] \\ G \arrow[r,hookrightarrow] & H
				\end{tikzcd} \quad \quad \quad
				\begin{tikzcd}
				G \ltimes X \arrow[d,twoheadrightarrow]  \arrow[r, "T" below,twoheadleftarrow] & \cG \arrow[d,"\mu\circ\pi"] \\ G \arrow[r,hookrightarrow] & H
			\end{tikzcd}
			\]
		where \(G \inj H\) is the inclusion map and \(G \ltimes X \surj G\) is the surjective covering homomorphism given by \((g,x) \mapsto g\).
	\end{enumerate}
	Then \(\mu(\cH) = G\), so that \(\mu\) defines a surjective morphism \(\cH \surj G\); the subgroupoid \(\ker \mu\) of \(\cH\) is wide and has trivial isotropy groups; the equivalence kernel \(\sim\) of the set function
	\[
		\left(x \mapsto [\pi \circ \Sigma(1_G,x)]_{\ker\mu}\right) : X \to \Ob(\cH/\ker\mu)
	\]
	is a \(G\)-invariant; and there exists a unique isomorphism \(\tilde{S} : G \ltimes (X/\sim) \iso \cH/\ker\mu\) with
	\[
		\forall (g,x) \in G \ltimes X, \quad \tilde{S}(g,[x]_\sim) = [\pi \circ S(g,x)]_{\ker \mu}.
	\]
\end{lemma}

\begin{proof}
	Let us first show that \(\mu(\cH) = G\). Let \(p \coloneqq \left((g,x) \mapsto g\right) : G \ltimes X \to G\), so that
	\[
		p = \mu \circ \pi \circ S, \quad p \circ T = \mu \circ \pi.
	\]
	Since \(\pi : \cG \to \cH\) and \(p : G \ltimes X \to G\) are surjections, it follows that
	\[
		G = p(G \ltimes X) = (p \circ T)(\cG) = (\mu \circ \pi)(\cG) = \mu(\cH).
	\]
	
	Next, let us show that \(\sim\) is \(G\)-invariant. Let \(x,y \in X\), and suppose that \(x \sim y\), so that there exists \(\left(f : \pi \circ S(x) \to \pi \circ S(y)\right) \in \ker\mu\). Given \(g \in G\), the arrow
	\[
		(\pi \circ S)(g,y) \cdot f \cdot (\pi \circ S)(\inv{g},gx) : \pi \circ S(gx) \to \pi \circ S(gy)
	\]
	in \(\cH\) satisfies
	\[
		\mu\mleft((\pi \circ S)(g,y) \cdot f \cdot (\pi \circ S)(\inv{g},gx)\mright) = g \cdot \id \cdot \inv{g} = \id;
	\]
	hence, \(gx \sim gy\).
	
	Next, let us check that \(\ker\mu\) has trivial isotropy groups; it is wide since \(\id_e \in \ker\mu\) for all \(e \in \Ob(\cH)\). Let \(e \in \Ob(\cH)\) and let \(h \in (\ker\mu)(e)\). Since \(\pi\) is a covering of groupoids, let \(\tilde{h} \in \pi^{-1}(h)\). Then
	\[
		\id = \mu(h) = (\mu \circ \pi)(\tilde{h}) = (p \circ T)(\tilde{h}),
	\]
	so that \(T(\tilde{h}) = (\id,T(s(\tilde{h})))\), and hence
	\[
		\tilde{e} \coloneqq s(\tilde{h}) = ST(s(\tilde{h})) = ST(t(\tilde{h})) = t(\tilde{h});
	\]
	in other words, \(\tilde{h} \in \cG(\tilde{e})\) and \(ST(\tilde{h}) = \id_{\tilde{e}}\) for some \(\tilde{e} \in \pi^{-1}(e)\). It now follows that
	\[
		h = \pi(\tilde{h}) = \pi(\eta_{\tilde{e}}^{-1} \cdot \id_{\tilde{e}} \cdot \eta_{\tilde{e}}) = \pi(\eta_{\tilde{e}})^{-1} \cdot \id_e \cdot \pi(\eta_{\tilde{e}}) = \id_e.
	\]
	Since \(\ker\mu\) is wide and has trivial isotropy groups, it now follows that \(\cH/\ker\mu\) is well-defined.
	
	Next, let us show that \(\tilde{S}\) is well-defined. First, let \(g \in G\), let \(x,y \in X\), and suppose that \(x \sim y\), so that there exists \(\left(f : \pi \circ S(x) \to \pi \circ S(y)\right)\in\ker\mu\). Since
	\[
		\mu\mleft((\pi \circ S)(g,y) \cdot f \cdot (\pi \circ S)(\inv{g},gx)\mright) = g \cdot \id \cdot \inv{g} = \id,
	\]
	it follows that \((\pi \circ S)(g,y) \circ f \circ (\pi \circ S)(\inv{g},gx) \in \ker\mu\), so that
	\begin{align*}
		[(\pi \circ S)(g,y)]_{\ker\mu} &= \left[\left((\pi \circ S)(g,y) \cdot f \cdot (\pi \circ S)(\inv{g},gx)\right) \cdot (\pi\circ S)(g,x) \cdot \id_{\pi \circ S(x)}\right]_{\ker\mu}\\ &= [(\pi \circ S)(g,x)]_{\ker\mu},
	\end{align*}
	so that \(\tilde{S}(g,[x]_\sim)\) is well-defined. A straightforward calculation now shows that \(\tilde{S}\), which is well-defined as a map between sets of arrows, is a homomorphism.
	
	Now, since \(\sim\) is \(G\)-invariant, the quotient map \(X \to X/\sim\) is \(G\)-equivariant and hence extends to a covering of groupoids \(\sigma \coloneqq (g,x) \mapsto (g,[x]) : G \ltimes X \to G \ltimes (X/\sim)\). We wish to show that there exists a unique homomorphism \(\tilde{T} : \cH/\ker\mu \to G \ltimes (X/\sim)\) with
	\[
			\forall f \in \cG, \quad \tilde{T}([\pi(f)]_{\ker\mu}) = (\sigma \circ T)(f).
	\]
	Let \(f_1,f_2 \in \cG\), and suppose that \([\pi(f_1)]_{\ker\mu} = [\pi(f_2)]_{\ker\mu}\), so that there exist \(a,b \in \ker\mu\), such that \(\pi(f_1) \cdot a = b \cdot \pi(f_2)\); we wish to show that \(\sigma \circ T(f_1) = \sigma \circ T(f_2)\). On the one hand,
	\[
		p \circ T(f_1) = \mu \circ \pi(f_1) = \mu\mleft(\pi(f_1)\cdot a\mright) = \mu\mleft(b \cdot \pi(f_2)\mright) = \mu \circ \pi (f_2) = p \circ T(f_2);
	\]
	on the other hand,
	\[
		\left((\pi \circ \eta_{s(f_2)}) \cdot a \cdot (\pi \circ \eta_{s(f_1)})^{-1} : (\pi \circ S)\mleft(T(s(f_1))\mright) \to (\pi \circ S)\mleft(T(s(f_2))\mright) \right) \in \ker \mu,
	\]
	so that \(T(s(f_1)) \sim T(s(f_2))\). Hence,
	\[
		\sigma \circ T(f_1) = \left(p \circ T(f_1),[T(s(f_1))]\right) = \left(p \circ T(f_2),[T(s(f_2))]\right) = \sigma \circ T(f_2).
	\]
	Again, a straightforward calculation now shows that \(\tilde{T}\), which is well-defined as a map between sets of arrows, is a homomorphism.
	
	Finally, let us show that \(\tilde{S}\) and \(\tilde{T}\) are mutually inverse homomorphisms. On the one hand, for every \((g,x) \in G \ltimes X\),
	\[
		\tilde{T} \circ \tilde{S}(g,[x]) = \tilde{T}\mleft([S(g,x)]_{\ker\mu}\mright) = (\sigma \circ T)(S(g,x)) = \sigma(g,x) = (g,[x]);
	\]
	On the other hand, for every \(f \in \cG\),
	\begin{multline*}
		\tilde{S} \circ \tilde{T}\mleft([\pi(f)]_{\ker\mu}\mright) = \tilde{S}\mleft(\sigma \circ T(f)\mright) = [\pi \circ S \circ T(f)]_{\ker\mu}\\ = [\pi(\eta_{t(f)}) \cdot \pi(f) \cdot \pi(\eta_{s(f)})^{-1}]_{\ker\mu} = [\pi(f)]_{\ker\mu}.
	\end{multline*}
	Thus, the homomorphisms \(\tilde{S}\) and \(\tilde{T}\) are mutually inverse.
\end{proof}

\end{document}